\documentclass
 [10pt]{article}
\usepackage
 [margin=.8in,left=.8in]{geometry}

\NewCommandCopy{\oldparagraph}{\paragraph}
\RenewDocumentCommand{\paragraph}{s O{#3} m}{%
  \IfBooleanTF{#1}
    {\oldparagraph*{#3.}}       
    {\oldparagraph[#2]{#3.}}    
}

 \newcommand{\bracket}[3]{%
  \stretchleftright
    {#1}
    {%
      \ensurestackMath{\addstackgap[1pt]{#2}}%
      \vrule width 0pt depth 2pt height 0pt
    }
    {#3}%
} 
\newcommand{\scaledbracket}[3]{%
  \ThisStyle{%
    \stretchleftright
      {#1}
      {
        \ensurestackMath{\addstackgap[1pt]{\SavedStyle #2}}%
        \vrule width 0pt depth 1.5pt height 0pt
      }
      {#3}%
  }%
}

\usepackage{amsthm}
\usepackage{amssymb}
\usepackage{amsmath}
\usepackage{mathtools}

\usepackage{stmaryrd}    
\usepackage{rotating}
\usepackage{xfrac}

\usepackage{adjustbox}
\usepackage[table]{xcolor}
\usepackage{tcolorbox}

\newtcolorbox{standout}{
  colback=gray!15,
  boxrule=0pt,
  left=.3cm,
  right=.3cm,
  top=.18cm,
  bottom=.18cm,
  boxsep=0pt
}

\usepackage
  [inline]
  {enumitem}
\setlist[enumerate,1]{label=\textbf{\textup{(\roman*)}}}
\setlist[enumerate,2]{label=\textbf{\textup{(\alph*)}}}
\setlist{
  topsep=1pt,
  itemsep=0pt,
  leftmargin=1.2cm,
  before=\leavevmode
}

\usepackage[safe]{tipa} 

\usepackage{scalerel} 

\usepackage{stackengine}

\usepackage{accents}
\newlength{\dhatheight}

\usepackage{tikz}
\usetikzlibrary{
  backgrounds, 
  decorations.pathmorphing,
  decorations.pathreplacing,
  decorations.shapes,
  nfold
}
\usetikzlibrary{cd}

\tikzcdset{
  arrow style=tikz, 
  diagrams={
    cramped,
    baseline=
    {([yshift=-2pt]current bounding box.center)},
    >={
      Computer Modern Rightarrow[
        length=4pt, width=4pt
      ]
    },
    hook/.style={
      {Hooks[right, length=2pt, width=8pt]}->
    },
    hook'/.style={
      {Hooks[left, length=2pt, width=8pt]}->
    },
    shorten=0pt,
  }
}

\newcommand{\inlinetikzcd}[1]{\begin{tikzcd}[sep=small, ampersand replacement=\&]#1\end{tikzcd}}

\usepackage{mathrsfs} 
\DeclareMathAlphabet{\mathpzc}{OT1}{pzc}{m}{it} 

\usepackage{hyperref}
\hypersetup{
  colorlinks = true,
  linkcolor=darkblue, citecolor=darkblue,
  urlcolor=black
}

\definecolor{darkblue}{rgb}{0.05,0.25,0.65}
\definecolor{darkgreen}{RGB}{20,140,10}
\definecolor{lightgray}{rgb}{0.9,0.9,0.9}
\definecolor{darkorange}{RGB}{200,100,5}
\definecolor{darkyellow}{rgb}{.91,.91,0}
\definecolor{lightolive}{RGB}{225, 220, 185}

\usepackage{cleveref}
\crefname{equation}{}{}
\crefname{section}{\S}{\S\S}
\crefname{subsection}{\S}{\S\S}
\crefname{subsubsection}{\S}{\S\S}
\crefname{definition}{Def.}{Defs.}
\crefname{theorem}{Thm.}{Thms.}
\crefname{corollary}{Cor.}{Cors.}
\crefname{lemma}{Lem.}{Lems.}
\crefname{proposition}{Prop.}{Props.}
\crefname{remark}{Rem.}{Rems.}
\crefname{notation}{Ntn.}{Ntns.}
\crefname{fact}{Fact}{Fact}
\crefname{example}{Ex.}{Exs.}
\crefname{figure}{Fig.}{Figs.}
\crefname{table}{Tab.}{Tabs.}
\crefname{footnote}{ftn.}{ftns.}
\Crefname{footnote}{Ftn.}{Ftns.}

\newtheorem{theorem}{Theorem}[section]
\newtheorem{lemma}[theorem]{Lemma}

\newtheorem{proposition}[theorem]{Proposition}

\theoremstyle{definition}
\newtheorem{definition}[theorem]{Definition}

\newtheorem{example}[theorem]{Example}

\newtheorem{remark}[theorem]{Remark}

\newcommand{\acts}{\raisebox{1.4pt}{\;\rotatebox[origin=c]{90}{$\curvearrowright$}}\hspace{.5pt}}

\newcommand{\rchi}{\raisebox{1pt}{$\chi$}}

\newcommand{\proofstep}[1]{\scalebox{.8}{#1}}

\newcommand{\HilbertSpace}[1]{\mathcal{#1}}

\newcommand{\hotype}[1]{\mathcal{#1}}

\newcommand{\unit}[2]{\mathrm{ret}^{\scalebox{.7}{$#1$}}_{\scalebox{.7}{$#2$}}}
\newcommand{\counit}[2]{\mathrm{obt}^{\scalebox{.7}{$#1$}}_{\scalebox{.7}{$#2$}}}

\let\PLAINthebibliography\thebibliography
\renewcommand\thebibliography[1]{
  \PLAINthebibliography{#1}
  \setlength{\parskip}{0.5pt}
  \setlength{\itemsep}{0.5pt plus .3ex}
}

\newcommand{\yields}{\vdash}

\newcommand{\shape}{
  \raisebox{0.6pt}{\normalfont\fontencoding{T3}\fontfamily{cmr}\selectfont\textesh}
}

\newcommand{\defneq}{\equiv}

\newcommand\bosonic[1]{\mathstrut\mkern2.5mu#1\mkern-14mu\raise1.7ex%
  \hbox{$\scriptstyle\rightsquigarrow$}}

\newcommand{\oppositeinterval}[3]{
\draw[gray, line width=3.3]
  (#1+.2,#2) -- (#1+#3-.2,#2);
\draw[draw=gray,line width=1.3,fill=black]
  (#1,#2) circle (.2);
\draw[draw=gray,line width=1.3,fill opacity=.6,fill=white]
  (#1+#3,#2) circle (.2);
}

\newcommand{\openinterval}[3]{
\draw[gray, line width=3.3]
  (#1+.2,#2) -- (#1+#3-.2,#2);
\draw[draw=gray,line width=1.3,fill opacity=.6,fill=white]
  (#1+#3,#2) circle (.2);
\draw[draw=gray,line width=1.3,fill opacity=.6,fill=white]
  (#1,#2) circle (.2);
}

\newcommand{\closedinterval}[3]{
\draw[gray, line width=3.3]
  (#1+.2,#2) -- (#1+#3-.2,#2);
\draw[draw=gray,line width=1.3,fill=black]
  (#1+#3,#2) circle (.2);
\draw[draw=gray,line width=1.3,fill=black]
  (#1,#2) circle (.2);
}

\newcommand{\halfcircle}{
\begin{scope}
\clip
  (-4,0) rectangle (4.1,-4.1);
\draw[
  gray,
  line width=30pt,
  draw opacity=.4
]
  (0,0) circle (3);
\end{scope}

\draw[draw=gray,line width=1.3,fill=black]
  (0.1,-3.5) circle (.2);
\draw[draw=gray,line width=1.3,fill opacity=.6,fill=white]
  (-0.1,-3.5) circle (.2);

\draw[gray, line width=3.3]
  (-1.7+.2,-3.1) -- (1.7-.2,-3.1);
\draw[draw=gray,line width=1.3,fill=black]
  (1.7,-3.1) circle (.2);
\draw[draw=gray,line width=1.3,fill opacity=.6,fill=white]
  (-1.7,-3.1) circle (.2);

\closedinterval{.1}{-2.5}{2.35}
\openinterval{-2.5}{-2.5}{2.4}

\closedinterval{1.7}{-1.83}{1.34}
\openinterval{-3}{-1.83}{1.34}

\closedinterval{2.2}{-1.12}{1.12}
\openinterval{-3.31}{-1.12}{1.12}

\closedinterval{2.45}{-.38}{1.05}
\openinterval{-3.5}{-.38}{1.05}
}

\newcommand{\halfcircleAdjusted}{
\begin{scope}
\clip
  (-4,0) rectangle (4.1,-4.1);
\draw[
  gray,
  line width=30pt,
  draw opacity=.4
]
  (0,0) circle (3);
\end{scope}

\draw[draw=gray,line width=1.3,fill=black]
  (0.1,-3.5) circle (.2);
\draw[draw=gray,line width=1.3,fill opacity=.6,fill=white]
  (-0.1,-3.5) circle (.2);
  
\draw[gray, line width=3.3]
  (-1.8+.2,-3.0) -- 
  (+1.8-.2,-3.0);
\draw[draw=gray,line width=1.3,fill=black]
  (1.8,-3) circle (.2);
\draw[draw=gray,line width=1.3,fill opacity=.6,fill=white]
  (-1.8,-3) circle (.2);

\closedinterval{.1}{-2.5}{2.35}
\openinterval{-2.5}{-2.5}{2.4}

\closedinterval{1.7}{-1.83}{1.34}
\openinterval{-3}{-1.83}{1.34}

\closedinterval{2.2}{-1.12}{1.12}
\openinterval{-3.31}{-1.12}{1.12}

\closedinterval{2.45}{-.38}{1.05}
\openinterval{-3.5}{-.38}{1.05}
}

\newcommand{\halfcircleover}{
\begin{scope}
\clip
  (-4,0) rectangle (4.1,-4.1);
\draw[
  white,
  line width=43pt,
]
  (0,0) circle (3);
\end{scope}
  \halfcircle
}

\newcommand{\halfcircleoverAdjusted}{
\begin{scope}
\clip
  (-4,0) rectangle (4.1,-4.1);
\draw[
  white,
  line width=43pt,
]
  (0,0) circle (3);
\end{scope}
  \halfcircleAdjusted
}

\newcommand{\grayunderbrace}[2]{\mathcolor{gray}{\underbrace{\mathcolor{black}{#1}}}_{\mathcolor{gray}{#2}}}

\newcommand{\grayoverbrace}[2]{\mathcolor{gray}{\overbrace{\mathcolor{black}{#1}}}^{\mathcolor{gray}{#2}}}

\newcommand{\plus}{\hspace{.8pt}{\adjustbox{scale={.5}{.77}}{$\sqcup$} \{\infty\}}}
\newcommand{\cpt}{\hspace{.8pt}{\adjustbox{scale={.5}{.77}}{$\cup$} \{\infty\}}}

\newcommand{\ActedIndexedOperator}[4]{          
       {#1}
          _{\hspace{-1.5pt}
           \scalebox{.6}{${#4}\hspace{-2pt}
           \left[
           \def\arraystretch{.8}
           \def\arraycolsep{0pt}
           \begin{array}{c}
            {#2} \\ #3
            \\[-12pt]
            {}
           \end{array}
           \right]
           $}
        }
}

\newcommand{\ActedWOperator}[3]{
  \ActedIndexedOperator{\widehat{W}}{#1}{#2}{#3}
}
\newcommand{\WOperator}[2]{
  \ActedWOperator{#1}{#2}{}
}

\newcommand{\level}{k}
\newcommand{\lattice}{K}
\newcommand{\denominator}{q}
\newcommand{\p}{p}
\newcommand{\braidingAngle}{\theta}
\newcommand{\topologicalSpin}{s}

\newcommand{\shapeEquivalence}{\simeq_{\scalebox{.7}{$\shape$}}}

\newcommand{\makevspace}{\mathclap{\phantom{\vert_{\vert_{\vert_{\vert_{\vert}}}}}}}

\newcommand{\fullRotation}{\mathrm{rot}}

\newcommand{\nocontentsline}[3]{}
\newcommand{\stoptoc}{\let\addcontentsline\nocontentsline}

\newcommand{\modulo}[2]{{#1}\hspace{1pt}\scalebox{.7}{$\mathrm{mod}\hspace{1.8pt}{#2}$}}

\begin{document}

\setlength{\abovedisplayskip}{3.5pt}
\setlength{\belowdisplayskip}{3.5pt}
\setlength{\abovedisplayshortskip}{-5pt}
\setlength{\belowdisplayshortskip}{3pt}

\title{
  Fractional Quantum Hall Anyons via
  \\
  the Algebraic Topology of Surplus Flux Quanta
}

\author{
  \def\arraystretch{.6}
  \begin{tabular}{c}
  Hisham Sati\rlap{${}^{\,\hyperlink{DoS}{a}, \hyperlink{Courant}{b}}$}
  \\
  {\color{gray}\footnotesize \tt hsati@nyu.edu}
  \end{tabular}
  \;\;\;\;
  \def\arraystretch{.6}
  \begin{tabular}{c}
  Urs Schreiber\rlap{${}^{\,\hyperlink{DoS}{a}}$}
  \\
  {\color{gray}\footnotesize\tt us13@nyu.edu}
  \end{tabular}
}

\maketitle

\begin{abstract}
Fractional quantum Hall systems (FQH), with their experimentally observed anyonic topological order, are a candidate future hardware platform for intrinsically stable quantum registers (``topological qbits'') and operations (``topological quantum gates''). But theoretical questions remain, notably concerning the desired realization of nonabelian anyon statistics. 

\smallskip 
Here we analyze a novel non-Lagrangian effective description of FQH anyons based on global quantization of surplus flux in extraordinary non-abelian cohomology theories, specifically in unstable Cohomotopy. Topological quantum order manifests as local systems of gapped Hilbert spaces over the covariantized quantized flux moduli spaces.  

\smallskip
Assuming that FQH surplus flux is effectively quantized in 2-Cohomotopy theory (a cousin of {\it Hypothesis H} in high-energy physics), we rigorously compute covariant flux monodromy and classify its irreducible unitary representations. This reproduces fine details such as of modular topological order on the torus; but it also predicts potential new routes to nonabelian statistics in FQH systems which  may inform experimental searches for topological quantum hardware.
\end{abstract}

\vfill

\hrule
\vspace{5pt}

{
\hypertarget{DoS}{}
\footnotesize
\noindent
\def\arraystretch{1}
\tabcolsep=0pt
\begin{tabular}{ll}
${}^a$\,
&
Mathematics, Division of Science; and
\\
&
Center for Quantum and Topological Systems,
\\
&
NYUAD Research Institute,
\\
&
New York University Abu Dhabi, UAE.  
\end{tabular}
\hfill
\adjustbox{raise=-15pt}{
\href{https://ncatlab.org/nlab/show/Center+for+Quantum+and+Topological+Systems}{\includegraphics[width=3cm]
{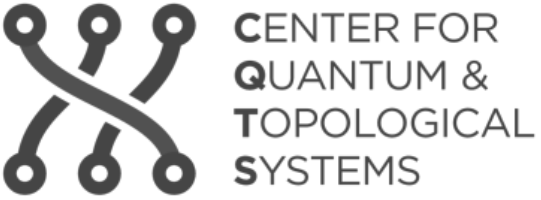}}
}

\vspace{1mm} 
\hypertarget{Courant}{}
\noindent 
\def\arraystretch{1}
\tabcolsep=0pt
\begin{tabular}{ll}
${}^b$\,The Courant Institute of Mathematical Sciences, NYU, NY.
\end{tabular}

\vspace{.2cm}

\noindent
The authors acknowledge the support of {\it Tamkeen} under the 
{\it NYU Abu Dhabi Research Institute grant} {\tt CG008}.
}

\newpage

\text{ }

\vspace{4cm}

\setcounter{tocdepth}{3}
\begin{center}
\begin{minipage}{9cm}
\tableofcontents
\end{minipage}
\end{center}

\section{Introduction \& Survey}
\label{IntroductionAndSurvey}

\noindent
{\bf Need for topological quantum protection.}
The potential promise of {\it quantum computers} \cite{NielsenChuang00}\cite{GrumblinHorowitz19} is enormous \cite{Gill24}\cite{Beterov24}\cite{Preskill24}, but their practicability hinges on finding and implementing methods to stabilize quantum registers and gates against decohering noise. 
Serious arguments 
\cite{Kak08}\cite{Dyakonov14}\cite{LauLimEtAl22}\cite{Ezratty23a}\cite{Ezratty23b}\cite{HoefletHaenerTroyer23}\cite{Gent23}\cite{Waintal24} and practical experience \cite{QBI} 
suggest that the currently dominant approach of {\it quantum error correction} at the software level (QEC \cite{LidarBrun13}\cite{Preskill23})  will need to be supplemented 
\cite{DasSarma22} 
\footnote{
  \cite{DasSarma22}:
  ``{The qubit systems we have today are a tremendous scientific achievement, but they take us no closer to having a quantum computer that can solve a problem that anybody cares about. [...] What is missing is the breakthrough [...] bypassing quantum error correction by using far-more-stable qubits, in an approach called topological quantum computing.}''
}
by more fundamental physical mechanisms of quantum error {\it protection} already at the hardware level,
 in the form of ``topological'' stabilization of quantum states (``topological qbits'') and operations (``topological quantum gates'') \cite{Kitaev03}\cite{FKLW03}\cite{NSSFD08}\cite{Sau17}\cite{MySS24-TQG}\cite{SatiValera25}.  While the general idea of topological quantum protection has become commonplace, its fine details have received less attention and are nowhere near as well understood as those of QEC --- this in contrast to its plausible necessity for scalable quantum computing.

\medskip 
\noindent
{\bf FQH systems: Surplus flux quanta.}
A practical contender for the required topological quantum hardware is
(cf. \cite{AverinGoldman01}\cite{DasSarmaFreedmanNayak05}\cite{Bravyi06}\cite{BarkeshliQi14}\cite{MongEtAl14})
{\it fractional quantum Hall systems} (FQH, cf. \cite{PrangeGirvin86}\cite{ChakrabortyPietilainen95}\cite{Stormer99}\cite{Jain07}\cite{Jain20} \cite{PapicBalram24}).
These are predominantly electron gases constrained to an effectively 2-dimensional surface $\Sigma^2$ 
(2DEG, e.g. realized on interfaces between semiconducting materials 
\cite{PapicBalram24})
at extremely low temperature and penetrated by transverse magnetic flux (cf. \cite[\S 2.2.1]{Griffiths}\cite[\S 2.1]{SS25-Flux}) so strong that the number of {\it flux quanta} (cf. \cite[(27)]{Kittel})
through the surface $\Sigma^2$ is an integer multiple $\lattice$ of the number of electrons confined to $\Sigma^2$, called the inverse of the {\it filling fraction} $\nu = \sfrac{1}{\lattice}$ (more generally a rational number).

\smallskip 
At unit-fraction filling $\nu=1/\lattice$, each electron in $\Sigma^2$ appears
--- in a heuristic picture, \cite{Jain89}\cite{Jain92}, cf. \cite[p. 882]{Stormer99} ---
to form a ``bound state'' of sorts with $\lattice$ flux quanta.
More generally, adiabatic insertion of further \emph{surplus} flux quanta transfers charge of magnitude $\nu e$, where $e$ is the magnitude of the electron charge:
\[
  \adjustbox{
    margin=3pt,
    bgcolor=lightolive
  }{$
  \mbox{FQH system at filling fraction $\nu$}
  \;\;
  :
  \;\;
  \begin{tikzcd}[
    decoration=snake,
    column sep=20pt
  ]
  \pm \;\mbox{surplus flux quanta}
  \;\;\;
  \ar[rr, decorate]
  &&
  \;\;\;
  \mp 
  \nu
  \;
  \mbox{units of electron charge}
  \,.
  \end{tikzcd}
  $}
\]
At Laughlin filling $\nu=1/\lattice$, this surplus flux insertion produces a fundamental
quasi-particle or quasi-hole. At other fillings it may produce a composite:
at $\nu=2/5$, a full flux insertion carries charge of magnitude $2e/5$,
whereas a fundamental quasi-particle carries $e/5$
(Rem. \ref{RelationBetweenBraidingPhasesAndFillingFraction}).
These surplus flux  quasi-particle sectors of FQH systems are candidates for the desired
topologically protected quantum states (exhibiting ``topological order''
\cite{Wen95}, cf. \cite[\S III]{ZengCHenZhouWen19}\cite{SS23-TopOrder}).

\vspace{-4mm}
\begin{equation}
\label{FQHSchematics}
\adjustbox{
  raise=-1.5cm
}{
\begin{tikzpicture}[
  scale=.75
]

\node
  at (-.8,.7)
  {
    \adjustbox{
      bgcolor=white,
      scale=.7
    }{
      \color{darkblue}
      \bf
      \def\arraystretch{.9}
      \begin{tabular}{c}
        unpaired surplus
        \\
        flux quantum:
        \\
        \color{purple}
        quasi-hole
      \end{tabular}
    }
  };

\draw[
  dashed,
  fill=lightgray
]
  (0,0)
  -- (8,0)
  -- (10+.3-.1,2+.3)
  -- (2.8+.3+.1,2+.3)
  -- cycle;

\begin{scope}[
  shift={(2.4,.5)}
]
\shadedraw[
  draw opacity=0,
  inner color=olive,
  outer color=lightolive
]
  (0,0) ellipse (.7 and .3);
\end{scope}

\begin{scope}[
  shift={(4.5,1.5)}
]

\begin{scope}[
 scale=1.8
]
\shadedraw[
  draw opacity=0,
  inner color=olive,
  outer color=lightolive
]
  (0,0) ellipse (.7 and .25);
\end{scope}

\begin{scope}[
 scale=1.45
]
\shadedraw[
  draw opacity=0,
  inner color=olive,
  outer color=lightolive
]
  (0,0) ellipse (.7 and .25);
\end{scope}

\shadedraw[
  draw opacity=0,
  inner color=olive,
  outer color=lightolive
]
  (0,0) ellipse (.7 and .25);

\begin{scope}[
  scale=.2
]
\draw[
  fill=black
]
  (0,0) ellipse (.7 and .25);
\end{scope}

\end{scope}

\begin{scope}[
  shift={(7.2,1)},
  scale=1.7
]
\shadedraw[
  draw opacity=0,
  inner color=white,
  outer color=lightgray
]
  (0,0) ellipse (.7 and .25);
\end{scope}

\draw[
  white,
  line width=2
]
  (-.4, 1.3)
  .. controls 
  (1,2) and 
  (2,2) ..
  (2.32,.7);
\draw[
  -Latex,
  black!70
]
  (-.4, 1.3)
  .. controls 
  (1,2) and 
  (2,2) ..
  (2.32,.7);

\node
  at (10,.5)
  {
    \adjustbox{
      scale=.7
    }{
      \color{darkblue}
      \bf
      \def\arraystretch{.9}
      \def\tabcolsep{-5pt}
      \begin{tabular}{c}
        deficit of a
        \\
        flux-quantum:
        \\
        \color{purple}
        quasi-particle
      \end{tabular}
    }
  };

\draw[
  white,
  line width=2
]
  (10,1.2) 
  .. controls 
  (9.5,1.8) and 
  (8.1,2.5) ..
  (7.5,1.18);
\draw[
  -Latex,
  black!70
]
  (10,1.2) 
  .. controls 
  (9.5,1.8) and 
  (8.1,2.5) ..
  (7.5,1.18);

\node
  at (1.3,2.7)
  {
    \adjustbox{
      scale=.7
    }{
      \color{darkblue}
      \bf
      \def\arraystretch{.9}
      \def\tabcolsep{-5pt}
      \begin{tabular}{c}
        $\lattice$ flux-quanta
        absorbed
        \\
        by each electron:
      \end{tabular}
    }
  };

\draw[
 line width=2.5pt,
  white
]
  (2.4, 3.1) .. controls 
  (2.8,3.3) and 
  (4,3.5) ..
  (4.3,1.8);

\draw[
  -Latex,
  black!70
]
  (2.4, 3.1) .. controls 
  (2.8,3.3) and 
  (4,3.5) ..
  (4.3,1.8);

\node at 
  (9,2.6)
  {
   \scalebox{.8}{
     \color{gray}
     (cf. \cite[Fig. 16]{Stormer99})  
   }
  };

\node[
  gray,
  rotate=-20,
  scale=.73
] 
  at (7.8,+.3) {$\Sigma^2$};

\end{tikzpicture}
}
\end{equation}

For identical Abelian quasi-particles of a specified species, a ``braiding
interchange'' of their worldlines contributes a statistical phase
$\zeta = e^{\pi \mathrm{i} \braidingAngle}$, {\it independent} of the local
details of the braiding process (\cite{Halperin84}\cite{ASW84}, as reviewed
in \cite[\S 9.8]{Jain07}):
\begin{equation}
  \label{TheBraidingPhase}
  \hspace{3cm}
  \adjustbox{raise=-1.4cm}{
  \begin{tikzpicture}[
    xscale=.7
  ]
    \draw[
      gray!30,
      fill=gray!30
    ]
      (-4.6,-1.5) --
      (+1.8,-1.5) --
      (+1.8+3-.5,-.4) --
      (-4.6+3+.5,-.4) -- cycle;

    \begin{scope}[
      shift={(-1,-1)},
      scale=1.2
    ]
    \shadedraw[
      draw opacity=0,
      inner color=olive,
      outer color=lightolive
    ]
      (0,0) ellipse (.7 and .1);
    \end{scope}

    \draw[
     line width=1.4
    ]
      (-1,-1) .. controls
      (-1,0) and
      (+1,0) ..
      (+1,+1);

  \begin{scope}
    \clip 
      (-1.5,-.2) rectangle (+1.5,1);
    \draw[
     line width=7,
     white
    ]
      (+1,-1) .. controls
      (+1,0) and
      (-1,0) ..
      (-1,+1);
  \end{scope}
  
    \begin{scope}[
      shift={(+1,-1)},
      scale=1.2
    ]
    \shadedraw[
      draw opacity=0,
      inner color=olive,
      outer color=lightolive
    ]
      (0,0) ellipse (.7 and .1);
    \end{scope}
    \draw[
     line width=1.4
    ]
      (+1,-1) .. controls
      (+1,0) and
      (-1,0) ..
      (-1,+1);

  \node[
    rotate=-25,
    scale=.7,
    gray
  ]
    at (-2.85,-1.25) {
      $\Sigma^2$
    };

  \draw[
    -Latex,
    gray
  ]
    (-3.4,-1.35) -- 
    node[
      near end, 
      sloped,
      scale=.7,
      yshift=7pt
      ] {time}
    (-3.4, 1.4);

  \node[
    scale=.7
  ] at 
    (0,-1.3)
   {\bf \color{darkblue} surplus flux quanta};

  \node[
    scale=.7
  ] at 
    (1.5,0)
   {\bf \color{darkgreen} braiding};

  \end{tikzpicture}
  }
  \;\;\;\;
  \def\arraystretch{1.4}
  \def\arraycolsep{2pt}
  \begin{array}{lll}
  \scalebox{.7}{
    \color{darkblue}
    \it braiding angle
  }
  &
  \braidingAngle := \tfrac{\p}{\lattice}
  &
  \in
  \mathbb{Q}
  \\
  \scalebox{.7}{
    \color{darkblue}
    \it braiding phase
  }
  &
  \zeta 
    := 
  e^{
    \pi \mathrm{i} \braidingAngle
  }
  &
  \in
  \mathrm{U}(1)
  \\
  \scalebox{.7}{
    \color{darkblue}
    \it topological spin
  }
  &
  \topologicalSpin
  \,=\,
  \tfrac{\braidingAngle}{2}
  &
  \in 
  \mathbb{Q}/\mathbb{Z}
  \mathrlap{\,.}
  \end{array}
\end{equation}
\vspace{.2cm}

\noindent
The phase depends on the FQH state, the quasi-particle species, and the
braid orientation. For fundamental Laughlin quasi-particles at
$\nu=1/\lattice$, it satisfies $\zeta=\pm e^{\pi\mathrm{i}\nu}$ in the
chosen orientation (\cite[\S 9.8.2]{Jain07}). Interferometry probes the
full-braid phase $\zeta^2=e^{2\pi\mathrm{i}\braidingAngle}$, without fixing
the sign of $\zeta$: the reported phases at $\nu=1/3$ and $\nu=2/5$ are
consistent with $\braidingAngle=1/3$ and $\braidingAngle=-2/5\equiv3/5$
modulo $1$, respectively (\cite{NakamuraEtAl20,NakamuraEtAl23};
Ex. \ref{ExperimentallyObservedBraidingPhases}).

\noindent
These fractional quasi-particles are called ``{\it anyons}'' \cite{ASW84}\cite{Stern08}\cite{Simon23}, in (somewhat inaccurate) reference \cite{Wilczek82} to {\it any} possible exchange phase between those of bosons ($\zeta = 1$, $\topologicalSpin = 0$) and those of fermions ($\zeta = -1$, $\topologicalSpin = \tfrac{1}{2}$), cf. \cite[\S 11.4]{Halvorson07}.
Experimental evidence for fractional statistics, from interferometry and
current-correlation measurements, has been reported by various groups and
in different materials: \cite{BartolomeiEtAl20}\cite{NakamuraEtAl19}\cite{NakamuraEtAl20}\cite{NakamuraEtAl23}\cite{Ruelle23}\cite{GlidicEtAl23}\cite{Kundu23}\cite{VeillonEtAl24}\cite{SamuelsonEtAl24}.

\medskip

The braiding phases \eqref{TheBraidingPhase} motivate examining FQH systems directly through their {\it flux quantization} \cite{SS25-Flux}:

\medskip

\noindent
{\bf Global flux quantization in FQH systems.}
Experiment shows abundantly that the fractional quantum Hall effect is a {\it universal} phenomenon
\cite{JeckelmannJeanneret04}\cite[p. 1]{Girvin04}\cite{GMP17}
in that its characteristic properties are independent of the microscopic nature of the host material and of impurities and irregularities of the sample. This suggests 
\cite{FroehlichKerler91}
the existence of accurate {\it effective} quantum field theoretic descriptions 
(cf. \cite{Fradkin13})
whose degrees of freedom reflect not any microscopic host particles but instead the nature of the universally emergent FQH quasi-particles (much like and closely related to how conformal field theory universally serves as an effective description of critical phenomena in statistical mechanics, cf. \cite[\S 3.2]{DiFrancescoMathieuSenechal97}). 

\smallskip 
Traditionally, this putative effective FQH theory is sought in the ancient and extensively studied realm of {\it Lagrangian} quantum field theories (cf. \cite{HenneauxTeitelboim92}\cite{GS25-Fields}), where one argues (\cite{Zee95}\cite{Wen95}, cf. \cite[\S7.3]{Wen07}\cite[\S 13.7]{Fradkin13}\cite[\S 2]{Witten16}\cite[\S 5]{Tong16}\cite[p. 5]{SS25-Srni}) that the relevant candidates are variants of abelian Chern-Simons theory \cite{BosNair89}\cite{Polychronakos90}\cite{Manoliu98}, see \S\ref{FractionalQuantumHallSystems}. 
A local gauge-field Lagrangian density $L=L(A)$ must be accompanied by global field and transition data to specify a quantum gauge theory. These data encode flux quantization and the topological sectors relevant to FQH systems (cf. \cite[\S 3.3]{SS25-Flux} and \hyperlink{FigureG}{\it Fig. G}). Chern-Simons theory has such global formulations; our construction instead begins with a proposed flux classifying space.

\vspace{.4cm}

\noindent
\hypertarget{FigureG}{}
\adjustbox{
  rndfbox=5pt
}{
  \includegraphics
  [width=16.5cm]
  {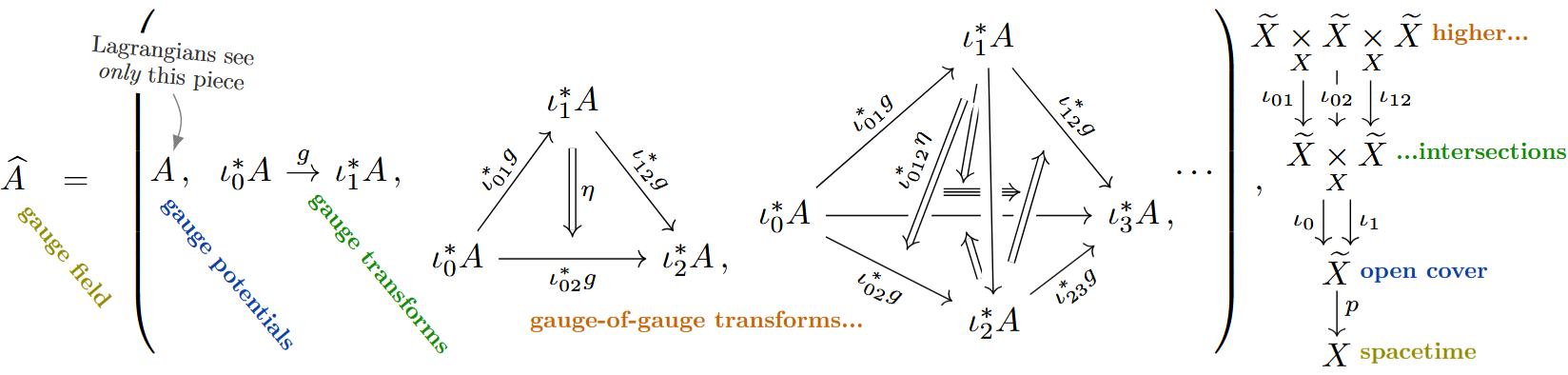}
}
\vspace{-.2cm}
\begin{center}
\begin{minipage}{16cm}
  \footnotesize
  {\bf Figure G.}
  The full non-perturbative data of a (higher) gauge field configuration $\widehat{A}$ on a spacetime $X$ consists not just of the gauge potentials $A$, which are only defined locally -- namely on an open cover $\widetilde{X}$ of $X$ by charts --, but also of ``transition data'' $g$ which gauge-transforms between coincident gauge potentials on different charts, and further of incrementally higher transition data which higher-gauge transforms between coincident transition data. 
  
  \;\;\;\;\;This (higher) transition data encodes the {\it flux quantization law} and the topological {\it charge} or {\it soliton sector} of the gauge field. It supplements the local expression $L=L(A)$ in a global formulation of the theory.
  
  \;\;\;\;\;For further exposition and pointers see \cite{Alvarez85} for the case of the ordinary electromagnetic field and \cite[\S 3.3]{SS25-Flux}\cite{SS26-Zagreb} for full generality. 
\end{minipage}
\end{center}

\smallskip

For example, locally eliminating the dynamical field from an FQH
Chern-Simons action yields a fractional electromagnetic response that is
not, by itself, a globally defined Chern-Simons theory of the probe field.
Retaining the dynamical gauge field and its topological sectors resolves
this distinction (\cite[\S 2.3]{Witten16}; see
\cref{ProblemWithFluxQuantizationForCSLagrangian}). Our alternative
construction seeks to derive state spaces and their monodromy directly
from the flux-quantization law. Its comparison with Chern-Simons
ground-state representations is discussed in
\S\ref{GroundStateMonodromyAndTopologicalOrder}.

\medskip

\noindent
{\bf Non-Lagrangian effective FQH theory based on flux quantization.}
We have developed a non-Lagrangian theory of topological quantum states in (higher) gauge theories
which is compatible with and in fact entirely based on consistent flux-quantization (survey in \cite{SS25-Flux}\cite{SS25-Srni}) --  the main insight here is (recalled below in \S\ref{FluxObservablesAndClassifyingSpaces}):
\begin{center}
\begin{minipage}{14.7cm}
\begin{itemize} 
\item[{\bf (a)}] flux-quantization laws are encoded by extra-ordinary \footnotemark\, 
non-abelian cohomology theories 
\cite[\S 2]{FSS23-Char}\cite[\S 1]{SS25-EqChar}
with {\it classifying spaces} $\hotype{A}$ whose ``rationalization'' reflects the duality-symmetric form of the gauge-field's Bianchi identities, cf. \cite[\S 3]{SS25-Flux} and \eqref{Rational2SphereAsCSFluxModel} below;

\item[{\bf (b)}] the {\it topological quantum observables on flux} depend only on the homotopy type of this classifying space $\hotype{A}$, and not on any other (local, microscopic) properties of the theory  \cite{SS24-QObs}.
\end{itemize} 
\end{minipage}
\end{center}
\footnotetext{
The term  {\it extra-ordinary cohomology theory} is standard (cf. \cite{Maunder63}\cite[\S 6]{FomenkoFuchs16})  for (Whitehead-)generalized abelian cohomology theories
(cf. \cite[Ex. 2.10]{FSS23-Char}) represented by spectra of spaces,  in contrast to the ordinary cohomology theories represented by (spectra of) Eilenberg-MacLane spaces \eqref{OrdinaryCohomologyViaClassifyingSpaces}. Here we use the term in the yet greater generality of non-ordinary {\it and} non-abelian cohomology \cite[\S 2]{FSS23-Char}\cite[\S 1]{SS25-EqChar}, as a more evocative version of the overused and now ambiguous term ``generalized cohomology''.
}

A quick way to understand the underlying principle is to recall the classical fact (cf. \cite[p. 263]{FomenkoFuchs16}\cite[Ex. 2.1]{FSS23-Char}) of algebraic topology (see \S\ref{SomeAlgebraicTopology})
that there exist classifying spaces --- here denoted $B^n \mathbb{Z}$ \footnote{
  These ordinary classifying spaces are known as {\it Eilenberg-MacLane spaces} and traditionally denoted ``$K(\mathbb{Z},n)$''.
}, characterized by $\pi_k B^n \mathbb{Z} \simeq \delta_{k}^n \mathbb{Z}$ ---  for (integral, reduced) ordinary cohomology:

\vspace{.4cm}
\begin{equation}
  \label{OrdinaryCohomologyViaClassifyingSpaces}
  \begin{tikzcd}
  \overset{
   \mathclap{
     \adjustbox{
       raise=5pt,
       scale=.7
     }{
       \color{darkblue}
       \bf
       Ordinary cohomology
     }
   }
  }{
    \widetilde H^n(X;\mathbb{Z})
  }
  \ar[
    rr,
    <->,
    "{ \sim }"
  ]
  \;\;
  &&
  \;\;
  \overset{
    \mathclap{
      \adjustbox{
        rotate=28,
        raise=3pt,
        scale=.7
      }{
        \color{darkblue}
        \bf
        \rlap{
          homotopy classes
        }
      }
    }
  }{
    \pi_0
  }
  \;
  \overset{
    \mathclap{
      \;\;\;\;\;
      \adjustbox{
        rotate=28,
        raise=3pt,
        scale=.7
      }{
        \color{darkblue}
        \bf
        \rlap{
          of maps to
        }
      }
    }
  }{
  \mathrm{Map}^\ast\big(
  }
    X
    ,\,
  \overset{
    \mathclap{
      \;\;
      \adjustbox{
        rotate=28,
        raise=3pt,
        scale=.7
      }{
        \color{darkblue}
        \bf
        \rlap{
          classifying space
        }
      }
    }
  }{
    B^n \mathbb{Z}
  }
  \big)
  \,,
  \end{tikzcd}
\end{equation}
and that the usual (Dirac) flux quantization of the electromagnetic field (cf. \cite{Alvarez85}\cite[Ex. 3.9]{SS25-Flux})
says that its underlying topological charge is a class in $\widetilde H^2(X;\, \mathbb{Z})$, hence represented by a map from spacetime to 
$B^2 \mathbb{Z} \,\shapeEquivalence\, B \mathrm{U}(1) \shapeEquivalence \mathbb{C}P^\infty$
\footnote{
  Our notation ``$\shapeEquivalence$'' stands for {\it weak homotopy equivalences} \eqref{WeakHomotopyEquivalenceNotation}. 
}
,
and that the latter is all that is needed to deduce
(cf. \cite[Ex. 2.2]{SS25-Flux})
that solitonic magnetic flux through a plane comes in integer units -- the {\it flux quanta} \eqref{FQHSchematics}:
\begin{equation}
  \label{Reduced2CohomologyOfPlane}
  \big\{
  \begin{tikzcd}
    \mathbb{R}^2_{\cpt} 
    \ar[
      r,
      "{ c }",
      "{
      \mathclap{
      \adjustbox{
        scale=.7,
        raise=-1pt,
      }{
        \color{darkblue}
        \bf
        \def\arraystretch{.9}
        \begin{tabular}{c}
          Hmtpy classes of 
          maps classifying
          \\
          solitonic magnetic flux
        \end{tabular}
      }
      }
      }"{yshift=8pt}
    ]
    &
    B^2 \mathbb{Z}
  \end{tikzcd}
  \big\}_{\big/\mathrm{hmtp}}
  \;\simeq\;
  \pi_0
  \,
  \mathrm{Map}^\ast\big(
    R^2_{\cpt}
    ,\,
    B^2 \mathbb{Z}
  \big)
  \;\simeq\;
  \pi_0
  \,
  \mathrm{Map}^\ast\big(
    S^2
    ,\,
    B^2 \mathbb{Z}
  \big)
  \;\simeq\;
  \pi_2(
    B^2 \mathbb{Z}
  )
  \;\simeq\;
  \overset{
    \mathclap{
      \adjustbox{
        raise=13pt,
        scale=.7
      }{
        \color{darkblue}
        \bf
        \def\arraystretch{.9}
        \begin{tabular}{c}
          number of 
          \\
          flux quanta
        \end{tabular}
      }
    }
  }{
    \mathbb{Z}
  }
  \,.
\end{equation}
In fact, the classifying space $\mathbb{B}^2 \mathbb{Z}$ moreover encodes the ordinary topological flux quantum observables through any surface $\Sigma^2$, as seen in Ex. \ref{TheCaseOfElectromagnetism} below.

\medskip

\noindent
{\bf The key role of algebraic topology.}
With this understanding, the question of an effective QFT description of FQH systems is not answered as traditionally (by choosing a Lagrangian whose equations of motion reflect local properties like the Hall current) but instead by finding an effective classifying space $\hotype{A}$ whose implied topological quantum observables reproduce the expected observations on global topological states, such as the emblematic (non-)commutation relation of Wilson line operators on the torus, shown in \eqref{TorusWilsonLineCommutatorForFQH} below.

\smallskip 
It turns out (in \S\ref{FluxObservablesAndClassifyingSpaces}) that this construction of topological flux quantum states  proceeds entirely by the analysis of ``{\it local systems}'' (cf. Rem. \ref{LocalSystemsAndQuantumStateSpaces} below)
on the (generally covariantized) homotopy type of moduli spaces \eqref{StateSpacesAsLocalSystemsOnModuliSpace}
of flux given by mapping spaces from the spacetime domain into the classifying space for the flux-quantization law --- and as such is squarely a problem in the mathematical subject of  homotopy theory and algebraic topology (for which we have compiled some background in \S\ref{SomeAlgebraicTopology}).

\smallskip

\noindent
{\bf Novel effective surplus flux quantization for FQH systems.}
Concretely, a candidate classifying space for the {\it surplus} magnetic flux through FQH systems (as seen by the effective quasi-particles/holes) turns out \cite{SS25-Seifert}\cite{SS24-AbAnyons}\cite{SS25-Srni} 
\footnote{\label{GeometricEngineering}
As explained in \cite{SS25-Seifert}\cite{GSS24-FluxOnM5}, following \cite{SS25-EqChar}, \, the 2-sphere here is a cousin of the 4-sphere which similarly serves as flux quantization of the higher gauge field in 11D supergravity \cite[\S 2.5]{Sati13}\cite{FSS20-H}\cite{GSS24-SuGra} (review in \cite{SS25-Flux}), where its choice as such is referred to as {\it Hypothesis H} \cite{FSS20-H}. While this is where our approach to FQH systems here comes from and is informed by \cite{SS24-AbAnyons}, for the present purpose the reader may ignore this geometric engineering of FQH systems on M5-probes of 11D SuGra. A review and discussion of the relevance to deeper questions of FQH systems (such as their hidden supersymmetry \cite{GromovMartinecRyu20}\cite{PBFGP23}) may be found in \cite[\S 2-3]{SS25-Srni}.
}
to be the 2-sphere 
$\hotype{A} \simeq S^2$ (see \S\ref{2CohomotopicalFluxThroughSurfaces}), modeling effective FQH flux in a variation of the ordinary classifying space 
\eqref{Reduced2CohomologyOfPlane}
(of which it is the ``2-skeleton''):
\smallskip 
\begin{equation}
  \label{2SphereInsideBU1}
  \adjustbox{
    scale=.7
  }{
    \color{darkblue}
    \bf
    \def\arraystretch{.9}
    \begin{tabular}{c}
      Classifying space for
      \\
      surplus FQH flux
    \end{tabular}
  }
  \begin{tikzcd}
  S^2 
  \,\simeq\,
  \mathbb{C}P^1
  \ar[r, hook]
  &
  \mathbb{C}P^{\infty}
  \,\shapeEquivalence\,
  B \mathrm{U}(1)
  \,\shapeEquivalence\,
  B^2 \mathbb{Z}
  \end{tikzcd}
  \adjustbox{
    scale=.7
  }{
    \color{darkblue}
    \bf
    \def\arraystretch{.9}
    \begin{tabular}{c}
      Classifying space for
      \\
      ordinary magnetic flux
    \end{tabular}
  }
\end{equation}

\smallskip 
\noindent In this article, we work out the monodromy groups and candidate state spaces determined by this classifying space together with the quantization prescription adopted in Def. \ref{TopologicalQuantumSectorsOfGeneralGaugeTheories}, and compare the resulting data with abelian Chern-Simons theory. Their physical realization requires the additional input discussed in \S\ref{PredictionsAndComparisons}. As motivation for this comparison, note (cf. \cite{FSS22-GS}) that the rationalization of the 2-sphere is encoded by the following differential equations (its ``Sullivan minimal model'', cf. \cite[\S 3.2]{SS25-Flux}), which are just those equations that characterize the Chern-Simons 3-form $H_3$ for a gauge field flux density $F_2$ as it appears in the Lagrangian formulation of Chern-Simons theory (cf. \cite[Prop. 1.27(b)]{Freed95} and \eqref{EffectiveCSLagrangian}):
\smallskip 
\begin{equation}
  \label{Rational2SphereAsCSFluxModel}
  \adjustbox{
    scale=.7,
    raise=1pt
  }{
    \color{darkblue}
    \bf
    \def\arraystretch{.9}
    \begin{tabular}{c}
      Rational model of
      \\
      classifying space for
      \\
      surplus FQH flux
    \end{tabular}
  }
  \mathrm{CE}(\mathfrak{l}S^2)
  \;\simeq\;
  \mathbb{R}_{{}_{\mathrm{d}}}
  \!
  \overset{
    \mathclap{
      \adjustbox{
        scale=.7,
        raise=-4pt
      }{
        \color{gray}
        \bf
        \def\arraystretch{.85}
        \begin{tabular}{c}
          flux
          \\
          density
        \end{tabular}
      }
    }
  }{
  \underset{
    \mathclap{
      \adjustbox{
        scale=.7,
        raise=-4pt
      }{
        \color{darkorange}
        \bf
        \def\arraystretch{.85}
        \begin{tabular}{c}
          effective
          \\
          higher flux
        \end{tabular}
      }
    }
  }{
  \left[
  \def\arraycolsep{1pt}
  \begin{array}{l}
    F_2
    \\
    H_3
  \end{array}
  \right]
  }
  }
  \Big/
  \left(\!\!
  \begin{array}{l}
  \mathrm{d}\, F_2
  \,=\,
  0
  \\
  \mathrm{d}\, 
    H_3 
  \,=\, F_2 F_2
  \end{array}
  \!\! \right)
  \adjustbox{
    scale=.7,
    raise=1pt
  }{
    \color{darkblue}
    \bf
    \def\arraystretch{.9}
    \begin{tabular}{c}
      Bianchi identities characterizing
      \\
      Chern-Simons 3-form /
      \\
      Green-Schwarz mechanism
    \end{tabular}
  }
\end{equation}

\smallskip 
Incidentally, it is in this sense that our effective description of the FQH effect is a mild form of {\it higher gauge theory} (cf. \cite{SS26-Zagreb}), since the Chern-Simons 3-form (traditionally understood as a Lagrangian density) here appears as a higher flux density --- the 3-form $H_3$ --- satisfying a Bianchi identity of the form known from {\it Green-Schwarz mechanisms}. \footnote{
  In the ``geometric engineering'' of our FQH model on M5-branes referred to in footnote \ref{GeometricEngineering}, this 3-form arises as the restriction to an orbi-singularity of the ``self-dual'' tensor field carried by these branes, which itself is quantized in a higher (and ``twistorial'') form of Cohomotopy, cf. \cite{GSS24-FluxOnM5}\cite{SS25-Seifert}.
}
Moreover,  $H_3$ is the rational image of the {\it Hopf fibration}, the generator of 
\begin{equation}
  \mathbb{Z}
  \;\simeq\;
  \pi_3(S^2)
  \;\simeq\;
  \pi_0 \mathrm{Map}^\ast( S^3, S^2 )
  \;\simeq\;
  \pi_1 \mathrm{Map}^\ast( S^2, S^2 )
  \,.
\end{equation}
This is the non-torsion class that disappears under passage to the ordinary classifying space \eqref{2SphereInsideBU1}, and it is this class which we find in 
\S\ref{2CohomotopicalFluxThroughSurfaces}, 
Prop. \ref{2CohomotopicalFluxMonodromyOnPlane},
to be identified with the observable $\widehat{\zeta}$ of fractional braiding phases \eqref{TheBraidingPhase}!

\medskip

\noindent
{\bf Results.} Within this proposed effective description, we derive the monodromy groups and representation-theoretic data over various surface geometries. We compare them with known FQH descriptions; assessing possible experimental distinctions requires the physical realization and stability conditions discussed in \S\ref{PredictionsAndComparisons}:

\vspace{-.1cm}
\begin{center}
\begin{tikzpicture}
\draw node at (0,0) {
\footnotesize
\def\arraystretch{2.8}
\def\tabcolsep{10pt}
\begin{tabular}{cc|cc}
 \rowcolor{lightgray}
  &
  {\bf Surface}
  &
  \multicolumn{2}{l}{
    \bf Results
    for 2-Cohomotopical flux quanta
  }
  \\[-5pt]
  \rowcolor{lightgray}
  \S\ref{2CohomotopicalFluxThroughSurfaces}
  & 
  $\Sigma^2$
  &
  {\bf 
    broadly
  }
  &
  {\bf fine-print}
  \\
  \hline
  \hline
  \S\ref{2CohomotopicalFluxThroughPlane}
  &
  \def\arraystretch{1.3}
  \def\tabcolsep{-5pt}
  \begin{tabular}{c}
    The plane
    \\
    $\Sigma^2_{0,0,1}\makevspace$
  \end{tabular}
  &
  fractional statistics
  &
  \begin{minipage}{5cm}
  framing regularization
  of link observables appears automatically
  \end{minipage}
  \\[1pt]
  \rowcolor{lightgray}
  \S\ref{2CohomotopicalFluxThroughTorus}
  &
  \def\arraystretch{1.1}
  \def\tabcolsep{-5pt}
  \begin{tabular}{c}
    The torus
    \\
    $\Sigma^2_{1,0,0}\makevspace$
  \end{tabular}
  &
  torus monodromy
  &
  \begin{minipage}{5cm}
  classification with irreducible Heisenberg restriction
  \end{minipage}
  \\[1pt]
  \adjustbox{
    raise=1pt
  }{
  \def\tabcolsep{-11pt}
  \def\arraystretch{1.1}
  \begin{tabular}{c}
    \S\ref{OnTheOpenAnnulus}
  \end{tabular}
  }
  &
  \def\arraystretch{1.1}
  \def\tabcolsep{-5pt}
  \begin{tabular}{c}
    The annulus
    \\
    $\Sigma^2_{0,0,2}\makevspace$
  \end{tabular}
  &
  dihedral representations
  &
  \begin{minipage}{5cm}
  irreducible unitary representations of dimensions $1$ and $2$
  \end{minipage}
  \\
  \hline
  \rowcolor{lightgray}
  \S\ref{QBitQuantumGatesOperableOn2PuncturedOpenDisk}
  &
  \def\arraystretch{1.1}
  \def\tabcolsep{-5pt}
  \begin{tabular}{c}
    2-Punctured open disk
    \\
    $\Sigma^2_{0,0,3}\makevspace$
  \end{tabular}
  &
  puncture monodromy
  &
  \begin{minipage}{5cm}
    irreducible unitary representations of dimensions $1$, $2$, $3$, and $6$
  \end{minipage}
  
\end{tabular}
};

\end{tikzpicture}
\end{center}

\noindent
{\bf Comparison to $\mathrm{U}(1)$-Chern-Simons theory.}
To a large extent, our construction turns out to give a curious and curiously direct (re-)derivation of the fine properties of $\mathrm{U}(1)$-Chern-Simons quantum field theory (cf. \cite{BosNair89}\cite{Polychronakos90}\cite{Manoliu98}\cite{GelcaUribe10}) by novel non-Lagrangian means, and as such the result seems of interest in its own right, beyond the topic of FQH systems (see Remarks \ref{ComparisonWithCSOnPlane} and \ref{ComparisonToModularDataOfAbelianCS} below).

For the one-component Laughlin state at filling $\nu=1/\denominator$,
the conventional Chern-Simons description has scalar parameter
$\lattice=\denominator=2\level$ in our normalization. Ordinary
$\mathrm{U}(1)$-Chern-Simons theory requires even $\lattice$, while
{\it spin Chern-Simons theory} also permits odd $\lattice$
\cite[\S 5]{DijkgraafWitten90}\cite[p. 381]{MooreRead91}, as summarized
below. The torus representations of \S\ref{2CohomotopicalFluxThroughTorus}
supply monodromy data for comparison with these theories; identifying
a physical Chern-Simons theory further requires matching the spin,
charge, and transport data. For multicomponent abelian states the comparison
uses a $K$-matrix, as recalled in \cref{KMatrixCSTheory}.
\begin{equation}
  \label{ChernSimonsLevel}
    \adjustbox{scale=0.85}{
    \def\arraystretch{1.6}
    \def\tabcolsep{9pt}
    \begin{tabular}{cccccc}
      \rowcolor{lightgray}
      & {\bf Symbol} 
      &
      {\bf at $\nu=\tfrac{1}{\denominator}$}
      & {\bf in ordinary CS} 
      & {\bf in ``spin'' CS}
      &
      $
        \exp\big(
          \tfrac{\mathrm{i}}{\hbar}
          S_{\mathrm{CS}}
        \big)
        =
      $
      \\
      {\bf CS level}
      &
      $\level$ 
      &
      {}
      & $\in \phantom{2}\mathbb{N}_{>0}$ & $\in \tfrac{1}{2}\mathbb{N}_{>0}$
      & 
      $e^{
        2 \pi \mathrm{i} 
        \, k \,
        \int\! A \, \mathrm{d}A
      }$
      \\
      \rowcolor{lightgray}
      &
      $\lattice \defneq 2\level$ 
      &
      $= \denominator$
      & 
      $\in 2\mathbb{N}_{> 0}$
      &
      $\in \phantom{\tfrac{1}{2}}\mathbb{N}_{> 0}$
      &
      $e^{
        \pi \mathrm{i} 
        \, \lattice \,
        \int\! A \, \mathrm{d}A
      }$
    \end{tabular}
    }
\end{equation}

\smallskip

\noindent
But our theory also differs from usual $\mathrm{U}(1)$-CS-theory in subtle respects:
\begin{itemize}
\item On the torus, we find (Prop. \ref{General2CohomotopicalStatesOverTorus}, \ref{QuantumStatesOverTheAASpinTorus}) fractional braiding phases $e^{\pi \mathrm{i} \tfrac{\p}{\lattice} }$ for general coprime $\p$ and $\lattice$. In one-component $\mathrm{U}(1)_{\lattice}$ CS-theory, charge $q$ gives the phase $e^{\pi\mathrm{i}q^2/\lattice}$, so nonunit numerators already occur there. Our family allows general coprime numerators, beyond the charge-square numerators available at a fixed one-component level. Comparison with multicomponent $K$-matrix theory (\cite{WenZee92}\cite{Wen95}, cf. \cref{KMatrixCSTheory}) is discussed in \S\ref{PredictionsAndComparisons}.

\item These $\lattice$-dimensional state spaces admit inequivalent linear extensions distinguished by scalar phases under modular transformations (Rem. \ref{FineTopologicalOrder}). Character twists of a fixed extension have the same projective action; their physical distinction depends on the scalar-phase data retained in a ground-state realization.

\item On $n$-punctured disks, the covariant flux monodromy group is the quotient of the framed braid group by the common framing shift (Prop. \ref{CovariantFluxMonodromyOnPuncturedDisks}). The corresponding puncture lattice is $Q_{n-1}=\mathbb{Z}^n/\mathbb{Z}(1,\cdots,1)$, with the dual integral standard action \eqref{DualIntegralStandardRepresentation}. The open annulus and the 3-punctured sphere give explicit representation-theoretic examples (\S\ref{OnTheOpenAnnulus} and \S\ref{QBitQuantumGatesOperableOn2PuncturedOpenDisk}).

\item The construction distinguishes solitonic flux observables from the transformations that move punctures. The action of these transformations on flux monodromy is computed in \S\ref{FluxThroughPuncturedSurface}; its realization by adiabatic motion of physical defects requires additional input.
\end{itemize}

\medskip

\noindent
{\bf Dimensionality.}
Our field theory is defined in $1+3$ dimensions, with spatial slices
$\Sigma^2\times\mathbb{R}$ (Def. \ref{Spacetime}). The solitonic condition
\eqref{SolitonicTopologicalObservablesMoreGenerally} requires the flux to
vanish at infinity, in particular towards both ends of the transverse
$\mathbb{R}$-factor.\footnote{
  In the geometric engineering on M5-branes, this transverse direction
  is the decompactified ``M-theory circle''; cf. footnote
  \ref{GeometricEngineering}.
}
The resulting topological observables are expressed in terms of
$\Sigma^2$ by \eqref{TopologicalGFluxObservables}. A physical realization
as an effectively $(1+2)$-dimensional system would additionally require
dynamics confining the relevant modes to a layer of finite transverse
extent. Vanishing at infinity alone fixes neither the location nor the
width of such a layer.

\medskip

\noindent
{\bf Conclusion.}
In summary, we hypothesize that FQH flux is effectively quantized in 2-Cohomotopy (\S\ref{2CohomotopicalFluxThroughSurfaces}), in particular identifying configurations of anyonic FQH quasi-holes/particles \eqref{FQHSchematics} with the configurations of signed points
(Ex. \ref{2CohomotopyOfSurfaces})
that the Pontrjagin/Segal theorems associate with charges in 2-Cohomotopy (\S\ref{2CohomotopicalFluxThroughPlane})
\footnote{
  As such, there are tantalizing relations of our results to the proposal \cite{MoravaRolfsen23}, see particularly the end of \S 2 there and footnote \ref{MoravaOnLoopSpaceOf2Sphere} below.
}.
Our main results are:

\begin{center}
\begin{minipage}{15cm}
\begin{itemize}[
  itemsep=3pt
]

\item[(i)] {\bf Consistency checks:} The construction recovers the solitonic braiding phase \eqref{TheBraidingPhase} and framed Wilson loop observables (\S\ref{2CohomotopicalFluxThroughPlane}), relates states on closed surfaces to Heisenberg representations (\S\ref{OnClosedSurfaces}--\ref{2CohomotopicalFluxThroughTorus}), and computes mapping class group actions on punctured-surface flux monodromy (\S\ref{FluxThroughPuncturedSurface}).

\item[(ii)] {\bf Further results:} 2-Cohomotopical flux quantization classifies unitary torus representations whose Heisenberg restriction is irreducible, whose comparison with K-matrix Chern-Simons theory away from unit filling fractions is discussed in \S\ref{PredictionsAndComparisons}. On punctured surfaces, the covariant flux monodromy groups admit non-abelian representations, with explicit classifications for the open annulus and the 3-punctured sphere (\S\ref{OnTheOpenAnnulus} and \S\ref{QBitQuantumGatesOperableOn2PuncturedOpenDisk}). Identifying these representations with protected states of physical defects requires additional input about their realization.
\end{itemize}
\end{minipage}
\end{center}

\medskip

The punctured-surface calculations therefore supply candidate symmetry
actions for defect systems. Establishing an implementation by protected
quantum states, and identifying the operations produced by adiabatically
moving the defects, requires additional physical input. Related proposals
are discussed in \cite[\S 3]{MySS24-TQG} and \cite{SS26-Islands}.

\begin{center}
\begin{minipage}{\linewidth}
\centering
\hypertarget{FigureD}{}
\begin{tikzpicture}[scale=0.8]

\draw[
  gray!30,
  fill=gray!15
]
  (0-.3,0) --
  (-.3+3,1) --
  (9.7+3, 1) --
  (9.7, 0) -- cycle;

\foreach \n in {1,...,13} {
  \draw[
    line width=2pt
  ]
    (\n*.7, -2) --
    (\n*.7, -1);
  \draw[
    line width=2pt
  ]
    (\n*.7, +1) --
    (\n*.7, +2);
}

\draw[
  line width=2pt
]
  (4*.7, -1) .. controls
  (4*.7, -.3) and
  (4*.7+.3,-.6) ..
  (4*.7+.3, 0);
\draw[
  line width=2pt
]
  (4*.7, +1) .. controls
  (4*.7, +.3) and
  (4*.7+.3,+.6) ..
  (4*.7+.3, 0);

\draw[
  line width=2pt
]
  (3*.7, -1) .. controls
  (3*.7, -.3) and
  (3*.7-.3,-.6) ..
  (3*.7-.3, 0);
\draw[
  line width=2pt
]
  (3*.7, +1) .. controls
  (3*.7, +.3) and
  (3*.7-.3,+.6) ..
  (3*.7-.3, 0);

\draw[
  line width=2pt
]
  (11*.7, -1) .. controls
  (11*.7, -.3) and
  (11*.7+.2,-.6) ..
  (11*.7+.2, 0);
\draw[
  line width=2pt
]
  (11*.7, +1) .. controls
  (11*.7, +.3) and
  (11*.7+.2,+.6) ..
  (11*.7+.2, 0);

\draw[
  line width=2pt
]
  (10*.7, -1) .. controls
  (10*.7, -.3) and
  (10*.7-.2,-.6) ..
  (10*.7-.2, 0);
\draw[
  line width=2pt
]
  (10*.7, +1) .. controls
  (10*.7, +.3) and
  (10*.7-.2,+.6) ..
  (10*.7-.2, 0);

\draw[
  line width=2pt
]
  (6*.7, -1) .. controls
  (6*.7, -.3) and
  (6*.7+.75,-.6) ..
  (6*.7+.75, 0);

\draw[
  line width=2pt
]
  (9*.7, -1) .. controls
  (9*.7, -.3) and
  (9*.7-.75,-.6) ..
  (9*.7-.75, 0);

\draw[
  line width=2pt
]
  (7*.7, -1) .. controls
  (7*.7, -.3) and
  (7*.7+.25,-.6) ..
  (7*.7+.25, 0);
\draw[
  line width=2pt
]
  (8*.7, -1) .. controls
  (8*.7, -.3) and
  (8*.7-.25,-.6) ..
  (8*.7-.25, 0);

\draw[
  gray,
  line width=2
]
  (-.3,0) --
  (2.05,0);
\draw[
  gray,
  line width=2
]
  (2.85,0) --
  (6.95,0);
\draw[
  gray,
  line width=2
]
  (7.75,0) --
  (9.7,0);
  
\shadedraw[
  draw opacity=0,
  inner color=olive,
  outer color=lightolive
]
  (7.5*.7,0) ellipse 
  (.58 and .12);

\draw[
  line width=2pt
]
  (7*.7, +1) .. controls
  (7*.7, +.3) and
  (7*.7+.25,+.6) ..
  (7*.7+.25, 0);

\draw[
  line width=2pt
]
  (8*.7, +1) .. controls
  (8*.7, +.3) and
  (8*.7-.25,+.6) ..
  (8*.7-.25, 0);

\draw[
  line width=2pt
]
  (6*.7, +1) .. controls
  (6*.7, +.3) and
  (6*.7+.75,+.6) ..
  (6*.7+.75, 0);

\draw[
  line width=2pt
]
  (9*.7, +1) .. controls
  (9*.7, +.3) and
  (9*.7-.75,+.6) ..
  (9*.7-.75, 0);

\draw[
  line width=2pt
]
  (1*.7, -1) --
  (1*.7, +1); 

\draw[
  line width=2pt
]
  (2*.7, -1) .. controls
  (2*.7, -.3) and
  (2*.7-.1,-.6) ..
  (2*.7-.1, 0);
\draw[
  line width=2pt
]
  (2*.7, +1) .. controls
  (2*.7, +.3) and
  (2*.7-.1,+.6) ..
  (2*.7-.1, 0);

\draw[
  line width=2pt
]
  (5*.7, -1) .. controls
  (5*.7, -.3) and
  (5*.7+.1,-.6) ..
  (5*.7+.1, 0);
\draw[
  line width=2pt
]
  (5*.7, +1) .. controls
  (5*.7, +.3) and
  (5*.7+.1,+.6) ..
  (5*.7+.1, 0);

\draw[
  line width=2pt
]
  (12*.7, -1) --
  (12*.7, +1); 

\draw[
  line width=2pt
]
  (13*.7, -1) --
  (13*.7, +1); 

\draw[
  gray!30,
  draw opacity=.5,
  fill=gray!25,
  fill opacity=.5
]
  (0-.3,0) --
  (-.3-3,-1) --
  (9.7-3, -1) --
  (9.7,0) -- cycle;

\draw[fill=white] 
  (3.5*.7,0) ellipse 
  (.57 and .07);
\draw[fill=white] 
  (10.5*.7,0) ellipse 
  (.4 and .07);

\begin{scope}

\clip
  (4,0) rectangle
  (+6,-1);

\shadedraw[
  draw opacity=0,
  inner color=olive,
  outer color=lightolive
]
  (7.5*.7,0) ellipse 
  (.58 and .12);

\end{scope}

\node at (2.4,-2.5)
 {
   \scalebox{.65}{
     \color{darkblue}
     \bf
     \def\arraystretch{.9}
     \begin{tabular}{c}
       flux-expelling
       \\
       defect (puncture):
       \\
       {\bf candidate} anyon
     \end{tabular}
   }
 };

\node at (5.3,-2.5)
 {
   \scalebox{.65}{
     \color{darkblue}
     \bf
     \def\arraystretch{.9}
     \begin{tabular}{c}
       flux quantum
       \\
       soliton (vortex):
       \\
     abelian anyon
     \end{tabular}
     }
 };

\end{tikzpicture}

\begin{minipage}{14cm}
  \hypertarget{FigureD}{}
  \footnotesize
  {\bf Figure D --- Solitonic and defect anyons.}
    Solitonic magnetic flux quanta give the abelian anyons discussed in Rem. \ref{ComparisonWithCSOnPlane}. Punctures model excluded loci of vanishing solitonic flux (cf. \hyperlink{FigureI}{Fig. I}). Their covariant flux monodromy has non-abelian representations (\S\ref{OnTheOpenAnnulus}--\ref{QBitQuantumGatesOperableOn2PuncturedOpenDisk}); realization by physical defect anyons requires additional input.
\end{minipage}
\end{minipage}

\end{center}

\vspace{.5cm}

\noindent
{\bf Acknowledgements.}
We thank 
Sadok Kallel,
Moishe Kohan,
Martin Palmer,
and 
Will Sawin for useful discussions concerning aspects of the algebro-topological analysis in \S\ref{2CohomotopicalFluxThroughSurfaces}.
We thank the late Jack Morava for further inspiring discussions. 


\section{Exotic Topological Flux Quanta}
\label
{FluxObservablesAndClassifyingSpaces}

Here we develop our main Definition \ref{TopologicalQuantumSectorsOfGeneralGaugeTheories} of quantum states of (ordinary and) extra-ordinary/exotic topological flux quanta.

\smallskip 
We begin by adopting a discrete group-algebra construction for ordinary Yang-Mills fluxes (Def. \ref{TheYMFluxQuantumObservables}), motivated by the flux-observable analysis of \cite{SS24-QObs}. Proposition \ref{TheTopologicalGFluxObservables} identifies this algebra in terms of the electric/magnetic classifying space $\hotype{A} \defneq B\big(G\ltimes(\mathfrak{g}/\Lambda)\big)$.

\smallskip 
In this algebro-topological formulation (cf. \S\ref{SomeAlgebraicTopology}), the result has an evident generalization to higher gauge theories with extra-ordinary flux quantization laws \cite{SS25-Flux}
classified by any other pointed space $\hotype{A}$, such as the 2-sphere \eqref{2SphereInsideBU1}
with its higher Chern-Simons flux form \eqref{Rational2SphereAsCSFluxModel}
in our motivating example of FQH systems, to be treated this way in \S\ref{2CohomotopicalFluxThroughSurfaces}.

\smallskip 
Following \cite[\S 4]{SS24-QObs}, we adopt this generalization as a quantization prescription for topological flux. This is Def. \ref{TopologicalQuantumSectorsOfGeneralGaugeTheories} below, where we successively refine the prescription by allowing punctured surfaces and accounting for diffeomorphism-covariance (``general covariance'').

\smallskip 
We analyze the effect of general covariance in \S\ref{ModuliSpacesOfTopologicalFlux}, obtaining the monodromy groups whose unitary representations are studied in \S\ref{2CohomotopicalFluxThroughSurfaces}.

\smallskip

This non-Lagrangian prescription is motivated by the topological flux algebra of ordinary Yang-Mills theory discussed in \cite{SS24-QObs}, its previous applications in formal high-energy physics (cf. \cite{SS22-Conf}\cite{CSS23}\cite[\S 4]{SS24-QObs}), and the representation-theoretic structures obtained below for FQH systems. The resulting discrete algebras and their representations are the mathematical objects considered here; their physical interpretation is discussed in \S\ref{PredictionsAndComparisons}.

\smallskip 
\subsection{Topological Quantum States}
\label{TopologicalQuantumStates}

First, we set up some notation:

\noindent

\begin{definition}[\bf Spacetime]
\label{Spacetime}
Throughout, we consider

\begin{itemize}[
  leftmargin=.5cm,
  itemsep=1pt,
  topsep=2pt
]
\item[$\circ$] $X^{1,3} := \mathbb{R}^{1,1} \times \Sigma^2$ a globally hyperbolic 4D spacetime,

\item[$\circ$]  with spatial slices $\mathbb{R}^1 \times \Sigma^2$,
to be thought of as a tubular neighborhood of:

\item[$\circ$]  $\Sigma^2$, a surface (here: a connected, oriented smooth 2D manifold with boundary) which at times is

\item[$\circ$]  specialized to $\Sigma^2 \defneq \Sigma^2_{\mathcolor{darkorange}{g},\mathcolor{darkgreen}{b},\mathcolor{purple}{n}}$, the unique (up to diffeomorphism) surface (cf. \S\ref{SurfacesAnd2Cohomotopy}):
\begin{itemize}[
  itemsep=0pt,
  topsep=1pt
]
  \item of genus $\mathcolor{darkorange}{g}$,

  \item with $\mathcolor{darkgreen}{b}$ boundary components,

  \item and $\mathcolor{purple}{n}$ punctures:

\end{itemize}
\begin{equation}
  \label{TheSurface}
  \Sigma^2_{
    \underset{
      \mathclap{
        \adjustbox{
          scale=.7,
          rotate=+40,
        }{
          \llap{
            \color{darkorange}
            \raisebox{24pt}{genus}
            \hspace{-11pt}
          }
        }
      }
    }{
      \mathcolor{darkorange}{g}
    },\;
    \underset{
      \mathclap{
        \adjustbox{
          scale=.7,
          rotate=+40,
        }{
          \llap{
            \color{darkgreen}
            \raisebox{22pt}{boundaries}
            \hspace{-12pt}
          }
        }
      }
    }{
     \mathcolor{darkgreen}{b}
    },\;
    \underset{
      \mathclap{
        \adjustbox{
          scale=.7,
          rotate=+40,
        }{
          \llap{
            \color{purple}
            \raisebox{21pt}{punctures}
            \hspace{-11pt}
          }
        }
      }
    }{
    \mathcolor{purple}{n}
    }
  }
  \quad\simeq\quad 
  \big(
    \grayunderbrace{
    \Sigma^2_{0,0,0}
    }{
      \mathclap{
      \adjustbox{scale=.7}{
        \color{darkblue}
        \begin{tabular}{c}
          sphere
        \end{tabular}
      }
      }    
    }
    \overset
    {
      \mathclap{
      \adjustbox{scale=.7}{
        \color{gray}
        \def\arraystretch{.8}
        \begin{tabular}{c}
          connected 
          \\
          sum
        \end{tabular}
      }
      }    
    }
    {
      \,\#\, 
    }
    \grayunderbrace{
    T^2
    \#
    \cdots
    \#
    T^2
    }{
      \mathclap{
      \adjustbox{scale=.7}{
        \color{darkblue}
        \begin{tabular}{c}
          $\mathcolor{darkorange}{g}$ connected summands
          \\
          of tori
        \end{tabular}
      }
      }
    }
  \big)
    \overset
    {
      \mathclap{
      \adjustbox{
        scale=.7,
        rotate=30
      }{
        \rlap{
        \hspace{-10pt}
        \color{gray}
        complement
        }
      }
      }    
    }
   {
    \setminus
  }
  \big\{
    \grayunderbrace{
    {D}^2 
    \sqcup
    \cdots
    \sqcup
    {D}^2
    }{
      \mathclap{
      \adjustbox{scale=.7}{
        \color{darkblue}
        \begin{tabular}{c}
          $\mathcolor{darkgreen}{b}$ disjoint summands
          \\
          of open disks
        \end{tabular}
      }
      }    
    }
    \overset
    {
      \mathclap{
      \adjustbox{
        scale=.7,
        rotate=30
      }{
        \rlap{
        \hspace{-10pt}
        \color{gray}
        disjoint union
        }
      }
      }    
    }
   {
    \;\sqcup\;
   }
    \grayunderbrace{
    \overline{D}^2 
    \sqcup
    \cdots
    \sqcup
    \overline{D}^2     
    }{
      \mathclap{
      \adjustbox{scale=.7}{
        \color{darkblue}
        \begin{tabular}{c}
          $\mathcolor{purple}{n}$ disjoint summands
          \\
          of closed disks
        \end{tabular}
      }
      }    
    }
  \big\}
    \,,
\end{equation}
\\
understood as modeling an effectively 2-dimensional sample of material. \footnote{
  \label{SurfacesInExperiment}
  Albeit routinely considered in theory (cf. \cite{WenNiu90}), the practicability of direct laboratory realizations of $\Sigma^2_{g,b,n}$ \eqref{TheSurface} with transversal magnetic flux is limited when $g > 0$. The case $g = 1$ (the torus) is readily realized (only) when considering momentum space (the Brillouin torus of a 2D crystal, cf. \cite{SS23-TopOrder}) instead of position space, but, while noteworthy in itself, this is not the case of FQH systems of concern here.  Alternatively, it was argued 
  \cite{BarkeshliJianQi13}
  that suitable defects, called ``genons'', in a crystal lattice could make a sample of nominal genus $g=0$ effectively behave like $g > 0$.  
  
  But irrespective of practicality,
  the theoretical possibility of $g > 0$ allows us to compare our topological quantum flux observables to those of abelian Chern-Simons theory in the case $\Sigma^2_{g > 0, 0,0}$, and their agreement in this theoretical case supports the validity of our observables also in the more practical cases of $g =0$, $n,b \neq 0$.
}
\end{itemize}
We abbreviate $\Sigma^2_{g,b} := \Sigma^2_{g,b,0}$ and $\Sigma^2_g := \Sigma^{2}_{g,0} \defneq \Sigma^2_{g,0,0}$.
\end{definition}

\begin{example}[\bf Some surfaces]
  We have (boundary-fixing) diffeomorphisms as follows:
  \begin{equation}
  \label{ExamplesOfSurfaces}
  \adjustbox{}{
  \def\arraystretch{1.6}
  \begin{tabular}{rll}
    sphere 
    &
    $\Sigma^2_{0,0,0} 
     \,\simeq\, S^2$
    \\
    torus   
    &
    $\Sigma^2_{1,0,0} \,\simeq\, T^2$
    \\
    closed disk
    &
    $\Sigma^2_{0,1,0} \,\simeq\, D^2$
    &
    \adjustbox{
      raise=-.2cm
    }{
    \begin{tikzpicture}[scale=.7]
      \draw[
        line width=1,
        fill=gray!70
      ]
        (0,0)
        ellipse
        (2 and .5);
    \end{tikzpicture}
    }
    \\[5pt]
    open disk
    &
    $\Sigma^2_{0,0,1} \,\simeq\, \mathbb{R}^2$
    &
    \adjustbox{
      raise=-.25cm
    }{
    \begin{tikzpicture}[scale=.7]
      \draw[
        draw opacity=0,
        fill=gray!70
      ]
        (0,0)
        ellipse
        (2 and .5);
    \end{tikzpicture}
    }
    \\
    \def\tabcolsep{0pt}
    \def\arraystretch{.9}
    \begin{tabular}{r}
      open annulus
      \\
      / punctured plane
    \end{tabular}
    &
    $\Sigma^2_{0,0,2} \,\simeq\, \mathbb{R}^2 \setminus \{0\}$
    &
    \adjustbox{
      raise=-.4cm
    }{
    \begin{tikzpicture}[scale=.7]
      \draw[
        draw opacity=0,
        fill=gray!70
      ]
        (0,0)
        ellipse
        (2 and .5);
      \draw[
        draw opacity=0,
        fill=white
      ]
        (0,0)
        ellipse
        (2.4*.17 and .5*.17);
    \end{tikzpicture}
    }
    \\
    half-open annulus
    &
    $\Sigma^2_{0,1,1}$
    &
    \adjustbox{
      raise=-.4cm
    }{
    \begin{tikzpicture}[scale=.7]
      \draw[
        draw opacity=0,
        fill=gray!70
      ]
        (0,0)
        ellipse
        (2 and .5);
      \draw[
        line width=1,
        fill=white
      ]
        (0,0)
        ellipse
        (2.4*.17 and .5*.17);
    \end{tikzpicture}
    }
    \\[5pt]
    closed annulus
    &
    $\Sigma^2_{0,2,0} \,\simeq\, A^2$
    &
    \adjustbox{
      raise=-.4cm
    }{
    \begin{tikzpicture}[scale=.7]
      \draw[
        line width=1,
        fill=gray!70
      ]
        (0,0)
        ellipse
        (2 and .5);
      \draw[
        line width=1,
        fill=white
      ]
        (0,0)
        ellipse
        (2.4*.17 and .5*.17);
    \end{tikzpicture}
    }
  \end{tabular}
  }
  \end{equation}
  Here the disks and annuli are the surface types readily and commonly realized in laboratory FQH experiments (cf. footnote \ref{SurfacesInExperiment}). 
\end{example}

\begin{remark}[\bf Spin structure]
\label{SpinStructure}
In fact, we regard spacetime $X^{1,3}$ (Def. \ref{Spacetime})
--- and therefore also the surfaces $\Sigma^2$ \eqref{TheSurface} ---
as equipped with a spin structure (cf. \cite{Milnor63}). Covariance for
a fixed spin structure uses the subgroup of mapping classes preserving it.
We specify the choices on the torus in
\S\ref{2CohomotopicalFluxThroughTorus} and on punctured genus-zero surfaces
in \S\ref{FluxThroughPuncturedSurface}. For example, of the four spin
structures on the 3-punctured sphere, only the one extending over the
sphere is invariant under the full $\mathrm{Sym}_3$; each of the other
three has stabilizer $\mathrm{Sym}_2$.
\end{remark}

\noindent
\begin{definition}[\bf Gauge group] 
\label{GaugeGroup}
For the following Def. \ref{TheYMFluxQuantumObservables} we consider 
$G$ a Lie group, with Lie algebra $\mathfrak{g}$ and with a choice of $\mathrm{Ad}$-invariant lattice $\Lambda \subset \mathfrak{g}$ (not necessarily full, possibly zero) --- 
but shortly we specialize this to $G \defneq \mathbb{R}$ and $\Lambda \defneq \mathbb{Z} \hookrightarrow \mathbb{R}$ \eqref{ChoiceOfOrdinaryEMGaugeGroup}.
\end{definition} 

The flux-observable analysis of \cite[\S 1 and \S A.1]{SS24-QObs} motivates considering the Fr{\'e}chet Lie group of smooth maps from $\Sigma^2$ to $G\ltimes_{\mathrm{Ad}}(\mathfrak{g}/\Lambda)$. Its relation to group convolution and strict deformation quantization provides background motivation. Here we adopt directly the following discrete construction, which is sufficient for the topological sector.

\begin{definition}[\bf Topological Yang-Mills flux algebra]
\label{TheYMFluxQuantumObservables}
  For a closed surface $\Sigma^2_g$, we define the algebra of topological
  $G$-flux observables to be the complex group algebra
  \eqref{GroupAlgebra} of the discrete group of connected components
  of the smooth mapping group:
  \begin{equation}
    \label{GFluxObservables}
    \def\arraystretch{2.2}
    \begin{array}{rcl}
    \mathllap{
      \smash{
      \scalebox{.7}{
        \color{darkblue}
        \bf
        \def\arraystretch{.9}
        \begin{tabular}{c}
          Algebra of topological 
          \\
          quantum observables
          \\
          on YM-flux through $\Sigma^2$
        \end{tabular}
      }
      }
    }
    \mathrm{TopFlxObs}
      _{\Sigma^2}
    &:=&
    \overset{
      \mathllap{
        \adjustbox{
          scale=.7,
          rotate=-17
        }{
          \llap{
            \rm
            \color{gray}
            \def\arraystretch{.9}
            \begin{tabular}{c}
              group 
              \\
              convolution 
              \\
              algebra
            \end{tabular}
          }
          \hspace{-30pt}
        }
      }
    }{\mathbb{C}}
    \Big[
      \pi_0\, C^\infty\big(
        \Sigma^2,
        \,
        G \ltimes_{\!{}_{\mathrm{Ad}}} (\mathfrak{g}/\Lambda)
      \big)
    \Big]
    \\
    &\simeq&
    \mathbb{C}\Big[
    \grayunderbrace{
      \pi_0\, C^\infty\big(
        \Sigma^2,
        \,
        G 
      \big)
    }{
      \scalebox{.7}{
        \rm 
        \def\arraystretch{.9}
        \def\tabcolsep{-4pt}
        \begin{tabular}{c}
          electric
        \end{tabular}
      }
    }
    \;
    \ltimes_{\!{}_{\mathrm{Ad}}}
    \;
    \grayunderbrace{
      \pi_0\, C^\infty\big(
        \Sigma^2,
        \,
        \tfrac{\mathfrak{g}}{\Lambda}
      \big)
    }{
      \scalebox{.7}{
        \rm 
        \def\arraystretch{.9}
        \def\tabcolsep{-4pt}
        \begin{tabular}{c}
          magnetic 
        \end{tabular}
      }
    }
    \Big]
    \mathrlap{\,.}
    \end{array}
  \end{equation}
  \vspace{-.5cm}
  
\end{definition}
The notation $\mathbb{C}[-]$ denotes the algebraic group $\ast$-algebra,
with involution $g^\ast=g^{-1}$ on group elements and conjugation on
coefficients. If a $C^\ast$-algebra is desired, we take the full group
$C^\ast$-completion of this discrete group algebra; its unitary
representations extend uniquely to that completion. No convolution
algebra or strict deformation quantization of the infinite-dimensional
mapping group is needed for this definition.
The component group in Def. \ref{TheYMFluxQuantumObservables} has the following equivalent homotopical descriptions.

\smallskip

\begin{proposition}[{\bf Topological YM flux quantum observables} {\cite[\S 3]{SS24-QObs}}]
\label{TheTopologicalGFluxObservables}
The discrete group algebra of Def. \ref{TheYMFluxQuantumObservables} is equivalently the group algebra \eqref{GroupAlgebra} of the fundamental group $\mathcolor{purple}{\pi_1}(-)$ \eqref{ConnectedComponentsOfLoopSpace} of the component of the constant map into the classifying space $\mathcolor{purple}{B}(-)$ \textup{(cf. \cite[\S 1]{AguilarGitlerPrieto02})}:
\begin{equation}
  \label{TopologicalGFluxObservables}
    \def\arraystretch{2}
    \begin{array}{ll}
    \mathclap{
        \adjustbox{
          scale=.7,
          raise=3pt
        }{
          \color{darkblue}
          \bf
          \def\arraystretch{.9}
          \begin{tabular}{c}
            Algebra of 
            {\color{purple}topological}
            \\
            flux
            observables
          \end{tabular}
        }     
      }
      \hspace{1.7cm}
    {
    \mathrm{TopFlxObs}
      _{\Sigma^2}
    }
    &\;:=\;
    \mathbb{C}\Big[
      \mathcolor{purple}{\pi_0}
      \,
      C^\infty\big(
        \Sigma^2,
        \,
        G \ltimes_{\mathrm{Ad}} \tfrac{\mathfrak{g}}{\Lambda}
      \big)
    \Big] 
    \\
    &\;\simeq\;
    \mathbb{C}\Big[
      \mathcolor{purple}{\pi_0}
      \,
      \mathrm{Map}(
        \Sigma^2,
        \,
        G \ltimes_{\mathrm{Ad}} 
        \tfrac{\mathfrak{g}}{\Lambda}
      \big)
    \Big]  
    \;\simeq\;
    \mathbb{C}\Big[
      \mathcolor{purple}{\pi_1}
      \,
      \mathrm{Map}_0
      \big(
        \Sigma^2
        ,\,
        \mathcolor{purple}{B}\big(
          G 
            \ltimes
          \tfrac{\mathfrak{g}}{\Lambda}
        \big)
      \big)
    \Big]
    \\
    &\;\simeq\;
    \mathbb{C}\Big[
      \pi_0
      \,
      \mathrm{Map}^\ast
      \big(
        (
          \Sigma^2
          \times
          \mathbb{R}
        )_{\cpt}
        ,\,
        \mathcolor{purple}{B}\big(
          G 
            \ltimes
          \tfrac{\mathfrak{g}}{\Lambda}
        \big)
      \big)
    \Big]
    \,.
  \end{array}
\end{equation}
\end{proposition}

\pagebreak[3]
\begin{remark}[\bf Asymptotic boundary localization of topological flux observables]
\label
{AsymptoticBoundaryLocalizationOfTopologicalFluxObservables}
\begin{itemize}[leftmargin=.75cm] 

\item[{\bf (i)}] In the second line of \eqref{TopologicalGFluxObservables}, we are showing several isomorphic incarnations of this group algebra of observables (these isomorphisms are explained in \cite[\S A.2]{SS24-QObs}, using basic facts also recalled in \S\ref{SomeAlgebraicTopology}), which each have their use in the following. 

\item[{\bf (ii)}] In particular, for a general topological group
$\widetilde G$, the suspension--loop adjunction gives
(recalling that $\Sigma^2$ here is closed, hence compact)
\begin{equation}
  \label{HomotopicalClutchingConstruction}
  \pi_0
  \,
  \mathrm{Map}^\ast\big(
    (\Sigma^2 \times \mathbb{R})_{\cpt}
    ,\,
    B \widetilde G
  \big)
  \;\simeq\;
  \pi_0\,
  \mathrm{Map}^\ast\big(
    S^1 \wedge (\Sigma^2)_{\plus}
    ,\,
    B \widetilde G
  \big)
  \;\simeq\;
  \pi_0
  \mathrm{Map}\big(
    \Sigma^2
    ,\,
    \widetilde G
  \big)
\end{equation}
Here $(\Sigma^2\times\mathbb{R})_{\cpt}$ is the {\it reduced suspension}
$S^1\wedge(\Sigma^2)_{\plus}$: equivalently, the ordinary suspension of
$\Sigma^2$ with its two tips identified to the single point $\infty$.
Thus \eqref{HomotopicalClutchingConstruction} classifies gauge bundles
with a chosen trivialization at $\infty$ by homotopy classes of
$\widetilde G$-valued maps on $\Sigma^2$.

\item[{\bf (iii)}] This is a relative form of the {\it clutching
construction} (cf. \cite[\S 7]{Husemoller66}). The compactified cylinder
is the union of two cones over $\Sigma^2$, meeting both along their
equator $\Sigma^2$ and at their common tip $\infty$. Trivialize the bundle
on each cone, with both trivializations agreeing with the chosen one
at $\infty$. The transition function is then the identity at $\infty$
and a map $\Sigma^2\to\widetilde G$ on the equator. Its homotopy class
is the last term of \eqref{HomotopicalClutchingConstruction}.
\end{itemize}

\begin{minipage}{6cm}
  \footnotesize
  {\bf Figure C.}
  The two ends of $\Sigma^2\times\mathbb{R}$ are identified to the
  single point $\infty$. With the bundle trivialized there, its
  clutching data are maps $\Sigma^2\to\widetilde G$ on the equator,
  up to homotopy \eqref{HomotopicalClutchingConstruction}.
  The dashed arrow marks the identification of the two displayed tips.
\end{minipage}
\;\;
\adjustbox{
  raise=-1.6cm
}{
\begin{tikzpicture}

  \draw[olive, densely dashed, <->]
    (-3,.14) .. controls (-2,1.15) and (2,1.15) .. (3,.14);
  \node[scale=.65, fill=white, inner sep=1pt] at (0,.93)
    {tips identified};

  \draw[olive]
    (0,0)
    ellipse
    (3 and .6);

  \draw[
    olive,
    line width=1,
    densely dotted
  ]
    (90:.2 and .6)
    arc
    (90:-90:.2 and .6);
  \draw[
    olive,
    line width=1,
  ]
    (-270:.2 and .6)
    arc
    (-270:-270+180:.2 and .6);

  \node[scale=.7] at 
    (0,0) 
    {$\Sigma^2$};

  \node[scale=.7] at 
    (.4,-.6) 
    {\adjustbox{
      margin=1pt,
      bgcolor=white
    }{$\mathbb{R}$}};

  \draw[fill=black] 
    (-3,0) circle (.04);

  \node[scale=.7] at (-3,0) 
    {  
      \adjustbox{
        margin= 0pt 1pt,
        bgcolor=white,
      }
        {$\infty$}
    };
  \node[scale=.7] at (+3,0) 
    {  
      \adjustbox{
        margin= 0pt 1pt,
        bgcolor=white,
      }
        {$\infty$}
    };

\node[scale=.9] at (5,0) 
  {$B \widetilde{G}$};

\draw[
  line width=5,
  white
] 
  (.5,-.1) .. controls
  (2,.7) and 
  (3.6,.7) ..
  (4.6,.1);
\draw[
  ->
] 
  (.5,-.1) .. controls
  (2,.7) and 
  (3.6,.7) ..
  (4.6,.1);

\begin{scope}[
  shift={(0, -1.7)}
]

 \begin{scope}
 \clip
   (-3.3,.7) rectangle 
   (+2.8,-.7);
  \draw[olive]
    (0,0)
    ellipse
    (3 and .6);
  \end{scope}

  \draw[fill=black] 
    (-3,0) circle (.04);

  \node[scale=.7] at (-3,0) 
    {  
      \adjustbox{
        margin= 0pt 1pt,
        bgcolor=white,
      }
        {$\infty$}
    };

 \draw[
    olive,
    line width=1,
    densely dotted
  ]
    (90:.2 and .6)
    arc
    (90:-90:.2 and .6);
  \draw[
    olive,
    line width=1,
  ]
    (-270:.2 and .6)
    arc
    (-270:-270+180:.2 and .6);

  \begin{scope}[
    shift={(2.8,0)},
    scale=.36
  ]
  \draw[
    olive,
    line width=1.3,
  ]
    (90:.2 and .6)
    arc
    (90:-90:.2 and .6);
  \draw[
    olive,
    line width=1.3,
  ]
    (-270:.2 and .6)
    arc
    (-270:-270+180:.2 and .6);
  \end{scope}

\node[
  scale=.7
] at (2.85, -.4)
{$\Sigma^2$};

\node[scale=.9] at (5,0) 
  {$\widetilde{G}$};

\draw[->]
  (3,-.07) .. controls
  (3.5,.2) and
  (4,.2) ..
  (4.7,.0);

\end{scope}

\begin{scope}[
  rotate=40,
  shift={(-1.35,-3)}
]

\begin{scope}
\clip
   (2.8,.7) rectangle 
   (+3.2,-.7);

  \draw[olive]
    (0,0)
    ellipse
    (3 and .6);
\end{scope}

   \begin{scope}[
    shift={(2.8,0)},
    scale=.36
  ]
  \draw[
    olive,
    line width=1.3,
  ]
    (90:.2 and .6)
    arc
    (90:-90:.2 and .6);
  \draw[
    olive,
    line width=1.3,
  ]
    (-270:.2 and .6)
    arc
    (-270:-270+180:.2 and .6);
  \end{scope}

  \node[
    scale=.7,
    rotate=40
  ] at (+3.07,0) 
    {  
      \adjustbox{
        margin= 0pt 1pt,
        bgcolor=white,
      }
        {$\infty$}
    };

\end{scope}

\end{tikzpicture}
}

\end{remark}

\begin{example}[{\bf The prediction of ordinary electromagnetism...}]
  \label{TheCaseOfElectromagnetism}
  For ordinary electromagnetic flux, subject to the usual Dirac charge quantization law (where the magnetic but not electric flux is quantized in integral cohomology, cf. \cite[(14)]{SS24-QObs}) the relevant choice in 
  Def. \ref{GaugeGroup}
  is:
  \begin{equation}
    \label{ChoiceOfOrdinaryEMGaugeGroup}
    \def\arraystretch{1.4}
    \def\arraycolsep{2pt}
    \begin{array}{cc}
      &
      \overset{
        \mathclap{
          \adjustbox{
            scale=.7,
            raise=10pt
          }{
            \color{gray}
            \def\arraystretch{.85}
            \begin{tabular}{c}
              no electric
              \\
              flux quantization
            \end{tabular}
          }
        }
      }{
        G := \mathbb{\mathbb{R}}
      }
      \,,
      \hspace{1cm}
      \overset{
        \mathclap{
          \adjustbox{
            scale=.7,
            raise=10pt
          }{
            \color{gray}
            \def\arraystretch{.85}
            \begin{tabular}{c}
              usual magnetic
              \\
              flux quantization
            \end{tabular}
          }
        }
      }{
      \Lambda 
        := 
      \mathbb{Z} 
        \hookrightarrow 
      \mathbb{R}
      }
      \\
      \smash{
      \scalebox{.7}{
        \color{gray}
        \def\arraystretch{.9}
        \begin{tabular}{c}
          classifying 
          \\
          space
        \end{tabular}
      }
      }
      &
    \hotype{A} 
      := 
    B\big( 
      \;
      \mathbb{R} 
        \;
        \ltimes 
        \;
      \tfrac{\mathbb{R}}{\mathbb{Z}} 
      \;
    \big)   
      \,\shapeEquivalence\, 
    B \mathrm{U}(1)\;.
    \end{array}
  \end{equation}
  In this case, the algebra 
  \eqref{TopologicalGFluxObservables}
  of observables on topological flux through a closed surface $\Sigma_g$ \eqref{TheSurface}
  is
  \begin{equation}
    \label{OrdinaryFluxObservablesOnClosedSurface}
    \def\arraystretch{1.2}
    \def\arraycolsep{5pt}
    \begin{array}{ccll}
      \mathllap{
        \smash{
          \adjustbox{scale=.7}{
            \color{darkblue}
            \bf
            \def\arraystretch{.9}
            \begin{tabular}{c}
              Algebra of topological
              \\
              flux observables for
              \\
              ordinary electromagnetism
            \end{tabular}
          }
        }
      }
      \mathrm{TopFluxObs}
        ^{B \mathrm{U}(1)}
        _{\Sigma^2_g}
      &\simeq&
      \mathbb{C}\big[
        \pi_0
        \,
        \mathrm{Map}\big(
          \Sigma^2
          ,\,
          \mathrm{U}(1)
        \big)
      \big]
      \\
      &\simeq&
      \mathbb{C}\big[
        \pi_0
        \,
        \mathrm{Map}\big(
          \Sigma^2
          ,\,
          B \mathbb{Z}
        \big)
      \big]
      \\
      &\simeq&
      \mathbb{C}
      \big[
        H^1\big(
          \Sigma^2_g
          ;\, 
          \mathbb{Z}
        \big)
      \big]
      &
      \proofstep{
        e.g. \cite[p. 263]{FomenkoFuchs16}
      }
      \\
      &\simeq&
      \mathbb{C}\big[
        \,
        \mathbb{Z}^{2g}
        \,
      \big]
      &
      \proofstep{
        e.g. 
        \cite[Thm. 6.13]{Giblin77},
      }
    \end{array}
    \hspace{-2cm}
  \end{equation}
and so a corresponding space of quantum states $\HilbertSpace{H}_{\Sigma^2_g}$ carries an action of linear operators $\WOperator{\vec{a}}{\vec{b}}$, for $ \vec a, \vec b \in \mathbb{Z}^g$, subject to
\begin{equation}
  \label{OperatorAlgebraOfOrdinaryEMObservables}
  \WOperator{\vec{a}}{\vec{b}}
  \circ
  \WOperator{\vec{a}'}{\vec{b}'}
  \;\;
  =
  \;\;
  \WOperator{\vec{a} + \vec{a}'}{\vec{b} + \vec{b}'}
  \;\;
  =
  \;\;
  \WOperator{\vec{a}'}{\vec{b}'}
  \circ
  \WOperator{\vec{a}}{\vec{b}}
  \,.
\end{equation}
\end{example}

\begin{remark}[\bf ...and its failure to describe FQH systems]
\begin{itemize} 
\item[{\bf (i)}]
While this algebra of observables \eqref{OperatorAlgebraOfOrdinaryEMObservables}
is  the prediction of the ordinary traditional theory of electromagnetism, it
is {\it not quite} the algebra of observables of magnetic flux actually seen in fractional quantum Hall systems!

\item[{\bf (ii)}]
Instead, over the torus ($g = 1$) the quantum observables of FQH systems are famously thought
to satisfy a {\it non-commutative deformation} of \eqref{OperatorAlgebraOfOrdinaryEMObservables} where the cross-terms pick up the square $\zeta^2$ of the braiding phase factor \eqref{TheBraidingPhase} when commuted past each other \cite[(4.9)]{WenNiu90}\cite[(4.14)]{IengoLechner92}\cite[(4.21)]{Fradkin13}\cite[(5.28)]{Tong16}:
\begin{equation}
  \label{TorusWilsonLineCommutatorForFQH}
  \WOperator{1}{0}
  \circ
  \WOperator{0}{1}
  \;=\;
  \zeta
  \,
  \WOperator{1}{1}
  \;=\;
  \zeta^2
  \, 
  \WOperator{0}{1}
  \circ
  \WOperator{1}{0}
  \,.
\end{equation}
\item[{\bf (iii)}]
This motivates adopting the universal group algebra of the
{\it integer Heisenberg group} $\widehat{\mathbb{Z}^{2}}$
\eqref{IntegerHeisenbergGroup}, a non-abelian central extension of
$\mathbb{Z}^{2}$ by $\mathbb{Z}$. At fixed central phase $\zeta$,
representations factor through the quotient by the relation
$z=\zeta\,1$, where $z$ denotes the central generator; further quotient
relations may hold in a particular representation. Thus
\eqref{TorusWilsonLineCommutatorForFQH} motivates this universal algebra,
while its identification with the algebra adopted here follows from
the cohomotopy calculation in
Prop. \ref{MonodromyOfFluxThroughClosedSurface}.
\end{itemize}
\end{remark}

This motivates looking for a coherent generalization of topological flux observables via  non-standard flux quantization laws such that they do capture effects like \eqref{TorusWilsonLineCommutatorForFQH} (we find this in \S\ref{2CohomotopicalFluxThroughTorus}).

\smallskip

Proposition \ref{TheTopologicalGFluxObservables} shows that the adopted topological flux algebra depends on the classifying space encoding the flux-quantization law (cf. Ex. \ref{TheCaseOfElectromagnetism}).

\noindent
This motivates constructing topological flux observables directly from the flux-quantization law specified by a classifying space
$\hotype{A}$.

\medskip

\noindent
{\bf The topological flux quantization prescription.}
We are thus led to the following Def. \ref{TopologicalQuantumSectorsOfGeneralGaugeTheories},
which is the quantization prescription adopted in our analysis. Its ordinary gauge-theory case is Def. \ref{TheYMFluxQuantumObservables} and Prop. \ref{TheTopologicalGFluxObservables}; the same group-algebra construction makes sense for general flux classifying spaces.
Note that the formula \eqref{TopologicalGFluxObservables} immediately generalizes from the case $\hotype{A} \defneq B(G \ltimes \tfrac{\mathfrak{g}}{\Lambda})$ to {\it any} pointed connected space $\hotype{A}$:

\begin{equation} \label{TopologicalObservablesMoreGenerally}
  \mathllap{
    \adjustbox{
      scale=.7
    }{
      \color{darkblue}
      \bf
      \def\arraystretch{.9}
      \begin{tabular}{c}
        Topological flux observables for
        \\
        flux quantized in 
        $\hotype{A}$-cohomology
      \end{tabular}
    }
  }
  \mathrm{TopFlxObs}
    ^{\hotype{A}}
    _{\Sigma^2}
  \;:=\;
  \mathbb{C}\big[
    \pi_1
    \,
    \mathrm{Map}_0
    \big(
      \Sigma^2
      ,\,
      \hotype{A}
    \big)
  \big]
  \,.
\end{equation}

\begin{remark}[\bf Classifying spaces for higher and generalized symmetries]
  \label{HigherAndGeneralizedSymmetry}
  For any pointed connected space $\hotype{A}$ its loop space $\Omega \hotype{A}$ carries a ``higher group'' structure (cf. \cite{NSS15}\cite[\S 2.2]{SS20-Orb}\cite[\S 3.1.2]{SS25-EBund})  under concatenation and reversal of loops, and $\hotype{A}$ is the classifying space for {\it that} higher group structure (cf. \cite[Prop. 2.2]{FSS23-Char}):
  $$
    \hotype{A} 
    \;\shapeEquivalence\;
    B(\Omega \hotype{A})
    \,.
  $$
  This way, the generalization \eqref{TopologicalObservablesMoreGenerally} corresponds to passage to {\it higher} gauge theories with {\it higher} (gauge) symmetry.  
  These days in physics the latter  is also referred to as ``generalized symmetry'', see \cite{Gwilliam25} for an account that makes contact with our perspective here.
\end{remark}

\smallskip 
\noindent 
This algebra \eqref{TopologicalObservablesMoreGenerally} of observables being a {\it group algebra} means that the corresponding spaces of (pure) quantum states --- which generally are modules over the observable algebra --- are actually modules of a group algebra and, as such, nothing but (unitary) representations of this group (cf. \cite[p. 11]{CiagliaIbortMarmo19}):

\begin{equation}
  \label{TopologicalQuantumStatesOnClosedSurface}
  \mathllap{
    \adjustbox{
      scale=.7
    }{
      \color{darkblue}
      \bf
      \def\arraystretch{.9}
      \begin{tabular}{c}
        Topological quantum states of
        \\
        flux quantized in 
        $\hotype{A}$-cohomology
      \end{tabular}
    }
  }
  \HilbertSpace{H}
    ^{\hotype{A}}
    _{\Sigma^2}
  \;\in\;
  \mathrm{URep}\big(
    \pi_1
    \,
    \mathrm{Map}_0(
      \Sigma^2
      ,\,
      \hotype{A}
    )
  \big)
  \,.
\end{equation}

\smallskip 
Further generalizations are immediately suggested by these formulas:
If $\Sigma^2 = \Sigma^2_{g,b,n}$ 
\eqref{TheSurface}
is possibly non-compact ($n > 0$), then the {\it solitonic} flux configurations 
(cf. \cite[\S 2.2]{SS25-Flux}\cite[\S A.2]{SS24-QObs})
are those which are {\it vanishing at infinity} and thus classified by {\it pointed} maps on the one-point compactification $(-)_{\cpt}$ (cf.  \S\ref{SomeAlgebraicTopology} and \hyperlink{FigureI}{\it Fig. I}), so that \eqref{TopologicalQuantumStatesOnClosedSurface} generalizes to:
\begin{equation}
  \label{SolitonicTopologicalObservablesMoreGenerally}
  \mathllap{
    \adjustbox{
      scale=.7
    }{
      \color{darkblue}
      \bf
      \def\arraystretch{.9}
      \begin{tabular}{c}
        ... for solitonic flux
        \\
        on non-compact spaces
      \end{tabular}
    }
  }
  \grayoverbrace{
  \HilbertSpace{H}
    ^{\hotype{A}}
    _{\Sigma^2}
  }{
    \mathclap{
      \scalebox{.7}{
        \def\arraystretch{.9}
        \begin{tabular}{c}
          space of
          \\
          quantum states
        \end{tabular}
      }
    }
  }
  \quad\in\quad 
  \grayunderbrace{
    \mathrm{URep}
  }{
    \mathclap{
      \scalebox{.7}{
        \def\arraystretch{.9}
        \begin{tabular}{c}
          unitary
          \\
          representations
        \end{tabular}
      }
    }
  }
  \Big(
    \grayunderbrace{
    \pi_1
    \,
    \grayoverbrace{
    \mathrm{Map}
      ^{\ast}
      _0
    \big(
      \Sigma^2_{\cpt}
      ,\,
      \hotype{A}
    \big)
    }{
      \mathclap{
        \scalebox{.7}{
          \def\arraystretch{.9}
          \begin{tabular}{c}
            moduli space of 
            \\
            solitonic topological flux
          \end{tabular}
        }
      }
    }
    }{
      \mathclap{
        \scalebox{.7}{
          \begin{tabular}{c}
            topological flux 
            \\
            monodromy group
          \end{tabular}
        }
      }
    }
  \Big)
  \,.
\end{equation}

\smallskip 
\noindent
\begin{minipage}{8.4cm}
  \hypertarget{FigureI}{}
  \footnotesize
  {\bf Figure I -- Solitonic flux and the point at infinity.}  
  To ``adjoin the point at infinity'' --
  $\Sigma^2 \mapsto \Sigma^2 _{\cpt}$ -- means to identify all {\it open ends} of a surface with a single point  ``$\infty$''. Regarding this as the basepoint of $\Sigma^2_{\cpt}$ and understanding the basepoint of $\hotype{A}$ to classify vanishing flux, then taking the classifying maps $\Sigma^2_{\cpt} \xrightarrow{\;} \hotype{A}$ to be {\it pointed maps} in $\mathrm{Map}^\ast_0(\Sigma^2_{\cpt}, \hotype{A})$ 
  \eqref{SolitonicTopologicalObservablesMoreGenerally}
  literally makes the flux {\it vanish at infinity}, which is the defining condition for solitonic flux. This way, punctures in $\Sigma^2$ represent islands where flux is expelled, cf. \hyperlink{FigureD}{\it Fig. D}.
\end{minipage}
\;\;\;\;
\adjustbox{
  raise=-1.7cm
}{
\begin{tikzpicture}[scale=.65]

\begin{scope}[shift={(-3.4,0)}]

  \shade[ball color=gray!40]
    (0,0) circle (2.5);

\begin{scope}[
    shift={(+.3,-.7)}
  ]
  \draw[
    line width=1.5,
    draw opacity=1,
    fill=white,
  ] 
    (145:.4) arc 
    (145:145+360:.4);
\end{scope}

 \begin{scope}[
    shift={(+.3-1.2,-.7)}
  ]
  \draw[
    draw opacity=0,
    fill=white,
  ] 
    (145:.4) arc 
    (145:145+360:.4);
  \node at (0,0) 
  { $\infty$ };
\end{scope}
 
  \begin{scope}[
    shift={(+1.1,+.5)}
  ]
  \draw[
    draw opacity=0,
    fill=white,
  ] 
    (145:.4) arc 
    (145:145+360:.4);
  \node at (0,0) { $\infty$ };
  \end{scope}

\end{scope}

  \node at (-.2,0) {$\simeq$};

\begin{scope}[shift={(3.2,0)}, scale=1.2]
  
      \draw[
        draw opacity=0,
        fill=gray!70,
      ]
        (125:2*1.2 and 1*1.2)
        arc
        (125:125-360:2*1.2 and 1*1.2);
     \node at (210:2*1.45 and 1*1.45) 
       {$\infty$};

 \begin{scope}[
   shift={(-.8,0)}
 ]
      \draw[
        draw opacity=0,
        fill=white,
      ]
        (145:2*.3 and 1*.3)
        arc
        (145:145+360:2*.3 and 1*.3);
     \node at (0,0) {$\infty$};
  \end{scope}

 \begin{scope}[
   shift={(+.8,0)}
 ]
      \draw[
        line width=1.5,
        draw opacity=1,
        fill=white,
      ]
        (145:2*.3 and 1*.3)
        arc
        (145:145+360:2*.3 and 1*.3);
  \end{scope}

 \end{scope}

\node at (.1,+1.7) {
  $(\Sigma^2_{0,1,2})_{\cpt}$
};

\node at (1.15,-2.2) {\footnotesize
  \color{darkblue} here flux is expelled
};

\draw[
  line width=1.6,
  olive!60,
  -Latex
] 
  (-.9, -2.2) .. controls
  (-2.5,-2.1) and
  (-4,-2.1) ..
  (-4.3, -.9);

\draw[
  line width=1.6,
  olive!60,
  -Latex
] 
  (-.8, -1.8) .. controls
  (-1, -1) and
  (-2, -1) ..
  (-2.3, .3);

\draw[
  line width=1.6,
  olive!60,
  -Latex
] 
 (-.4, -2) --
 (.15, -1.1);

\draw[
  line width=1.6,
  olive!60,
  -Latex
] 
 (-0, -2) .. controls
 (1,-1) and
 (2.3,-1.4) ..
 (2.3, -.2);

\end{tikzpicture}
}

\medskip 
Moreover, if the gauge field is to be regarded as ``generally covariant'' in the sense of physics --- namely covariant with respect to diffeomorphisms, such as for gravitational and topological systems --- so that two topological flux configurations are to be regarded as gauge equivalent if one is obtained from the other by precomposition with a diffeomorphism, then (cf. \cite[Def. 1.1]{Dul23}) the true moduli space is the {\it homotopy quotient} \eqref{HomotopyQuotientViaBorel} of the flux moduli space \eqref{SolitonicTopologicalObservablesMoreGenerally} by the action of the diffeomorphism group $\mathrm{Diff}(\Sigma^2)$ (see Def. \ref{ModuliSpaces}).
Therefore we set:

\begin{definition}[\bf Generally covariant topological quantum states of exotic flux]
\label{TopologicalQuantumSectorsOfGeneralGaugeTheories}
The possible spaces of quantum states of generally covariant topological flux quantized in $\hotype{A}$-cohomology are the irreducible unitary representations of the covariantized flux monodromy group: 
\begin{equation}
  \label{StateSpacesAsLocalSystemsOnModuliSpace}
  \hspace{-1cm}
  \adjustbox{
    rndfbox=5pt,
    bgcolor=lightolive
  }{$
    \scalebox{.7}{
      \color{darkblue}
      \bf
      \def\arraystretch{.9}
      \begin{tabular}{c}
        Generally covariantized          
        \\
        quantum states of
        \\
        topological solitonic flux 
      \end{tabular}
    }
    \;
    \HilbertSpace{H}
      ^{\hotype{A}}
      _{\Sigma^2}
    \hspace{.2cm}\in\hspace{.2cm}
    \grayoverbrace{
    \grayunderbrace{    \mathclap{\phantom{\vert_{\vert_{\vert}}}}
    \mathrm{URep}
    }{
      \mathclap{
        \scalebox{.7}{
          \color{gray}
          \def\arraystretch{.9}
          \begin{tabular}{c}
            reps
            \\
            of
          \end{tabular}
        }
      }
    }
    \bigg(
    \grayunderbrace{
      \pi_1
      \mathclap{\phantom{\vert_{\vert_{\vert}}}}
    }{
      \mathclap{
        \adjustbox{scale=.7}{
          \def\arraystretch{.9}
          \begin{tabular}{l}
            fundamental
            \\
            group of
          \end{tabular}
        }
      }
    }
    }{
      \mathclap{
        \scalebox{.7}{
          \color{gray}
          \def\arraystretch{.9}
          \begin{tabular}{c}
            local systems of
            \\
            state spaces on
          \end{tabular}
        }
      }
    }
    \adjustbox{
      scale=1.4,
      raise=-1pt
    }{$($}
    \grayoverbrace{
    \grayunderbrace{
    \mathclap{\phantom{\vert_{\vert_{\vert}}}}
    \mathrm{Map}
      ^{\ast}
      _0
    \big(
       \Sigma^2_{\cpt}
      ,\,
      \hotype{A}
    \big)
    }{
      \mathclap{
        \hspace{2pt}
        \scalebox{.7}{
          \def\arraystretch{.9}
          \begin{tabular}{c}
            plain moduli space
            \\
            of topological solitonic flux
            \\
            quantized in $\hotype{A}$-cohomology
          \end{tabular}
        }
      }
    }
    \grayunderbrace{
    \mathclap{\phantom{\vert_{\vert_{\vert}}}}
    \sslash
    \mathrm{Diff}^{+,\partial}(
      \Sigma^2
    )
    }{
      \mathclap{
        \scalebox{.7}{
          \def\arraystretch{.9}
          \begin{tabular}{c}
            covariantized
            \\
            under diffeos
          \end{tabular}
        }
      }
    }\mathclap{\phantom{\bigg(}}
    }{
      \scalebox{.7}{
        moduli space of topological 
        solitonic flux
      }
    }
    \adjustbox{
      scale=1.4,
      raise=-1pt
    }{$)$}
    \bigg)
  $}
  \,.
\end{equation}

\end{definition}

In order to get a handle on these groups
\eqref{StateSpacesAsLocalSystemsOnModuliSpace}
of generally covariantized flux monodromy, the first general result we prove below (Prop. \ref{ExtensionOfMCGByFluxMonodromy}, not surprising, but important) is that they are  equivalently the semi-direct product of the plain flux monodromy with the surface's mapping class group:
  \begin{equation}
    \label{SemidirectProductStructureInIntroduction}
      \pi_1
      \Big(
      \grayoverbrace{
      \mathrm{Map}
        ^\ast
        _0
      \big(
         \Sigma^2_{\cpt}
        ,\,
        \hotype{A}
      \big)
      \;\sslash\;
      \mathrm{Diff}^{+,\partial}(
        \Sigma^2
      )
      }{
        \scalebox{.7}{
          \color{gray}
          \rm
          covariantized flux moduli space
          \eqref{StateSpacesAsLocalSystemsOnModuliSpace}
        }
      }
      \Big)
      \;\;\simeq \;\;
      \grayunderbrace{
      \pi_1
      \Big(
      \mathrm{Map}
        ^\ast
        _0
      \big(
        \Sigma^2_{\cpt}
        ,\,
        \hotype{A}
      \big)
      \Big)
      }{
        \scalebox{.7}{
          \color{gray}
          \rm
          flux monodromy
        }
      }
      \;\rtimes\;
      \grayunderbrace{
        \pi_0
        \Big(
          \mathrm{Diff}^{+,\partial}(
            \Sigma^2
          )
        \Big)
      }{
        \mathclap{
        \scalebox{.7}{
          \color{gray}
          \rm
          mapping class group
          \eqref{MappingClassGroup}
        }
        }
      }.
  \end{equation}

Some remarks on the import of this construction \eqref{StateSpacesAsLocalSystemsOnModuliSpace}:

\begin{remark}[\bf Local systems of vector spaces as quantum state spaces]
\label{LocalSystemsAndQuantumStateSpaces}
\begin{enumerate}[leftmargin=.8cm]  
\item Representations of fundamental groups $\pi_1(\hotype{M}_0)$ as in \eqref{StateSpacesAsLocalSystemsOnModuliSpace}, equivalently (since $\hotype{M}_0$ is a connected component) of the {\it fundamental groupoid} $\Pi_1(\hotype{M}_0)$ \eqref{DeloopingGroupoid} 
are also known as {\it local systems} (cf. \cite[\S I.1]{Deligne70}\cite[p. 257]{Whitehead78}\cite[\S 2.5]{Dimca04}\cite[\S 2.6]{Szamuely09}) or {\it flat bundles} (cf. \cite[\S 9.2.1]{Voisin02}) of vector spaces over $\hotype{M}_0$ 
(cf. \cite[Lit. 2.22]{MySS24-TQG}). 
\[
  \begin{tikzcd}[decoration=snake]
    &
    &
    \HilbertSpace{H}_{m_{{}_1}}
    \ar[
      dr,
      shorten=-2pt,
      "{
        U(\gamma_{12})
      }"{sloped}
    ]
    &
    {}
    \ar[
      d,
      phantom,
      "{
        \mathrlap{
          \scalebox{.7}{
            \color{darkorange}
            \bf
            \def\arraystretch{.9}
            \begin{tabular}{c}
              unitary operators
              \\
              between fiber spaces
            \end{tabular}
          }
        }      
      }"
    ]
    \\
    \HilbertSpace{H}_{\bullet}
    \ar[
      dd,
      "{
        \scalebox{.7}{
          \color{darkblue}
          \bf
          \begin{tabular}{c}
            Local system/
            \\
            flat bundle
            \\
            of vector spaces
          \end{tabular}
        }
      }"{swap, pos=0}
    ]
    &
    \HilbertSpace{H}_{m_{{}_0}}
    \ar[
      rr,
      shorten=-2pt,
      "{ U(\gamma_{02}) }"
    ]
    \ar[
      ur,
      shorten=-2pt,
      "{
        U(\gamma_{01})
      }"{sloped}
    ]
    &&
    \HilbertSpace{H}_{m_{{}_2}}
    \\
    & & 
    m_{{}_1}
    \ar[
      dr, 
      shorten=-2pt,
      decorate,
      "{ 
        \gamma_{01} 
      }"{sloped, yshift=2pt}
    ]
    &
    {}
    \ar[
      d,
      phantom,
      "{
        \mathrlap{
          \scalebox{.7}{
            \color{darkorange}
            \bf
            \def\arraystretch{.9}
            \begin{tabular}{c}
              paths in 
              \\
              base space
            \end{tabular}
          }
        }      
      }"
    ]
    \\
    \mathllap{
      \smash{
      \scalebox{.7}{
        \def\arraystretch{.9}
        \begin{tabular}{c}
          \color{darkblue}
          \bf over base space 
          \\
          (here: moduli space
          \\
          of topological flux)
        \end{tabular}
      }
      }
    }
    \hotype{M}_0
    &
    m_{{}_0}
    \ar[
      rr, 
      shorten=-2pt,
      decorate,
      "{ \gamma_{02} }"{yshift=2pt}
    ]
    \ar[
      ur, 
      shorten=-2pt,
      decorate,
      "{ 
        \gamma_{01} 
      }"{sloped, yshift=2pt}
    ]
    &&
    m_{{}_2}
  \end{tikzcd}
\]

\item There is deep relevance 
\cite{SS23-EoS}
in identifying these as quantum state spaces subject to symmetries and classical control in general (cf. \cite{SS23-Monadology}\cite{CiagliaDiCosmoIbortMarmo20}) and specifically so concerning the nature of anyonic topological quantum gates (cf. \cite[\S 3]{MySS24-TQG} and the next Rem. \ref{DefectAnyonsAndBraidReps}).

\item In particular, (higher, cf. \cite{SS23-EoS}) local systems are considered as (relative) spaces of quantum states in generic ``pull-push''-constructions of (extended) topological quantum field theories (\cite{FHLT10}\cite{Morton15}\cite{Sc14}\cite[\S A.2]{FreedMooreTeleman23}). In this respect the point of \eqref{StateSpacesAsLocalSystemsOnModuliSpace} is the appropriate identification of the local system's domain moduli space for the case of topological flux quanta (and the remaining ``pull-push''-propagation along cobordisms, receiving so much attention in the topological field theory literature, plays little to no role in laboratory experiment, where the given slab of material commonly just sits there without spontaneously changing its topology!).

\item In our situation \eqref{StateSpacesAsLocalSystemsOnModuliSpace}, the base space is the covariantized moduli space of topological fluxes quantized in some $\hotype{A}$-cohomology, and paths $\begin{tikzcd}[decoration=snake] m \ar[r, decorate] & m'\end{tikzcd}$ are {\it combinations} of 
\begin{enumerate}[
  itemsep=1pt
]
\item gauge transformations of the (higher) gauge field

(from the homotopy quotient's numerator in \eqref{StateSpacesAsLocalSystemsOnModuliSpace})

\item diffeomorphisms, hence gauge transformations of the ``gravitational field'',

(from the homotopy quotient's denominator in \eqref{StateSpacesAsLocalSystemsOnModuliSpace}).
\end{enumerate}

The Hilbert space(s) of states $\HilbertSpace{H}^{\hotype{A}}_{\Sigma^2}$ forming a local system/flat bundle means that both of these kinds of gauge transformations are implemented as unitary quantum operators/observables in a coherent way.

\item
The local system \eqref{StateSpacesAsLocalSystemsOnModuliSpace} records the action of these transformations on the flux state spaces. Identifying which transformations are implemented as physical operations requires additional data specifying the physical realization.
\end{enumerate} 
\end{remark}

Now note that these paths here are paths in moduli space, hence not of single flux quanta but of all flux quanta that are present at once. The same holds for punctures/defects:

\pagebreak
\begin{remark}[\bf Defect anyons and braid representations]
\label{DefectAnyonsAndBraidReps}
$\,$

\begin{minipage}{13cm}
\begin{itemize}[leftmargin=.8cm]  
\item[{\bf (i)}] Punctures of $\Sigma^2$ model \emph{defects} in the slab of
quantum material (cf. \cite{Mermin79}). Their braided \emph{worldlines},
illustrated on the right, define a braid group and a point-pushing
homomorphism to the mapping class group in
\eqref{SemidirectProductStructureInIntroduction}. This homomorphism
need not be injective (see \eqref{BirmanSequence}).

\item[{\bf (ii)}] Pulling back the mapping-class action on flux quantum states
\eqref{StateSpacesAsLocalSystemsOnModuliSpace} gives linear representations
of this braid group ({\it braid representations}, cf.
\cite{KauffmanLomonaco04}\cite{Abad15}). These encode a possibly non-abelian generalization of the braiding phases of solitonic anyons \eqref{TheBraidingPhase}, exhibiting the defects as {\it defect anyons} (\cite[p. 4]{Kitaev06}\cite{SS23-TopOrder}). 
\end{itemize}
\end{minipage}
\quad 
\adjustbox{
  raise=-1.6cm
}{
  \includegraphics[width=3cm]
  {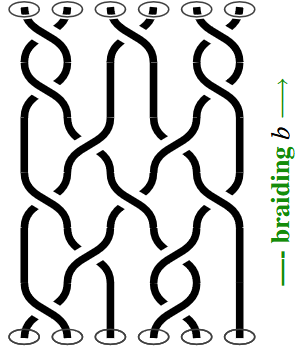}
}

\vspace{0.5mm} 
\begin{itemize}[leftmargin=1.3cm]  
\item[{\bf (iii)}]
Extrinsic defects and their controlled braiding have been studied in FQH-related systems, including genons \cite{BarkeshliJianQi13} and parafermion defects \cite{ClarkeAliceaShtengel13}. Here we investigate defect monodromy arising from the proposed cohomotopical flux model and its candidate topological quantum gates (cf. \cite{KauffmanLomonaco04}\cite{HoSOlomonOh10}\cite{Chen12}\cite{Kauffman18}). The feasibility of adiabatic control, for either defects or solitonic excitations, depends on a device-specific physical realization.
\end{itemize} 
\end{remark}

\begin{center}
\hypertarget{FigureA}{}
\begin{minipage}{6cm}
  \footnotesize
  {\bf Figure B -- Topological braid gates.}
  Covariant flux quantum states on a sufficiently punctured surface, by their diffeomorphism equivariance \eqref{StateSpacesAsLocalSystemsOnModuliSpace}, carry unitary representations of the {\it braid group} \eqref{SurfaceBraidGroup}
  of joint motions of the defects around each other, exhibiting the punctures as {\it defect anyons}.

  If a gapped physical realization has robust adiabatic transport realizing this braid representation, controlled braiding implements the corresponding unitary operators as {\it quantum gates}. Protection against deformations of the control paths requires flat, or appropriately projectively flat, topological transport that is stable under the allowed perturbations; a spectral gap alone does not ensure this (cf. \S\ref{PredictionsAndComparisons}).
\end{minipage}
\hspace{.2cm}
\adjustbox{
  scale=.8,
  raise=-3cm
}{
\begin{tikzpicture}

  \shade[right color=lightgray, left color=white]
    (3,-3)
      --
      node[above, yshift=-1pt, 
      xshift=14pt,
      sloped]{
        \scalebox{.7}{
          \color{darkblue}
          \bf
          $\sim$ 2D material
        }
      }
    (-1,-1)
      --
    (-1.21,1)
      --
    (2.3,3);

  \draw[]
    (3,-3)
      --
    (-1,-1)
      --
    (-1.21,1)
      --
    (2.3,3)
      --
    (3,-3);

\draw[-Latex]
  ({-1 + (3+1)*.3},{-1+(-3+1)*.3})
    to
  ({-1 + (3+1)*.29},{-1+(-3+1)*.29});

\draw[-Latex]
    ({-1.21 + (2.3+1.21)*.3},{1+(3-1)*.3})
      --
    ({-1.21 + (2.3+1.21)*.29},{1+(3-1)*.29});

\draw[-Latex]
    ({2.3 + (3-2.3)*.5},{3+(-3-3)*.5})
      --
    ({2.3 + (3-2.3)*.49},{3+(-3-3)*.49});

\draw[-latex]
    ({-1 + (-1.21+1)*.53},{-1 + (1+1)*.53})
      --
    ({-1 + (-1.21+1)*.54},{-1 + (1+1)*.54});

  \begin{scope}[rotate=(+8)]
   \draw[dashed]
     (1.5,-1)
     ellipse
     ({.2*1.85} and {.37*1.85});
   \begin{scope}[
     shift={(1.5-.2,{-1+.37*1.85-.1})}
   ]
     \draw[->, -Latex]
       (0,0)
       to
       (180+37:0.01);
   \end{scope}
   \begin{scope}[
     shift={(1.5+.2,{-1-.37*1.85+.1})}
   ]
     \draw[->, -Latex]
       (0,0)
       to
       (+37:0.01);
   \end{scope}
   \begin{scope}[shift={(1.5,-1)}]
     \draw (.43,.65) node
     { \scalebox{.8}{$
     $} };
  \end{scope}
  \draw[fill=white, draw=gray]
    (1.5,-1)
    ellipse
    ({.2*.3} and {.37*.3});
  \draw[line width=3.5, white]
   (1.5,-1)
   to
   (-2.2,-1);
  \draw[line width=1.1]
   (1.5,-1)
   to node[
     above, 
     yshift=-4pt, 
     pos=.85]{
     \;\;\;\;\;\;\;\;\;\;\;\;\;
     \rotatebox[origin=c]{7}
     {
     \scalebox{.7}{
     \color{darkorange}
     \bf
     \colorbox{white}{anyon worldline}
     }
     }
   }
   (-2.2,-1);
  \draw[
    line width=1.1
  ]
   (1.5+1.2,-1)
   to
   (3.5,-1);
  \draw[
    line width=1.1,
    densely dashed
  ]
   (3.5,-1)
   to
   (4,-1);

  \draw[line width=3, white]
   (-2,-1.3)
   to
   (0,-1.3);
  \draw[-latex]
   (-2,-1.3)
   to
   node[
     below, 
     yshift=+2pt,
     xshift=-7pt
    ]{
     \scalebox{.7}{
       \rotatebox{+7}{
       \color{darkblue}
       \bf
       time
       }
     }
   }
   (0,-1.3);
  \draw[dashed]
   (-2.7,-1.3)
   to
   (-2,-1.3);

 \draw
   (-3.15,-.8)
   node{
     \scalebox{.7}{
       \rotatebox{+7}{
       \color{darkgreen}
       \bf
       braiding
       }
     }
   };

  \end{scope}

  \begin{scope}[shift={(-.2,1.4)}, scale=(.96)]
  \begin{scope}[rotate=(+8)]
  \draw[dashed]
    (1.5,-1)
    ellipse
    (.2 and .37);
  \draw[fill=white, draw=gray]
    (1.5,-1)
    ellipse
    ({.2*.3} and {.37*.3});
  \draw[line width=3.1, white]
   (1.5,-1)
   to
   (-2.3,-1);
  \draw[line width=1.1]
   (1.5,-1)
   to
   (-2.3,-1);
  \draw[line width=1.1]
   (1.5+1.35,-1)
   to
   (3.6,-1);
  \draw[
    line width=1.1,
    densely dashed
  ]
   (3.6,-1)
   to
   (4.1,-1);
  \end{scope}
  \end{scope}

  \begin{scope}[shift={(-1,.5)}, scale=(.7)]
  \begin{scope}[rotate=(+8)]
  \draw[dashed]
    (1.5,-1)
    ellipse
    (.2 and .32);
  \draw[fill=white, draw=gray]
    (1.5,-1)
    ellipse
    ({.2*.3} and {.32*.3});
  \draw[line width=3.1, white]
   (1.5,-1)
   to
   (-1.8,-1);
\draw
   (1.5,-1)
   to
   (-1.8,-1);
  \draw
    (5.23,-1)
    to
    (6.4-.6,-1);
  \draw[densely dashed]
    (6.4-.6,-1)
    to
    (6.4,-1);
  \end{scope}
  \end{scope}

\draw (1.73,-1.06) node
 {
  \scalebox{.8}{
    $z_{{}_{i}}$
  }
 };

\begin{scope}
[ shift={(-2,-.55)}, rotate=-82.2  ]

 \begin{scope}[shift={(0,-.15)}]

  \draw[]
    (-.2,.4)
    to
    (-.2,-2);

  \draw[
    white,
    line width=1.1+1.9
  ]
    (-.73,0)
    .. controls (-.73,-.5) and (+.73-.4,-.5) ..
    (+.73-.4,-1);
  \draw[
    line width=1.1
  ]
    (-.73+.01,0)
    .. controls (-.73+.01,-.5) and (+.73-.4,-.5) ..
    (+.73-.4,-1);

  \draw[
    white,
    line width=1.1+1.9
  ]
    (+.73-.1,0)
    .. controls (+.73,-.5) and (-.73+.4,-.5) ..
    (-.73+.4,-1);
  \draw[
    line width=1.1
  ]
    (+.73,0+.03)
    .. controls (+.73,-.5) and (-.73+.4,-.5) ..
    (-.73+.4,-1);

  \draw[
    line width=1.1+1.9,
    white
  ]
    (-.73+.4,-1)
    .. controls (-.73+.4,-1.5) and (+.73,-1.5) ..
    (+.73,-2);
  \draw[
    line width=1.1
  ]
    (-.73+.4,-1)
    .. controls (-.73+.4,-1.5) and (+.73,-1.5) ..
    (+.73,-2);

  \draw[
    white,
    line width=1.1+1.9
  ]
    (+.73-.4,-1)
    .. controls (+.73-.4,-1.5) and (-.73,-1.5) ..
    (-.73,-2);
  \draw[
    line width=1.1
  ]
    (+.73-.4,-1)
    .. controls (+.73-.4,-1.5) and (-.73,-1.5) ..
    (-.73,-2);

 \draw
   (-.2,-3.3)
   to
   (-.2,-2);
 \draw[
   line width=1.1,
   densely dashed
 ]
   (-.73,-2)
   to
   (-.73,-2.5);
 \draw[
   line width=1.1,
   densely dashed
 ]
   (+.73,-2)
   to
   (+.73,-2.5);

  \end{scope}
\end{scope}

\begin{scope}[shift={(-5.6,-.75)}]

  \draw[line width=3pt, white]
    (3,-3)
      --
    (-1,-1)
      --
    (-1.21,1)
      --
    (2.3,3)
      --
    (3, -3);

  \shade[right color=lightgray, left color=white, fill opacity=.7]
    (3,-3)
      --
    (-1,-1)
      --
    (-1.21,1)
      --
    (2.3,3);

  \draw[]
    (3,-3)
      --
    (-1,-1)
      --
    (-1.21,1)
      --
    (2.3,3)
      --
    (3, -3);

\draw (1.73,-1.06) node
 {
  \scalebox{.8}{
    $z_{{}_{i}}$
  }
 };

\draw[-Latex]
  ({-1 + (3+1)*.3},{-1+(-3+1)*.3})
    to
  ({-1 + (3+1)*.29},{-1+(-3+1)*.29});

\draw[-Latex]
    ({-1.21 + (2.3+1.21)*.3},{1+(3-1)*.3})
      --
    ({-1.21 + (2.3+1.21)*.29},{1+(3-1)*.29});

\draw[-Latex]
    ({2.3 + (3-2.3)*.5},{3+(-3-3)*.5})
      --
    ({2.3 + (3-2.3)*.49},{3+(-3-3)*.49});

  \draw[-latex]
    ({-1 + (-1.21+1)*.53},{-1 + (1+1)*.53})
      --
    ({-1 + (-1.21+1)*.54},{-1 + (1+1)*.54});

  \begin{scope}[rotate=(+8)]
   \draw[dashed]
     (1.5,-1)
     ellipse
     ({.2*1.85} and {.37*1.85});
   \begin{scope}[
     shift={(1.5-.2,{-1+.37*1.85-.1})}
   ]
     \draw[->, -Latex]
       (0,0)
       to
       (180+37:0.01);
   \end{scope}
   \begin{scope}[
     shift={(1.5+.2,{-1-.37*1.85+.1})}
   ]
     \draw[->, -Latex]
       (0,0)
       to
       (+37:0.01);
   \end{scope}
  \draw[fill=white, draw=gray]
    (1.5,-1)
    ellipse
    ({.2*.3} and {.37*.3});
 \end{scope}

   \begin{scope}[shift={(-.2,1.4)}, scale=(.96)]
  \begin{scope}[rotate=(+8)]
  \draw[dashed]
    (1.5,-1)
    ellipse
    (.2 and .37);
  \draw[fill=white, draw=gray]
    (1.5,-1)
    ellipse
    ({.2*.3} and {.37*.3});
\end{scope}
\end{scope}

  \begin{scope}[shift={(-1,.5)}, scale=(.7)]
  \begin{scope}[rotate=(+8)]
  \draw[dashed]
    (1.5,-1)
    ellipse
    (.2 and .32);
  \draw[fill=white, draw=gray]
    (1.5,-1)
    ellipse
    ({.2*.3} and {.37*.3});
\end{scope}
\end{scope}

\begin{scope}
[ shift={(-2,-.55)}, rotate=-82.2  ]

 \begin{scope}[shift={(0,-.15)}]

 \draw[line width=3, white]
   (-.2,-.2)
   to
   (-.2,2.35);
 \draw
   (-.2,.5)
   to
   (-.2,2.35);
 \draw[dashed]
   (-.2,-.2)
   to
   (-.2,.5);

\end{scope}
\end{scope}

\begin{scope}
[ shift={(-2,-.55)}, rotate=-82.2  ]

 \begin{scope}[shift={(0,-.15)}]

 \draw[
   line width=3, white
 ]
   (-.73,-.5)
   to
   (-.73,3.65);
 \draw[
   line width=1.1
 ]
   (-.73,.2)
   to
   (-.73,3.65);
 \draw[
   line width=1.1,
   densely dashed
 ]
   (-.73,.2)
   to
   (-.73,-.5);
 \end{scope}
 \end{scope}

\begin{scope}
[ shift={(-2,-.55)}, rotate=-82.2  ]

 \begin{scope}[shift={(0,-.15)}]

 \draw[
   line width=3.2,
   white]
   (+.73,-.6)
   to
   (+.73,+3.7);
 \draw[
   line width=1.1,
   densely dashed]
   (+.73,-0)
   to
   (+.73,+-.6);
 \draw[
   line width=1.1 ]
   (+.73,-0)
   to
   (+.73,+3.71);
\end{scope}
\end{scope}

\end{scope}

\draw
  (-2.2,-4.2) node
  {
    \scalebox{1.2}{
      $
       \mathllap{
          \raisebox{1pt}{
            \scalebox{.58}{
              \color{darkblue}
              \bf
              \def\arraystretch{.9}
              \begin{tabular}{c}
                quantum state for
                \\
                fixed anyon positions
                \\
                $z_1, z_2, \cdots$
                at time
                {\color{purple}$t_1$}
              \end{tabular}
            }
          }
          \hspace{-5pt}
       }
        \big\vert
          \psi({\color{purple}t_1})
        \big\rangle
      $
    }
  };

\draw[|->]
  (-1.3,-4.1)
  to
  node[
    sloped,
    yshift=5pt
  ]{
    \scalebox{.7}{
      \color{darkgreen}
      \bf
      unitary adiabatic transport
    }
  }
  node[
    sloped,
    yshift=-5pt,
    pos=.4
  ]{
    \scalebox{.7}{
      }
  }
  (+2.4,-3.4);

\draw
  (+3.2,-3.85) node
  {
    \scalebox{1.2}{
      $
        \underset{
          \raisebox{-7pt}{
            \scalebox{.55}{
              \color{darkblue}
              \bf
              \def\arraystretch{.9}
               \begin{tabular}{c}
                quantum state for
                \\
                fixed anyon positions
                \\
                $z_1, z_2, \cdots$
                at time
                {\color{purple}$t_2$}
              \end{tabular}
            }
          }
        }{
        \big\vert
          \psi({\color{purple}t_2})
        \big\rangle
        }
      $
    }
  };

\node[
  draw=gray,
  line width=1.5,
  rounded corners
] 
at (-6.7,2.4) {
\scalebox{.7}{
  \hspace{-12pt}
  \begin{tabular}{c}
    topological quantum gate
    \\
    by controlled adiabatic
    \\
    braiding of anyon worldlines
  \end{tabular}
  \hspace{-12pt}
}
};

\end{tikzpicture}
}
\end{center}

\begin{remark}[\bf Focus on irreducible representations: Superselection sectors of anyons]
\label
{FocusOnIrreducibleRepresentations}
\begin{itemize}[leftmargin=.8cm]  
\item[{\bf (i)}]
  From the perspective of Def. \ref{TopologicalQuantumSectorsOfGeneralGaugeTheories},
  spaces of pure quantum states (of topological flux) are defined as representations of (modules over) the algebra of observables. From this  perspective, it is the (unitary) {\it irreducible} representations that matter in characterizing the actual quantum system, in that reducible representations behave like parallel copies of a system, which could just as well be discussed separately: ``superselection sectors'' (cf. \cite[Def. 2.1]{FrohlichGabbianiMarchetti90}\cite[p. 273]{Baker13}).
  
\item[{\bf (ii)}]
  While this may (and should) seem obvious, it is in some contrast to common practice: For punctured surfaces $\Sigma^2$ the mapping class group appearing in \eqref{SemidirectProductStructureInIntroduction} contains {\it braiding} operations (cf.  \cref{HomotopyTypeOfDiffeomorphismGroup}) and it is common to consider Yang-Baxter representations for these (``R-matrices'', cf. \cite{KauffmanLomonaco04}\cite{ZhangKauffmanGe05}), which generally are reducible (cf. \cite[p. 15]{LechnerPennigWood19}).
  \end{itemize} 
\end{remark}

\subsection{General Covariance of Flux}
\label{ModuliSpacesOfTopologicalFlux}

We discuss here the effect of  the covariantization \eqref{StateSpacesAsLocalSystemsOnModuliSpace} on the plain moduli spaces of solitonic topological flux. After recalling basics of mapping class groups and braid groups, the main result here is Prop. \ref{ExtensionOfMCGByFluxMonodromy} below, which was already announced as \eqref{SemidirectProductStructureInIntroduction}. 

\medskip

\noindent
{\bf Braid groups and mapping class groups.}

\begin{definition}[{\bf Configuration spaces and Braid groups}, {cf. \cite{FoxNeuwirth62}\cite[Lit. 2.18, 2.20]{MySS24-TQG}}]
\begin{itemize}[
  itemsep=2pt
]
  \item[{\bf (i)}] For $\Sigma$ a smooth manifold, possibly with boundary, and $n \in \mathbb{N}$, the {\it configuration space of $n$ points in $\Sigma$} is the topological space
  $$
    \mathrm{Conf}_n(\Sigma)
    \;:=\;
    \Big\{
      (s_1, \cdots, s_n)
      \,\in\,
      \Sigma^{\times^n}
      \,\Big\vert\,
      \underset{i \neq j}{\forall}
      s_i \neq s_j
    \Big\}
    \Big/
    \mathrm{Sym}_n
  $$
  (topologized as the quotient space of a subspace of a product space).

  \item[{\bf (ii)}]  
  For $\Sigma$ connected and of dimension at least $2$, the fundamental group of this space is the {\it braid group on $n$ strands in $\Sigma$} (cf. \cite[\S 9]{FarbMargalit12}), which as such comes equipped with a forgetful map to the symmetric group:
  \begin{equation}
    \label{SurfaceBraidGroup}
    \mathrm{Br}_n(\Sigma)
    \;:=\;
    \begin{tikzcd}
      \pi_1 \mathrm{Conf}_n(\Sigma)
      \ar[r, ->>]
      &
      \mathrm{Sym}_n \,.
    \end{tikzcd}
  \end{equation}
  \item[\bf (iii)] The kernel of \eqref{SurfaceBraidGroup} is the {\it pure braid} group $\mathrm{PBr}$:
  \begin{equation}
    \label{PureBraidGroup}
    \begin{tikzcd}[column sep=12pt]
      1
      \ar[r]
      &
      \mathrm{PBr}_n(\Sigma)
      \ar[rr, hook]
      &&
      \mathrm{Br}_n(\Sigma)
      \ar[
        rr,
        ->>
      ]
      &&
      \mathrm{Sym}_n
      \ar[r]
      &
      1
    \end{tikzcd}
  \end{equation}
\end{itemize}
\end{definition}
\begin{example}[{\bf Artin presentation of braid groups}, {cf. \cite[\S 7]{FoxNeuwirth62}\cite[Lit. 2.20]{MySS24-TQG}}]
\label{ArtinPresentationOfBraidGroups}
For $n \geq 2$, the surface braid group \eqref{SurfaceBraidGroup} of the disk (the default case of braid groups) has the following finite presentation:
\begin{equation}
  \label{ArtinPresentationOfPlainBraids}
  \mathrm{Br}_n
  \;:=\;
  \mathrm{Br}_n\big(
    \Sigma^2_{0,1,0}
  \big)
  \;\simeq\;
  F\langle
    b_1, \cdots, b_{n-1}
  \rangle\big/
  \Big(\,
    \underset{
      {i + 1 < j}
    }{\forall}
    \big(
      b_i b_j
      =
      b_j b_i
    \big)
    \,,
    \underset{
      1 \leq i < n - 1 
    }{\forall}
    \underset{
      \mathclap{
        \adjustbox{
          scale=.7,
          color=gray
        }{
          Yang-Baxter relation
        }
      }
    }{
    \big(
      b_i \, b_{i + 1} \, b_i
      \,=\,
      b_{i + 1} \, b_i \, b_{i + 1}
    \big)
    }
  \Big)
  \,,
\end{equation}
in terms of which its canonical homomorphism to the symmetric group is the quotient map by one further set of relations:
\begin{equation}
  \label{ArtinGeneratorsForSymmetricGroup}
  \begin{tikzcd}
    \mathrm{Br}_n
    \ar[
      r,
      ->>
    ]
    &
    \mathrm{Sym}_n
    \,:=\,
    \mathrm{Br}_n \big/
    \Big(
      \underset{i}{\forall}
      \big(
        b_i b_i = \mathrm{e}
      \big)
    \Big)
    \,.
  \end{tikzcd}
\end{equation}
The general surface braid group $\mathrm{Br}_n(\Sigma^2)$ may be presented by adjoining to these {\it Artin generators} $b_i$ further generators (corresponding to moving single strands along cycles in the surface) and further relations. In each case, there is a projection to the symmetric group by retaining the Artin generators:
\[
  \begin{tikzcd}
    \mathrm{Br}_n(\Sigma^2)
    \ar[r, ->>]
    &
    \mathrm{Sym}_n
  \end{tikzcd}
\]
\end{example}
\begin{example}[{\bf Presentation of spherical braid group} {\cite{FadellVanBuskirk61}, cf. \cite{Tan24}}]
The surface braid group \eqref{SurfaceBraidGroup} of the sphere (often: ``spherical braid group'') is presented as a quotient of the Artin presentation \eqref{ArtinPresentationOfPlainBraids} by one further relation:
\begin{equation}
  \label{PresentationOfSphericalBraidGroup}
  \mathrm{Br}_n(S^2)
  \;\simeq\;
  \mathrm{Br}_n \big/
  \big(
    (b_1 \cdots b_{n-1})(b_{n-1}\cdots b_1)
  \big)
  \,.
\end{equation}
\end{example}

\begin{definition}[\bf Diffeomorphism Group and Mapping Class Group]
\label{DiffeomorphismGroup}
For $\Sigma$ an oriented manifold, possibly with boundary, we write
\begin{equation}
  \label{DiffeoAndHomeoGroups}
  \begin{tikzcd}[
    column sep=20pt,
    row sep=12pt
  ]
    \mathrm{Homeo}^{+,\partial}(\Sigma)
    \ar[
      r, hook
    ]
    &
    \mathrm{Homeo}(\Sigma)
    \ar[r, hook]
    &
    \mathrm{Map}(\Sigma,\Sigma)
    \\
    \mathrm{Diff}^{+,\partial}(\Sigma)
    \ar[
      u, 
      hook, 
      "{ \iota }"{swap, pos=.4}
    ]
    \ar[r, hook]
    &
    \mathrm{Diff}(\Sigma)    
    \ar[
      u, 
      hook, 
      "{ \iota }"{swap, pos=.4}
    ]
  \end{tikzcd}
\end{equation}
for its topological groups of homeomorphisms and diffeomorphisms, respectively for the further subgroups of maps preserving the orientation $(+)$ and restricting to the identity on the boundary ($\partial$).

For $\Sigma \defneq \Sigma^2$ an orientable surface and choosing any one of its orientations, the group of connected components of the latter diffeo group is known as the {\it mapping class group} \cite[\S 1]{Ivanov02}\cite[\S 3]{Morita07}\cite[p. 45]{FarbMargalit12}:
\begin{equation}
  \label{MappingClassGroup}
  \mathrm{MCG}(\Sigma^2)
  \;:=\;
  \pi_0\big(
    \mathrm{Diff}^{+,\partial}(\Sigma^2)
  \big)
  \,.
\end{equation}
\end{definition}
\noindent
(For a fixed spin structure, we use the {\it spin} mapping class subgroup preserving it, with the choices specified in Rem. \ref{SpinStructure}.)

\begin{example}[{\bf Mapping class groups of closed oriented surfaces}, {cf. \cite[\S 6]{Morita07}\cite[\S 6]{FarbMargalit12}}]
\label
{MappingClassGroupsOfClosedOrientedSurfaces}
  The mapping class group of the torus is
  \begin{equation}
    \label{MCGOfTorus}
    \mathrm{MCG}(\Sigma^2_1)
    \;\simeq\;
    \mathrm{Sp}_2(\mathbb{Z})
    \,\simeq\,
    \mathrm{SL}_2(\mathbb{Z})
    \,,
  \end{equation}
  which is generated by the two elements
  \cite[Thm VII.2 p. 78]{Serre73}\cite[Thm. 1.1]{Conrad-SL2Z}:
  \begin{equation}
    \label{ModularGenerators}
    S
    \,:=\,
    \left[
    \def\arraystretch{1}
    \def\arraycolsep{2pt}
    \begin{array}{cc}
      0 & 1 
      \\
      -1 & 0
    \end{array}
    \right]
    \,,
    \hspace{.7cm}
    T
    \,:=\,
    \left[
    \def\arraystretch{1}
    \def\arraycolsep{2pt}
    \begin{array}{cc}
      1 & 1 
      \\
      0 & 1
    \end{array}
    \right]
    \mathrlap{\,,}
  \end{equation}
  and presented subject to the following relations \cite[p. 126]{KoecherKrieg07}\cite[\S 2.1]{BlancDeserti12}:
  \begin{equation}
    \label{PresentationOfSL2Z}
    \mathrm{SL}_2(\mathbb{Z})
    \,\simeq\,
    \big\langle
      S,\,T
      \,\big\vert\,
      S^4
        \,=\,
      (T S)^3 
        \,=\,
      \mathrm{e}
      ,\,
      S^2\, (T S) 
        \,=\,
      (T S)\, S^2
    \big\rangle
    \,.
  \end{equation}
  
  More generally, the mapping class group of $\Sigma^2_g$ \eqref{TheSurface}, for $g \in \mathbb{N}$, sits in a short exact sequence (cf. \cite[\S 6]{FarbMargalit12})
  \begin{equation}
    \label{MCGOfClosedSurface}
    \begin{tikzcd}[
      column sep=15pt,
      row sep=-4pt
    ]
      1 
      \ar[r]
      &
      I_g
      \ar[
        rr,
        hook
      ]
      &&
      \mathrm{MCG}(\Sigma^2_g)
      \ar[
        rr,
        ->>
      ]
      &&
      \mathrm{Sp}_{2g}(\mathbb{Z})
      \ar[r]
      &
      1
      \mathrlap{\,,}
      \\
      &
      \mathclap{
      \scalebox{.7}{
        \color{gray}
        Torelli group
      }
      }
      && 
      \mathclap{
      \scalebox{.7}{
        \color{gray}
        mapping class group
      }
      }
      &&
      \mathclap{
      \scalebox{.7}{
        \color{gray}
        symplectic group
      }
      }
    \end{tikzcd}
  \end{equation}
  where the (pre-)composition action of $\mathrm{MCG}(\Sigma^2_g)$ on 
  \[
    H_1(\Sigma^2_g;\mathbb{Z}) 
    \;\simeq\;
    H^1(\Sigma^2_g;\mathbb{Z}) 
    \;\simeq\; 
    \mathbb{Z}^{g} \times \mathbb{Z}^g
  \]
  is through the defining action of the integer symplectic group $\mathrm{Sp}_{2g}(\mathbb{Z})$.

  The set of spin structures is a torsor for $H^1(\Sigma^2_g;\mathbb{Z}_2)$
  (cf. \cite[p.~199]{Milnor63}); equivariantly, its elements are the quadratic
  refinements $q$ of the mod-2 intersection form, characterized by
  \[
    q(x+y)=q(x)+q(y)+x\cdot y
    \qquad (x,y\in H_1(\Sigma^2_g;\mathbb{Z}_2)).
  \]
  The mapping class action is $q\mapsto q\circ f_\ast^{-1}$.
  After choosing an origin in the torsor this is an affine action,
  whose linear part is the induced action on $H^1$; for $g\geq2$
  there is no origin fixed by the entire mapping class group
  (\cite[\S\S 3--5]{Johnson80}).
  On the torus, the unique odd spin structure $\mathrm{pp}$ is invariant.
  Choosing it as origin, and using Poincar{\'e} duality, identifies the
  four spin structures equivariantly with $H_1(\Sigma^2_1;\mathbb{Z}_2)$:  
  \begin{equation}
    \label{SpinStructuresOnTheTorus}
    \{\mathrm{pp}, \mathrm{aa}, \mathrm{ap}, \mathrm{pa}\}
    \;:\simeq\;
    \mathbb{Z}_2^2
    \;\simeq\;
    H_1\big(\Sigma^2_1;\, \mathbb{Z}_2\big)
  \end{equation}
  (with ``periodic'' or ``antiperiodic'' boundary conditions for spinors along the two basis 1-cycles, cf. \cite[\S 2]{AGMV86}) and $\mathrm{MCG}(\Sigma^2_1) 
    \simeq 
  \mathrm{SL}_2(\mathbb{Z})$ \eqref{MCGOfTorus} preserves $\mathrm{pp}$ and transitively permutes the other three. On the other hand, the stabilizer subgroup of the $\mathrm{aa}$ structure, hence the {\it spin mapping class group} of diffeomorphisms preserving $\mathrm{aa}$, is generated by $S$ and the {\it square} of $T$ 
  (cf. \cite[p. 3]{BRZW18})
  subject to the following relations:
  \begin{equation}
    \label
    {aaSpinMappingClassGroupOfTorus}
    \begin{tikzcd}[
      row sep=5pt
    ]
      \mathrm{MCG}(\Sigma^2_1)^{\mathrm{aa}}
      \ar[
        d,
        phantom,
        "{ \simeq }"{sloped}
      ]
      \ar[
        r, 
        hook
      ]
      &
      \mathrm{MCG}(\Sigma^2_1)
      \ar[
        d,
        phantom,
        "{ \simeq }"{sloped}
      ]
      \\
      \big\langle
        S, T^2
        \,\vert\,
        S^4 = [S^2, T^2] = \mathrm{e}
      \big\rangle
      \ar[
        r,
        hook
      ]
      &
      \mathrm{SL}_2(\mathbb{Z})
      \mathrlap{\,.}
    \end{tikzcd}
  \end{equation}
\end{example}

\begin{definition}[\bf Moduli spaces of solitonic topological fluxes]
\label{ModuliSpaces}
The underlying homeomorphisms of the diffeomorphisms 
\eqref{DiffeoAndHomeoGroups}
of surfaces $\Sigma^{2}_{g,b,n}$ \eqref{TheSurface} extend functorially to the one-point compactification (by Prop. \ref{OnePointCompactificationFunctorialOnProperMaps}) to make a topological group homomorphism
$$
  \begin{tikzcd}
    \mathrm{Diff}^{(+,\partial)}\big(
      \Sigma^2_{g,b,n}
    \big)
    \ar[r, "{ \iota }"]
    &
    \mathrm{Homeo}^{(+,\partial)}\big(
      \Sigma^2_{g,b,n}
    \big)
    \ar[
      r,
      "{ (-)_{\cpt} }"
    ]
    &[+14pt]
    \mathrm{Aut}_{\mathrm{Top}^\ast}\big(
      (\Sigma^2_{g,b,n})_{\cpt}
    \big)    
    \,.
  \end{tikzcd}
$$
Via the latter's action (by pre-composition) on pointed mapping spaces \eqref{SolitonicTopologicalObservablesMoreGenerally}
we obtain the homotopy quotient \eqref{HomotopyQuotientViaBorel}
of the pointed mapping space
\footnote{
\label{DiffeosPreserveHopfDegree}
For every pointed coefficient space $\hotype{A}$, precomposition fixes
the constant map to its basepoint and hence preserves the connected
component $\mathrm{Map}^{\ast}_0\big((\Sigma^2_{g,b,n})_{\cpt},\hotype{A}\big)$.
This applies also to surfaces with boundary and requires no degree
classification of the other components.
}
\begin{equation}
  \label{CovariantizedModuliSpace}
  \mathrm{Map}
    ^\ast
    _0
  \big(
    (\Sigma^2_{g,b,n})_{\cpt}
    ,\,
    \hotype{A}
  \big)
 \! \sslash \!
  \mathrm{Diff}^{+,\partial}\big(
    \Sigma^2_{g,b,n}
  \big)
  \;\;
  \in
  \;\;
  \mathrm{Top}^\ast
  \,,
\end{equation}
identified in \eqref{StateSpacesAsLocalSystemsOnModuliSpace} as the {\it covariantized moduli space} of $\hotype{A}$-quantized solitonic topological fluxes on $\Sigma^2_{g,b,n}$.
\end{definition}

The following is classical but somewhat scattered in the literature:
\begin{proposition}[\bf Homotopy type of Diffeomorphism groups]
\label[proposition]
{HomotopyTypeOfDiffeomorphismGroup}
\begin{enumerate}
  \item 
  For compact oriented surfaces $\Sigma^2_{g,b}$ \eqref{TheSurface}, the homotopy type of their diffeomorphism group \eqref{DiffeoAndHomeoGroups} is:
  \begin{equation}
    \label{HomotopyTypeOfDiffeosOfClosedSurfaces}
    \hspace{-.5cm}
    \def\arraycolsep{7pt}
    \def\arraystretch{1.5}
    \begin{array}{lc@{\hspace{18pt}}l@{\hspace{8pt}}c@{\hspace{8pt}}l}
      \mathrm{Diff}^{+}\big(
        \Sigma^2_{0,\,0,\,0}
      \big)
      \;\shapeEquivalence\;
      \,
      \mathrm{SO}(3)
      &\Rightarrow&
      \mathrm{MCG}(\Sigma^2_{0,\,0,\,0})
      \,\simeq\,
      1
      &
      \mbox{\rm and}
      &
      \pi_1
      \,
      \mathrm{Diff}^{+}\big(
        \Sigma^2_{0,\,0,\,0}
      \big)
      \,\simeq\,
      \mathbb{Z}_2
      \\
      \mathrm{Diff}^{+}\big(
        \Sigma^2_{1,\,0,\,0}
      \big)
      \;\shapeEquivalence\;
      \mathrm{SL}_2(\mathbb{Z})
      \times 
      T^2
      &\Rightarrow&
      \mathrm{MCG}(\Sigma^2_{1,\,0,\,0})
      \,\simeq\,
      \mathrm{SL}_2(\mathbb{Z})
      &
      \mbox{\rm and}
      &
      \pi_1
      \,
      \mathrm{Diff}^{+}\big(
        \Sigma^2_{1,\,0,\,0}
      \big)
      \,\simeq\,
      \mathbb{Z} \times \mathbb{Z}
      \\
      \mathrm{Diff}^{+}_0\big(
        \Sigma^2_{g\geq 2,\,0,\,0}
      \big)
      \;\shapeEquivalence\;
      \ast
      &\Rightarrow&
      &
      &
      \pi_1
      \,
      \mathrm{Diff}^{+}\big(
        \Sigma^2_{g\geq 2,\,0,\,0}
      \big)
      \,\simeq\,
      1
      \\[+10pt]
      \mathrm{Diff}^{+,\partial}_0\big(
        \Sigma^2_{g,\,b\geq 1,\, 0}
      \big)
      \;\shapeEquivalence\;
      \ast
      &\Rightarrow&
      &
      &
      \pi_1
      \,
      \mathrm{Diff}^{+,\partial}\big(
        \Sigma^2_{g,\, b\geq 1,\, 0}
      \big)
      \,\simeq\,
      1\,.
    \end{array}
  \end{equation}
  In particular:
  \begin{equation}
    \pi_1\, \mathrm{Diff}^{+, \partial}
    (\Sigma^2_{g,b})
    \simeq
    1
    \;\;\;
    \;\text{\rm except for}\;
    (g,b) \in 
    \bracket\{{
      (0,0), (1,0)
    }\}
    \mathrlap{\,.}
  \end{equation}

  \item  For punctured oriented surfaces $\Sigma^2_{g,b,\geq 1}$, the map from their mapping class group to that of $\Sigma^2_{g,b}$
  {\rm (obtained by uniquely extending the underlying homeomorphisms over
  the punctures and using the identification of topological and smooth
  mapping class groups of surfaces)}
  sits in a long exact sequence {\rm (``generalized Birman sequence'')}
  with the surface's braid group \eqref{SurfaceBraidGroup}, of this form:
  \begin{equation}
    \label{BirmanSequence}
    \begin{tikzcd}[
      column sep=11pt
    ]
      \pi_1
      \,
      \mathrm{Diff}^{+,\partial}\big(
        \Sigma^2_{g,b}
      \big)
      \ar[rr]
      &&
      \mathrm{Br}_n\big(
        \Sigma^2_{g,b}
      \big)
      \ar[
        rr
      ]
      &&
      \mathrm{MCG}\big(
        \Sigma^2_{g,b,n}
      \big)
      \ar[rr]
      &&
      \mathrm{MCG}\big(
        \Sigma^2_{g,b}
      \big)
      \,.
    \end{tikzcd}
  \end{equation}
  \begin{enumerate}
  \item
  
  Hence when $\pi_1 \mathrm{Diff}^{+,\partial}(\Sigma^2_{g,b}) = 1$ the mapping class group sits in a short exact sequence of the form
  \begin{equation}
    \label{ExtensionByBraidGroup}
    \begin{tikzcd}[
      column sep=10pt
    ]
      1
      \ar[r]
      &
      \mathrm{Br}_{n \geq 1}\big(
        \Sigma^2_{g,b}
      \big)
      \ar[
        rr
      ]
      &&
      \mathrm{MCG}(\Sigma^2_{g,b,n \geq 1})      
      \ar[
        rr
      ]
      &&
      \mathrm{MCG}(\Sigma^2_{g,b})      
      \ar[
        r
      ]
      &
      1
      \,,
    \end{tikzcd}
  \end{equation}
  and exhausts the homotopy type of the diffeomorphism group:
  \begin{equation}
    \label{HomotopyTypeOfDiffeosOfPuncturedSurface}
    \mathrm{Diff}^{+,\partial}\big(
      \Sigma^2_{g,b,n\geq 1}
    \big)
    \;\shapeEquivalence\;
    \mathrm{MCG}\big(
      \Sigma^2_{g,b,n\geq 1}
    \big)
    \;\;\;\;\;\;
    \Rightarrow
    \;\;\;\;\;\;
    \pi_1
    \,
    \mathrm{Diff}^{+,\partial}\big(
      \Sigma^2_{g,b,n\geq 1}
    \big)
    \;\simeq\;
    1\,.
  \end{equation}
  \item For $g=0$, $b=0$, and $n\geq3$ --- {\rm where the assumption in {\rm (iia)} fails by \eqref{HomotopyTypeOfDiffeosOfClosedSurfaces}} --- the spherical braid group surjects onto the mapping class group with kernel $\mathbb{Z}_2$, the image of $\pi_1\mathrm{Diff}^+(S^2)$ {\rm (generated by the full rotation braid $\fullRotation$)}:
  \begin{equation}
    \label{SphericalBraidGroupSequence}
    \begin{tikzcd}[
      column sep=13pt
    ]
      1 
      \ar[r]
      &
      \mathbb{Z}_2
      \ar[
        rr, 
        hook,
        "{
          \fullRotation
        }"
      ]
      &&
      \mathrm{Br}_{n \geq 3}(\Sigma^{2}_0)
      \ar[rr, ->>]
      &&
      \mathrm{MCG}(\Sigma^2_{0,0,n})
      \ar[
        r
      ]
      &
      1
      \mathrlap{\,.}
    \end{tikzcd}
  \end{equation}
  \item
  For $n=1,2$, the map from the spherical braid group to the mapping class group is an isomorphism; hence the kernel in \eqref{SphericalBraidGroupSequence} is trivial in these exceptional cases. Concretely, for $g=b=0$ we have:
  \begin{equation}
    \label{MappingClassGroupsOfPuncturedSpheres}
    \def\arraystretch{1.3}
    \begin{array}{rcl}
      \mathrm{MCG}(\Sigma^2_{0,0,1})
      &\simeq&
      1
      \,\simeq\,
      \mathrm{Br}_1(S^2)
      \\
      \mathrm{MCG}(\Sigma^2_{0,0,2})
      &\simeq&
      \mathbb{Z}_2
      \,\simeq\,
      \mathrm{Br}_2(S^2)
      \\
      \mathrm{MCG}(\Sigma^2_{0,0,3})
      &\simeq&
      \mathrm{Sym}_3
      \,\not\simeq\,
      \mathrm{Br}_3(S^2)
      \\
      \mathrm{MCG}(\Sigma^2_{0,0,4})
      &\simeq&
      \mathrm{PSL}_2(\mathbb{Z})
      \ltimes
      (\mathbb{Z}_2 \times \mathbb{Z}_2)\,.
    \end{array}
  \end{equation}
  \end{enumerate}
  \end{enumerate}
\end{proposition}
\begin{proof}
In \eqref{HomotopyTypeOfDiffeosOfClosedSurfaces}
the  first statement  is due to \cite{Smale59}, the first three 
were proven by \cite{EarleEells67}\cite{EarleEells69}\cite{Gramain73}, and the fourth is \cite[Thm. 1D p. 170]{EarleSchatz70}.
The statement \eqref{HomotopyTypeOfDiffeosOfPuncturedSurface} follows 
with \cite{Yagasaki00}\cite[Thm. 1.1]{Yagasaki01}.
\footnote{
  The surfaces in \cite{Yagasaki00}\cite{Yagasaki01} are assumed to be without boundary, but equipped with marked closed subcomplexes to be fixed by the diffeomorphisms. Under this definition, a puncture surrounded by a marked circle behaves just as a boundary for the purpose of computing the homotopy type of the diffeomorphism group.
} 
The generalized Birman sequence \eqref{BirmanSequence} is named in honor of \cite{Birman69}, cf. \cite[Thm. 3.13]{Massuyeau21}.
In its implication of the short exact sequence \eqref{ExtensionByBraidGroup} this is reviewed in \cite[Thm. 9.1]{FarbMargalit12}.
The spherical braid group extension 
\eqref{SphericalBraidGroupSequence}
is discussed in \cite[(9.1)]{FarbMargalit12}
and the identifications \eqref{MappingClassGroupsOfPuncturedSpheres}
of its quotients are proven, for instance, in \cite[Prop. 2.3]{FarbMargalit12}\cite{Bloomquist24}. 
\end{proof}

\begin{example}[\bf Mapping class groups of the $n$-punctured disk]
  Since the mapping class group of the disk $\Sigma^2_{0,1, 0}$ 
  is trivial by \eqref{HomotopyTypeOfDiffeosOfClosedSurfaces}, the exact sequence \eqref{ExtensionByBraidGroup} shows that the mapping class group of its punctured versions is the plain braid group \eqref{ArtinPresentationOfPlainBraids}:
  \begin{equation}
    \label{PlainBraidGroupAsMappingClassGroup}
    \def\arraystretch{1.2}
    \begin{array}{rcl}
    \mathrm{MCG}(\Sigma^2_{0,1,n})
    &\simeq&
    \mathrm{Br}_n\,.
    \end{array}
  \end{equation}
\end{example}

\begin{example}[\bf Mapping class group of the closed annulus]
The mapping class group \eqref{MappingClassGroup} of the closed annulus \eqref{ExamplesOfSurfaces}
is the integers (cf. \cite[Prop. 2.4]{FarbMargalit12}), generated by the {\it boundary Dehn twist} indicated in \hyperlink{FigureDehnTwist}{\it Fig. DT}.
\begin{equation}
  \label{MCGOfClosedAnnulus}
  \mathrm{MCG}\big(\Sigma^2_{0,2}\big)
  \;\simeq\;
  \mathbb{Z}
  \,.
\end{equation}
\end{example}
\begin{center}
\hypertarget{FigureDehnTwist}{}
\begin{minipage}{8cm}
  \footnotesize
  {\bf Figure DT -- The mapping class group of the closed annulus} is freely generated by a single element, called the {\it boundary Dehn twist}, which is, up to deformation, the diffeomorphism that at each radius is a circular rotation, radially interpolating between the zero-rotation on one boundary and the unit rotation on the other. 
\end{minipage}
\;\;
    \adjustbox{
      raise=-.4cm,
      scale=1.26
    }{
    \begin{tikzpicture}[scale=2]
      \draw[
        line width=1,
        fill=gray!70
      ]
        (0,0)
        ellipse
        (2*.5 and .5*.5);

     \def\resolution{6}

     \foreach \n in {0,..., \resolution} {
      \draw[
        line width=.01,
        green!30,
        -Latex
      ]
        (260:{2.4*(.17+.21-.21*\n/\resolution)} and 
        {.5*(.17+.21-.21*\n/\resolution)})
        arc
        (260:260+10+350*\n/\resolution:{2.4*(.17+.21-.21*\n/\resolution)} and 
        {.5*(.17+.21-.21*\n/\resolution)});
     }

      \draw[
        line width=1.3,
        green!30,
        fill=white
      ]
        (260:{2.4*.17} and 
        {.5*.17})
        arc
        (260:260+10+350:{2.4*(.17+.21-.21)} and 
        {.5*(.17+.21-.21)});
      \draw[
        line width=.4,
        green!30,
        -Latex,
        fill=white
      ]
        (260:{2.4*.17} and 
        {.5*.17})
        arc
        (260:260+10+350:{2.4*(.17+.21-.21)} and 
        {.5*(.17+.21-.21)});

    \end{tikzpicture}
    }
\end{center}

\medskip
\medskip

\noindent
{\bf Covariant flux monodromy.} With all this in hand, we come to the main statement to be proven in this section, announced as \eqref{SemidirectProductStructureInIntroduction}.

\begin{proposition}[\bf Extension of mapping class group by flux monodromy]
  \label{ExtensionOfMCGByFluxMonodromy}
  For every $\Sigma^2_{g,b,n}$ 
  \eqref{TheSurface}
  we have a split short exact sequence of groups
  \[
    \begin{tikzcd}[
      column sep=10pt
    ]
      1 
      \ar[
        r,
        shorten=-2pt
      ]
      &
      \pi_1
      \Big(
      \grayunderbrace{
      \mathrm{Map}
        ^\ast
        _0
      \big(
        (\Sigma^2_{g,b,n})_{\cpt}
        ,\,
        \hotype{A}
      \big)
      }{
        \scalebox{.7}{
          \color{gray}
          \rm
          moduli space
        }
      }
      \Big)
      \ar[
        rr,
        shorten=-2pt
      ]
      &&
      \pi_1
      \Big(
      \grayunderbrace{
      \mathrm{Map}
        ^\ast
        _0
      \big(
        (\Sigma^2_{g,b,n})_{\cpt}
        ,\,
        \hotype{A}
      \big)
      \sslash
      \mathrm{Diff}^{+,\partial}(
        \Sigma^2_{g,b,n}
      )
      }{
        \scalebox{.7}{
          \color{gray}
          \rm
          covariantized moduli space
          \eqref{CovariantizedModuliSpace}
        }
      }
      \Big)
      \ar[
        rr,
        shorten=-2pt
      ]
      &&
      \grayunderbrace{
      \mathrm{MCG}(\Sigma^2_{g,b,n})
      }{
        \mathclap{
        \scalebox{.7}{
          \color{gray}
          \rm
          mapping class group
          \eqref{MappingClassGroup}
        }
        }
      }
      \ar[
        ll,
        bend right=20
      ]
      \ar[
        r,
        shorten=-2pt
      ]
      &
      1
      \,,
    \end{tikzcd}
  \]
  exhibiting an action of the mapping class group on the fundamental group of the moduli space, so that we have the corresponding semidirect product:
  \begin{equation}
    \label{ActionOfMCGOnModuliMonodromy}
      \pi_1
      \Big(
      \grayunderbrace{
      \mathrm{Map}
        ^\ast
        _0
      \big(
        (\Sigma^2_{g,b,n})_{\cpt}
        ,\,
        \hotype{A}
      \big)
      \sslash
      \mathrm{Diff}^{+,\partial}(
        \Sigma^2_{g,b,n}
      )
      }{
        \scalebox{.7}{
          \color{gray}
          \rm
          covariantized moduli space
          \eqref{CovariantizedModuliSpace}
        }
      }
      \Big)
      \;\;\simeq\;\;\;
      \grayunderbrace{
      \mathrm{MCG}(\Sigma^2_{g,b,n})
      }{
        \mathclap{
        \scalebox{.7}{
          \color{gray}
          \rm
          mapping class group
          \eqref{MappingClassGroup}
        }
        }
      }
      \;\ltimes\;
      \pi_1
      \Big(
      \grayunderbrace{
      \mathrm{Map}
        ^\ast
        _0
      \big(
        (\Sigma^2_{g,b,n})_{\cpt}
        ,\,
        \hotype{A}
      \big)
      }{
        \scalebox{.7}{
          \color{gray}
          \rm
          moduli space
        }
      }
      \Big)\,.
  \end{equation}
\end{proposition}
\begin{proof}
For notational convenience, we abbreviate
\[
  \def\arraystretch{1.4}
  \begin{array}{l}
  \hotype{F} := 
  \mathrm{Map}
    ^\ast
    _0
  \big(
    (\Sigma^2_{g,b,n})_{\cpt}
    ,\,
    \mathcal{A}
  \big)
  \\
  \hotype{D} := 
  \mathrm{Diff}^{+,\partial}\big(
    \Sigma^2_{g,b,n}
  \big)
  \,,
  \end{array}
\]
whence the claim to be proven is split exactness of
\begin{equation}
  \label{AbbreviatedExtensionOfMCGBypi1F}
  \begin{tikzcd}[
    column sep=12pt
  ]
    1
    \ar[r]
    &
    \pi_1(\hotype{F})
    \ar[rr]
    &&
    \pi_1(\hotype{F}\sslash\hotype{D})
    \ar[rr]
    &&
    \pi_0(\hotype{D})
    \ar[ll, bend right=20]
    \ar[r]
    &1\,.
  \end{tikzcd}
\end{equation}
To this end, the Borel homotopy fiber sequence \eqref{BorelFiberSequence}
$$
  \begin{tikzcd}[
    column sep=43pt
  ]
  \hotype{F}
  \ar[r]
  &
  \hotype{F} \sslash \hotype{D}
  \ar[r]
  &
  \ast \sslash \hotype{D}
  \ar[
    l, bend right=20
  ]
  \end{tikzcd}
$$
(split by picking the zero-map)
induces a long exact sequence of homotopy groups \eqref{LongExactSequenceOfHomotopyGroups}
of this form:
\begin{equation}
\label{LESForFOverD}
\begin{tikzcd}[row sep=12pt, 
  column sep=20pt
]
  &
  &
  {}
  \ar[
    d,
    phantom,
    "{}"{coordinate, name=mid1}
  ]
  &
  \overset{
    \mathclap{
      \adjustbox{scale=.7, raise=+3pt}{
        \color{gray}
        by\eqref{HomotopyGroupsOfBG}
      }
    }
  }{
   \pi_1\big(
     \hotype{D}
   \big)
  }
  \ar[
    dll,
    rounded corners,
    to path={
      -- ([xshift=8pt]\tikztostart.east)
      |- (mid1)
      -| ([xshift=-10pt]\tikztotarget.west)
      -- (\tikztotarget)
    }
  ]
  \\
  &
  \pi_1(
    \hotype{F}
  )
  \ar[
    d,
    phantom,
    "{}"{coordinate, name=mid2}
  ]
  \ar[
    r
  ]
  & 
  \pi_1\big(
    \hotype{F} \sslash \hotype{D}
  \big)
  \ar[
    r
  ]
  &
  \pi_0\big(
    \hotype{D}
  \big)
  \ar[
    dll,
    rounded corners,
    to path={
      -- ([xshift=8pt]\tikztostart.east)
      |- (mid2)
      -| ([xshift=-10pt]\tikztotarget.west)
      -- (\tikztotarget)
    }
  ]
  \\
  &
  \pi_0
  (
    \hotype{F}
  )
  \ar[
    r,
    shorten=-2pt,
    "{ \sim }"
  ]
  & 
  \pi_0(
   \hotype{F}
    \sslash 
    \hotype{D}
  )  
  \mathrlap{\,.}
  \end{tikzcd}
\end{equation}
Here the last map shown is an isomorphism
by \eqref{TrivialComponentActionMeansHoQuotientPreservesPi0} (cf. footnote \ref{DiffeosPreserveHopfDegree}),
whence the exact sequence truncates to
$$
  \begin{tikzcd}
    \pi_1(\hotype{D})
    \ar[r]
    &
    \pi_1(\hotype{F})
    \ar[r]
    &
    \pi_1\big(\hotype{F} \sslash \hotype{D}\big)
    \ar[rr]
    &&
    \pi_0(\hotype{D})
    \ar[ll, bend right=20]
    \ar[r]
    &
    1
    \mathrlap{\,.}
  \end{tikzcd}
$$
If, at this point, we invoke Prop. \ref{HomotopyTypeOfDiffeomorphismGroup}
then the claim \eqref{AbbreviatedExtensionOfMCGBypi1F} follows for most surfaces, namely those for which $\pi_1(\hotype{D}) \simeq 1$.
But in fact, the claim follows generally by observing that the first connecting map in \eqref{LESForFOverD}
factors through the trivial group:
  \[
    \begin{tikzcd}[
      row sep=-3pt, column sep=30pt
    ]
      \pi_1(\hotype{D})
      \,\equiv\,
      \pi_1\big(
        \mathrm{Diff}^{+}(\Sigma^2_{g,b,n})
      \big)
      \ar[
        dr
      ]
      \ar[
        rr
      ]
      &&
      \pi_1\Big(
        \mathrm{Map}
          _0
          ^\ast
        \big(
          (\Sigma^2_{g,b,n})_{\cpt}
          ,\,
          \hotype{A}
        \big)
      \Big)
      \,\defneq\,
      \pi_1(\hotype{F})
      \mathrlap{\,.}
      \\
      &
      1
      \ar[
        ur
      ]
    \end{tikzcd}
  \]
  Namely, by \eqref{BorelFiberSequence}, the map is given by taking a given loop of diffeomorphisms to the loop of maps obtained by composing these diffeos with the constant map $(\Sigma^2_{g,b,n})_{\cpt} \xrightarrow{\;} \hotype{A}$ --- but that gives the constant loop representing the neutral element of $\pi_1$.  
\end{proof}

This Proposition \ref{ExtensionOfMCGByFluxMonodromy} is our main tool for analyzing the covariantized topological quantum states on $\hotype{A}$-quantized flux according to \eqref{StateSpacesAsLocalSystemsOnModuliSpace}. In the next section, we specialize $\hotype{A}$ to $S^2$ and work out the consequences.


\section{Flux quantized in 2-Cohomotopy}
\label
{2CohomotopicalFluxThroughSurfaces}

We now specify the classifying space
$\hotype{A}$ \eqref{StateSpacesAsLocalSystemsOnModuliSpace}
to the 2-sphere, $\hotype{A} \defneq S^2$ so that flux is classified by the non-abelian cohomology theory called {\it 2-Cohomotopy}
\footnote{
  \label{TwistedCohomotopy}
  More precisely, we expect that flux in FQH systems is classified by the {\it tangentially twisted} version of 2-cohomotopy according to \cite[Ex. 3.8]{FSS23-Char}, but in the present context the tangential twisting is trivial in the cases of main interest, namely on the plane and on the torus. The surfaces on which tangential twisting would be relevant are just those on which experimental realizability of FQH systems is dubious. If and when this experimental situation changes, a dedicated discussion of the tangentially twisted analogues of the results in \S\ref{OnClosedSurfaces} would be called for. It is fairly clear how this will proceed: The relevant tangentially twisted version of the Pontrjagin theorem (recalled here as Prop. \ref{PontrjaginTheorem}) is given in \cite[\S 2]{FSS20-H} and the tangentially twisted version of the May-Segal theorem is reviewed as \cite[Thm. 4.2]{Kallel24}.
}
\eqref{Cohomotopy},
and we work out (according to \S\ref{FluxObservablesAndClassifyingSpaces}) the resulting covariant topological quantum observables on and quantum states of 2-cohomotopically quantized flux through various surfaces $\Sigma^2$, using the results of \S\ref{ModuliSpacesOfTopologicalFlux}. 

\smallskip 
Remarkably, in the case of $\Sigma^2 \defneq S^2$ the sphere or $\Sigma^2 \defneq T^2$ the torus, we find reproduced (in \S\ref{2CohomotopicalFluxThroughPlane} and \S\ref{2CohomotopicalFluxThroughTorus}, respectively) 
the situation traditionally argued via quantized $\mathrm{U}(1)$-Chern-Simons theory over these surfaces, including fine-print such as regularization of Wilson-loop observables by framings, modular equivariance and refinement to ``spin'' Chern-Simons theory.

\smallskip 
Then, by instead choosing punctured surfaces, we similarly work out the 2-Cohomotopically quantized flux through the annulus (\S\ref{OnTheOpenAnnulus}) and the 2-punctured disk (\S\ref{QBitQuantumGatesOperableOn2PuncturedOpenDisk}).
\medskip

\begin{definition}[{\bf Cohomotopy}, {cf. \cite{Spanier49}\cite[\S VII]{STHu59}\cite[Ex. 2.7]{FSS23-Char}}]
  \label{OnCohomotopy}
  The generalized non-abelian cohomology theory \cite[\S 2]{FSS23-Char} whose classifying spaces are the $n$-spheres $S^n$ is called {\it Cohomotopy}
  \footnote{
    The original literature \cite{Borsuk36}\cite{Spanier49} on Cohomotopy (also \cite{Pontrjagin38}, without using that terminology) is focused on equipping the Cohomotopy sets $\pi^n(X)$ with group structure, which is possible when $X$ is a CW-complex of dimension $\leq 2n-2$, and hence speaks of {\it cohomotopy groups} in these special situations. For the present purpose this group structure plays no role but instead the perspective of generalized non-abelian cohomology is the natural one: Also ordinary non-abelian cohomology \cite{Grothendieck55} $\widetilde{H}^1(X;\, G) \,\simeq\, \pi_0 \, \mathrm{Map}^\ast\big(X,\, B G \big)$ 
    (cf. \cite[Ex. 2.2]{FSS23-Char}\cite[Thm. 4.1.13]{SS25-EBund})
    has no group structure (unless $G$ happens to be abelian).
  }
  , 
  denoted
  \begin{equation}
    \label{Cohomotopy}
    \widetilde{\pi}^n(X)
    \;:=\;
    \pi_0
    \,
    \mathrm{Map}^\ast\big(
      X
      ,\,
      S^n
    \big)
    \,,
   \mathrlap{
    \;\;\;\;\;\;\;
    \mbox{for $X \in \mathrm{Top}^\ast$}
    \,.
    }
  \end{equation}
Here the terminology and notation indicate the ``duality'' with the {\it homotopy} groups
$
  \pi_n(X) \,\simeq\,
  \pi_0\,
  \mathrm{Map}^\ast\big(S^n,\, X\big)
$.
\end{definition}

\begin{remark}[\bf Relation to the $\mathrm{O}(3)$/$\mathbb{C}P^1$ sigma-model]
  \label{RelationToCP1Model}
  Taking flux-quantization
  \eqref{TopologicalObservablesMoreGenerally} to be in 2-Cohomotopy,
  $\hotype{A}\defneq S^2$ (Def. \ref{OnCohomotopy}), so that topological
  charges are represented by continuous maps $\Sigma^2\xrightarrow{}S^2$,
  has some similarity with the classification of solitons (or \emph{lumps})
  in the $\mathbb{C}P^1$ model, equivalently the $\mathrm{O}(3)$ nonlinear
  sigma-model with target $S^2$. Its fields are \emph{smooth} maps to $S^2$
  (cf. \cite[\S VI.A]{Forte92}\cite[\S 6.1]{MantonSutcliffe04}).
  Planar solitons with a fixed value at spatial infinity carry the degree
  charge in $\pi_2(S^2)\simeq\mathbb{Z}$. The ``Hopf term''
  (cf. \cite[(6)]{WilzcekZee83}\cite[\S II.C]{Forte92}
  \cite{ChakrabortyGoshMalik01}), essentially the abelian Chern-Simons term
  \eqref{EffectiveCSLagrangian}, instead measures topology of their
  spacetime histories, as discussed in Rem. \ref{RelationToHopfions} below.
  See also the discussion of continuous 2-Cohomotopy in the
  Skyrme-Faddeev model by \cite{AK05,AK10}.
  
  However, our re-identification of such maps $\Sigma^2 \xrightarrow{\;} S^2$ as classifying exotically quantized flux according to \S\ref{FluxObservablesAndClassifyingSpaces}, \cite{SS25-Flux}, puts these arguments in a coherent perspective concerning flux quanta in FQH systems, cf. the following Ex. \ref{2CohomotopyOfSurfaces}.
\end{remark}

\begin{remark}[\bf Generalized higher symmetry group of 2-cohomotopical flux]
\begin{itemize} 
\item[\bf (i)] In view of Rem. \ref{HigherAndGeneralizedSymmetry}, the choice of classifying space $\hotype{A} \,\defneq\, S^2 \,\shapeEquivalence\, B (\Omega S^2)$ corresponds to considering as (homotopy type of the) gauge group the loop group $\Omega S^2$  of the 2-sphere (under concatenation and reversal of loops) which is a ``higher group'' (``$\infty$-group'') exhibiting ``generalized symmetry''.   
  See also footnote \ref{MoravaOnLoopSpaceOf2Sphere} below.

\item[\bf (ii)]   The looping of the canonical comparison map $1^2 : S^2 \xrightarrow{\;} B^2 \mathbb{Z}$ exhibits this generalized symmetry group as a deformation of (the homotopy type of) the standard electromagnetic gauge group $\mathrm{U}(1)$:
  $$
    \begin{tikzcd}
      \Omega S^2
      \ar[
        rr,
        "{
          \Omega 1^2
        }"
      ]
      &&
      \Omega B^2 \mathbb{Z}
      \,\shapeEquivalence\,
      B \mathbb{Z}
      \,\shapeEquivalence\,
      \mathrm{U}(1)
      \,.
    \end{tikzcd}
  $$
  This map induces an isomorphism on $\pi_1$, 
  but while $\pi_{> 1}\big(\mathrm{U}(1)\big) \,\simeq\, 0$, the deformation $\Omega S^2$ on the left has non-trivial homotopy groups in arbitrarily high degree, in particular 
  a non-finite contribution
  $$
    \pi_2\big(
      \Omega S^2
    \big)
    \underset{
      \eqref{SpheresFromSmashProducts}
    }{
      \;\simeq\;
    }
    \pi_3(S^2)
    \;\simeq\;
    \mathbb{Z}
  $$
  generated by the Hopf fibration.
  We already remarked after \eqref{Rational2SphereAsCSFluxModel} that this is the homotopical avatar of the Chern-Simons form, and we will see now that it is also the origin of the appearance of anyonic braiding phases of flux solitons quantized in 2-Cohomotopy: This is seen in Prop. \ref{2CohomotopicalFluxMonodromyOnPlane}, Prop. \ref{MonodromyOfFluxThroughClosedSurface}  and Prop. \ref{IdentifyingCentralGeneratorAcrossSurfaces} below.
  \end{itemize}
\end{remark}

\subsection{On the plane}
\label{2CohomotopicalFluxThroughPlane}

We recall here (from \cite{SS24-AbAnyons}) how solitonic flux through the plane $\mathbb{R}^2 \,\simeq\,\Sigma^2_{0,0,1}$ \eqref{ExamplesOfSurfaces} quantized in 2-cohomotopy reproduces the Wilson loop observables of anyonic links as predicted by abelian Chern-Simons theory (Rem. \ref{ComparisonWithCSOnPlane}
below).
But to start with, we briefly recall the Pontrjagin construction that serves to relate Cohomotopy charge to solitonic flux density in general.

\medskip

\noindent
{\bf 2-Cohomotopical flux solitons via the Pontrjagin construction.}
Among generalized non-abelian cohomology theories, (unstable) {\it Cohomotopy} $\pi^n$ (cf. \cite{Pontrjagin38}\cite{Spanier49}\cite[\S VII]{STHu59}\cite[Ex. 2.7]{FSS23-Char}), 
whose classifying spaces are the $n$-spheres $S^n \simeq \mathbb{R}^n_{\cpt}$ \eqref{CompactificationOfRn},
$$
  \widetilde{\pi}^n(-) 
  \,:=\, 
  \pi_0 \, \mathrm{Map}^\ast\big(-,\mathbb{R}^n_{\cpt}\big)
  \,,
$$
stands out in that it accurately characterizes the solitonic flux configurations of given charge \cite{SS20-Tad}\cite{SS24-AbAnyons} --- this may be understood as the content of the original unstable Pontrjagin theorem (which these days is more famous as the {\it Pontrjagin-Thom theorem} pertaining only to the {\it stable} case which is of little concern to us here):

\begin{proposition}[{\bf Pontrjagin theorem -- Cohomotopy charge}, {cf. \cite[\S II.16]{Bredon93}\cite[\S IX]{Kosinski93}}]
  \label{PontrjaginTheorem}
  Given a smooth $d$-manifold $\Sigma^d$ without boundary and $n \in \mathbb{N}_{>0}$ with $n \leq d$, there is a natural bijection between:
  \begin{itemize}
  \item[1.] the reduced $n$-Cohomotopy of the one-point compactification $\Sigma^d_{\cpt}$,

  \item[2.] the classes of compact normally framed submanifolds $Q^{d-n} \xhookrightarrow{\;} \Sigma^d$ of co-dimension=$n$, under compact normally framed cobordisms in $\Sigma^d\times[0,1]$
  \end{itemize}

  \begin{equation}
    \label{PontrjaginIsomorphism}
    \begin{tikzcd}[column sep=30pt]
      \scalebox{.7}{
        \color{darkblue}
        \bf
        \def\arraystretch{.9}
        \begin{tabular}{c}
          Reduced $n$-Cohomotopy
          \\
          of 1pt compactification
        \end{tabular}
      }
      \widetilde{\pi}^n\big(
        \Sigma^d_{\cpt}
      \big)
      \ar[
        rr,
        shift left=2,
        "{
          \scalebox{.7}{
            \color{darkgreen}
            \bf
            regular pre-image of \scalebox{1.4}{$0$}
          }
        }"
      ]
      \ar[
        rr,
        phantom,
        "{ \sim }"
      ]
      \ar[
        from=rr,
        shift left=2,
        "{
          \scalebox{.7}{
            \color{darkgreen}
            \bf
            cohomotopy charge
          }
        }"{yshift=-1pt}
      ]
      &\phantom{----}&
      \mathrm{Cob}^n_{\mathrm{Fr}}\big(
        \Sigma^d
      \big)
      \scalebox{.7}{
        \color{darkblue}
        \bf
        \def\arraystretch{.9}
        \begin{tabular}{c}
          Cobordism classes of
          \\
          normally framed sub-
          \\
          manifolds of codim\scalebox{1.2}{$=n$}\rlap{,}
        \end{tabular}
      }
    \end{tikzcd}
  \end{equation}
  Here the \emph{Cohomotopy charge} $[c] \in \widetilde{\pi^n}\big(\Sigma^d_{\cpt}\big)$ of $Q^{d-n} \subset \Sigma^d$ with normal framing 
  $
    \begin{tikzcd}[sep=10pt]
      N Q 
        \ar[
          r, 
          shorten=-2pt, 
          "{ \mathrm{fr} }", 
          "{\sim}"{swap, pos=.35}
        ] 
      & 
      Q \times \mathbb{R}^n 
        \ar[
          ->>,
          r, 
          shorten=-2pt,
          "{ p }"{pos=.35}
        ] 
      & 
      \mathbb{R}^n
    \end{tikzcd}
  $ 
    is represented for any choice of tubular neighborhood $N Q \xhookrightarrow{ \iota } \Sigma$ by the ``scanning map''
  $$
    \begin{tikzcd}[
      sep=0pt,
      ampersand replacement=\&
    ]
      \Sigma^d 
      \ar[rr, "{ c }"]
        \&\& 
      \mathbb{R}^{n}_{\cpt}
      \\
      s 
        \&\longmapsto\& 
      \left\{\!\!
      \def\arraystretch{1.1}
      \begin{array}{ccl}
        p\big(\mathrm{fr}(s)\big)
        & \vert &
        s \in \iota(N Q)
        \\
        \infty 
        &\vert&
        \mathrm{otherwise}
        \mathrlap{\,.}
      \end{array}
      \right.
    \end{tikzcd}
    \quad 
    \adjustbox{
      raise=-1.4cm
    }{
    \begin{tikzpicture}
      \node at (0,0) {
        \includegraphics[width=9.2cm]
        {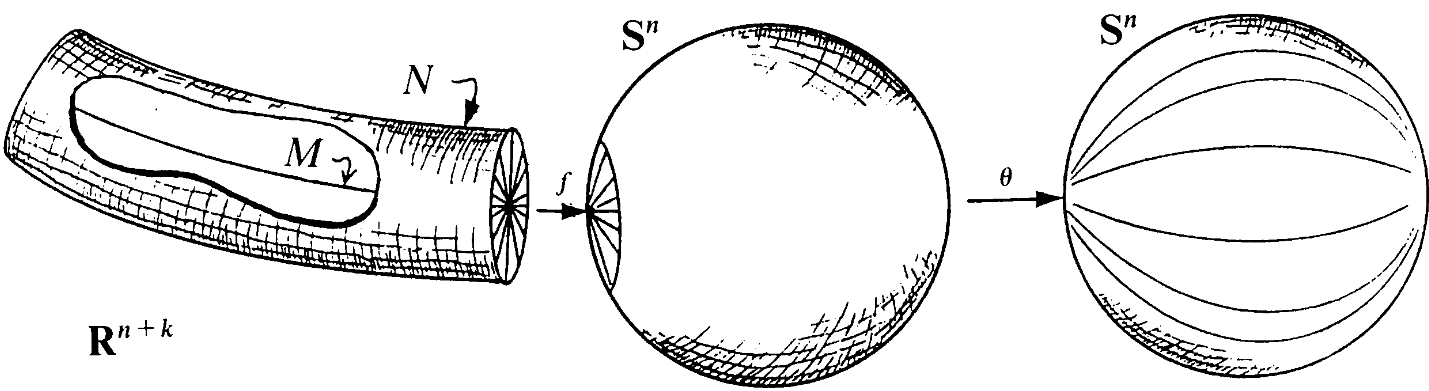}
      };

      \draw[
        draw opacity=0,
        fill=white
      ]
        (2.58,1.08) circle 
        (.17);
      \draw[
        draw opacity=0,
        fill=white
      ]
        (-.5,1.08) circle 
        (.17);
     \node[
       scale=.6
     ] 
       at (1.85,.8) {
       $\mathbb{R}^d_{\cpt}
       \,\simeq\, S^d$
     };

      \node[
        rotate=-90,
      ] at (5,0) {
       \scalebox{.65}{
        \color{gray}\rm
        \def\arraystretch{.9}
        \begin{tabular}{c}
         figure adapted 
         \\
         from \cite[Fig. II-13]{Bredon93}
        \end{tabular}
        }
      };

      \draw[
        draw opacity=0,
        fill=white
      ]
        (-1,0.085) circle 
        (.13);
      \draw[
        draw opacity=0,
        fill=white
      ]
        (1.87,0.15) circle 
        (.13);

      \draw[
        draw opacity=0,
        fill=white
      ]
        (-1.9,0.75) circle 
        (.15);
     \node[
       scale=.6
     ] 
       at (-1.9,0.75)
     { $N Q$ };

      \draw[
        draw opacity=0,
        fill=white
      ]
        (-2.65,0.28) circle 
        (.15);
    \node[scale=.6] at 
      (-2.58,0.24) 
    { $Q$ };

      \draw[
        draw opacity=0,
        fill=white
      ]
        (-3.76,-.9) circle 
        (.3);

    \draw[
      line width=4,
      white
    ]
      (-1.15,-.1-.1) .. controls
      (-1,-.35) and 
      (1.6,-.4) ..
      (2.2,0-.1);
    \draw[
      line width=.5,
      -Latex
    ]
      (-1.15,-.1-.1) .. controls
      (-1,-.35) and 
      (1.6,-.4) ..
      (2.2,0-.1);
    \node[scale=.6]
      at (1.75,-.36) { $c$ };

   \node[scale=.6] 
     at (-3.7,-.8)
     { $\Sigma$ };
    
    \end{tikzpicture}
    }
  $$
 \end{proposition}

\smallskip
\noindent
\begin{minipage}{7.6cm}
\begin{example}[\bf 2-Cohomotopy of surfaces]
\label{2CohomotopyOfSurfaces}
For surfaces without boundary, with flux quantized in 2-Cohomotopy, we have $d = n = 2$ in Prop. \ref{PontrjaginTheorem}, whence the Pontrjagin theorem identifies solitonic flux concentrated around {\it points} (soliton cores) with oriented tubular neighborhoods reflecting either positive or negative units of (magnetic) flux.

The {\it total flux} is the sum of the charges of these solitons (the orientations $\in \{\pm 1\}$ of their tubular neighborhoods), identified with the {\it Hopf degree} (Def. \ref{HopfDegree}) of the classifying map.

\end{example}
\end{minipage}
\hspace{.4cm}
 \adjustbox{
    raise=-2.3cm,
    scale=1
  }{
   \begin{tikzpicture}
\begin{scope}[
  scale=.8,
  shift={(.7,-4.9)}
]

\draw[
  line width=.8,
  ->,
  darkgreen
]
  (1.5,1) 
  .. controls (2,1.6) and (3,2.6) .. 
  (6,1);

\node[
  scale=.7,
  rotate=-18
] at (4.7,1.8) {
  \color{darkblue}
  \bf
  classifying map
};

\node[
  scale=.7,
  rotate=-18
] at (4.5,1.4) {
  \color{darkblue}
  \bf
  $n$
};

  \shade[
    right color=gray, left color=lightgray,
    fill opacity=.9
  ]
    (3,-3)
      --
    (-1,-1)
      --
        (-1.21,1)
      --
    (2.3,3);

  \draw[dashed]
    (3,-3)
      --
    (-1,-1)
      --
    (-1.21,1)
      --
    (2.3,3)
      --
    (3,-3);

  \node[
    scale=1
  ] at (3.2,-2.1)
  {$\infty$};

  \begin{scope}[rotate=(+8)]
  \shadedraw[
    dashed,
    inner color=olive,
    outer color=lightolive,
  ]
    (1.5,-1)
    ellipse
    (.2 and .37);
  \draw
   (1.5,-1)
   to 
    node[above, yshift=-1pt]{
     \;\;\;\;\;\;\;\;\;\;\;
     \rotatebox[origin=c]{7}{
     \scalebox{.7}{
     \color{darkorange}
     \bf
       anyon
     }
     }
   }
    node[below, yshift=+6.3pt]{
     \;\;\;\;\;\;\;\;\;\;\;\;
     \rotatebox[origin=c]{7}{
     \scalebox{.7}{
     \color{darkorange}
     \bf
       worldline
     }
     }
   }
   (-2.2,-1);
  \draw
   (1.5+1.2,-1)
   to
   (4,-1);
  \end{scope}

  \begin{scope}[shift={(-.2,1.4)}, scale=(.96)]
  \begin{scope}[rotate=(+8)]
  \shadedraw[
    dashed,
    inner color=olive,
    outer color=lightolive,
  ]
    (1.5,-1)
    ellipse
    (.2 and .37);
  \draw
   (1.5,-1)
   to
   (-2.3,-1);
  \draw
   (1.5+1.35,-1)
   to
   (4.1,-1);
  \end{scope}
  \end{scope}
  \begin{scope}[shift={(-1,.5)}, scale=(.7)]
  \begin{scope}[rotate=(+8)]
  \shadedraw[
    dashed,
    inner color=olive,
    outer color=lightolive,
  ]
    (1.5,-1)
    ellipse
    (.2 and .32);
  \draw
   (1.5,-1)
   to
   (-1.8,-1);
  \end{scope}
  \end{scope}

\end{scope}

\node[
  scale=.73,
  rotate=-27
] at (2.33,-5.5) {
  \color{darkblue}
  \bf
  \def\arraystretch{.9}
  \begin{tabular}{l}
    flux
    \\
    $F_2 =$ 
    \\
    $\;\;n^\ast(\mathrm{dvol}_{S^2})$
  \end{tabular}
};

\node[
  scale=.73,
] at (2.05,-5) {
  \color{darkblue}
  \bf
};

\node[
  rotate=-140
] at (6,-4) {
  \includegraphics[width=2cm]
  {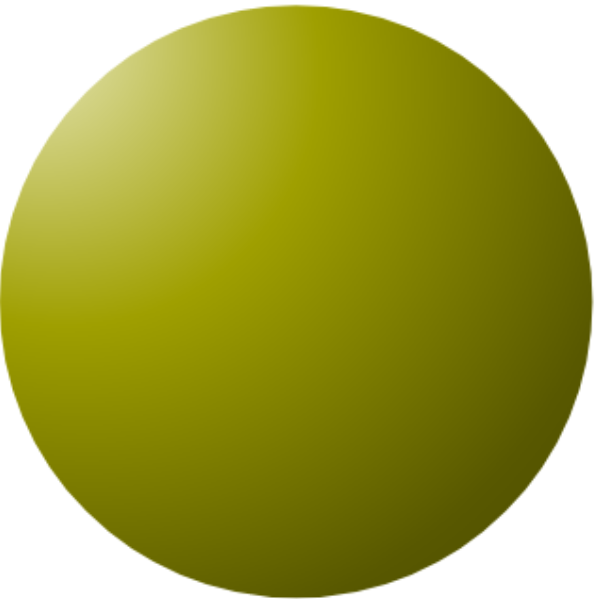}
};

\node[
  scale=.7
] at (6,-5.2) {
  \color{darkblue}
  \bf
  2-sphere $S^2$
};
    
\end{tikzpicture}
}

\begin{definition}[{\bf Hopf degree}, {cf. \cite[\S IX, Cor 5.8]{Kosinski93}}]
\label{HopfDegree}
For $n \in \mathbb{N}_{>0}$,
and 
\begin{equation}
  \label{BnZUnitClass}
  \begin{tikzcd}
    S^n 
    \ar[r, "{ 1^n }"]
    &
    B^n \mathbb{Z}
  \end{tikzcd}
\end{equation}
a map representing the generator $1 \in \mathbb{Z} \simeq \pi_n(B^n \mathbb{Z})$, the induced generalized cohomology operation from $n$-Cohomotopy to ordinary integral $n$-cohomology
\begin{equation}
  \label{OperationFromCohomotopyToCohomology}
  \begin{tikzcd}
    \pi^n(X)
    \;\defneq\;
    \pi_0\, 
    \mathrm{Map}\big(
      X
      ,\,
      S^n
    \big)
    \ar[
      rr,
      "{
        \pi_0
        (1^n)_\ast
      }"
    ]
    &&
    \pi_0\, 
    \mathrm{Map}\big(
      X
      ,\,
      B^n \mathbb{Z}
    \big)
    \;\simeq\;
    H^n(X;\, \mathbb{Z})
  \end{tikzcd}
\end{equation}
is a bijection when $X$ is a connected closed oriented manifold of
dimension $n$. Evaluation on its fundamental class identifies
$H^n(X;\mathbb{Z})$ with $\mathbb{Z}$; the resulting integer is the
{\it Hopf degree} of a map $X\to S^n$.

For a connected oriented $n$-manifold $X$ without boundary, the
corresponding compactly supported degree uses the one-point compactification:
\[
  \widetilde{\pi}^n(X_{\cpt})
  \;\simeq\; \widetilde H^n(X_{\cpt};\mathbb{Z})
  \;\simeq\; H^n_{\mathrm{cpt}}(X;\mathbb{Z})
  \;\simeq\; \mathbb{Z}.
\]
Here the spaces considered are of finite type, as in \eqref{TheSurface}.
For example, maps $\mathbb{R}^2\to S^2$ are null-homotopic;
the integer degree arises after compactification.

For a surface with boundary, compactification at infinity does not fix
the field on that boundary. Thus the compact disk has
$\widetilde\pi^2((D^2)_{\cpt})=0$. Fixing the boundary value instead gives
$[D^2/\partial D^2,S^2]_\ast\simeq\mathbb{Z}$, the relative degree.
The punctured closed disks in \S\ref{OnPuncturedDisks} retain the former
boundary convention: the field is unrestricted on the physical boundary.
\end{definition}

\smallskip

On the other hand, the flux density underlying (sourced by) a given Cohomotopy charge is characterized by the cohomotopical character map (the cohomotopical analog of the Chern-character map on K-cohomology, \cite{FSS23-Char}\cite{SS25-Flux}):
\begin{definition}[\bf Cohomotopical character map]
For $n = d$ the {\it character map} on cohomotopy 
\begin{equation}
  \label{ScanningMap}
  \begin{tikzcd}[sep=0pt]
    \widetilde{\pi}^d\big(
      \Sigma^d_{\cpt}
    \big)
    \ar[
      rr,
      "{
        \mathrm{ch}
      }"
    ]
    &&
    H^d_{\mathrm{dR}}\big(
      \Sigma^d
    \big)_{\mathrm{cpt}}
    \\
    {[c]}
    &\longmapsto&
    \big[
      c^\ast \mathrm{vol}_n
    \big]
  \end{tikzcd}
\end{equation}
takes $[c] \in \widetilde{\pi}^n\big(\Sigma^d_{\cpt}\big)$ --- for any representative $c : \Sigma^d_{\cpt} \xrightarrow{\;} \mathbb{R}^d_{\cpt}$ which is smooth on $c^{-1}(\mathbb{R}^d)$, such as the scanning maps in Prop. \ref{PontrjaginTheorem} --- to the class in compactly supported de Rham cohomology of the pullback of a $d$-form $\mathrm{vol} \in \Omega^n_{\mathrm{dR}}(\mathbb{R}^d)$ compactly supported on a neighborhood of $0 \in \mathbb{R}^d$ and of unit integral.
\end{definition}

\medskip 
\begin{remark}[\bf Flux density quantized in Cohomotopy]
\label{FluxDensityQuantizedInCohomotopy}
\begin{itemize} 
\item[{\bf (i)}] For a connected oriented $d$-manifold $\Sigma^d$ without boundary, this means that Cohomotopy charge $[c] \in
 \mathbb{Z} \simeq \widetilde \pi^d\big(\Sigma^d_{\cpt}\big) \equiv \pi_0 \mathrm{Map}^\ast\big( \Sigma^d_{\cpt},\, S^d \big)$ may be understood as sourcing a solitonic flux density $F_d \in \Omega^d_{\mathrm{dR}}\big(\Sigma^d\big)$ (solitonic in that it vanishes at infinity) which is supported with unit weight near $n_+ \in \mathbb{N}$ points in $\Sigma^d$ (all points outside each other's supporting neighborhoods) and with a negative unit weight near $n_- \in \mathbb{N}$ points (anti-solitons) such that $[c] = n_+ - n_-$.

 \item[{\bf (ii)}]
 For the case $d = 2$ of interest here, this is just the kind of magnetic flux distribution concentrated around solitonic vortex cores as seen in type II superconducting and in fractional quantum Hall semiconducting materials $\Sigma^2$, while any punctures in the surface $\Sigma^2$ \eqref{TheSurface} behave as loci from which flux is expelled (cf. \hyperlink{FigureI}{\it Fig. I}):
\end{itemize} 

 \smallskip 
 \hypertarget{FigureFluxFromPontrjagin}{}
 \noindent
 \begin{minipage}{3.9cm}
   \footnotesize
   {\bf Figure F.} Via the Pontrjagin theorem, 2-cohomotopical quantization of flux through a surface exhibits $N$ {\it flux quanta} as a concentration of flux density supported on the tubular neighborhoods of  $N$ disjoint points.
 \end{minipage}
 \adjustbox{
   raise=-.8cm, scale=0.93
 }{
 \begin{tikzpicture}
\draw[
  dashed,
  fill=lightgray
]
  (0,0)
  -- (8,0)
  -- (10,2)
  -- (2.8,2)
  -- cycle;

\begin{scope}[
  shift={(2.4,.5)}
]
\shadedraw[
  draw opacity=0,
  inner color=olive,
  outer color=lightolive
]
  (0,0) ellipse (.7 and .3);
\end{scope}

\begin{scope}[
  shift={(6.3,.6)}
]
\shadedraw[
  draw opacity=0,
  inner color=olive,
  outer color=lightolive
]
  (0,0) ellipse (.7 and .3);
\end{scope}

\begin{scope}[
  shift={(4.5,1.5)}
]
\shadedraw[
  draw opacity=0,
  inner color=olive,
  outer color=lightolive
]
  (0,0) ellipse (.7 and .25);
\end{scope}

\begin{scope}[shift={(4.8,.7)}]
\draw[
  fill=white,
  draw=gray
]
  (0,0) ellipse 
  (.17 and 
  0.5*0.17);
\end{scope}

\begin{scope}[shift={(7.5,1.2)}]
\draw[
  fill=white,
  draw=gray
]
  (0,0) ellipse 
  (.17 and 
  0.5*0.17);
\end{scope}

\draw[
  white,
  line width=2
]
  (-1.4, .7)
  .. controls (1,2) and (2,2) ..
  (2.32,.7);
\draw[
  -Latex,
  black
]
  (-1.4, .7)
  .. controls (1,2) and (2,2) ..
  (2.32,.7);

\node
  at (-1.1,.7)
  {
    \adjustbox{
      bgcolor=white,
      scale=.7
    }{
      \color{darkblue}
      \bf
      \def\arraystretch{.9}
      \def\tabcolsep{-5pt}
      \begin{tabular}{c}
        field \color{purple}solitons/
        \\
        quasi-particles/ 
        \\
        -holes/vortices:
        \\
        frmd submanifolds
      \end{tabular}
    }
  };

\draw[
  white,
  line width=2
]
  (10,1.1) 
  .. controls 
  (9.5,1.8) and 
  (8.5,1.8) ..
  (7.7,1.3);
\draw[
  -Latex
]
  (10,1.1) 
  .. controls 
  (9.5,1.8) and 
  (8.5,1.8) ..
  (7.7,1.3);

\node
  at (10.2,.5)
  {
    \adjustbox{
      scale=.7
    }{
      \color{darkblue}
      \bf
      \def\arraystretch{.9}
      \def\tabcolsep{-5pt}
      \begin{tabular}{c}
        flux-expelling 
        {\color{purple}defects}:
        \\
        punctures in the surface
      \end{tabular}
    }
  };
  
\end{tikzpicture}
}

\end{remark}

\medskip

\noindent
{\bf 2-Cohomotopical flux monodromy.}
For the quantum flux observables \eqref{SolitonicTopologicalObservablesMoreGenerally}, we need not just the connected components $\pi_0$ but the fundamental group $\pi_1$ of the moduli space of Cohomotopical flux, which we may understand as ``{\it 2-Cohomotopy in negative degree 1}'', classified by the loop space $\Omega S^2$ of the 2-sphere: 
\footnote{
  \label{MoravaOnLoopSpaceOf2Sphere}
  The loop space $\Omega S^2$ \eqref{2CohomotopyInNegativeDegree1} of the 2-sphere
  has received attention as a classifying space also in \cite[Def. 1.1]{Morava23}, there called the classifying space for ``{\l}ine bundles'' (with a Polish ``{\l}''), and has been related to configuration spaces of points in \cite{MoravaRolfsen23}\cite{Morava23}, reminiscent of the role they play for us in relating to group-completed configuration spaces in \S\ref{2CohomotopicalFluxThroughPlane} below.

  In \cite[p. 94]{CohenWu03} the homotopy groups of $\Omega S^2$ are recognized as natural sub-quotients of braid groups, which is a tantalizing observation in our context, whose further relevance however remains unclear to us at this point.
}
\begin{equation}
  \label{2CohomotopyInNegativeDegree1}
  \pi_1\, 
  \mathrm{Map}^\ast\big(
    -
    ,\,
    S^2
  \big)
  \underset{
    \eqref{FundamentalGroupOfPointedMappingSpace}
  }{
  \;\;
  \simeq
  \;\;
  }
  \pi_0\, 
  \mathrm{Map}^\ast\big(
    -
    ,\,
    \Omega S^2
  \big)
  \,.
\end{equation}
Just like the 2-sphere has a canonical comparison map $1^2 : S^2 \xrightarrow{\;}B^2 \mathbb{Z}$ \eqref{BnZUnitClass} whose induced cohomology operation \cite[Def. 2.3]{FSS23-Char} extracts ordinary 2-cohomology classes from 2-cohomotopy charges
$$
  \begin{tikzcd}
    \widetilde{\pi}^2(-)
    \;\simeq\;
    \pi_0\, \mathrm{Map}^\ast\big(
      -
      ,\,
      S^2
    \big)
    \ar[
      rr,
      "{ (1^2)_\ast }"
    ]
    &&
    \pi_0\, \mathrm{Map}^\ast\big(
      -
      ,\,
      B^2 \mathbb{Z}
    \big)
    \;\simeq\;
    \widetilde{H}^2(-;\, \mathbb{Z})
    \,,
  \end{tikzcd}
$$
so its loop space has the looped comparison map $\Omega 1^2 : \Omega S^2 \xrightarrow{\;} \Omega B^2 \mathbb{Z} \shapeEquivalence B \mathbb{Z}$ inducing the cohomology operation
$$
  \begin{tikzcd}
    \pi_0\, \mathrm{Map}^\ast\big(
      -
      ,\,
      \Omega S^2
    \big)
    \ar[
      rr,
      "{ (\Omega 1^2)_\ast }"
    ]
    &&
    \pi_0\, \mathrm{Map}^\ast\big(
      -
      ,\,
      B \mathbb{Z}
    \big)
    \;\simeq\;
    \widetilde{H}^1(-;\, \mathbb{Z})
  \end{tikzcd}
$$
which makes precise how 2-cohomotopical flux observables refine ordinary electromagnetic flux observables \eqref{OrdinaryFluxObservablesOnClosedSurface}.

It is now immediate to compute the observables on  covariantized 2-Cohomotopical solitonic flux on the plane, and there turns out to be essentially a single such observable (Prop. \ref{2CohomotopicalFluxMonodromyOnPlane} below),  to be denoted $\widehat{\zeta}$ ---
but it will take us the better part of the remainder of this section to identify this observable with the braiding phase \eqref{TheBraidingPhase}.

\begin{proposition}[\bf 2-Cohomotopical flux monodromy on the plane]
  \label{2CohomotopicalFluxMonodromyOnPlane}
  The spaces $\HilbertSpace{H}$ of topological quantum states
 \eqref{StateSpacesAsLocalSystemsOnModuliSpace}
  of solitonic flux quantized in 2-Cohomotopy on the plane are representations of the group of integers:
  \begin{equation}
    \label{IntegerObservablesOnPlane}
    \pi_1\Big(
      \mathrm{Map}^\ast_0\big(
        \mathbb{R}^2_{\cpt}
        ,\,
        S^2
      \big)
      \sslash
      \mathrm{Diff}^+(\mathbb{R}^2)
    \Big)
    \;\simeq\;
    \mathbb{Z}
    \,,
  \end{equation}
  hence defined by a single unitary operator $\widehat{\zeta}$:
  \vspace{-2mm} 
  \begin{equation}
    \label{BraidObservableOnPlane}
    \begin{tikzcd}[row sep=-2pt, column sep=0pt]
      \mathbb{Z}
      \ar[rr]
      &&
      \mathrm{U}(\HilbertSpace{H})
      \\
      n 
        &\longmapsto&
      \big(\widehat{\zeta}\,\big)^n
      \mathrlap{\,.}
    \end{tikzcd}
  \end{equation}
\end{proposition}
\begin{proof}
  The mapping class group of the plane is trivial (Prop. \ref{HomotopyTypeOfDiffeomorphismGroup}), so that by Prop. \ref{ExtensionOfMCGByFluxMonodromy} the only contribution is from the flux monodromy group itself, which is readily found to be 
  $$
    \def\arraystretch{1.2}
    \begin{array}{rcll}
    \pi_1\,
      \mathrm{Map}^\ast_0\big(
        \mathbb{R}^2_{\cpt}
        ,\,
        S^2
      \big)
      &\simeq&
      \pi_1
      \,
      \mathrm{Map}^\ast_0\big(
        S^2
        ,\,
        S^2
      \big)
      &
      \proofstep{
        by
        \eqref{CompactificationOfRn}
      }
      \\
      &\simeq&
      \pi_0
      \,
      \mathrm{Map}^\ast\big(
        S^3
        ,\,
        S^2
      \big)
      &
      \proofstep{
        by \eqref{SpheresFromSmashProducts}
      }
      \\
      &\simeq&
      \pi_3(S^2)
      \\
      &\simeq&
      \mathbb{Z}
      &
      \proofstep{
        Hopf fibration,
      }
    \end{array}
  $$
  identifying the observable $\widehat{\zeta}$ with the representation image of the flux monodromy which is classified by the {\it Hopf fibration}.
\end{proof}

Below in \S\ref{OnClosedSurfaces} we identify this same observable
$\widehat{\zeta}$ on every closed oriented surface. On the torus,
Lem. \ref{FinDimHeisenbergIrrepsAndRootsOfUnity} shows that in a
finite-dimensional irreducible Heisenberg representation it acts by a
root of unity, $\zeta=e^{\pi\mathrm{i}\p/\lattice}$ with
$\gcd(\p,\lattice)=1$. Requiring a plane sector to occur in such a
torus representation therefore restricts its phase to these values,
as expected for FQH braiding phases \eqref{TheBraidingPhase}.

\smallskip 
However, in the remainder of this subsection here, we work out what the observable $\widehat{\zeta}$ actually observes about 2-cohomotopical flux, and show that these indeed are braiding processes of flux quanta.

\medskip

\noindent
{\bf Understanding solitonic flux processes.}
In view of the Pontrjagin construction, we  regard $\mathrm{Map}^\ast\big(\mathbb{R}^2_{\cpt},\, S^2\big)$ as the (moduli) space of solitonic flux on the plane, quantized in 2-Cohomotopy, and hence its loop space $\Omega \, \mathrm{Map}^\ast_0\big(\mathbb{R}^2_{\cpt},\, S^2\big)$ --- where loops begin and end on the constant map, cf. \eqref{ConnectedComponentOfMappingSpace}, representing the {\it flux vacuum} ---  as the space of ``{\it vacuum scattering processes}'', where flux solitons (of positive charge) and anti-solitons (of negative charge) pairwise emerge out of the vacuum, move around, and finally pair-annihilate back into the vacuum.

\smallskip 
Of these vacuum processes, our observables \eqref{IntegerObservablesOnPlane}
detect their homotopy classes $[-]$ (hence their ``topological'' or ``deformation'' class in physics jargon), labeled by the integers:
$$
  \adjustbox{
    scale=.7,
    raise=1pt
  }{
    \color{darkblue}
    \bf
    \begin{tabular}{c}
      space of vacuum processes of
      \\
      solitonic flux on the plane
    \end{tabular}
  }
  \begin{tikzcd}
    \Omega
    \,
    \mathrm{Map}^\ast_0
    \big(
      \mathbb{R}^2_{\cpt}
      ,\,
      S^2
    \big)
    \ar[
      r,
      ->>,
      "{ [-] }"
    ]
    &
    \pi_1
    \,
    \mathrm{Map}^\ast\big(
      \mathbb{R}^2_{\cpt}
      ,\,
      S^2
    \big)    
    \;\simeq\;
    \mathbb{Z}
  \adjustbox{
    scale=.7,
    raise=1pt
  }{
    \color{darkblue}
    \bf
    \begin{tabular}{c}
      topological deformation
      \\
      classes of these processes
    \end{tabular}
  }
  \end{tikzcd}
$$
Our task is hence to understand these processes, on the left, and how they are observed, on the right.

\smallskip 
To this end, the first step is to better understand the moduli space 
$
    \mathrm{Map}^\ast_0
    \big(
      \mathbb{R}^2_{\cpt}
      ,\,
      S^2
    \big)
$
itself:

\medskip
\noindent
{\bf 2-Cohomotopical moduli via the Segal-Okuyama theorem.}
By the Pontrjagin theorem (Prop. \ref{PontrjaginTheorem}), one might na{\"i}vely expect that $\mathrm{Map}^\ast\big((\mathbb{R}^2)_{\cpt}, S^2\big)$ is the {\it configuration space} (cf. \cite{Cohen09}\cite{Kallel24}) 
of {\it signed} points (hence of $\pm$ unit charged soliton cores) in the plane, topologized such that continuous curves in the space reflect creation/annihilation of oppositely charged pairs as illustrated on the left of \hyperlink{FigureP}{Fig. P}.
However, this is not quite correct as it misses the normal framing carried also by the cobordisms, according to Pontrjagin's theorem \ref{PontrjaginTheorem}. 

\smallskip 
A correct model \cite{Okuyama05} is by configurations of {\it intervals with signed endpoints} (stringy solitons between unit charged ``quarks'') all parallel to one coordinate axis and topologized so as to reflect creation/annihilation of oppositely charged pairs of endpoints. 

\begin{center}
\hypertarget{FigureP}{}
\begin{minipage}{9cm}
  \footnotesize
  {\bf Figure P -- Solitonic flux processes and framed links.}
  
  Indicated on the left is the na{\"i}ve process of pair annihilation of oppositely signed points (of oppositely charged soliton core locations). 

  But the configuration space of signed points \cite[p. 94]{McDuff75} which is topologized to make these processes continuous paths (with their reverse paths modelling the dual pair-creation processes) turns out 
  \cite[p. 6]{McDuff75}
  to {\it not quite} have the correct homotopy type of the 2-Cohomotopical flux moduli space $\mathrm{Map}^\ast\big(\mathbb{R}^2_{\cpt},\, S^2\big)$.

  \smallskip

  Indicated on the right is a variant situation where the previous charges are located at the endpoints of intervals, and where their pair annihilation makes the corresponding intervals merge. 

  The configuration space of such intervals (all parallel to one fixed coordinate axis) topologized to make these processes continuous paths (cf. \cite[Def. 3.1-2]{Okuyama05}) {\it does} have the homotopy type of the 2-Cohomotopical flux moduli space! This follows from the result of \cite[Thm. 1.1]{Okuyama05}.

  \smallskip

  The upshot is that where  a continuous loop in the space on the left is an oriented link, a loop in the space on the right is an oriented link that is also equipped with a {\it framing}: a {\it framed oriented link}. 

  \smallskip

  In traditional Chern-Simons theory, the enhancement of links to framed links is a standard but {\it ad hoc} way to ``regularize'' the corresponding {\it Wilson loop observable}, which at face value actually diverges according to traditional quantization methods. Here with 2-Cohomotopical flux quantization, this framing correction automatically arises from the moduli space of solitonic flux quantized in 2-Cohomotopy (a point that is further expanded on in \cite{SS26-Renormalization}).

\end{minipage}
\hspace{.4cm}
\adjustbox{
  scale=.6
}
{
\def\arraystretch{1.3}
\begin{tabular}{|c|c|}
\hline
\multicolumn{2}{|c|}{\bf 
\large Configurations of charged}
\\
\bf \large points 
& \bf \large intervals
\\
\hline
\hline
&
\\[-10pt]
\adjustbox{raise=2cm}
{
\begin{tikzpicture}[decoration=snake]

\draw[line width=1pt,draw=gray,fill=black]
  (+.7,0) circle (.2);
\draw[
  line width=1pt,draw=gray,fill=white,
  fill opacity=.5
]
  (-0.7,0) circle (.2);

\draw[
  decorate,
  ->
] (0,-.3) -- (0,-1);

\draw[line width=1pt,draw=gray,fill=black]
  (.2,-1.4) circle (.2);
\draw[
  line width=1pt,draw=gray,fill=white,
  fill opacity=.5
]
  (0,-1.4) circle (.2);

\begin{scope}[shift={(0,-1.5)}]
\draw[
  decorate,
  ->
] (0,-.3) -- (0,-1);
\end{scope}

\node at (0,-2.8)
  {$
    \varnothing
  $};

\end{tikzpicture}
}

&

\begin{tikzpicture}[decoration=snake]

\begin{scope}[shift={(2,-2)}]

\begin{scope}[shift={(0,0)}]
\draw[line width=2pt, gray]
  (0,0) -- (1,0);
\draw[line width=1pt,draw=gray,fill=white]
  (0,0) circle (.2);
\draw[line width=1pt,draw=gray,fill=white]
  (1,0) circle (.2);
\end{scope}

\begin{scope}[shift={(1.6,0)}]
\draw[line width=2pt, gray]
  (0,0) -- (1,0);
\draw[line width=1pt,draw=gray,fill=black]
  (0,0) circle (.2);
\draw[line width=1pt,draw=gray,fill=black]
  (1,0) circle (.2);
\end{scope}

\begin{scope}[shift={(1.6,-1.3)}]
\draw[line width=2pt, gray]
  (-.05,0) -- (1,0);
\draw[line width=1pt,draw=gray,fill=black]
  (-.2,0) circle (.2);
\draw[line width=1pt,draw=gray,fill=black]
  (1,0) circle (.2);
\end{scope}

\begin{scope}[shift={(0,-1.3)}]
\draw[line width=2pt, gray]
  (0,0) -- (1,0);
\draw[line width=1pt,draw=gray,fill=white]
  (0,0) circle (.2);
\draw[line width=1pt,draw=gray,fill=white, fill opacity=.5]
  (1.2,0) circle (.2);
\end{scope}

\draw[
  decorate,
  ->
] (1.3,-.3) -- (1.3,-1);

\draw[
  decorate,
  ->
] (1.3,-1.7) -- (1.3,-2.4);

\begin{scope}[shift={(0,-1.2)}]
\draw[
  decorate,
  ->
] (1.3,-1.7) -- (1.3,-2.4);
\end{scope}

\begin{scope}[shift={(0,-2.4)}]
\draw[
  decorate,
  ->
] (1.3,-1.7) -- (1.3,-2.4);
\end{scope}

\begin{scope}[shift={(0,-3.9)}]
\draw[
  decorate,
  ->
] (1.3,-1.7) -- (1.3,-2.4);
\end{scope}

\begin{scope}[shift={(0,-2.6)}]
\draw[line width=2pt, gray]
  (0,0) -- (2.6,0);
\draw[line width=1pt,draw=gray,fill=white]
  (0,0) circle (.2);
\draw[line width=1pt,draw=gray,fill=black]
  (2.6,0) circle (.2);
\end{scope}

\begin{scope}[shift={(0,-3.9)}]
\draw[line width=2pt, gray]
  (.6,0) -- (2,0);
\draw[line width=1pt,draw=gray,fill=white]
  (.6,0) circle (.2);
\draw[line width=1pt,draw=gray,fill=black]
  (2,0) circle (.2);
\end{scope}

\begin{scope}[shift={(.3,-5.2)}]
\draw[line width=1pt,draw=gray,fill=black]
  (1.07,0) circle (.2);
\draw[line width=1pt,draw=gray,fill=white, fill opacity=.5]
  (.93,0) circle (.2);
\end{scope}

\draw (1.3, -6.6) node {$\varnothing$};

\end{scope}
  
\end{tikzpicture}
\\
\hline
\multicolumn{2}{|c|}{
  \bf \large
  \hspace{-1.4cm}
  tracing out
}
\\
\bf \large
links 
& 
\bf \large framed links
\\
\hline
&
\\[-8pt]
\adjustbox{
  raise=1.5cm
}{
\begin{tikzpicture}[decoration=snake]

\draw[line width=1pt,draw=gray,fill=black]
  (+.7,0) circle (.2);
\draw[
  line width=1pt,draw=gray,fill=white,
  fill opacity=.5
]
  (-0.7,0) circle (.2);

\begin{scope}[shift={(0,.5)}]
\draw[line width=1pt,draw=gray,fill=black]
  (.2,-1.4) circle (.2);
\draw[
  line width=1pt,draw=gray,fill=white,
  fill opacity=.5
]
  (0,-1.4) circle (.2);
\end{scope}

\node at (0.1,-2)
  {$
    \varnothing
  $};

\draw[gray]
  (-.7,0) -- (0.1,-1.2); 
\draw[gray]
  (+.71,0) -- (0.1,-1.2); 
\draw[
  gray,
  dashed
]
  (0.1,-1.2) -- (0.1,-1.8);

\end{tikzpicture}
}
&
\begin{tikzpicture}

\begin{scope}[
  shift={(2,-2)}
]

\begin{scope}[shift={(0,0)}]
\draw[line width=2pt, gray]
  (0,0) -- (1,0);
\draw[line width=1pt,draw=gray,fill=white]
  (0,0) circle (.2);
\draw[line width=1pt,draw=gray,fill=white]
  (1,0) circle (.2);
\end{scope}

\begin{scope}[shift={(2,0)}]
\draw[line width=2pt, gray]
  (0,0) -- (1,0);
\draw[line width=1pt,draw=gray,fill=black]
  (0,0) circle (.2);
\draw[line width=1pt,draw=gray,fill=black]
  (1,0) circle (.2);
\end{scope}

\begin{scope}[shift={(2,-.8)}]
\draw[line width=2pt, gray]
  (-.2,0) -- (1,0);
\draw[line width=1pt,draw=gray,fill=black]
  (-.4,0) circle (.2);
\draw[line width=1pt,draw=gray,fill=black]
  (1,0) circle (.2);
\end{scope}

\begin{scope}[shift={(0,-.8)}]
\draw[line width=2pt, gray]
  (0,0) -- (1.2,0);
\draw[line width=1pt,draw=gray,fill=white]
  (0,0) circle (.2);
\draw[line width=1pt,draw=gray,fill=white, fill opacity=.5]
  (1.4,0) circle (.2);
\end{scope}

\begin{scope}[shift={(0,-1.6)}]
\draw[line width=2pt, gray]
  (0,0) -- (2.8,0);
\draw[line width=1pt,draw=gray,fill=white]
  (0.2,0) circle (.2);
\draw[line width=1pt,draw=gray,fill=black]
  (2.8,0) circle (.2);
\end{scope}

\begin{scope}[shift={(.15,-2.4)}]
\draw[line width=2pt, gray]
  (.5,0) -- (2,0);
\draw[line width=1pt,draw=gray,fill=white]
  (.5,0) circle (.2);
\draw[line width=1pt,draw=gray,fill=black]
  (2.2,0) circle (.2);
\end{scope}

\begin{scope}[shift={(.5,-3.2)}]
\draw[line width=1pt,draw=gray,fill=black]
  (1.07,0) circle (.2);
\draw[line width=1pt,draw=gray,fill=white, fill opacity=.5]
  (.93,0) circle (.2);
\end{scope}

\draw (1.5, -4) node {$\varnothing$};

\end{scope}

\draw[
  gray,
  smooth,
  fill=gray,
  fill opacity=.3,
  draw opacity=.3
]
  plot 
  coordinates{
    (2,-1.06) 
    (2,-2.8)
    (2.2,-3.6)
    (2.65,-4.4)
    (3.5,-5.3)
    (4.35,-4.4)
    (4.8,-3.6)
    (5,-2.8)
    (5,-1.06)
  }
  -- (4.05,-1.06)
  plot 
  coordinates {
    (4.05,-1.06)
    (4,-2)
    (3.6, -2.8)
    (3.4, -2.8)
    (3,-2.1)
    (2.95,-1.06)
  }
  --(2,-1.06);

\draw[
  line width=2pt,
  white
]
  (1.9,-1.5) -- (5.1,-1.5);
\draw[
  line width=2pt,
  white
]
  (1.9,-1.34) -- (5.1,-1.34);
\draw[
  line width=2pt,
  white
]
  (1.9,-1.18) -- (5.1,-1.18);
  
\end{tikzpicture}
\\
\hline
\end{tabular}
}
\end{center}

Under this identification of (the homotopy type of) our flux moduli space $\mathrm{Map}^\ast\big(\mathbb{R}^2_{\cpt},\, S^2\big)$ with a configuration space of charged intervals in $\mathbb{R}^2$ parallel to a fixed axis, loops in this space are identified with {\it framed oriented link diagrams}, or {\it framed links}, for short --- and loop homotopies are identified with ``link cobordism'':
\begin{center}
\def\tabcolsep{2pt}
\adjustbox{
  raise=+.1cm,
  margin=-5pt,
  scale=.37,
  fbox
}{
\begin{tabular}{ccc}
\scalebox{2}{\bf Framed link diagrams}
&
\scalebox{2}{$\leftrightarrow$}
&
\scalebox{2}{\bf Loops in moduli space}
\\
\adjustbox{raise=-1.5cm}{
\begin{tikzpicture}
  \draw[
    line width=1.4,
    -Latex
  ]
    (0:1.6) arc (0:180:1.6);
  \draw[
    line width=1.4,
    -Latex
  ]
    (180:1.6) arc (180:360:1.6);
\end{tikzpicture}
}
&
\scalebox{1.6}{
$\leftrightarrow$
}
&
\adjustbox{
  scale=.5,
  raise=-2cm
}{
\begin{tikzpicture}
 \begin{scope}
   \halfcircle
 \end{scope}
 \begin{scope}[yscale=-1]
   \halfcircle
 \end{scope} 
\end{tikzpicture}
}
\\
\adjustbox{raise=-1.5cm}{
\begin{tikzpicture}
  \draw[
    line width=1.4,
    -Latex
  ]
    (0:1.6) arc (0:180:1.6);

\begin{scope}[shift={(2,0)}]
  \draw[
    line width=10,
    white
  ]
    (180:1.6) arc (180:0:1.6);
  \draw[
    line width=1.4,
    -Latex,
  ]
    (180:1.6) arc (180:0:1.6);
  \draw[
    line width=1.4,
    -Latex,
  ]
    (0:1.6) arc (0:-180:1.6);
\end{scope}
  
  \draw[
    line width=10,
    white
  ]
    (180:1.6) arc (180:360:1.6);
  \draw[
    line width=1.4,
    -Latex
  ]
    (180:1.6) arc (180:360:1.6);

\end{tikzpicture}
}
&
\scalebox{1.6}{
$\leftrightarrow$
}
&
\adjustbox{
  scale=.5,
  raise=-1.7cm
}{
\begin{tikzpicture}
\begin{scope}[yscale=-1]
  \halfcircle
\end{scope}
\begin{scope}[
  shift={(3.4,.2)},
  xscale=-1
]
\halfcircle
\end{scope}
\begin{scope}[shift={(0,0)}]
\halfcircleover
\end{scope}
\begin{scope}[
  yscale=-1,
  shift={(3.4,-.2)},
  xscale=-1
]
\halfcircleover
\end{scope}
  
\end{tikzpicture}
}
\\
\adjustbox{raise=-1.5cm}{
\begin{tikzpicture}
  \draw[
    line width=1.4,
    -Latex
  ]
    (0:1.6) arc (0:180:1.6);

\begin{scope}[shift={(3.2,0)}]
  \draw[
    line width=12,
    white
  ]
    (180:1.6) arc (180:0:1.6);
  \draw[
    line width=1.4,
    -Latex,
  ]
    (180:1.6) arc (180:0:1.6);
  \draw[
    line width=1.4,
    -Latex,
  ]
    (0:1.6) arc (0:-180:1.6);
\end{scope}
  
  \draw[
    line width=12,
    white
  ]
    (180:1.6) arc (180:360:1.6);
  \draw[
    line width=1.4,
    -Latex
  ]
    (180:1.6) arc (180:360:1.6);

\end{tikzpicture}
}
&
\scalebox{1.6}{
$\leftrightarrow$
}
&
\adjustbox{
  scale=.5,
  raise=-1.5cm
}{
\begin{tikzpicture}

\begin{scope}[yscale=-1]
  \halfcircle
\end{scope}
\begin{scope}[shift={(6,0)},scale=-1]
  \halfcircleover
\end{scope}
\begin{scope}[shift={(6,0)}, xscale=-1]
\halfcircle
\end{scope}
\begin{scope}
\halfcircleover
\end{scope}

\end{tikzpicture}
}
\\
\adjustbox{
  raise=-3.5cm
}{
\begin{tikzpicture}[
  rotate=-90
]
  \draw[
    line width=1.4,
    -Latex
  ]
    (0:1.6) arc (0:180:1.6);

\begin{scope}[shift={(3.2,0)}]
  \draw[
    line width=12,
    white
  ]
    (180:1.6) arc (180:0:1.6);
  \draw[
    line width=1.4,
    -Latex,
  ]
    (180:1.6) arc (180:0:1.6);
  \draw[
    line width=1.4,
    -Latex,
  ]
    (0:1.6) arc (0:-180:1.6);
\end{scope}
  
  \draw[
    line width=12,
    white
  ]
    (180:1.6) arc (180:360:1.6);
  \draw[
    line width=1.4,
    -Latex
  ]
    (180:1.6) arc (180:360:1.6);

\end{tikzpicture}
}

&
\adjustbox{scale=1.6}{
$\leftrightarrow$
}
&

\adjustbox{
  raise=-7cm,
  scale=.5
}{
\begin{tikzpicture}

\begin{scope}[
  shift={(.1,6)},
  yscale=-1
]
\halfcircle
\end{scope}

\begin{scope}[
  shift={(0,6)},
]
\clip (+4,-4) rectangle 
      (+0,+0);
\halfcircleAdjusted
\end{scope}

\begin{scope}[
  shift={(0,0)},
  yscale=-1,
  xscale=-1
]
\clip (+4,-4) rectangle 
      (+0,+0);
\halfcircleAdjusted
\end{scope}

\begin{scope}[xscale=-1]
\halfcircle
\end{scope}
\begin{scope}[
  xscale=-1,
  yscale=-1
]
\clip (-4,-4) rectangle 
      (+0,+0);
\halfcircleoverAdjusted
\end{scope}

\begin{scope}[
  shift={(.2,6)},
]
\clip (-4,-4) rectangle 
      (-.1,+0);
\halfcircleoverAdjusted
\end{scope}

\end{tikzpicture}
}
\\[-10pt]
&&
\end{tabular}
}
\;\;\;
\begin{tabular}{c}
\begin{minipage}{5cm}
  \hypertarget{FigureFL}{}
  \footnotesize
  {\bf Figure FL.}
  \footnotesize
  A framed link is a ribbon link, hence a ``worldsheet of intervals''. When flattened out on a plane (``blackboard framing'') the inner twisting of the ribbon is entirely reflected in its self-crossings and thus the interval's extension may be disregarded again.
\end{minipage}
\hspace{-.6cm}
\adjustbox{
  raise=-1cm
}{
\begin{tikzpicture}

\begin{scope}[
  shift={(-1.5,0)}
]

\begin{scope}[
  yscale=-1,
  shift={(0,-2)}
]
\draw[
  draw=white,
  fill=white
]
  (-.08,0) rectangle
  (+.08,1);
\draw[
  line width=.9
]
 (-.08,0) --
 (-.08,1);
\draw[
  line width=.9
]
 (+.08,0) --
 (+.08,1);
\draw[
  line width=.6,
  gray
]
  (-.08,0) --
  (+.08,0);

\begin{scope}[
  shift={(.28,1)}
]
   \clip
    (-1,0) rectangle 
    (+1,.8);
  \draw[
    line width=.9,
    fill=white
  ] 
    (0,0) circle 
    (.2+.16); 
  \draw[line width=.9] 
    (0,0) circle 
    (.2); 
\end{scope}
\end{scope}

\begin{scope}
\draw[
  draw=white,
  fill=white
]
  (-.08-.055,0) rectangle
  (+.08+.055,1);
\draw[
  line width=.9
]
 (-.08,0) --
 (-.08,1);
\draw[
  line width=.9
]
 (+.08,0) --
 (+.08,1);
\draw[
  line width=.6,
  gray
]
  (-.08,0) --
  (+.08,0);

\begin{scope}[
  shift={(.28,1)}
]
   \clip
    (-1,0) rectangle 
    (+1,.8);
  \draw[
    draw=white,
    fill=white,
  ] 
    (0,0) circle 
    (.2+.22); 
  \draw[
    line width=.9,
  ] 
    (0,0) circle 
    (.2+.16); 
  \draw[line width=.9] 
    (0,0) circle 
    (.2); 
\end{scope}
\end{scope}

\end{scope}

\begin{scope}
\node at
  (-.5,1) {$\simeq$};

\begin{scope}
\draw[
  line width=.9
]
  (+.2,0) .. controls
  (+.2,.5) and
  (-.2,.5) ..
  (-.2,1);

\draw[
  line width=5,
  white
]
  (-.2,0) .. controls
  (-.2,.5) and
  (+.2,.5) ..
  (+.2,1);

\draw[
  line width=.9
]
  (-.2,0) .. controls
  (-.2,.5) and
  (+.2,.5) ..
  (+.2,1);
\end{scope}

\begin{scope}[shift={(0,1)}]
\draw[
  line width=.9
]
  (+.2,0) .. controls
  (+.2,.5) and
  (-.2,.5) ..
  (-.2,1);

\draw[
  line width=5,
  white
]
  (-.2,0) .. controls
  (-.2,.5) and
  (+.2,.5) ..
  (+.2,1);

\draw[
  line width=.9
]
  (-.2,0) .. controls
  (-.2,.5) and
  (+.2,.5) ..
  (+.2,1);
\end{scope}

\draw[
  line width=.6,
  gray
]
 (-.2,0) -- (+.2,0);
\draw[
  line width=.6,
  gray
]
 (-.2,2) -- 
 (+.2,2);

\end{scope}

\node at
  (+.5,1) {$\simeq$};

\begin{scope}[
  shift={(1.5,2)},
  rotate=180
]

\begin{scope}[
  yscale=-1,
  shift={(0,-2)}
]
\draw[
  draw=white,
  fill=white
]
  (-.08,0) rectangle
  (+.08,1);
\draw[
  line width=.9
]
 (-.08,0) --
 (-.08,1);
\draw[
  line width=.9
]
 (+.08,0) --
 (+.08,1);
\draw[
  line width=.6,
  gray
]
  (-.08,0) --
  (+.08,0);

\begin{scope}[
  shift={(.28,1)}
]
   \clip
    (-1,0) rectangle 
    (+1,.8);
  \draw[
    line width=.9,
    fill=white
  ] 
    (0,0) circle 
    (.2+.16); 
  \draw[line width=.9] 
    (0,0) circle 
    (.2); 
\end{scope}
\end{scope}

\begin{scope}
\draw[
  draw=white,
  fill=white
]
  (-.08-.055,0) rectangle
  (+.08+.055,1);
\draw[
  line width=.9
]
 (-.08,0) --
 (-.08,1);
\draw[
  line width=.9
]
 (+.08,0) --
 (+.08,1);
\draw[
  line width=.6,
  gray
]
  (-.08,0) --
  (+.08,0);

\begin{scope}[
  shift={(.28,1)}
]
   \clip
    (-1,0) rectangle 
    (+1,.8);
  \draw[
    draw=white,
    fill=white,
  ] 
    (0,0) circle 
    (.2+.22); 
  \draw[
    line width=.9
  ] 
    (0,0) circle 
    (.2+.16); 
  \draw[line width=.9] 
    (0,0) circle 
    (.2); 
\end{scope}
\end{scope}

\end{scope}

\end{tikzpicture}
}
\\
\\[-5pt]
\adjustbox{
  scale=0.5,
  fbox
}{$
\adjustbox{}{
\hspace{-.24cm}
\def\tabcolsep{10pt}
\begin{tabular}{c||c}
\multicolumn{2}{c}{
\large
\hspace{1.7cm}
{\bf Link cobordism }
\hspace{1.1cm}
$\leftrightarrow$
\hspace{1.3cm}
{\bf Loop homotopy in moduli space}
}
\\
\adjustbox{
  raise=-1cm,
}{
\begin{tikzpicture}[yscale=-1]
\draw[
  line width=1.2,
  ]
  (-1,1)
  edge[bend left=40]
  (-1,-1);
\draw[-Latex]
  [line width=1.2]
  (-1+.38,+0)
  --
  (-1+.38,-.0001);

\begin{scope}[
  xscale=-1
]
\draw[
  line width=1.2,
  ]
  (-1,1)
  edge[bend left=40]
  (-1,-1);
\draw[-Latex]
  [line width=1.2]
  (-1+.38,0)
  --
  (-1+.38,+.0001);
\end{scope}
\end{tikzpicture}
}
\;\;\;
\begin{tikzpicture}[decoration=snake]
\draw[decorate,->]
  (0,0) -- (+.55,0);
\draw[decorate,->]
  (0,0) -- (-.55,0);
\end{tikzpicture}
\adjustbox{
  raise=-1cm
}{
\begin{tikzpicture}[xscale=-1]
\begin{scope}[
  shift={(4.5,0)}
]
\draw[
  line width=1.2,
  ]
  (-1,1)
  edge[bend right=40]
  (+1,+1);
\draw[
  line width=1.2,
  -Latex
]
  (0,+.63)
  --
  (0.0001,+.63);

\begin{scope}[
  yscale=-1
]
\draw[
  line width=1.2,
  ]
  (-1,1)
  edge[bend right=40]
  (+1,+1);
\draw[
  line width=1.2,
  -Latex
]
  (+.0001,+.63)
  --
  (0,+.63);
\end{scope}

\end{scope}

\end{tikzpicture}
}
&
\adjustbox{
  scale=.6,
  raise=-1.4cm
}{
\begin{tikzpicture}
\draw[
  gray,
  line width=30pt,
  draw opacity=.4,
]
 (0,2) 
 --
 (0,-2);

\closedinterval{-.5}{1.6}{1}
\closedinterval{-.5}{.8}{1}
\closedinterval{-.5}{0}{1}
\closedinterval{-.5}{-.8}{1}
\closedinterval{-.5}{-1.6}{1}

\begin{scope}[
  xshift=3.3cm,
  xscale=-1
]
\draw[
  gray,
  line width=30pt,
  draw opacity=.4
]
 (0,2) 
 --
 (0,-2);

\openinterval{-.5}{1.6}{1}
\openinterval{-.5}{.8}{1}
\openinterval{-.5}{0}{1}
\openinterval{-.5}{-.8}{1}
\openinterval{-.5}{-1.6}{1}

\end{scope}
\end{tikzpicture}
}
\;\;\;\;
\adjustbox{scale=.8}{
\begin{tikzpicture}[decoration=snake]
\draw[decorate,->, line width=1]
  (0,0) -- (+0.55,0);
\draw[decorate,->, line width=1]
  (0,0) -- (-0.55,0);
\end{tikzpicture}
}
\adjustbox{
  scale=.6,
  raise=-1.1cm
}{
\begin{tikzpicture}

\begin{scope}
\clip
  (-.6,0) rectangle (5,2);
\draw[
  gray,
  line width=28,
  draw opacity=.4
]
  (1.75,2.6) circle (1.8);
\end{scope}
\closedinterval{-.4}{1.75}{1.2}
\closedinterval{-.1}{1.3}{1.75}
\begin{scope}[
  xshift=3.5cm,
  xscale=-1
]
\openinterval{-.4}{1.75}{1.2}
\openinterval{-.1}{1.3}{1.75}
\end{scope}
\oppositeinterval{.32}{.8}{2.85}
\oppositeinterval{1.65}{.33}{.2}

\begin{scope}[yscale=-1]
\begin{scope}
\clip
  (-.6,0) rectangle (5,2);
\draw[
  gray,
  line width=28,
  draw opacity=.4
]
  (1.75,2.6) circle (1.8);
\end{scope}
\closedinterval{-.4}{1.75}{1.2}
\closedinterval{-.1}{1.3}{1.75}
\begin{scope}[
  xshift=3.5cm,
  xscale=-1
]
\openinterval{-.4}{1.75}{1.2}
\openinterval{-.1}{1.3}{1.75}
\end{scope}
\oppositeinterval{.32}{.8}{2.85}
\oppositeinterval{1.65}{.33}{.2}
\end{scope}
  
\end{tikzpicture}
\hspace{-1.7cm}
}
\\[30pt]
\hline
\\[-12pt]
\adjustbox{
  raise=-1cm,
}{
\begin{tikzpicture}
\draw[
  line width=1.2,
  ]
  (-1,1)
  edge[bend left=40]
  (-1,-1);
\draw[-Latex]
  [line width=1.2]
  (-1+.38,+0)
  --
  (-1+.38,-.0001);

\begin{scope}[
  xscale=-1
]
\draw[
  line width=1.2,
  ]
  (-1,1)
  edge[bend left=40]
  (-1,-1);
\draw[-Latex]
  [line width=1.2]
  (-1+.38,0)
  --
  (-1+.38,+.0001);
\end{scope}
\end{tikzpicture}
}
\;\;\;
\begin{tikzpicture}[decoration=snake]
\draw[decorate,->]
  (0,0) -- (+.55,0);
\draw[decorate,->]
  (0,0) -- (-.55,0);
\end{tikzpicture}
\adjustbox{
  raise=-1cm
}{
\begin{tikzpicture}
\begin{scope}[
  shift={(4.5,0)}
]
\draw[
  line width=1.2,
  ]
  (-1,1)
  edge[bend right=40]
  (+1,+1);
\draw[
  line width=1.2,
  -Latex
]
  (0,+.63)
  --
  (0.0001,+.63);

\begin{scope}[
  yscale=-1
]
\draw[
  line width=1.2,
  ]
  (-1,1)
  edge[bend right=40]
  (+1,+1);
\draw[
  line width=1.2,
  -Latex
]
  (+.0001,+.63)
  --
  (0,+.63);
\end{scope}

\end{scope}

\end{tikzpicture}
}
&
\hspace{-.8cm}
\adjustbox{
  scale=.6,
  raise=-1.4cm
}{
\begin{tikzpicture}
\draw[
  gray,
  line width=30pt,
  draw opacity=.4,
]
 (0,2) 
 --
 (0,-2);

\openinterval{-.5}{1.6}{1}
\openinterval{-.5}{.8}{1}
\openinterval{-.5}{0}{1}
\openinterval{-.5}{-.8}{1}
\openinterval{-.5}{-1.6}{1}

\begin{scope}[
  xshift=3.3cm,
  xscale=-1
]
\draw[
  gray,
  line width=30pt,
  draw opacity=.4
]
 (0,2) 
 --
 (0,-2);

\closedinterval{-.5}{1.6}{1}
\closedinterval{-.5}{.8}{1}
\closedinterval{-.5}{0}{1}
\closedinterval{-.5}{-.8}{1}
\closedinterval{-.5}{-1.6}{1}

\end{scope}
\end{tikzpicture}
}
\;\;\;\;
\adjustbox{scale=.8}{
\begin{tikzpicture}[decoration=snake]
\draw[decorate,->, line width=1]
  (0,0) -- (+0.55,0);
\draw[decorate,->, line width=1]
  (0,0) -- (-0.55,0);
\end{tikzpicture}
}
\hspace{-.7cm}
\adjustbox{
  scale=.6,
  raise=-1.1cm
}{
\begin{tikzpicture}[xscale=-1]

\begin{scope}
\clip
  (-.6,0) rectangle (5,2);
\draw[
  gray,
  line width=28,
  draw opacity=.4
]
  (1.75,2.6) circle (1.8);
\end{scope}
\closedinterval{-.4}{1.75}{1.2}
\closedinterval{-.1}{1.3}{1.75}
\begin{scope}[
  xshift=3.5cm,
  xscale=-1
]
\openinterval{-.4}{1.75}{1.2}
\openinterval{-.1}{1.3}{1.75}
\end{scope}
\oppositeinterval{.32}{.8}{2.85}
\oppositeinterval{1.65}{.33}{.2}

\begin{scope}[yscale=-1]
\begin{scope}
\clip
  (-.6,0) rectangle (5,2);
\draw[
  gray,
  line width=28,
  draw opacity=.4
]
  (1.75,2.6) circle (1.8);
\end{scope}
\closedinterval{-.4}{1.75}{1.2}
\closedinterval{-.1}{1.3}{1.75}
\begin{scope}[
  xshift=3.5cm,
  xscale=-1
]
\openinterval{-.4}{1.75}{1.2}
\openinterval{-.1}{1.3}{1.75}
\end{scope}
\oppositeinterval{.32}{.8}{2.85}
\oppositeinterval{1.65}{.33}{.2}
\end{scope}
  
\end{tikzpicture}
\hspace{-1.7cm}
}
\\
\hline
&
\\[-8pt]
\adjustbox{raise=-.8cm}{
\begin{tikzpicture}
  \draw[
    line width=1.2,
    -Latex
  ]
    (0:1) arc (0:182:1);
  \draw[
    line width=1.2,
    -Latex
  ]
    (180:1) arc (180:362:1);
\end{tikzpicture}
}
\;\;
\begin{tikzpicture}[decoration=snake]
\draw[decorate,->]
  (0,0) -- (+.55,0);
\draw[decorate,->]
  (0,0) -- (-.55,0);
\node at (1.8,0) {$\varnothing$};
\end{tikzpicture}
\hspace{.7cm}
&
\adjustbox{
  scale=.35,
  raise=-1.2cm
}{
\begin{tikzpicture}
  \halfcircle
  \begin{scope}[yscale=-1]
  \halfcircle
  \end{scope}
\end{tikzpicture}
}
\;\;
\begin{tikzpicture}[decoration=snake]
\draw[decorate,->]
  (0,0) -- (+.55,0);
\draw[decorate,->]
  (0,0) -- (-.55,0);
\node at (1.8,0) {$\varnothing$};
\end{tikzpicture}
\hspace{.7cm}
\end{tabular}
}
$}
\end{tabular}
\end{center}

\begin{definition}[\bf Crossing-, Linking- and Framing numbers]
\label{LinkingNumber}
\begin{itemize}
\item[{\bf (i)}] Any crossing in a framed oriented link diagram $L$ locally has one of the following forms, to which we assign the {\it crossing number} $\pm 1$, respectively, as shown:

\vspace{-.1cm}
\begin{equation}
\label{TotalLinkingNumber}
\#\left(\!
\adjustbox{raise=-.63cm, scale=.5}{
\begin{tikzpicture}

\draw[
  line width=1.2,
  -Latex
]
  (-.7,-.7) -- (.7,.7);
\draw[
  line width=7,
  white
]
  (+.7,-.7) -- (-.7,.7);
\draw[
  line width=1.2,
  -Latex
]
  (+.7,-.7) -- (-.7,.7);
 
\end{tikzpicture}
}
\!\right)
\;=\;
+1
\,,
\hspace{1cm}
\#\left(\!
\adjustbox{raise=-.63cm, scale=0.5}{
\begin{tikzpicture}[xscale=-1]

\draw[
  line width=1.2,
  -Latex
]
  (-.7,-.7) -- (.7,.7);
\draw[
  line width=7,
  white
]
  (+.7,-.7) -- (-.7,.7);
\draw[
  line width=1.2,
  -Latex
]
  (+.7,-.7) -- (-.7,.7);
 
\end{tikzpicture}
}
\!\right)
\;\;
=
\;\;
-1
\,.
\end{equation}

\item[{(\bf ii)}] For $(L_i)_{i =1}^N$ the connected components of $L$, the {\it linking number} $\mathrm{lnk}(L_i, L_j)$ is {\it half} the sum of crossing numbers between $L_i$ and $L_j$ (cf. \cite[p. 7]{Ohtsuki01}).

\item[{(\bf iii)}] The {\it framing number} $\mathrm{fr}(L_i)$ is the sum of crossing numbers of $L_i$ with itself.

\item[{\bf (iv)}] The sum $\#L$ of the crossing numbers of all crossings of $L$ is hence the sum of all the framing and linking numbers:
\vspace{0mm} 
\begin{equation}
  \label{TotalCrossingNumber}
  \#(L)
  \;\;:=\;\;
  \underset{
    \mathclap{
      \scalebox{.8}{$
      \def\arraystretch{.1}
      \begin{array}{c}
        c \in 
        \\
        \mathrm{crssngs}(L)
      \end{array}
      $}
    }
  }{\sum}
  \;
  \#(c)
  \;\;=\;\;
  \sum_{i}
  \mathrm{frm}(L_i)
  +
  \sum_{i \neq j} \mathrm{lnk}(L_i, L_j)
  \,.
\end{equation}
\end{itemize}
\end{definition}

This has the following effect. 

\begin{proposition}[{\bf Vacuum loops of cohomotopical flux through the plane} {\cite[Thm. 3.18, Rem. 4.4]{SS24-AbAnyons}}] 
\label[proposition]{VacuumLoopsOf2CohomotopicalFlux}
\begin{enumerate} 
\item  
\label
{VacuumLoopsOf2CohomotopicalFluxAndWrithe}
Loops of 2-cohomotopical flux moduli on the plane are identified with framed links topologized to reflect link cobordism, whence their homotopy class is identified with the framed link's total crossing number $\# L$ {\rm (its ``writhe'', cf. \cite[p. 523]{Ohtsuki01})}:
\begin{equation}
\label{writhe}
    \begin{tikzcd}[
      row sep=0pt, column sep=15pt
    ]
    \Omega 
    \,
    \mathrm{Map}^\ast_0\big(
      \mathbb{R}^2_{\cpt}
      ,\,
      S^2
    \big)
    \ar[
      rr,
      "{ [-] }"
    ]
    &&
    \pi_1 
    \, \mathrm{Map}^\ast_0\big(
      \mathbb{R}^2_{\cpt}
      ,\,
      S^2
    \big)    
    \,\simeq\,
    \mathbb{Z}\;.
    \\
    L &\longmapsto&
    \# L
    \mathrlap{\,.}
    \\
    \mathclap{
      \scalebox{.7}{
        \bf
        \color{darkblue}
        framed link
      }
    }
    &&
    \mathclap{
      \scalebox{.7}{
        \bf
        \color{darkblue}
        writhe
      }
    }
    \end{tikzcd}
  \end{equation}

 \item
  Moreover, the pure states on these observables are labeled by $\zeta \in \mathrm{U}(1)$ and give the expectation values
  \begin{equation}
    \label{PureStateOnThePlane}
    \langle \zeta \vert L \vert \zeta \rangle
    \;=\;
     \zeta^{
      \# L 
    }
    \,.
  \end{equation}
  \end{enumerate}
\end{proposition}

\medskip 
\begin{example}[\bf Writhe for basic knots and links]
$\,$

\hspace{1cm}
\adjustbox{}{
\begin{tikzcd}[
  sep=0pt
]
\adjustbox{
  scale=.7
}{
\begin{tikzpicture}[
  scale=1
]

\begin{scope}[
  shift={(1,0)}
]
\draw[line width=2, -Latex]
  (0:1) arc (0:180:1);
\end{scope}

\draw[line width=7,white]
  (0:1) arc (0:180:1);
\draw[line width=2, -Latex]
  (0:1) arc (0:180:1);

\draw[line width=2, -Latex]
  (180:1) arc (180:360:1);

\begin{scope}[shift={(1,0)}]
\draw[line width=7, white]
  (180:1) arc (180:360:1);
\draw[line width=2, -Latex]
  (180:1) arc (180:360:1);
\end{scope}

\node[gray]
  at (.5,.5) {\color{red} 
    \scalebox{1.1}{$-$}
  };
\node[gray]
  at (.5,-.64) {\color{red}
    \scalebox{1.1}{$-$}
  };
  
\end{tikzpicture}
}
&[-4pt]
\overset
{\mathclap{\adjustbox{raise=2pt, scale=.75}{$\mathrm{writhe}$}}}
{\longmapsto}
&
-2\,,
\qquad 
\adjustbox{
  scale=.6
}{
\begin{tikzpicture}
\foreach \n in {0,1,2} {
\begin{scope}[
  rotate=\n*120-4
]
\draw[
  line width=2,
  -Latex
]
 (0,-1)
   .. controls
   (-1,.2) and (-2,2) ..
 (0,2)
   .. controls
   (1,2) and (1,1) ..
  (.9,.7);
\end{scope}

\node[darkgreen]
  at (\n*120+31:.6) {
    \scalebox{1}{$+$}
  };

}
\end{tikzpicture}
}
&[-16pt]
\overset
{\mathclap{\adjustbox{raise=2pt, scale=.75}{$\mathrm{writhe}$}}}
{\longmapsto}
&
+3\,,
\qquad 
\adjustbox{
}{
\begin{tikzpicture}

\node at (0,-1.2) {
  \includegraphics[width=2.2cm]
  {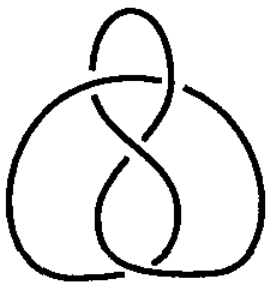}
};

\draw[line width=1.2pt,-Latex]
  (-.276,-.8) -- 
  (-.274,-.8-.02);
\draw[line width=1.2pt,-Latex]
  (.35, -.5) --
  (.35, -.47);
\draw[line width=1.2pt,-Latex]
  (-.09,1.09) --
  (-.09+.05, 1.13);
\draw[line width=1.2pt,-Latex]
  (-.18,.54) --
  (-.19-.01, .54);

\node[red]
  at (0,-1.3) {
    \scalebox{.7}{$-$}
  };
\node[red]
  at (.3,0) {
    \scalebox{.7}{$-$}
  };
\node[darkgreen]
  at (.5,.64) {
    \scalebox{.7}{$+$}
  };
\node[darkgreen]
  at (-.57,.64) {
    \scalebox{.7}{$+$}
  };

\end{tikzpicture}
}
&[-10pt]
\overset
{\mathclap{\adjustbox{raise=2pt, scale=.75}{$\mathrm{writhe}$}}}
{\longmapsto}
&
0\,.
\end{tikzcd}
}
\vspace{-.4cm}

\end{example}

\begin{remark}[\bf Comparison to Wilson loop link observables of abelian Chern-Simons theory]
\label[remark]{ComparisonWithCSOnPlane}
\begin{enumerate} 
\item For Chern-Simons theory with abelian gauge group $\mathrm{U}(1)$ it is widely understood by appeal to path-integral arguments  (\cite[p. 363]{Witten89}\cite[p. 169]{FroehlichKing89} following \cite{Polyakov88}) that 
the quantum observables are labeled by framed links $L$, 
often considered as equipped with labels (charges) $q_i$ on their $i$th connected component $L_i$
and the expectation value of these observables in these states is the charge-weighted exponentiated framing- and linking numbers (Def. \ref{LinkingNumber}) as follows
(\cite[p. 363]{Witten89}, cf. review in \cite[(5.1)]{MPW19}):
\begin{equation}
  \label{CSWilsonLoopObservable}
  \def\arraystretch{2}
  \begin{array}{ccl}
  \mathrm{W}_{\!k}(L)
  &=&
  \exp\bigg(\!
    \frac{\pi \mathrm{i}}{K}
    \Big(
    \sum_{i}
    q_i^2
    \,
    \mathrm{frm}(L_i)
    +
    \sum_{i \neq j}
    q_i q_j
    \,
    \mathrm{lnk}(L_i, L_j)
    \Big)
  \!\bigg)
  \,.
  \end{array}
\end{equation}

\item For integer charges, replace each component $L_i$ with
$\vert q_i\vert$ unit-charged parallel copies using its given framing,
reversing their orientations when $q_i<0$ and omitting the component when
$q_i=0$. With the induced framings, this preserves both the
$q_iq_j$ linking contributions and the $q_i^2$ self-framing contribution
in \eqref{CSWilsonLoopObservable}. Denoting the resulting unit-charged
framed link again by $L$, we find that its Chern-Simons expectation value
coincides with our pure topological quantum state
\eqref{PureStateOnThePlane}:
$$
  W_k(L)
   \;=\;
  \exp\bigg(\!
    \tfrac{\pi \mathrm{i}}{K}
    \Big(
    \sum_{i}
    \mathrm{frm}(L_i)
    +
    \sum_{i \neq j}
    \mathrm{lnk}(L_i, L_j)
    \Big)
  \! \bigg)
  \;
  \underset{
    \mathclap{
      \adjustbox{
        raise=-2pt,
        scale=.7
      }{
        \eqref{TotalCrossingNumber}
      }
    }
  }{
    =
  }
  \;
  \exp\Big(
    \tfrac{\pi \mathrm{i}}{K}
    \#(L)
  \Big).
$$
\end{enumerate}
\end{remark}

\begin{remark}[\bf Comparison to the literature on Hopfions]
  \label{RelationToHopfions}
  The Hopf invariant in Prop. \ref{2CohomotopicalFluxMonodromyOnPlane}
  also classifies static ``Hopf-charged solitons'' (or \emph{Hopfions})
  on three-dimensional space, with fixed value at spatial infinity:
  their configurations define maps $S^3\to S^2$ and hence elements of
  $\pi_3(S^2)\simeq\mathbb{Z}$. Such textures occur in the
  Skyrme-Faddeev model (\cite{FadeevNiemi97}, review in
  \cite{Fadeev02}\cite[\S 9.11]{MantonSutcliffe04}) and in
  three-dimensional $\mathbb{C}P^1$ models \cite{RaduTchrakianYang13}.
  In the planar setting of Rem. \ref{RelationToCP1Model}, it is instead
  suitable spacetime histories of degree-charged solitons that carry
  this Hopf invariant. For the vacuum loops considered here, the
  adjunction in Prop. \ref{2CohomotopicalFluxMonodromyOnPlane}
  identifies these histories with based maps $S^3\to S^2$.
  The analysis of \cite{WilzcekZee83} (review in \cite[\S II.C]{Forte92})
  relates the Hopf term for planar soliton histories to fractional spin
  and anyonic statistics in $2+1$ dimensions.
  What seems not to have been noticed before in previous literature is the finer identification of Prop. \ref{VacuumLoopsOf2CohomotopicalFlux} using the finer theorems by Segal-Okuyama \cite{Okuyama05}\cite{SS24-AbAnyons} 
  and the resulting identification with Chern-Simons Wilson loop observables (Rem. \ref{ComparisonWithCSOnPlane}), which, we suggest, fully nails down this identification on Euclidean space.
\end{remark}
Moreover, what seems to have received no attention before is that this identification of anyons further refines to more general 3-spaces of the form $\mathbb{R}^1 \times \Sigma^2$ (Def. \ref{Spacetime}), which is what we proceed to discuss.

\subsection{On the sphere}
\label{OnTheSphere}

We briefly discuss 2-cohomotopical flux on the 2-sphere $\Sigma^2_0 \simeq S^2$, meaning the actual 2-sphere whose point-at-infinity is disjoint, in contrast to the 2-sphere $(\Sigma^2_{0,0,1})_{\cpt} \simeq \mathbb{R}^2_{\cpt}$ that arose as the one-point compactification of the plane in \S\ref{2CohomotopicalFluxThroughPlane}. 

\smallskip 
In order for this actual 2-sphere to be realized as an FQH system in the laboratory, one would not only need to produce a 2-dimensional electron gas of spherical topology, but also make it enclose the endpoint of a very long and thin solenoid to approximate a magnetic monopole at its center that would produce magnetic flux going {\it radially} through the spherical electron gas --- which is a tall order. Nevertheless, in FQH theory this situation is often considered as an instructive hypothetical case study. 

\smallskip 
Also for us here, the following analysis of the 2-sphere case serves in  \S\ref{FluxThroughPuncturedSurface} as an intermediate step in identifying the braiding phases of solitonic anyons on the plane, that we obtained in the previous \S\ref{2CohomotopicalFluxThroughPlane}, with corresponding braiding phases on $n$-punctured disks, hence in experimentally accessible situations.

\medskip

While it is classical that
$$
  \pi_1 \, \mathrm{Map}^\ast\big(     
    (\Sigma^2_{0,0,1})_{\cpt} ,\, S^2 
  \big)
  \underset{
    \scalebox{.7}{ \eqref{SpheresFromSmashProducts} }
  }{
  \;\;
  \simeq
  \;\;
  }
  \pi_3(S^2)
  \;\simeq\;
  \mathbb{Z}
$$
(generated by the Hopf fibration $S^3 \to S^2$), the analogous statement for the un-based sphere needs another argument:

\begin{lemma}[{\bf Fundamental group of unpointed endomaps of the 2-sphere} {\cite[Thm. 5.3(1)]{STHu45}\cite[Lem. 3.1]{Koh60}}]
\label{FundamentalGroupOfUnpointedEndomapsof2Sphere}
  The fundamental group of the space of {\rm (un-pointed)} maps $S^2 \xrightarrow{\;} S^2$ is, in the connected component of maps of Hopf degree $k \in \mathbb{Z}$ (Def. \ref{HopfDegree}), isomorphic to:
  $$ 
    \pi_1
    \Big(
      \mathrm{Map}\big(
        S^2,\, S^2
      \big)
      ,\,
      \mathrm{deg} = k
    \Big)
    \;\;
    \simeq
    \;\;
    \mathbb{Z}/(2 k)
    \,.
  $$
\end{lemma}
In our notation \eqref{TheSurface}
and \eqref{ConnectedComponentOfMappingSpace}
this means, in particular, that in the component of vanishing Hopf degree we have:
\begin{equation}
  \label{FundamentalGroupOfUnpointedS2Endomaps}
  \pi_1 \, 
  \mathrm{Map}\big( \Sigma^2_{0} ,\, S^2 \big)
  \;\simeq\;
  \mathbb{Z}
  \,.
\end{equation}
\begin{lemma}[\bf Solitonic 2-cohomotopical flux monodromy on plane and 2-sphere are identified]
\label{Solitonic2CohomotopicalFluxOnPlaneAndSphereIdentified}
  The canonical map
  \vspace{-2mm} 
  \begin{equation}
    \label{IdentifyingFluxOverPlaneAnd2Sphere}
    \begin{tikzcd}
      \pi_1\, 
      \mathrm{Map}^\ast\big(
        (\Sigma^2_{0,0,1})_{\cpt}
        ,\,
        S^2
      \big)
      \ar[
        rr,
        "{
          \pi_1(p^\ast)
        }",
        "{ \sim }"{swap}
      ]
      &&
      \pi_1\, 
      \mathrm{Map}^\ast\big(
        (\Sigma^2_{0})_{\cpt}
        ,\,
        S^2
      \big)
      \underset{
        \mathclap{
          \scalebox{.7}{
            \eqref{UnpointedMapsFromPointed}
          }
        }
      }{
        \;\simeq\;
      }
      \pi_1\, 
      \mathrm{Map}\big(
        \Sigma^2_0
        ,\,
        S^2
      \big)      
    \end{tikzcd}
  \end{equation}
  is an isomorphism.
\end{lemma}
\begin{proof}
  The long exact sequence of homotopy groups
  \eqref{LongExactSequenceOfHomotopyGroups}
  induced by the evaluation map
  $\mathrm{Map}(S^2, S^2) \xrightarrow{\mathrm{ev}} S^2$
  \eqref{EvaluationSequence} is, in the relevant part, of the form
  $$
    \begin{tikzcd}[
      row sep=10pt, column sep=40pt
    ]
    \pi_2(S^2)
    \ar[r, "{ n }"]
    \ar[
      d, 
      phantom, 
      "{ 
        \rotatebox[origin=c]{90}{$\simeq$}
        \mathrlap{
          \scalebox{.7}{
            Hopf degree
          }
        } 
      }"
    ]
    &
    \pi_1\, 
    \mathrm{Map}^\ast\big(
      S^2 ,\, S^2
    \big)
    \ar[r, ->>, "{ \pi_1(p^\ast) }"]
    \ar[
      d, 
      phantom, 
      "{ 
        \rotatebox[origin=c]{90}{$\simeq$}
        \mathrlap{
          \scalebox{.7}{
            Hopf fibration
          }
        } 
      }"
    ]
    &
    \pi_1\, 
    \mathrm{Map}\big(
      S^2 ,\, S^2
    \big)    
    \ar[r]
    \ar[
      d, 
      phantom, 
      "{ 
        \rotatebox[origin=c]{90}{$\simeq$}
        \mathrlap{
          \scalebox{.7}{
            \eqref{FundamentalGroupOfUnpointedS2Endomaps}
          }
        } 
      }"
    ]
    &
    \pi_1(S^2)
    \ar[
      d, 
      phantom, 
      "{ 
        \rotatebox[origin=c]{90}{$\simeq$}
      }"
    ]
    \\
    \mathbb{Z}
    &
    \mathbb{Z}
    &
    \mathbb{Z}
    &
    1
    \mathrlap{\,,}
    \end{tikzcd}
  $$
  where the map on the left must be multiplication by some integer $n$, by the freeness of $\mathbb{Z}$. But then exactness on the left implies that the middle map must send $n$ to $0$, while exactness on the right means that the middle map is surjective hence that $n = 0$, which by exactness on the left implies that the middle map is also injective, hence bijective. 
\end{proof}

\begin{remark}[\bf Identifying braid phase observable on sphere]
\label{IdentifyingBraidPhaseObservableOnSphere}
In terms of 2-cohomotopically quantized flux, this says that the algebras of topological flux observables on the plane and on the sphere are both isomorphic to $\mathbb{C}[\mathbb{Z}]$ and canonically identified as such, whence the discussion in \S\ref{2CohomotopicalFluxThroughPlane} gives that in an irreducible representation on a (1-dimensional) Hilbert space 
the generator $1 \in \mathbb{Z}$ acts as multiplication by some phase factor $\zeta \in \mathrm{U}(1) \subset \mathbb{C}$:
$$
  \begin{tikzcd}[row sep=-2pt, 
   column sep=0pt
  ]
    \mathbb{Z}
    \ar[rr]
    &&
   \;\; \mathrm{U}(\HilbertSpace{H}_{S^2})
    \\
    1 &\longmapsto&
   \hspace{-10mm} \widehat{\zeta} 
    &[-20pt]
   \hspace{-5mm} :
    \vert \psi \rangle
    &\mapsto&
    \zeta 
    \vert \psi \rangle
    \mathrlap{\,.}
  \end{tikzcd}
$$
\end{remark}
In \S\ref{OnClosedSurfaces} and \S\ref{2CohomotopicalFluxThroughTorus}
we identify this central observable across surfaces. If the sphere sector
is required to occur in a finite-dimensional torus representation, its
phase must be a root of unity, by
Lem. \ref{FinDimHeisenbergIrrepsAndRootsOfUnity}, as expected in FQH systems.

\subsection{On closed surfaces}
\label{OnClosedSurfaces}

While spherical FQH systems as in \S\ref{OnTheSphere} are just barely plausible as having experimental realizations, for closed surfaces of more general genus $g \in \mathbb{N}$ this quickly becomes even less plausible as $g$ increases. Nevertheless, the notoriously rich theoretical predictions for these somewhat hypothetical FQH systems on closed surfaces are crucial intermediate stages in understanding FQH systems in general and hence also in experimentally accessible situations. In particular, requiring a sector on the open disk from
\S\ref{2CohomotopicalFluxThroughPlane} to occur in a finite-dimensional
torus representation constrains its braiding phase $\zeta$ to a root of
unity. Proposition \ref{IdentifyingCentralGeneratorAcrossSurfaces}
identifies the central generators, while
Lem. \ref{FinDimHeisenbergIrrepsAndRootsOfUnity} supplies the
finite-dimensionality constraint. This is the compatibility condition
we impose when comparing these sectors across surfaces, as expected
for FQH systems \eqref{TheBraidingPhase}.

\medskip

The new key we now offer for understanding FQH systems on closed surfaces via topological flux quantization is the following Prop. \ref{MonodromyOfFluxThroughClosedSurface}, whose roots in algebraic topology date back half a century (\cite{Hansen74}, following \cite{STHu45}), but which gains new meaning when understood now as being about observables on topological flux quantized in 2-cohomotopy: 

\begin{definition}[{\bf Integer Heisenberg group, level 2}, {(cf. \cite[p. 7]{GelcaUribe10}\cite[Def. 2.4]{GelcaHamilton15a}\cite[Def. 2.3]{GelcaHamilton15b}\cite[(8)]{BlanchetPalmerShaukat21}\cite[(1.2)]{BlanchetPalmerShaukat25}}]
By the {\it integer Heisenberg group at level=2} 
\footnote{
  By the ``level'' we here mean the integer $n$ in the nontrivial group commutators $[a_i,b_i] = n$. 
  The integer Heisenberg group at level $=1$  
  (cf. \cite[p. 232]{LeePacker96}\cite[p. 213]{DrutuKapovich18})
  is isomorphic to groups of certain upper triangular integer matrices (cf. \cite[p. 35]{DummitFoote03}\cite[p. 299]{DrutuKapovich18}), and in this form is commonly considered in pure algebra and group theory (cf. \cite[(1.1)]{LeePacker96}). In contrast, the case of relevance here, with level $= 2$ --- which is the case of subgroups of the actual eponymous Heisenberg group from quantum mechanics --- seems not to have found much attention in the pure algebra/group theory literature. 
},
to be denoted $\widehat{\mathbb{Z}^{2g}}$
for $g \in \mathbb{N}$, we refer to the group
\begin{equation}
  \label{IntegerHeisenbergGroup}
  \widehat{\mathbb{Z}^{2g}}
  \;
  :=
  \;
  \left\{
    \big(\vec a, \vec b, n\big)
    \,\in\,
    \mathbb{Z}^g \times \mathbb{Z}^g
    \times \mathbb{Z}  
   \,,\;\;
  \def\arraystretch{1.3}
  \begin{array}{c}
    \big(\vec a, \vec b, n\big)
    \cdot
    \big(\vec a\,', \vec b\,', n'\big)
    :=
    \big(\,
      \vec a + \vec a\,'
      ,\,
      \vec b + \vec b\,'
      ,\,
      n + n'
      +
      \mathcolor{purple}{
      \vec a \cdot \vec b\,'
      -
      \vec a\,' \cdot \vec b \,
      }
      \,
    \big)
  \end{array}
  \!\!\! \right\}.
\end{equation}
which is the central extension of the free abelian group $\mathbb{Z}^{2g}$ by
(cf. \cite[\S IV]{Brown82}) the 2-cocycle shown shaded in \eqref{IntegerHeisenbergGroup}, hence by the restriction of the canonical symplectic form on $\mathbb{R}^{2g}$ along $\mathbb{Z}^{2g} \xhookrightarrow{\;} \mathbb{R}^{2g}$, hence is the subgroup $\widehat{\mathbb{Z}^{2g}} \xhookrightarrow{\;} \widehat{\mathbb{R}^{2g}}$ of the ordinary Heisenberg group  (cf. \cite[\S 9.5]{PressleySegal88}) on integer-valued elements.
\end{definition}

\begin{proposition}
[\bf Monodromy of 2-cohomotopical flux through closed surfaces]
\label
{MonodromyOfFluxThroughClosedSurface}
The 2-cohomotopical flux monodromy 
\eqref{SolitonicTopologicalObservablesMoreGenerally}
over a closed oriented surface $\Sigma^2_g$ \eqref{TheSurface}, for $g \in \mathbb{N}$,  forms the $\mathbb{Z}$-extension of the free abelian group $\mathbb{Z}^{2g}$
\eqref{OrdinaryFluxObservablesOnClosedSurface}
\begin{equation}
\label
{ExtensionSequenceOfFLuxMonodromyOnClosedSurface}
  \begin{tikzcd}[
    column sep=10pt,
    row sep=10pt
  ]
    1
    \ar[r]
    &
    \pi_1\,
      \mathrm{Map}^\ast_0\big(
        S^2_{\cpt}
        ,\,
        S^2
      \big)
    \ar[
      rr
    ]
    \ar[
      d,
      shorten <=-2pt,
      "{ \sim }"{sloped, pos=.4}
    ]
    &&
    \pi_1 \,
      \mathrm{Map}^\ast_0\big(
        (\Sigma^2_g)_{\cpt}
        ,\,
        S^2
      \big)
    \ar[
      rr
    ]
    \ar[
      d,
      shorten <=-2pt,
      "{ \sim }"{sloped, pos=.4}
    ]
    &&
    \pi_1 \,
      \mathrm{Map}^\ast_0\big(
        \bigvee_g
        (S^1_a \vee S^1_b)
        ,\,
        S^2
      \big)
    \ar[r]
    \ar[
      d,
      shorten <=-2pt,
      "{ \sim }"{sloped, pos=.4}
    ]
    &
    1
    \\
    &
    \mathbb{Z}
    \ar[
      rr,
      hook
    ]
    &&
    \widehat{\mathbb{Z}^{2g}}
    \ar[
      rr,
      ->>
    ]
    &&
    \mathbb{Z}^{2g}
    \mathrlap{\,,}
  \end{tikzcd}
\end{equation}
that is the \emph{integer Heisenberg group} at level=2 \eqref{IntegerHeisenbergGroup}.
\end{proposition}
\begin{proof}
  The top short exact sequence is due to \cite[Thm. 1 \& p. 6]{Hansen74}, recalled as Lem. \ref{HansenThm1} in the appendix.
  The identification of the group on the left is by Lem.  \ref{FundamentalGroupOfUnpointedEndomapsof2Sphere}.
  Now \cite[Thm. 1]{LarmoreThomas80} shows that the central extension in question has as its only non-trivial group commutators
  $[a_i, b_i] = 2$, for $i \in \{1, \cdots, g\}$, in accordance with \eqref{IntegerHeisenbergGroup}. \end{proof}

A more elegant proof of Prop. \ref{MonodromyOfFluxThroughClosedSurface}, for the case of the torus ($g = 1$), was recently found in \cite[\S 3]{KSS26}.

\subsubsection
{Modular action}

  Our way of casting Prop. \ref{MonodromyOfFluxThroughClosedSurface} --- with the extension cocycle highlighted in \eqref{IntegerHeisenbergGroup} identified as the standard symplectic form on $\mathbb{Z}^{2g}$ (instead of the cohomologous $2 \, \vec a \cdot \vec b'$ used in the original derivations \cite[Thm. 1]{LarmoreThomas80}\cite[Cor. 7.6]{Kallel01}) --- makes manifest that the integral symplectic group $\mathrm{Sp}_{2g}(\mathbb{Z})$ acts by group automorphisms on $\widehat{\mathbb{Z}^{2g}}$ (covering its defining action on $\mathbb{Z}^{2g}$ and acting trivially on the center $\mathbb{Z} \xhookrightarrow{\;} \widehat{\mathbb{Z}^{2g}}$, cf. \cite[(2.6)]{GelcaHamilton15a}):
  \begin{equation}
    \label{LiftingModularActionToIntegerHeisenberg}
    \begin{tikzcd}[
      column sep=45pt,
      row sep=12pt
    ]
      &
      \mathrm{Aut}\big(
        \widehat{\mathbb{Z}^{2g}}
      \big)
      \ar[
        d
      ]
      \\
      \mathrm{Sp}_{2g}(\mathbb{Z})
      \ar[
        r,
        "{ \rho }",
        "{
          \mathrm{canonical}
        }"{swap}
      ]
      \ar[
        ur,
        dashed,
        "{
          \exists
          \,
          \widehat{\rho}
        }"{sloped}
      ]
      &
      \mathrm{Aut}\big(
        \mathbb{Z}^{2g}
      \big)
      \mathrlap{\,.}
    \end{tikzcd}
  \end{equation}

Better yet, with \eqref{SemidirectProductStructureInIntroduction}, one such lifted action is singled out by the geometrical interpretation as the diffeomorphism action on the 2-cohomotopical flux monodromy over closed surfaces. This we identify now:

\begin{proposition}[\bf Mapping class group action on 2-cohomotopical flux monodromy on a closed surface]
  \label[proposition]
  {MCGActionOnFluxMonodromyOverClosedSurface}
  Under the identification of Prop. \ref{MonodromyOfFluxThroughClosedSurface}, the action
  \eqref{ActionOfMCGOnModuliMonodromy}
  of the mapping class group of $\Sigma^2_g$ on the 2-cohomotopical flux monodromy over $\Sigma^2_g$ is via its  symplectic representation \eqref{MCGOfClosedSurface} 
  on the underlying $\mathbb{Z}^{2g}$ 
  \begin{equation}
    \begin{tikzcd}[row sep=15pt]
    \mathrm{MCG}(\Sigma^2_g)
    \ar[
      rr,
      "{
        \scalebox{.7}{
          \eqref{ActionOfMCGOnModuliMonodromy}
        }
      }"
    ]
    \ar[
      dd, 
      ->>,
      "{
        \scalebox{.7}{
          \eqref{MCGOfClosedSurface}
        }
      }"{swap}
    ]
    &&
    \mathrm{Aut}\big(
    \pi_1 \, \mathrm{Map}_0(
      \Sigma^2_g
      ,\,
      S^2
    )
    \big)
    \\
    &&
    \mathrm{Aut}\big(\,
      \widehat{\mathbb{Z}^{2g}}
    \,\big)
    \ar[
      u,
      "{
        \scalebox{.7}{
          \eqref{ExtensionSequenceOfFLuxMonodromyOnClosedSurface}
        }
      }"{swap},
      "{ \sim }"{sloped}
    ]
    \ar[
      d
    ]
    \\
    \mathrm{Sp}_{2g}(\mathbb{Z})
    \ar[
      urr,
      dashed,
      "{ 
         \exists!
         \,
         \widehat{\rho}
      }"{sloped}
    ]
    \ar[
      rr,
      "{
        \scalebox{.7}{
          \eqref{LiftingModularActionToIntegerHeisenberg}
        }
      }"{yshift=7pt,xshift=20pt},
      "{
        \rho
      }"
    ]
    &
    &
    \mathrm{Aut}\big(
      \mathbb{Z}^{2g}
    \big)
    \mathrlap{\,.}
    \end{tikzcd}
  \end{equation}
  Moreover, there is a choice of the top right isomorphism $\pi_1 \mathrm{Map}_0\bracket({ \Sigma^2_g, S^2 }) \simeq \widehat{\mathbb{Z}^{2g}}$ \cref{ExtensionSequenceOfFLuxMonodromyOnClosedSurface} under which 
  \begin{equation}
    \label{ShearlessMCGAction}
    \widehat{\rho}(M)
    :
      (v,n)
        \mapsto
      \bracket({
        M v, n
      })
      \mathrlap{\,.}
  \end{equation}
\end{proposition}
\begin{proof}
  Commutativity of the outer diagram follows by inspection of the construction of \eqref{ExtensionSequenceOfFLuxMonodromyOnClosedSurface} as spelled out in \S\ref{SurfacesAnd2Cohomotopy}: 
  By Lem. \ref{FLuxMonodromyMappingToOrdinaryCohomology} there, the induced action of $\mathrm{MCG}(\Sigma^2_g)$ on $\inlinetikzcd{ \pi_1\, \mathrm{Map}\big(\Sigma^2_g,\, S^2 \big) \ar[r,->>] \& \mathbb{Z}^{2g} }$ is identified with its action on $H^1\big(\Sigma^2_g;\, \mathbb{Z}\big)$, for which it is classical \eqref{MCGOfClosedSurface} that it is through $\mathrm{Sp}_{2g}(\mathbb{Z})$, as claimed.

  We proceed to analyze the nature of $\widehat{\rho}$. Here
  the case $g = 0$ is immediate, therefore we assume now that $g \geq 1$.
Here we proceed by comparing the MCG-action by precomposition with the postcomposition action by     an orientation-reversing involutive diffeomorphism fixing the base point of the 2-sphere:
\begin{equation}
  \inlinetikzcd{
    S^2
    \ar[
      r,
      "{ r }"
    ]
    \&
    S^2
    \mathrlap{\,.}
  }
\end{equation}
Postcomposition with $r$ induces an automorphism of $\pi_1 \mathrm{Map}\bracket({ \Sigma^2_g, S^2 })$
which acts as -1 on $\mathbb{Z}^{2g}$, hence in particular preserves the symplectic form, and  which commutes with the precomposition action of any $[f] \in \mathrm{MCG}\bracket({ \Sigma^2_g })$. 

\begin{enumerate}
\item 
To see that $\widehat \rho$ exists uniquely:

Observe that every automorphism $\Phi$ of $\widehat{\mathbb{Z}^{2g}}$ preserving the symplectic form must be the identity on the center:
First, being an automorphism, $\Phi$ must restrict to an automorphism of the center, $\mathbb{Z} \simeq \langle \zeta \rangle$, hence must act on the center as $\zeta \mapsto \zeta^{\epsilon}$ for $\epsilon \in \{+1, -1\}$. But furthermore, with $\Phi$ also being a symplectomorphism, we have
\begin{equation}
  \zeta^{2 \epsilon}
    =
  \Phi\bracket({
    \bracket[{a_i, b_i}]
  })
  =
  \bracket[{
    \Phi(a_i),
    \Phi(b_i)
  }]
  =
  \zeta^{
    2 \omega\scaledbracket({
       \Phi(a_i),
       \Phi(b_i)
    })
  }
  =
  \zeta^{2}
  \mathrlap{\,,}
\end{equation}
which implies $\epsilon = 1$, as claimed.

With this in hand, \cref{ShearInAutomorphismsOfHeisenberg} gives 
\begin{subequations}
\label{TheShears}
  \begin{align}
    \label{ReflectionShear}
    r \circ [v,n] 
    \phantom{{}\circ [f]^{-1}}
    &=
    \bracket({
      -v
      ,\;
      n + 
      \lambda(v)
    })
    \\
    \label{MCGShear}
    \phantom{r \circ }
    [v,n]
    \circ [f]^{-1}
    & 
    =
    \bracket({
      M_f 
      v
      ,\;
      n 
        + 
      \ell_f(v)
    })
    \mathrlap{\,,}
  \end{align}
\end{subequations}
whence we deduce
\begin{subequations}
  \begin{align}
    &
    r 
      \circ 
    \bracket({
      [v,n]
      \circ 
      [f]
    })
    =
    \bracket({
      r 
      \circ 
      [v,n]
    })
      \circ 
      [f]
    \\
    \Rightarrow
    \;\;\;
    &
    \ell_f(v)
    +
    \lambda\bracket({
      M_f v
    })
    \,=\,
    \lambda(v)
    -
    \ell_f(v)
    \\
    \label{ellInTermsOfLambda}
    \Leftrightarrow
    \;\;\;
    &
    \lambda\bracket({
      M_f v
    })
    -
    \lambda(v)
    \,=\,
    -2\ell_f(v)
    \mathrlap{\,.}
  \end{align}
\end{subequations}
This shows that all $\ell_f$ are fixed by $\lambda$ as functions of $f$ and as such only depend on $M_f$. This uniquely determines, through \cref{MCGShear}, the action $\widehat \rho$, as claimed.

\item
To see that the shear $\ell_f$ may be taken to vanish upon another choice of isomorphism, as claimed in \cref{ShearlessMCGAction}:

The result \cref{ellInTermsOfLambda} also shows that $\lambda$ is invariant modulo 2 under symplectic transformations. But such invariance under the specific symplectomorphisms
\begin{equation}
  \begin{tikzcd}[
    sep=0pt
  ]
    \mathbb{Z}^{2g}
    \ar[rr]
    &&
    \mathbb{Z}^{2g}
    \\
    a_i &\mapsto& a_i + b_i
    \\
    b_i &\mapsto& b_i
  \end{tikzcd}
  \;\;\;\;\;\;\;\;
  \text{and}
  \;\;\;\;\;\;\;\;
  \begin{tikzcd}[
    sep=0pt
  ]
    \mathbb{Z}^{2g}
    \ar[rr]
    &&
    \mathbb{Z}^{2g}
    \\
    a_i &\mapsto& a_i
    \\
    b_i &\mapsto& b_i + a_i
  \end{tikzcd}
\end{equation}
implies, by linearity, that $\lambda$ already vanishes modulo 2, hence that
\begin{equation}
  \lambda 
    =
  2 \mu
  \,,
  \;\;\;
  \text{ for some linear }
  \;
  \inlinetikzcd{
    \mathbb{Z}^{2g}
    \ar[r, "{ \mu }"]
    \&
    \mathbb{Z}
  }
  \mathrlap{\,.}
\end{equation}

Using this we may absorb the shears in \cref{TheShears} in an
automorphism
\begin{equation}
  \begin{tikzcd}[
    sep=0pt
  ]
    \widehat{\mathbb{Z}^{2g}}
    \ar[
      rr,
      "{ U }"
    ]
    &&
    \widehat{\mathbb{Z}^{2g}}
    \\
    (v,n)
    &\mapsto&
    \bracket({
      v, n + \mu(v)
    })
    \mathrlap{\,,}
  \end{tikzcd}
\end{equation}
in that
\begin{equation}
  \begin{tikzcd}[
    sep=0pt
  ]
    \widehat{\mathbb{Z}^{2g}}
    \ar[
      rr,
      "{
       U 
         \scaledbracket
           ({(-) \circ [f]^{-1}}) 
       U^{-1}
      }"
    ]
    &&
    \widehat{\mathbb{Z}^{2g}}
    \\
    (v,n)
    &\longmapsto&
    \bracket({
      M_f v
      , 
      \bracket({
        n
          -
        \mu(v) 
          +
        \ell_f(v)
          +
        \mu\bracket({ M_f v })
      })
    })
    \\
    && =
    \bracket({
      M_f v, n
    })
  \end{tikzcd}
\end{equation}
as claimed in \cref{ShearlessMCGAction}.
\end{enumerate}
\end{proof}
In this proof of Prop. \ref{MCGActionOnFluxMonodromyOverClosedSurface} we used the following independent observation: 
\begin{lemma}
\label[lemma]
{ShearInAutomorphismsOfHeisenberg}
\begin{enumerate}
\item
  The automorphisms $\Phi$ of $\widehat{\mathbb{Z}^{2g}}$ which are the identity on the center, $\Phi(0,n) = (0,n)$, are of the form
  \begin{equation}
    \label
    {PhiIsEndoAndShear}
    \Phi(v, n)
    =
    \bracket({
      M v, n + \ell(v)
    })
    \mathrlap{\,,}
  \end{equation}
  for a symplectic matrix $M$ and any linear map $\ell$, \textup{the shear}, hence this subgroup of the automorphism group is:
  \begin{equation}
    \mathrm{Aut}_{\mathrm{ctr}}
    \bracket({
      \widehat{
        \mathbb{Z}^{2g}
      }
    })
    \;\simeq\;
    \mathrm{Hom}\bracket({
      \mathbb{Z}^{2g},
      \mathbb{Z}
    })
    \rtimes
    \mathrm{Sp}_{2g}(\mathbb{Z})
    \mathrlap{\,.}
  \end{equation}
\item
Such an automorphism is inner precisely if $M = 1$ and $\ell \in 2 \, \mathrm{Hom}\bracket({ \mathbb{Z}^{2g}, \mathbb{Z} })$.
\end{enumerate}
\end{lemma}
\begin{proof}
\begin{enumerate}
\item
Given such an automorphism $\Phi$, on the quotient by the center, it induces an automorphism $M \in \mathrm{Aut}(\mathbb{Z}^{2g})$ (not \emph{a priori} symplectic) and hence  $\Phi(v,0) = \bracket({ M v, \ell(v) })$ for some function $\ell$ (not \emph{a priori} linear). 
But the homomorphism property $\Phi\bracket({(u,0) \cdot (v,0)}) = \Phi(u,0) \cdot \Phi(v,0)$ demands that  
\[
  \begin{aligned}
  &
  \ell(u + v)
  +
  \omega(u,v)
  \;=\;
  \ell(u) + \ell(v)
  +
  \omega\bracket({
    M u, M v
  })
  \\
  \Leftrightarrow
  \;\;\;
  &
  \ell(u + v)
  - \ell(u) 
  - \ell(v)
  =
  \omega\bracket({
    M u, M v
  })
  -
  \omega(u,v)
  \mathrlap{\,.}
  \end{aligned}
\]
Here either side is a $\mathbb{Z}$-valued function that is both symmetric and skew-symmetric in its arguments $(u,v)$. This implies that both sides vanish for all arguments, which says that $\ell$ is linear and that $M$ is symplectic.
 
With this, the homomorphism property on general elements,
 $\Phi(v,n) \cdot \Phi(v',n') = \Phi\bracket({ (v,n) \cdot (v',n') })$, is already implied:
 \[
  \begin{aligned}
    & 
    \bracket({
      M v + M v'
      ,
      \;
      n + n' +
      \ell(v)
      +
      \ell(v')
      +
      \omega\bracket({
        M v, M v'
      })
    })
    \\
    =\;\; &
    \bracket({
      M(v + v')
      ,
      \;
      n + n' +
      \omega(v,v')
      +
      \ell(v + v')
    })
    \mathrlap{\,,}
  \end{aligned}
\]
and hence the claim follows.

\item
The inner automorphism induced by $(u,m)$ acts as
\[
  (v,n)
  \,\mapsto\,
  (u,m) \cdot (v,n)
  \cdot
  (-u,-m)
  =
  \bracket({
    v,
    \,
    n + 2 \, \omega(u,v)
  })
  \mathrlap{\,.}
\]
Hence the claim follows by observing  that every linear $\ell(-)$ as above is of the form $\omega(u,-)$ for
\[
  u :=
  \textstyle{\sum_i} 
  \bracket({
    \ell(b_i) a_i
    -
    \ell(a_i) b_i
  })
  \mathrlap{\,.}
  \tag*{\qedhere}
\]
\end{enumerate}
\end{proof}

\subsubsection
{Functoriality}

Next, we observe some form of functoriality 
in maps between closed surfaces
of the result of Prop. \ref{MonodromyOfFluxThroughClosedSurface}, the crucial implication being the identification of the braid phase observable $\zeta$ \eqref{TheBraidingPhase} across all closed surfaces.
To this end, write:
\begin{equation}
  \label{ProjectionsOfClosedSurfaces}
  \begin{tikzcd}
    \Sigma_g
    \ar[
      from=rr, 
      ->>, 
      "{ q^{g+1}_g }"{swap}
    ]
    &&
    \Sigma_{g + 1}
  \end{tikzcd}
\end{equation}
for the surjective continuous map 
\eqref{ConsecutiveProjectionMapsBetweenClosedSurfaces}
given by contracting one pair of edges in the standard fundamental polygon (cf. Prop. \ref{FundamentalPolygonOfClosedOrientedSurfaces}) of $\Sigma_{g + 1}$.
For instance, for $g \in \{0,1\}$ the maps 
\vspace{-4mm}
\begin{equation}
\label{SurjectionsOfClosedSurfaces}
\begin{tikzcd}
\adjustbox{
  scale=.8,
  raise=.3cm
}{
\begin{tikzpicture}
  \node  at (1,2.8)
  {
    \scalebox{.9}{
      \color{darkblue}
      \bf
      sphere
    }
  };
  \draw[
    line width=3,
    fill=lightgray
  ]
    (0,0)
    rectangle 
    (2,2);

\draw[fill=black]
  (0,0) circle (.08);

\draw[fill=black]
  (2,0) circle (.08);

\draw[fill=black]
  (0,2) circle (.08);

\draw[fill=black]
  (2,2) circle (.08);

\node
  at (1,.9) {$
    \Sigma^2_0 \simeq S^2
  $};

\end{tikzpicture}
}
\ar[
  from=rr,
  ->>,
  "{ 
    q^1_0 
  }"{swap}
]
&&
\adjustbox{
  scale=.8,
  raise=+.3cm
}{
\begin{tikzpicture}
  \node  at (1,2.8)
  {
    \scalebox{.9}{
      \color{darkblue}
      \bf
      torus
    }
  };
  \draw[
    draw opacity=0,
    fill=lightgray
  ]
    (0,0)
    rectangle 
    (2,2);

\draw[
  dashed,
  color=darkgreen,
  line width=1.6,
]
  (0,2) -- 
  node[yshift=8pt]{
    \scalebox{1}{
      \color{black}
      $a$
    }
  }
  (2,2);
\draw[
  -Latex,
  darkgreen,
  line width=1.6
]
  (1.3-.01,2) -- 
  (1.3,2);

\draw[
  dashed,
  color=darkgreen,
  line width=1.6,
]
  (0,0) -- (2,0);
\draw[
  -Latex,
  darkgreen,
  line width=1.6
]
  (1.3-.01,0) -- 
  (1.3,0);

\draw[
  dashed,
  color=olive,
  line width=1.6,
]
  (0,0) -- (0,2);
\draw[
  -Latex,
  olive,
  line width=1.6
]
  (0, .8+.01) -- 
  (0,.8);

\draw[
  dashed,
  color=olive,
  line width=1.6,
]
  (2,0) -- 
  node[xshift=8pt]{
    \scalebox{1}{
      \color{black}
      $b$
    }
  }
  (2,2);
\draw[
  -Latex,
  olive,
  line width=1.6
]
  (2, .8+.01) -- 
  (2,.8);

\draw[fill=black]
  (0,0) circle (.08);

\draw[fill=black]
  (2,0) circle (.08);

\draw[fill=black]
  (0,2) circle (.08);

\draw[fill=black]
  (2,2) circle (.08);

\node
  at (1,.9)
  {
    $\Sigma^2_1 \simeq \mathbb{T}^2$
  };

\end{tikzpicture}
}
\ar[
  from=rr,
  ->>,
  "{
    q^2_1
  }"{swap}
]
&&
\hspace{-20pt}
\adjustbox{
  raise=3pt
}{
  \begin{tikzpicture}[
    scale=.65
  ]

  \node[
    rotate=39
  ] at (-1.8,1.8) {
    \scalebox{.64}{
      \color{darkblue}
      \bf
      \def\arraystretch{.7}
      \begin{tabular}{c}
        2-handled
        \\
        torus
      \end{tabular}
    }
  };

  \draw[
    draw opacity=0,
    fill=lightgray
  ]
     (90-22.5:2) --
    (45+90-22.5:2) --
    (2*45+90-22.5:2) --
    (3*45+90-22.5:2) --
    (4*45+90-22.5:2) --
    (5*45+90-22.5:2) --
    (6*45+90-22.5:2) --
    (7*45+90-22.5:2) --
    cycle;

    \draw[
      line width=1.6,
      dashed,
      darkgreen
    ]
      (90-22.5:2)
      --
      (90+22.5:2);

    \draw[
      line width=1.6,
      dashed,
      olive
    ]
      (-45+90-22.5:2)
      --
      (-45+90+22.5:2);

    \draw[
      line width=1.6,
      dashed,
      darkgreen
    ]
      (-90+90-22.5:2)
      --
      (-90+90+22.5:2);

    \draw[
      line width=1.6,
      dashed,
      olive
    ]
      (-135+90-22.5:2)
      --
      (-135+90+22.5:2);

    \draw[
      line width=1.6,
      dashed,
      darkblue
    ]
      (-180+90-22.5:2)
      --
      (-180+90+22.5:2);

    \draw[
      line width=1.6,
      dashed,
      darkblue
    ]
      (-180+90-22.5:2)
      --
      (-180+90+22.5:2);

    \draw[
      line width=1.6,
      dashed,
      purple
    ]
      (-225+90-22.5:2)
      --
      (-225+90+22.5:2);

    \draw[
      line width=1.6,
      dashed,
      darkblue
    ]
      (-270+90-22.5:2)
      --
      (-270+90+22.5:2);

    \draw[
      line width=1.6,
      dashed,
      purple
    ]
      (-315+90-22.5:2)
      --
      (-315+90+22.5:2);

\draw[
  line width=1.3,
  -Latex,
  darkgreen
]
  (0.3,1.83) -- (.35,1.83);

\begin{scope}[
  rotate=-45
]
\draw[
  line width=1.3,
  -Latex,
  olive
]
  (0.24,1.83) -- 
  (.25,1.83);
\end{scope}
\begin{scope}[
  xscale=-1,
  rotate=+90
]
\draw[
  line width=1.3,
  -Latex,
  darkgreen
]
  (0.25,1.83) -- 
  (.3,1.83);
\end{scope}
\begin{scope}[
  xscale=-1,
  rotate=+135
]
\draw[
  line width=1.3,
  -Latex,
  olive
]
  (0.25,1.83) -- 
  (.3,1.83);
\end{scope}

\begin{scope}[
  rotate=+180
]
\draw[
  line width=1.3,
  -Latex,
  darkblue
]
  (.23,1.83) -- 
  (.33,1.83);
\end{scope}

\begin{scope}[
  rotate=-225
]
\draw[
  line width=1.3,
  -Latex,
  purple
]
  (.2,1.83) -- 
  (.3,1.83);
\end{scope}

\begin{scope}[
  rotate=-270
]
\draw[
  line width=1.3,
  -Latex,
  darkblue
]
  (-.20,1.83) -- 
  (-.25,1.83);
\end{scope}
\begin{scope}[
  rotate=-315
]
\draw[
  line width=1.3,
  -Latex,
  purple
]
  (-.20,1.83) -- 
  (-.25,1.83);
\end{scope}

\foreach \n in {1,...,8} {
  \draw[
    fill=black
  ]
    (22.5+\n*45:2)
    circle
    (.11);
}

\node[
  scale=.9
] at
  (0,2.3) {
    $a_1$
  };

\begin{scope}[
  rotate=-45
]
  \node[scale=.9] at
    (0,2.3) {
      $b_1$
    };
\end{scope}

\begin{scope}[
  rotate=-180
]
  \node[scale=.9] at
    (0,2.3) {
      $a_2$
    };
\end{scope}

\begin{scope}[
  rotate=-180-45
]
  \node[scale=.9] at
    (0,2.3) {
      $b_2$
    };
\end{scope}

\node
  at (0,0) {
    $\Sigma^2_2$
  };
  
\end{tikzpicture}
}
\end{tikzcd}
\end{equation}

\vspace{-2mm} 
\noindent are given, for $q^2_1$, by sending the purple and blue colored edges to the point $\bullet$, and for $q^1_0$ by sending also the remaining edges to the point.

\begin{lemma}[\bf Pullback of 2-cohomotopical flux monodromy on closed surfaces]
\label{PullbackOf2CohomotopicalFluxOnClosedSurfaces}
  Under the above identification \eqref{ExtensionSequenceOfFLuxMonodromyOnClosedSurface},
  the surjections $q$ \eqref{ProjectionsOfClosedSurfaces} map to the canonical inclusion of Heisenberg groups obtained by adjoining the generators corresponding to the contracted edges:
  \begin{equation}
    \label{InclusionOfFluxMonodromyOnClosedSurfaces}
    \begin{tikzcd}[
      row sep=10pt, column sep=30pt
    ]
      \pi_1 
      \mathrm{Map}^\ast\big(
        \Sigma^2_g
        ,\,
        S^2
      \big)
      \ar[
        rr,
        "{ 
          (q^{g+1}_g)^\ast 
        }"
      ]
      \ar[
        d,
        shorten=-1pt,
        "{ \sim }"{sloped, pos=.4}
      ]
      &&
      \pi_1 
      \mathrm{Map}^\ast\big(
        \Sigma^2_{g+1}
        ,\,
        S^2
      \big)
      \ar[
        d,
        shorten=-1pt,
        "{ \sim }"{sloped, pos=.4}
      ]
      \\
      \widehat{\mathbb{Z}^{2g}}
      \ar[
        rr,
        hook
      ]
      &&
      \widehat{\mathbb{Z}^{2g + 2}}
    \end{tikzcd}
  \end{equation}
\end{lemma}
\begin{proof}
  The proof of Lem. \ref{HansenThm1} shows that the short exact sequence \eqref{ExtensionSequenceOfFLuxMonodromyOnClosedSurface} is natural in $\Sigma^2_g$, whence we have a commuting diagram of this form:
$$
  \begin{tikzcd}[
    column sep=15pt,
    row sep=10pt
  ]
    &
    \mathbb{Z}
    \ar[
      rr,
      hook
    ]
    \ar[
      d,
      shorten=-2pt,
      "{ \sim }"{sloped}
    ]
    &&
    \widehat{\mathbb{Z}^{2g}}
    \ar[
      rr,
      ->>
    ]
    \ar[
      d,
      shorten=-2pt,
      "{ \sim }"{sloped}
    ]
    &&
    \mathbb{Z}^{2g}
    \ar[
      d,
      shorten=-2pt,
      "{ \sim }"{sloped}
    ]
    \\
    1
    \ar[r]
    &
    \pi_1\,
      \mathrm{Map}_0\big(
        S^2
        ,\,
        S^2
      \big)
    \ar[
      rr
    ]
    \ar[
      d,
      equals
    ]
    &&
    \pi_1 \,
      \mathrm{Map}_0\big(
        \Sigma^2_{g}
        ,\,
        S^2
      \big)
    \ar[
      rr
    ]
    \ar[
      d,
      shorten <=-1pt,
      "{
        (q^{g+1}_g)^\ast
      }"
    ]
    &&
    \pi_1 \,
      \mathrm{Map}^\ast_0\big(
        \bigvee_{g}
        (S^1_a \vee S^1_b)
        ,\,
        S^2
      \big)
    \ar[r]
    \ar[
      d,
      shorten <=-2pt,
      hook
    ]
    &
    1
    \\[+10pt]
    1
    \ar[r]
    &
    \pi_1\,
      \mathrm{Map}_0\big(
        S^2
        ,\,
        S^2
      \big)
    \ar[
      rr
    ]
    \ar[
      d,
      shorten <=-2pt,
      "{ \sim }"{sloped, pos=.4}
    ]
    &&
    \pi_1 \,
      \mathrm{Map}_0\big(
        \Sigma^2_{g \mathcolor{purple}{+1}}
        ,\,
        S^2
      \big)
    \ar[
      rr
    ]
    \ar[
      d,
      shorten <=-2pt,
      "{ \sim }"{sloped, pos=.4}
    ]
    &&
    \pi_1 \,
      \mathrm{Map}^\ast_0\big(
        \bigvee_{g \mathcolor{purple}{+1}}
        (S^1_a \vee S^1_b)
        ,\,
        S^2
      \big)
    \ar[r]
    \ar[
      d,
      shorten <=-2pt,
      "{ \sim }"{sloped, pos=.4}
    ]
    &
    1
    \\
    &
    \mathbb{Z}
    \ar[
      rr,
      hook
    ]
    &&
    \widehat{\mathbb{Z}^{2g\mathcolor{purple}{+2}}}
    \ar[
      rr,
      ->>
    ]
    &&
    \mathbb{Z}^{2g\mathcolor{purple}{+2}}
    \mathrlap{\,.}
  \end{tikzcd}
$$
This gives the claim.
\end{proof}

Combination of Lem. \ref{Solitonic2CohomotopicalFluxOnPlaneAndSphereIdentified} and Lem. \ref{PullbackOf2CohomotopicalFluxOnClosedSurfaces} leads to:

\begin{proposition}[\bf Identifying central braid phase generator across surfaces]
\label{IdentifyingCentralGeneratorAcrossSurfaces}
The canonical comparison map between solitonic flux monodromy on the plane and on $\Sigma_g$, $g \in \mathbb{N}$, identifies the central generators as the braiding phase observable of
\S\ref{2CohomotopicalFluxThroughPlane}:
$$
  \begin{tikzcd}[  
    column sep=15pt,
    row sep=0pt
  ]
    \pi_1\, 
    \mathrm{Map}^\ast\big(
      (\mathbb{R}^2)_{\cpt}
      ,\,
      S^2
    \big)
    \ar[
      d, "{\sim}"
    ]
    \ar[
      rr,
      "{
        \pi_1( q^\ast )
      }",
      "{
        \scalebox{.7}{
          \eqref{IdentifyingFluxOverPlaneAnd2Sphere}
        }
      }"{swap}
    ]
    &&
    \pi_1\, 
    \mathrm{Map}\big(
      \Sigma^2_0
      ,\,
      S^2
    \big)
    \ar[
      d, "{\sim}"
    ]
    \ar[
      rr,
      "{ 
        \pi_1\big( 
          (q^{g}_0)^\ast 
        \big)
      }",
      "{
        \scalebox{.7}{
          \eqref{InclusionOfFluxMonodromyOnClosedSurfaces}
        }
      }"{swap}
    ]
    &&
    \pi_1\, 
    \mathrm{Map}\big(
      \Sigma^2_g
      ,\,
      S^2
    \big)\,.
    \ar[
      d, "{\sim}"
    ]
    \\[15pt]
    \mathbb{Z}
    \ar[
      rr,
      "{ \sim }"
    ]
    &&
    \mathbb{Z}
    \ar[
      rr,
      hook
    ]
    &&
    \widehat{\mathbb{Z}^{2g}}
    \\
    1 &\longmapsto&
    1 &\longmapsto&
        \scalebox{1.4}{$($}
    \adjustbox{
      scale=.6,
      raise=1.5pt
    }{$
    \left[
    \adjustbox{raise=-1pt}{$
    \def\arraystretch{.9}
    \def\arraycolsep{0pt}
    \begin{array}{c}
      0 \\ 0
    \end{array}
    $}
    \right]
    $}
    ,
    1
    \scalebox{1.4}{$)$}
  \end{tikzcd}
$$
\end{proposition}

\begin{remark}[\bf Braiding phases of solitonic anyons on closed surfaces]
\label{IdentifyingBraidingPhases}
In generalization of Rem. \ref{IdentifyingBraidPhaseObservableOnSphere}, Prop. \ref{IdentifyingCentralGeneratorAcrossSurfaces} says that 
given 2-cohomotopical flux quantum states on $\Sigma^2_g$, hence 
a unitary representation of the 2-cohomotopical flux monodromy on $\Sigma^2_g$, hence of the
integer Heisenberg group
$\widehat{\mathbb{Z}^{2g}}$ \eqref{ExtensionSequenceOfFLuxMonodromyOnClosedSurface}, for $g \in \mathbb{N}$,
then the central observable $\widehat{\zeta}$
\begin{equation}
  \label{GenericRepOfIntegerHeisenberg}
  \begin{tikzcd}[
    column sep=0pt,
    row sep=-3pt,
  ]
    \mathbb{Z}
    \ar[
      rr, 
      hook
    ]
    &&
    \widehat{\mathbb{Z}^{2g}}
    \ar[
      rr
    ]
    &&
    \mathrm{U}(\HilbertSpace{H})
    &[-8pt]
    &[-1pt]
    \\
    1
    &\longmapsto&
    \scalebox{1.4}{$($}
    \adjustbox{
      scale=.6,
      raise=1.5pt
    }{$
    \left[
    \adjustbox{raise=-1pt}{$
    \def\arraystretch{.9}
    \def\arraycolsep{0pt}
    \begin{array}{c}
      0 \\ 0
    \end{array}
    $}
    \right]
    $}
    ,
    1
    \scalebox{1.4}{$)$}
    &\longmapsto&
    \widehat{\zeta}
    &:&
    \vert \psi \rangle 
    &\mapsto&
    \zeta \, \vert \psi \rangle \,,
    \\
    &&
    \scalebox{1.4}{$($}
    \adjustbox{
      scale=.6,
      raise=1.5pt
    }{$
    \left[
    \adjustbox{raise=-1pt}{$
    \def\arraystretch{.9}
    \def\arraycolsep{0pt}
    \begin{array}{c}
      \vec a \\ \vec b
    \end{array}
    $}
    \right]
    $}
    ,
    0
    \scalebox{1.4}{$)$}
    &\longmapsto&
    \WOperator{\vec a}{\vec b}
  \end{tikzcd}
\end{equation}
is to be understood as observing the braiding phases of solitonic flux according to \S\ref{2CohomotopicalFluxThroughPlane}, where the braiding happens within an open disk inside the surface.
\end{remark}
For $g\geq 1$, finite-dimensional irreducible representations restrict
this central phase to a root of unity: apply
Lem. \ref{FinDimHeisenbergIrrepsAndRootsOfUnity} to an irreducible
summand of the restriction to one handle's Heisenberg subgroup.
The torus case is analyzed in \S\ref{2CohomotopicalFluxThroughTorus},
with $\zeta=e^{\pi\mathrm{i}\p/\lattice}$ as in the FQH braiding phase
\eqref{TheBraidingPhase}.

\begin{remark}[\bf Braiding phases of defect anyons on closed surfaces]
  There is the following alternative way to exhibit the central generator of \eqref{ExtensionSequenceOfFLuxMonodromyOnClosedSurface} as witnessing braiding processes:
  For $g \geq 1$ and $n \geq 3$, let $\widehat{\mathbb{Z}^{2g}}'$ denote the level=2 extension of $\mathbb{Z}^{2g}$ as in \eqref{IntegerHeisenbergGroup} but by the finite cyclic group $\mathbb{Z}_{2(n+g-1)}$. This is evidently the quotient of the integer Heisenberg group \eqref{IntegerHeisenbergGroup} by the $2(n+g-1)$st power of the central element (whence its irreducible unitary representations $\rho$ force the braiding phase $\zeta$ \eqref{GenericRepOfIntegerHeisenberg} to be a $(2(n+g-1))$st root of unity), but it turns out to also be \cite[(10)]{BellingeriGervaisGuaschi08} (cf. also \cite[Cor. 8]{BlanchetPalmerShaukat21}\cite[Prop. 1.1]{BlanchetPalmerShaukat25})
  the quotient of the surface braid group $\mathrm{Br}_n(\Sigma^2_g)$ \eqref{SurfaceBraidGroup} 
  which identifies each Artin braid generator $b_i$ \eqref{ArtinGeneratorsForSymmetricGroup} with this central generator:
  $$
    \begin{tikzcd}[row sep=-2pt, column sep=0pt]
      \mathrm{Br}_n\big(\Sigma^2_g\big)
      \ar[rr, ->>]
      &&
      \widehat{\mathbb{Z}^{2g}}'
      \ar[rr, "{ \rho }"]
      &&
      \mathrm{U}(\HilbertSpace{H})
      \\
      b_i 
      &\longmapsto&
    \scalebox{1.4}{$($}
    \adjustbox{
      scale=.6,
      raise=1.5pt
    }{$
    \left[
    \adjustbox{raise=-1pt}{$
    \def\arraystretch{.9}
    \def\arraycolsep{0pt}
    \begin{array}{c}
      0 \\ 0
    \end{array}
    $}
    \right]
    $}
    ,
    1
    \scalebox{1.4}{$)$}
    &\longmapsto&
    \widehat{\zeta}
    \,.
    \end{tikzcd}
  $$
  The surface braid group describes motions of {\it defect anyons}
  (cf. Rem. \ref{DefectAnyonsAndBraidReps} and
  \S\ref{QBitQuantumGatesOperableOn2PuncturedOpenDisk}). The displayed
  quotient identifies a common algebraic description with solitonic flux,
  subject to the additional central relation: writing $z$ for the central
  generator and $m:=2(n+g-1)$, we have
  \[
    \mathbb{C}\big[\widehat{\mathbb{Z}^{2g}}'\big]
    \;\simeq\;
    \mathbb{C}\big[\widehat{\mathbb{Z}^{2g}}\big]
    \big/\big(z^m-1\big).
  \]
  Thus a solitonic representation descends to this quotient precisely
  when $z^m$ acts as the identity. At a fixed central phase this is the
  condition $\zeta^m=1$; the resulting surface-braid representation sends
  every Artin generator to the scalar $\zeta$. This compares the
  corresponding representations of the quotient observable algebra;
  their realization by physical defect motion requires additional input.
\end{remark}

\begin{remark}[\bf Comparison to expectations in FQH systems]
\label{ComparisonOfBraidingPhaseToExpectationInFQH}
It is folklore in the literature on anyons that the phase $\zeta$ \eqref{TorusWilsonLineCommutatorForFQH} in the quantum algebra of observables on the torus may be understood as witnessing anyon braiding phases in FQH systems (cf. \cite[(30)]{NSSFD08}\cite[(3.33)]{Tong16}\cite[\S 4.3]{Simon23}). The above results substantiate this intuition by a rigorous derivation of the statement from 2-cohomotopical flux quantization and as such lend support to our hypothesis that 2-Cohomotopy is the correct effective flux-quantization law for FQH systems.
\end{remark}

\begin{remark}[\bf Instability of the Hopf monodromy under target enlargement]
  Comparison of \eqref{OrdinaryFluxObservablesOnClosedSurface} with 
  \eqref{ExtensionSequenceOfFLuxMonodromyOnClosedSurface}
  amplifies how the anyon braiding
  phase $\widehat{\zeta}$ 
  \eqref{GenericRepOfIntegerHeisenberg}
  appears due to the exotic flux quantization in $S^2 \simeq \mathbb{C}P^1$ and disappears with the ordinary flux quantization in $\mathbb{C}P^\infty \simeq B \mathrm{U}(1)$ along the inclusion $\iota : S^2 \simeq \mathbb{C}P^1 \xhookrightarrow{\phantom{--}} \mathbb{C}P^\infty \simeq B \mathrm{U}(1)$ \eqref{2SphereInsideBU1}:
  \[
    \begin{tikzcd}[sep=0pt]
      \pi_1\Big(
        \mathrm{Map}\big(
          \Sigma^2_g
          ,\,
          \mathbb{C}P^1
        \big)
      \Big)
      \ar[
        rr,
        "{ \iota_\ast }"
      ]
      &&
      \pi_1\Big(
        \mathrm{Map}\big(
          \Sigma^2_g
          ,\,
          \mathbb{C}P^\infty
        \big)
      \Big)
      \\
      \widehat{\mathbb{Z}^{2g}}
      \ar[
        rr,
        ->>
      ]
      &&
      \mathbb{Z}^{2g}
      \\
      \WOperator{\vec a}{\vec b} 
        &\longmapsto&
      \WOperator{\vec a}{\vec b}
      \\[-2pt]
      \widehat{\zeta} 
        &\longmapsto&
      0
      \mathrlap{\,.}
    \end{tikzcd}
  \]
  In the Hall-dual situation of anomalous Hall effects in Chern insulators
  \cite{SS25-FQAH}, maps to $\mathbb{C}P^1$ describe a single occupied band
  in a two-band system. Enlarging the target to $\mathbb{C}P^\infty$
  gives the classifying space for complex line bundles and retains the
  stable first Chern class. In particular, the inclusion induces an
  isomorphism
  $\pi_2(\mathbb{C}P^1)\xrightarrow{\;\simeq\;}\pi_2(\mathbb{C}P^2)
  \simeq\mathbb{Z}$, so a nonzero Chern number survives this enlargement.
  It is the additional Hopf invariant, and hence the central monodromy
  generator displayed above, that disappears: the map
  $\pi_3(\mathbb{C}P^1)\simeq\mathbb{Z}\longrightarrow
  \pi_3(\mathbb{C}P^2)=0$ is zero.

  This is an instance of unstable two-band topology. The band-theory
  distinction between \emph{fragile} topology (cf. \cite{Bouhon2020})
  and \emph{delicate} topology specifies which occupied and unoccupied
  bands may be added \cite{BrouwerDwivedi23}; the two-band Hopf invariant
  is an example of the latter \cite{NelsonEtAl22}.
\end{remark}
The next proposition describes the resulting monodromy after target
enlargement: the central Hopf contribution disappears already at
$S^2\simeq\mathbb{C}P^1\xhookrightarrow{\;}\mathbb{C}P^2$,
while the $\mathbb{Z}^{2g}$ contribution remains:
\begin{proposition}[{\bf Effect of intermediate flux-quantization laws}]
For $g\in\mathbb{N}$ and $N\in\mathbb{Z}$ we have:
\begin{equation}
  N \geq 2 
  \;\;\;\;\;\;\;
  \Rightarrow
  \;\;\;\;\;\;\;
  \pi_1\Big(
    \mathrm{Map}_0\big(
      \Sigma^2_g
      ,\,
      \mathbb{C}P^N
    \big)
  \Big)
  \;\simeq\;
  \mathbb{Z}^{2g}
  \,.
\end{equation}
\end{proposition}
\begin{proof}
This is \cite[Lem. 7.5]{Kallel01}.
\end{proof}

\medskip

\subsection{On the torus}
\label
{2CohomotopicalFluxThroughTorus}

While transverse magnetic flux through a toroidal 2D electron gas is not readily realized experimentally (cf. footnote \ref{SurfacesInExperiment}), effective field theories of flux on arbitrary surfaces tend to be characterized by their theoretical predictions for the torus, notably through the dimension and modular transformation properties of the Hilbert space of states (transport data characterizing topological order \cite{Wen90}\cite{Wen95}, cf. Prop.~\ref{2CohomotopicalQuantumStatesOnTorus} below and the qualifications in \S\ref{GroundStateMonodromyAndTopologicalOrder}). 

\smallskip 
Therefore, a major example of the 
phenomena in \S\ref{ModuliSpacesOfTopologicalFlux}
is the following derivation of quantum states of 2-cohomotopically quantized topological flux on the torus, which reproduces the {\it modular data} \cite{Gannon05} of $\mathrm{U}(1)$-Chern-Simons theory
(Rem.~\ref{ComparisonToModularDataOfAbelianCS} below) -- in fact of {\it spin} 
Chern-Simons theory \eqref{ChernSimonsLevel}.

\subsubsection
{Flux monodromy classification}
\label
{FluxMonodromyClassification}

Our task in identifying the 2-cohomotopical flux quantum states over the torus is,
by Prop.~\ref{MonodromyOfFluxThroughClosedSurface}, to study finite-dimensional
unitary representations of the integer Heisenberg group
\eqref{IntegerHeisenbergGroup} after covariantization
\eqref{StateSpacesAsLocalSystemsOnModuliSpace}.
In this subsection we classify those representations which remain irreducible
upon restriction to the Heisenberg subgroup; the scope of this assumption is
discussed in Rem.~\ref{FineTopologicalOrder}\,\textup{(iii)}.

\begin{proposition}[\bf Diffeomorphism action over the torus is the canonical modular action]
\label
{DiffeomorphismActionOverTorusIsModularAction}
  The action \eqref{ActionOfMCGOnModuliMonodromy}
  of $\mathrm{MCG}(\Sigma^2_1) \,\simeq\, \mathrm{SL}_2(\mathbb{Z})$ \eqref{MCGOfTorus} on $\pi_1 \mathrm{Map}^\ast_0\big((\Sigma^2_1)_{\cpt},\,S^2\big) \,\simeq\, \widehat{\mathbb{Z}^2}$ \eqref{IntegerHeisenbergGroup} is the defining action of $\mathrm{Sp}_2(\mathbb{Z}) \,\simeq\, \mathrm{SL}_2(\mathbb{Z})$ on $\mathbb{Z}^2$ and trivial on the center, whence the flux monodromy group \eqref{ActionOfMCGOnModuliMonodromy} over the torus with its $\mathrm{pp}$-spin structure \eqref{SpinStructuresOnTheTorus} is
  \begin{equation}
    \label{CohomotopicalFluxMonodromyOverTorus}
    \mathrm{MCG}(\Sigma^2_{1})
    \;\ltimes\;
    \pi_1
    \Big(
    \mathrm{Map}
      ^\ast
      _0
    \big(
      (\Sigma^2_{1})_{\cpt}
      ,\,
      S^2
    \big)
    \Big)
    \;\;\simeq\;\;
    \mathrm{SL}_2(\mathbb{Z})
    \,\ltimes\,
    \widehat{\mathbb{Z}^2}
  \end{equation}
  while for the torus equipped with the $\mathrm{aa}$-spin structure it is correspondingly the subgroup $\mathrm{MCG}(\Sigma^2_1)^{\mathrm{aa}} \ltimes \widehat{\mathbb{Z}^2}$ \eqref{aaSpinMappingClassGroupOfTorus}.
\end{proposition}
\begin{proof}
  This is a special case of \cref{MCGActionOnFluxMonodromyOverClosedSurface}.
\end{proof}

Representations of $\mathrm{SL}_2(\mathbb{Z})$ and of
$\widehat{\mathbb{Z}^2}$ separately are well-studied.
Here we consider their compatibility in the semidirect product
\eqref{CohomotopicalFluxMonodromyOverTorus}.
We find representations which recover the familiar
$\mathrm{U}(1)$-Chern-Simons modular matrices at unit-fractional braiding
angles $\pi/\lattice$ and extend the construction to general rational
braiding angles $\pi\p/\lattice$.

\smallskip

We proceed incrementally: after some generalities, we construct the states
at $\braidingAngle=1/\lattice$ for even $\lattice$
(Prop.~\ref{2CohomotopicalQuantumStatesOnTorus}) and compare them with
Chern-Simons/CFT theory (Rem.~\ref{ComparisonToModularDataOfAbelianCS}).
We then treat all rational braiding angles for the $\mathrm{pp}$-spin structure
(Prop.~\ref{General2CohomotopicalStatesOverTorus}) and both allowed scalar
sectors for the $\mathrm{aa}$-spin structure
(Prop.~\ref{QuantumStatesOverTheAASpinTorus}).
The conclusion is Thm.~\ref{ClassificationOf2CohomotopicalFluxQuantumStatesOverTorus}.

\smallskip

The condition organizing this discussion is that a Heisenberg representation
extend to its semidirect product with the relevant mapping class group
\eqref{ActionOfMCGOnModuliMonodromy}, hence admit a \emph{covariantization}
in the sense of Def.~\ref{TopologicalQuantumSectorsOfGeneralGaugeTheories}:
\begin{definition}
[\bf Covariantizable representations of the integer Heisenberg group]
\label
{CovariantizableRepsOfHeisenberg}
  We say that a linear representation of the integer Heisenberg group $\widehat{\mathbb{Z}^2}$ is {\it covariantizable} for a given spin structure if it extends to a representation of the semidirect product with (cf. Prop.~\ref{DiffeomorphismActionOverTorusIsModularAction})
  \begin{enumerate}
    \item for the $\mathrm{pp}$-spin structure: 
    \newline
    the torus mapping class group \eqref{MCGOfTorus} generated by $S$ and $T$ \eqref{ModularGenerators},
    \item for the $\mathrm{aa}$-spin structure:
    \newline
    the spin mapping class subgroup \eqref{aaSpinMappingClassGroupOfTorus} generated only by $S$ and $T^2$.
  \end{enumerate}
\end{definition}

In order to analyze such extended representations, we first note the following elementary facts:   

\begin{lemma}[{\bf Extending representations along normal subgroup inclusions}, {cf. \cite[p.~175]{Isaacs76}}]
\label{ClassOfExtendibleRepIsGInvariant}
  Given $\iota:H \hookrightarrow G$ a normal subgroup inclusion and 
  $\rho : H \xrightarrow{\;} \mathrm{GL}(\HilbertSpace{H})$
  a finite-dimensional $\mathbb{C}$-linear representation, say that $\widehat{\rho} : G \xrightarrow{\;} \mathrm{GL}(\HilbertSpace{H})$ is an \emph{extension} if $\iota^\ast \widehat{\rho} = \rho$. --- We have:

  \noindent
  {\bf (i)}
  Any extension $\widehat{\rho}$ exhibits $\rho$ as  isomorphic to its $g$-translate 
  representations:
  \begin{equation}
    \underset{g \in G}{\forall}
    \;\;\;
    \rho \;\simeq\; \rho^g    
    \,,
    \;\;\;\;\;\;\;
    \mbox{\rm
      where 
      $\rho^g := \rho\big(\mathrm{Ad}_g(-)\big)$
    }.
  \end{equation}

  \noindent
  {\bf (ii)}
  If $\rho$ is irreducible then
  any two extensions $\widehat{\rho}$, $\widehat{\rho}\,'$ differ at most by tensoring with some multiplicative character of $G/H$:
  \[
    \underset{g \in G}{\forall}
    \;\;\;\;
    \widehat{\rho}\,'(g)
    \;=\;
    d(g H)
    \cdot
    \widehat{\rho}\,(g)
    \,\;,
    \;\;\;\;\;\;\;
    \mbox{
      for some 
      $d : G/H \xrightarrow{\;} \mathbb{C}^\times$.
    }
  \]
\end{lemma}
\begin{proof}
{\bf (i)}
We have for $h \in H$:
\vspace{-2mm} 
\[
  \def\arraystretch{1.4}
  \begin{array}{ccl}
    \widehat{\rho}\,(g)
    \circ
    \rho(h)
    \circ
    \widehat{\rho}\,(g)^{-1}
    &=&
    \widehat{\rho}\,(g)
    \circ
    \widehat{\rho}\,(h)
    \circ
    \widehat{\rho}\,(g^{-1})    
    \\
    &=&
    \widehat{\rho}\,(g h g^{-1})
    \\
    &=&
    \rho(g h g^{-1})
    \\
    &\defneq&
    \rho^g(h)
    \,,
  \end{array}
\]
showing that $\widehat{\rho}(g)$ serves as an intertwiner that exhibits the claimed isomorphism.

\noindent {\bf (ii)} We have, for $h \in H$ and $g \in G$:
\[
  \def\arraystretch{1.5}
  \begin{array}{rcl}
    \big(
      \widehat{\rho}\,'(g)
      \circ
      \widehat{\rho}\,(g)^{-1}
    \big)
    \circ
    \rho(h)
    \circ
    \big(
      \widehat{\rho}\,'(g)
      \circ
      \widehat{\rho}\,(g)^{-1}
    \big)^{-1}
    &=&
    \widehat{\rho}\,'(g)
    \circ
    \widehat{\rho}\,(g^{-1})
    \circ
    \rho(h)
    \circ
    \widehat{\rho}\,(g)
    \circ
    \widehat{\rho}\,'(g^{-1})
    \\
    &=&
    \widehat{\rho}\,'(g)
    \circ
    \rho(g^{-1} h g)
    \circ
    \widehat{\rho}\,'(g^{-1})
    \\
    &=&
    \rho(h)
    \,,
  \end{array}
\]
showing that the ratio of the two extensions for any $g$ commutes with all the $\rho$-operators. Therefore Schur's lemma \eqref{SchurLemma} implies that this ratio is a scalar $d(g)$ times the identity operator. The representation laws imply that $d$ is multiplicative, and agreement of the extensions on $H$ implies that $d$ is trivial there. Thus $d$ factors through $G/H$. If both extensions are unitary, then $d$ takes values in $\mathrm{U}(1)$.
\end{proof}

\begin{lemma}[\bf Finite-dimensional Heisenberg irreps and roots of unity]
  \label{FinDimHeisenbergIrrepsAndRootsOfUnity}
  If a unitary representation of the integer Heisenberg group, $\widehat{(-)} : \widehat{\mathbb{Z}^2} \xrightarrow{\phantom{--}} \mathrm{U}(\HilbertSpace{H})$, is finite-dimensional and irreducible, then the central observable $\widehat{\zeta}$ \eqref{GenericRepOfIntegerHeisenberg} is given by multiplication with a root of unity: $\widehat \zeta \,=\, \zeta \cdot \mathrm{id}$ and $\underset{n \in \mathbb{N}_{> 0}}{\exists}\;\; \zeta^n = 1$.
\end{lemma}
\begin{proof}
  First, with the representation assumed irreducible and since $\widehat{\zeta}$ commutes with all other representation operators, Schur's lemma \eqref{SchurLemma} implies that there is $\zeta \in \mathbb{C}$ with $\widehat{\zeta} = \zeta \cdot \mathrm{id}$.
  Further, since $\WOperator{1}{0}$ is unitary
we may find an eigenvector, to be denoted $\vert 0 \rangle$, with non-vanishing eigenvalue to be denoted $\xi$: 
  \begin{equation}
    \label{InitialEigenvector}
    \WOperator{1}{0}\vert 0 \rangle
    \;=\;
    \xi \vert 0 \rangle
    \,.
  \end{equation} 
  Now the commutation relation says that $\WOperator{0}{1}$ is a corresponding raising operator, in that the elements
  \begin{equation}
    \label{GenericStates}
    \vert
      n
    \rangle
    \;:=\;
    \big(\WOperator{0}{1}\big)^n 
    \vert 0 \rangle
    \;\;
    \in
    \;\;
    \HilbertSpace{H}
  \end{equation}
  are further eigenvectors of $\WOperator{1}{0}$ with eigenvalue $\zeta^{2n} \xi$:
  $$
    \WOperator{1}{0}
    \vert n \rangle
    \;\defneq\;
    \WOperator{1}{0}
    \big(\WOperator{0}{1}\big)^n 
    \vert 0 \rangle
    \underset{
      \scalebox{.7}{
        \eqref{HeisenbergGroupCommutator}
      }
    }{
      \;=\;
    }
    \zeta^{2n}
    \,
    \big(\WOperator{0}{1}\big)^n 
    \,
    \WOperator{1}{0}
    \vert 0 \rangle
    \underset{
      \scalebox{.7}{
        \eqref{InitialEigenvector}
      }
    }{
      \;=\;
    }
    \zeta^{2n}
    \,
    \big(\WOperator{0}{1}\big)^n 
    \,
    \xi
    \vert 0 \rangle
    \;\defneq\;
    \zeta^{2n}\xi
    \, 
    \vert n \rangle
    \,.
  $$
  But by the assumption of finite-dimensionality there can only be finitely many distinct eigenvalues, which is evidently equivalent to $\zeta$ being a root of unity.
\end{proof}

\medskip 
\begin{definition}[\bf Cyclic Heisenberg group]
For $g,o \in \mathbb{N}_{>0}$, we denote the $o$-cyclic version of the integer Heisenberg group \eqref{IntegerHeisenbergGroup} by
\begin{equation}
  \label{CyclicHeisenbergGroup}
  \widehat{\mathbb{Z}^{2g}_{\mathcolor{purple}{o}}}
  \;
  :=
  \;
  \left\{
    \big([\vec a], [\vec b], [n]\big)
    \,\in\,
    \mathbb{Z}^g_{\mathcolor{purple}{o}} 
      \times 
    \mathbb{Z}^g_{\mathcolor{purple}{o}}
      \times 
    \mathbb{Z}_{\mathcolor{purple}{o}}
   \,, \;\;
   \def\arraystretch{1.5}
   \begin{array}{c}
    \big([\vec a], [\vec b], [n]\big)
    \cdot
    \big([\vec a'], [\vec b\,'], [n']\big)
    :=
    \\
    \big(\,
      [\vec a + \vec a\,']
      ,\,
      [\vec b + \vec b\,']
      ,\,
      \big[
      n + n'
      +
      \vec a \cdot \vec b\,'
      -
      \vec a\,' \cdot \vec b \,
      \big]
      \,
    \big)
    \end{array}
 \!\!\! \right\}
  \,,
\end{equation}
where $\mathbb{Z}_o := \mathbb{Z}/o\mathbb{Z}$ denotes the $o$-cyclic group and $[-] \,\colon\,\mathbb{Z} \twoheadrightarrow \mathbb{Z}_o$ denotes the quotient map.
\end{definition}
\begin{lemma}
[\bf Covariantizable irreps factor through cyclic Heisenberg groups]
\label
{IrrepsFactorThroughCyclicHeisenberg}
If a finite-dimensional irreducible unitary representation $\widehat{(-)} : \widehat{\mathbb{Z}^2} \xrightarrow{\phantom{-}} \mathrm{U}(\HilbertSpace{H})$
is covariantizable {\rm (Def.~\ref{CovariantizableRepsOfHeisenberg})} for the $\mathrm{pp}$-spin structure,
then it is the pullback of a representation of the $o$-cyclic Heisenberg group \eqref{CyclicHeisenbergGroup} 
along the quotient coprojection
$
  [-]
  :
  \begin{tikzcd}[
    column sep=8pt
  ]
    \widehat{\mathbb{Z}^2}
    \ar[r, ->>, shorten=-1pt]
    &
    \widehat{\mathbb{Z}^2_o}
    \,,
  \end{tikzcd}
$
where $o$ may be taken to equal
\vspace{-2mm} 
\begin{equation}
  \label{OrderOfRootOfUnity}
  \mathrm{ord}(\zeta)
  :=
  \min\big\{n\in\mathbb{N}_{>0}\;\big\vert\;\zeta^n=1\big\}
  \,,
\end{equation}
the order of $\zeta$ in $\widehat{\zeta} = \zeta \, \mathrm{id}$ {\rm(Lem.~\ref{FinDimHeisenbergIrrepsAndRootsOfUnity})}.

Moreover, for the $\mathrm{aa}$-spin structure the analogous statement holds with $o$ replaced by $2o$.
\end{lemma}
\begin{proof}
Put $U:=\WOperator{1}{0}$ and $V:=\WOperator{0}{1}$.
The Heisenberg group law gives $UV=\zeta^2VU$.
Since $\zeta^o=1$, the operators $U^o$ and $V^o$ commute with
the entire Heisenberg representation.
By Schur's lemma \eqref{SchurLemma}, write
\[
  U^o=\alpha\,\mathrm{id},
  \qquad
  V^o=\beta\,\mathrm{id},
  \qquad
  \alpha,\beta\in\mathrm{U}(1).
\]
For any $a,b\in\mathbb{Z}$, the group law then gives
\[
  \WOperator{oa}{ob}
  =\zeta^{-o^2ab}U^{oa}V^{ob}
  =\alpha^a\beta^b\,\mathrm{id}.
\]
Thus these scalar values form a character of $(o\mathbb{Z})^2$.
Covariance preserves them under the relevant modular action,
by Lem.~\ref{ClassOfExtendibleRepIsGInvariant}\,\textup{(i)}.
For both spin structures, $S$-covariance gives
\begin{equation}
  \label{ImplicationOfSInvarianceOnScalars}
  \alpha=\beta^{-1},\qquad \beta=\alpha
  \qquad\Longrightarrow\qquad
  \alpha=\beta\in\{+1,-1\}.
\end{equation}
For the $\mathrm{pp}$-spin structure, $T$ sends
$\WOperator{0}{o}$ to $\WOperator{o}{o}$, so that
\begin{equation}
  \label{ImplicationOfTInvarianceOnScalars}
  \beta=\alpha\beta
  \qquad\Longrightarrow\qquad
  \alpha=\beta=1.
\end{equation}
Hence $U^o=V^o=\widehat{\zeta}^{\,o}=\mathrm{id}$, and the
representation factors through $\widehat{\mathbb{Z}^2_o}$.
For the $\mathrm{aa}$-spin structure, $T^2$ only imposes
$\beta=\alpha^2\beta$, which already follows from
\eqref{ImplicationOfSInvarianceOnScalars}.
In this case $U^{2o}=V^{2o}=\widehat{\zeta}^{\,2o}=\mathrm{id}$,
so the representation factors through
$\widehat{\mathbb{Z}^2_{2o}}$.
\end{proof}

\begin{lemma}[\bf Pullback from cyclic Heisenberg preserves irreducibility]
\label{PullbackFromCyclicHeisenbergPreservesIrreducibility}
For $o \in \mathbb{N}_{>0}$, given a representation of the cyclic Heisenberg group $\rho : \widehat{\mathbb{Z}^2_o} \xrightarrow{\;} \mathrm{GL}(\HilbertSpace{H})$, its pullback $p^\ast \rho : \widehat{\mathbb{Z}^2} \overset{p}{\twoheadrightarrow} \widehat{\mathbb{Z}^2_o} \xrightarrow{\rho} \mathrm{GL}(\HilbertSpace{H})$ is irreducible iff $\rho$ is.
\end{lemma}
\begin{proof}
  This holds generally for pullback along surjective group homomorphisms.
\end{proof}

\begin{lemma}
[\bf Some irreps of the integer Heisenberg group]  
\label
{SomeIrreps}
Let $\zeta\in\mathrm{U}(1)$ be any root of unity and put
\[
  o:=\mathrm{ord}(\zeta),
  \qquad
  D:=\mathrm{ord}(\zeta^2)
  =\frac{o}{\gcd(o,2)}.
\]
For every $s,t\in\{0,1\}$, the following formulas define a
$D$-dimensional irreducible unitary representation of the integer
Heisenberg group \eqref{IntegerHeisenbergGroup}:
\begin{equation}
\label
{GenericUnitaryRepOfIntHeisenberg}
  \def\arraystretch{1.5}
  \begin{array}{l}
  {}
  \\[5pt]
  \HilbertSpace{H}
  \;:=\;
  \mathbb{C}^D
  \\
  \,\simeq\,
    \mathrm{Span}_{\mathbb{C}}\big(
      \vert 0 \rangle
      ,\,
      \vert 1 \rangle
      ,\,
      \cdots
      ,\,
      \vert D \!-\! 1 \rangle
    \big)
    \\[10pt]
    \big\langle 
      n_1 
    \big\vert
      n_2
    \big\rangle
    \;\;
    :=
    \;\;
    \delta_{n_1 n_2}
  \end{array}
  \;\;\;
  \left\{
  \def\arraystretch{1}
  \begin{tikzcd}[
    ampersand replacement=\&,
    column sep=0pt,
    row sep=-3pt,
    /tikz/column 7/.append style={anchor=base west}
  ]
    \widehat{\mathbb{Z}^2}
    \ar[
      rr
    ]
    \&\&
    \mathrm{U}(\HilbertSpace{H})
    \&[-8pt]
    \&[-1pt]
    \\
    \scalebox{1.4}{$($}
    \adjustbox{
      scale=.6,
      raise=1.5pt
    }{$
    \left[
    \adjustbox{raise=-1pt}{$
    \def\arraystretch{.9}
    \def\arraycolsep{0pt}
    \begin{array}{c}
      1 \\ 0
    \end{array}
    $}
    \right]
    $}
    ,
    0
    \scalebox{1.4}{$)$}
    \&\longmapsto\&
    \WOperator{1}{0}
    \&:\&    
    \vert n \rangle
    \&\mapsto\&
    \;\;\;
    e^{
      s \pi \mathrm{i}/D
    }
    \zeta^{
      2n
    } 
    \, \vert n \rangle
    \\
    \scalebox{1.4}{$($}
    \adjustbox{
      scale=.6,
      raise=1.5pt
    }{$
    \left[
    \adjustbox{raise=-1pt}{$
    \def\arraystretch{.9}
    \def\arraycolsep{0pt}
    \begin{array}{c}
      0 \\ 1
    \end{array}
    $}
    \right]
    $}
    ,
    0
    \scalebox{1.4}{$)$}
    \&\longmapsto\&
    \WOperator{0}{1}
    \&:\&    
    \vert n \rangle
    \&\mapsto\&
    \begin{cases}
      \vert n + 1 \rangle
      & 
      \textup{ if } 
      \; n < D - 1
      \\
      (-1)^{s + t}
      \vert 0 \rangle
      & 
      \textup{ if } 
      \; n = D - 1
    \end{cases}
    \\
    \scalebox{1.4}{$($}
    \adjustbox{
      scale=.6,
      raise=1.5pt
    }{$
    \left[
    \adjustbox{raise=-1pt}{$
    \def\arraystretch{.9}
    \def\arraycolsep{0pt}
    \begin{array}{c}
      0 \\ 0
    \end{array}
    $}
    \right]
    $}
    ,
    1
    \scalebox{1.4}{$)$}
    \&\longmapsto\&
    \widehat{\zeta}
    \&:\&
    \vert n \rangle 
    \&\mapsto\&
    \;\;\;\zeta 
    \, \vert n \rangle \,,
  \end{tikzcd}
  \right.
\end{equation}
Here $n\in\{0,\ldots,D-1\}$, and the scalar powers of the two
standard generators are
\[
  \big(\WOperator{1}{0}\big)^D=(-1)^s\,\mathrm{id},
  \qquad
  \big(\WOperator{0}{1}\big)^D=(-1)^{s+t}\,\mathrm{id}.
\]
\end{lemma}
Here, the representation of general group elements follows from applying the group law to the above generators, for instance:

\begin{equation}
  \label{ModularRepresentationOf11}
  \zeta^{-1}
  \,
  \WOperator{1}{0}
  \WOperator{0}{1}
  \;=\;
  \WOperator{1}{1}
  \;=\;
  \zeta^{+1}
  \,
  \WOperator{0}{1}
  \WOperator{1}{0}
  \,.
\end{equation}
\begin{proof}
  First note that the operators are unitary and that the generating group commutators
  in $\widehat{\mathbb{Z}^2}$ \eqref{IntegerHeisenbergGroup} are evidently respected by the formulas \eqref{GenericUnitaryRepOfIntHeisenberg}, 
  so that they do define a representation of $\widehat{\mathbb{Z}^2}$: 
  \begin{equation}
    \label{HeisenbergGroupCommutator}
    \WOperator{1}{0}
    \circ 
    \WOperator{0}{1}
    \;=\;
    \zeta^2
    \,
    \WOperator{0}{1}
    \circ 
    \WOperator{1}{0}
    \,.
  \end{equation}  

To see that this is irreducible,
put
\[
U := \WOperator{1}{0},
\qquad
V := \WOperator{0}{1},
\qquad
\tilde o := 2D.
\]
Since
\[
U^D = (-1)^s \,\mathrm{id},
\qquad
V^D = (-1)^{s+t} \,\mathrm{id},
\qquad
\zeta^{2D} = 1,
\]
these representations factor through
$\widehat{\mathbb{Z}^2_{\tilde o}}$.
Hence, by Lem.
\ref{PullbackFromCyclicHeisenbergPreservesIrreducibility},
we may equivalently show that they are irreducible as
$\widehat{\mathbb{Z}^2_{\tilde o}}$-representations.
This being a finite group
(of order
$\vert \widehat{\mathbb{Z}^2_{\tilde o}} \vert = \tilde o^3$),
we may invoke Schur-orthonormality
\eqref{SchurOrthonormality} of irreps.

By the group law, we have
\[
\WOperator{a}{b}
\circ
\widehat{\zeta}^{\,c}
=
\zeta^{c-ab} U^a V^b.
\]
If $b \neq 0 \,\mathrm{mod}\, D$, this operator has
vanishing diagonal entries, hence vanishing trace.
Otherwise,
\[
V^b
=
(-1)^{(s+t)b/D}\,\mathrm{id},
\]
and therefore
\[
\mathrm{tr}\Big(
\WOperator{a}{b}
\circ
\widehat{\zeta}^{\,c}
\Big)
=
\zeta^{c-ab}
e^{s\pi\mathrm{i}a/D}
(-1)^{(s+t)b/D}
\sum_{n=0}^{D-1}
\zeta^{2na}.
\]
Since $\mathrm{ord}(\zeta^2)=D$, the sum on the right
is $D$ if $a=0 \,\mathrm{mod}\, D$ and vanishes otherwise,
by \eqref{DiscreteFourierTransformOfKroneckerDelta}.
Thus the
$\widehat{\mathbb{Z}^2_{\tilde o}}$ character components,
for $a,b,c \in \{0,\cdots,\tilde o-1\}$, are:
\begin{equation}
\label{TheRepCharacter}
\rchi_{
\big(
\hspace{-1pt}
\adjustbox{
scale=.65,
raise=+.7pt
}{$
\Big[
\def\arraycolsep{0pt}
\def\arraystretch{.9}
\begin{array}{c}
a
\\
b
\end{array}
\Big]
$}
,\,
c
\hspace{-1pt}
\big)
}
\;:=\;
\mathrm{tr}\Big(
\WOperator{a}{b}
\circ
\widehat{\zeta}^{\,c}
\Big)
\;=\;
\begin{cases}
D\,\zeta^{c-ab}
(-1)^{
s\frac{a}{D}
+
(s+t)\frac{b}{D}
}
&
\text{ if }
a = b = 0 \mod D
\\
0
&
\text{ otherwise. }
\end{cases}
\end{equation}
There are exactly $\tilde o/D$ choices each of $a$ and $b$
for which this character is nonzero.
On these elements its absolute value is $D$.
Hence its Schur-norm square is unity:
\[
\frac{1}{
\big\vert
\widehat{\mathbb{Z}^2_{\tilde o}}
\big\vert
}
\sum_{a,b,c=0}^{\tilde o-1}
\Big\vert
\rchi_{
\big(
\hspace{-1pt}
\adjustbox{
scale=.65,
raise=+.7pt
}{$
\Big[
\def\arraycolsep{0pt}
\def\arraystretch{.9}
\begin{array}{c}
a
\\
b
\end{array}
\Big]
$}
,\,
c
\hspace{-1pt}
\big)
}
\Big\vert^2
  \;=\;
\frac{1}{\tilde o^3}
\sum_{\substack{
0\leq a,b<\tilde o
\\
a=b=0 \,\mathrm{mod}\, D
}}
\sum_{c=0}^{\tilde o-1}
D^2
\;=\;
\frac{1}{\tilde o^3}
\left(\frac{\tilde o}{D}\right)^2
\tilde o\,D^2
\;=\;
1,
\]
which proves irreducibility.
\end{proof}

\begin{proposition}
[\bf Classification of covariantizable irreps of the integer Heisenberg group]
\label
{ClassificationOfCovariantizableIrrepsOfIntegerHeisenberg}
Any finite-dimensional irreducible unitary representation of
$\widehat{\mathbb{Z}^2}$ which is covariantizable
\textup{(Def.~\ref{CovariantizableRepsOfHeisenberg})}
is isomorphic to one from \textup{Lem.~\ref{SomeIrreps}}.
In the notation of that lemma, $t=0$ for both spin structures, and
\begin{itemize}
  \item for the $\mathrm{pp}$-spin structure,
  \[
    (-1)^s=\zeta^{D^2},
    \qquad\textup{equivalently}\qquad
    s=
    \begin{cases}
      1 & \textup{if }o\equiv2\pmod4,\\
      0 & \textup{otherwise};
    \end{cases}
  \]
  \item for the $\mathrm{aa}$-spin structure, either $s=0$ or $s=1$.
\end{itemize}
\end{proposition}

\begin{proof}
By Lem.~\ref{FinDimHeisenbergIrrepsAndRootsOfUnity},
the central generator $\widehat{\zeta}$ acts by multiplication
with a root of unity $\zeta$.
Put
\[
U := \WOperator{1}{0},
\qquad
V := \WOperator{0}{1},
\qquad
D := \mathrm{ord}(\zeta^2).
\]
Since
$U V = \zeta^2 V U$,
the operators $U^D$ and $V^D$ commute with all
Heisenberg representation operators.
Hence Schur's lemma implies that
\[
U^D=\alpha\,\mathrm{id},
\qquad
V^D=\beta\,\mathrm{id},
\qquad
\alpha,\beta\in \mathrm{U}(1).
\]

Next we see that $S$-covariance forces $t=0$
for both spin structures. Namely, the modular generator $S$ acts by
\[
U\longmapsto V^{-1},
\qquad
V\longmapsto U.
\]
Covariance therefore exhibits the given representation
as isomorphic to its translate under this action.
Since scalar operators are preserved under conjugation,
taking $D$th powers gives
$
\alpha=\beta^{-1}$
and
$\beta=\alpha$, and hence
\[
  \alpha
    =
  \beta\in\{+1,-1\}.
\]
Writing these signs in the $(s,t)$-parameterization
of Lem.~\ref{SomeIrreps},
\[
\alpha=(-1)^s,
\qquad
\beta=(-1)^{s+t},
\qquad
s,t\in\{0,1\},
\]
the equality $\alpha=\beta$ says precisely that
\[
(-1)^t=1,
\qquad\text{hence}\qquad
t=0.
\]
Thus $S$-covariance leaves only the two possibilities
\[
(\alpha,\beta)
=
\big((-1)^s,(-1)^s\big),
\qquad s\in\{0,1\}.
\]

We now identify the given representation with the
corresponding representation of Lem.~\ref{SomeIrreps}.
Since
\[
U^{2D}=V^{2D}=\mathrm{id},
\qquad
\zeta^{2D}=1,
\]
the representation factors through the cyclic
Heisenberg group
$\widehat{\mathbb{Z}^2_{\tilde o}}$
with $\tilde o:=2D$.
By Lem.
\ref{PullbackFromCyclicHeisenbergPreservesIrreducibility},
the resulting representation of this finite group
is irreducible.

The full center of
$\widehat{\mathbb{Z}^2_{\tilde o}}$
is generated by the distinguished central generator
and by the $D$th powers of the two standard generators.
Its character in the given representation is therefore
specified by
\[
\big(\zeta,\alpha,\beta\big)
=
\big(\zeta,(-1)^s,(-1)^s\big).
\]
Moreover,
\[
\widehat{\mathbb{Z}^2_{\tilde o}}
\big/
Z\big(
\widehat{\mathbb{Z}^2_{\tilde o}}
\big)
\;\simeq\;
(\mathbb{Z}/D\mathbb{Z})^2,
\]
and the commutator pairing induced by this central
character is
\[
\big((a,b),(a',b')\big)
\longmapsto
\zeta^{\,2(ab'-a'b)}.
\]
This pairing is nondegenerate because
$\zeta^2$ has order $D$.
Hence the finite
{\it Stone-von Neumann--Mackey theorem},
applied with this full central character
(cf. \cite[\S 4.1]{Prasad11}),
gives uniqueness of the irreducible representation
with these central-character data.
The representation of Lem.~\ref{SomeIrreps}
with parameters $(s,t)=(s,0)$ has precisely these data,
so the given representation is isomorphic to it.

Finally, $T$-covariance has the effect of fixing
the remaining parameter $s$ for the $\mathrm{pp}$-spin structure.
Namely, the modular generator $T$ acts by
\[
U\longmapsto U,
\qquad
V\longmapsto
\WOperator{1}{1}
=
\zeta^{-1}UV.
\]
Using the commutation relation, we obtain
\[
\big(\zeta^{-1}UV\big)^D
=
\zeta^{-D^2}U^DV^D
=
\zeta^{-D^2}\alpha\beta\,\mathrm{id}.
\]
Since covariance preserves the scalar operator $V^D$,
this implies
\[
\beta
=
\zeta^{-D^2}\alpha\beta,
\qquad\text{hence}\qquad
(-1)^s=\alpha=\zeta^{D^2}.
\]
If $o=\mathrm{ord}(\zeta)$ is odd, then $D=o$,
so that $\zeta^{D^2}=1$.
If $o$ is even, then $o=2D$ and $\zeta^D=-1$,
so that
\[
\zeta^{D^2}=(-1)^D.
\]
Therefore, in the $\mathrm{pp}$ case,
$s=1$ precisely when $o=2\,\mathrm{mod}\,4$,
and $s=0$ otherwise.

For the $\mathrm{aa}$-spin structure,
only $T^2$-covariance is required.
Its corresponding scalar condition is
\[
\beta
=
\zeta^{-2D^2}\alpha^2\beta
=
\beta,
\]
which imposes no further restriction on $s$.
\end{proof}

\medskip

We now turn to the actual construction of the modular covariantization of these representations. 

\paragraph
{Covariantization for $\mathrm{pp}$ spin structure}

We begin with the $\mathrm{pp}$-spin structure and the case
$(s,t)=(0,0)$ in \eqref{GenericUnitaryRepOfIntHeisenberg}, taking
$\p=1$ and even $\lattice$.

\begin{proposition}
[\bf Basic 2-Cohomotopical quantum states over the torus]
\label
{2CohomotopicalQuantumStatesOnTorus}
$\,$

  \noindent   
  Unitary representations of 
  $\mathrm{SL}_2(\mathbb{Z})\ltimes\widehat{\mathbb{Z}^2}$\;\eqref{CohomotopicalFluxMonodromyOverTorus} --- and hence spaces of quantum states \eqref{StateSpacesAsLocalSystemsOnModuliSpace}
  for 2-cohomotopical flux over the torus ---
  irreducible already in their restriction to $\widehat{\mathbb{Z}^2}$,
  are obtained for all even positive integers
  \begin{equation}
    \label{kAndZetaInBasicCase}
    \lattice \in 2\mathbb{N}_{> 0}
    \;\;\;\;\;\;\;\;
    \mbox{\rm with}
    \;\;\;\;\;\;\;\;
    \zeta 
      \,:=\, 
    e^{
      \tfrac
        { \pi \mathrm{i} }
        { \lattice }
    }
    \,,
  \end{equation}
  by the following formulas:
  \vspace{-.1cm} 
  \begin{equation}
    \label{IrrepOfSemidirectProduct}
  \begin{array}{l}
    \HilbertSpace{H}_{T^2}
    \;\coloneqq\;
    \mathbb{C}^{\lattice}
    \\
    \,\simeq\,
    \mathrm{Span}\big(
      \vert 0 \rangle
      ,\,
      \vert 1 \rangle
      ,\,
      \cdots
      ,\,
      \vert \lattice \!-\! 1 \rangle
    \big)
  \end{array}
    \;\;\;
    \left\{\!\!
    \adjustbox{raise=4pt}{
    \begin{tikzcd}[
      ampersand replacement=\&,
      row sep=-4pt,
      column sep=0pt
    ]
      \mathrm{SL}_2(\mathbb{Z})
      \,\ltimes\,
      \widehat{\mathbb{Z}^2}
      \ar[
        rr
      ]
      \&\&
      \mathrm{U}\big(
        \HilbertSpace{H}_{T^2}
      \big)
      \&[-13pt]
      \\
      \Big(
        \mathrm{I}
        ,\,
        \scalebox{1.4}{$($}
        \adjustbox{
          scale=.6,
          raise=1.5pt
        }{$
        \left[
        \adjustbox{raise=-1pt}{$
        \def\arraystretch{.9}
        \def\arraycolsep{0pt}
        \begin{array}{c}
          1 \\ 0
        \end{array}
        $}
        \right]
        $}
        ,
        0
        \scalebox{1.4}{$)$}
      \Big)
      \&\longmapsto\& 
        \WOperator{1}{0}
      \&:\&
      \big\vert n \big\rangle
      \&\mapsto\&
      \zeta^{2 n }
      \,
      \big\vert n \big\rangle
      \\
      \Big(
        \mathrm{I}
        ,\,
        \scalebox{1.4}{$($}
        \adjustbox{
          scale=.6,
          raise=1.5pt
        }{$
        \left[
        \adjustbox{raise=-1pt}{$
        \def\arraystretch{.9}
        \def\arraycolsep{0pt}
        \begin{array}{c}
          0 \\ 1
        \end{array}
        $}
        \right]
        $}
        ,
        0
        \scalebox{1.4}{$)$}
      \Big)
      \&\longmapsto\&
      \WOperator{0}{1}
      \&:\&
      \big\vert n \big\rangle
      \&\mapsto\&
      \big\vert 
        (n+1)
        \mod
        \lattice 
      \big\rangle
      \\
      \Big(
        \mathrm{I}
        ,\,
        \scalebox{1.4}{$($}
        \adjustbox{
          scale=.6,
          raise=1.5pt
        }{$
        \left[
        \adjustbox{raise=-1pt}{$
        \def\arraystretch{.9}
        \def\arraycolsep{0pt}
        \begin{array}{c}
          0 \\ 0
        \end{array}
        $}
        \right]
        $}
        ,
        1
        \scalebox{1.4}{$)$}
      \Big)
      \&\longmapsto\&
      \widehat{\zeta}
      \&:\&
      \big\vert n \big\rangle
      \&\mapsto\&
      \zeta
      \;
      \big\vert n \big\rangle
      \\
      \Big(
        S
        ,\,
        \scalebox{1.4}{$($}
        \adjustbox{
          scale=.6,
          raise=1.5pt
        }{$
        \left[
        \adjustbox{raise=-1pt}{$
        \def\arraystretch{.9}
        \def\arraycolsep{0pt}
        \begin{array}{c}
          0 \\ 0
        \end{array}
        $}
        \right]
        $}
        ,
        0
        \scalebox{1.4}{$)$}
      \Big)
      \&\longmapsto\&
      \widehat{S}
      \&:\&
      \big\vert
        n
      \big\rangle
      \&\mapsto\&
      \tfrac{1}{\sqrt{\lattice}}
      \sum_{\widehat{n}=0}^{\lattice-1}
      \zeta^{
          2
          n
          \widehat{n}
      }
      \big\vert \widehat{n} \big\rangle
      \\
       \Big(
        T
        ,\,
        \scalebox{1.4}{$($}
        \adjustbox{
          scale=.6,
          raise=1.5pt
        }{$
        \left[
        \adjustbox{raise=-1pt}{$
        \def\arraystretch{.9}
        \def\arraycolsep{0pt}
        \begin{array}{c}
          0 \\ 0
        \end{array}
        $}
        \right]
        $}
        ,
        0
        \scalebox{1.4}{$)$}
      \Big)
     \&\longmapsto\&
      \widehat{T}
      \&:\&
        \big\vert
          n
        \big\rangle
      \&\mapsto\&
      \tfrac
        {1}
        {c_\zeta}
      \,
      \zeta^{
        (n^2)
      }
      \,
      \big\vert
        n
      \big\rangle
      \mathrlap{\,,}
    \end{tikzcd}
    }
    \right.
  \end{equation}
  where the normalization factor
  $c_\zeta \in \mathbb{C}$ in the last line is any cube root \begin{equation}
  \label
  {TheNormalizationFactor}
  c_{\zeta}
  \;:=\;
  \Big(
  \tfrac{1}{\sqrt{\lattice}}
  \textstyle{\sum_{n=0}^{\lattice-1}}
  \,
  \zeta^{(n^2)}
  \Big)^{1/3}
  \,,
\end{equation}
which with \cref{kAndZetaInBasicCase} and by \cref{QuadraticGaussSumWithHalfExponent} may here be taken to be
$  \big(
    e^{\pi \mathrm{i}/4}
    \,
  \big)^{1/3}
  =
  e^{ \pi \mathrm{i} / 12}
  $.
\end{proposition}
\begin{proof}
  {\bf (i)}
  The restriction to $\widehat{\mathbb{Z}^2}$ is the irreducible unitary
  representation of Lem.~\ref{SomeIrreps} with $D=\lattice$ and
  $(s,t)=(0,0)$. Indeed, $\mathrm{ord}(\zeta^2)=\lattice$ and
  $\dim(\HilbertSpace{H}_{T^2})=\lattice$; the scalar sector is the one
  required by Prop.~\ref{ClassificationOfCovariantizableIrrepsOfIntegerHeisenberg}
  for these parameters.

\noindent   
{\bf (ii)} The Fourier operator $\widehat S$ is unitary by
\eqref{DiscreteFourierTransformOfKroneckerDelta}, and $\widehat T$ is
unitary because $|c_\zeta|=1$ for the normalization above.
To see that we also have a representation of the subgroup $\mathrm{SL}_2(\mathbb{Z})$, it is sufficient to show that the operators $\widehat{S}$ and $\widehat{T}$ respect the relations \eqref{PresentationOfSL2Z}.

Now first, we find
\begin{equation}
  \label{SquareOfHatS}
  \begin{array}{rcll}
  \widehat{S}
  \widehat{S}
  {\big\vert n \big\rangle}
  &\defneq&
  \widehat{S}
  \Big(
  \frac{1}{\sqrt{ \lattice }}
  \sum_{\widehat n}
  \,
  e^{
    \tfrac
      { 2 \pi \mathrm{i} }
      { \lattice }
    \widehat{n} \, n
  }
  {\big\vert \widehat{n} \big\rangle}
  \Big)
  \\
  &\defneq&
  \textstyle{
    \sum_{\widehat{\widehat n}}
  }
  \;
  \grayunderbrace{
  \tfrac{1}{\lattice}
  \textstyle{\sum_{\widehat n}}
  \,
  e^{
    \tfrac
      { 2 \pi \mathrm{i} }
      { \lattice }
    \widehat{n} \, 
      (n + \widehat{\widehat{n}})
  }
  }{
    \delta_0\big(
      n + \widehat{\widehat{n}}
      \,\mathrm{mod}\,
      \lattice
    \big)
  }
  \,
  {\big\vert
    \widehat{\widehat{n}}
  \big\rangle}
  \\
  &=&
  {\big\vert -n \,\mathrm{mod}\, \lattice \big\rangle}
  &
  \proofstep{
    by
    \eqref{DiscreteFourierTransformOfKroneckerDelta}.
  }
  \end{array}
\end{equation}
This immediately implies that $\widehat{S}^4 = \mathrm{id}$ and that, with 
\cref{TheNormalizationFactor},
\[
  \widehat{T} 
  \,
  \widehat{S}
    {\big\vert n \big\rangle}
  \;=\;
  \tfrac{1}{ 
    \mathclap{\phantom{\vert^{\vert}}}
    \lattice^{1/2} c_{\zeta}
  }
  \textstyle{
    \sum_{
        \widehat n 
    }
  }
  \,
  e^{
     \tfrac{
       \pi \mathrm{i}
     }{\lattice}
     (
       \widehat{n}^2
       +
       2\widehat{n} \, n
     )
  }
  {\big\vert \widehat{n} \big\rangle}
  \,,
\]
also $\widehat{S}^2 (\widehat{T} \widehat{S}) = (\widehat{T}\widehat{S}) \widehat{S}^2$: the diagonal phase of $\widehat T$ is even and $\lattice$-periodic since $\lattice$ is even. Hence the only remaining relation to check is $(\widehat{T}\widehat{S})^3 = \mathrm{id}$ or equivalently that 
\[
  \widehat{T}^{-1}
  \circ
  \widehat{S}^{-1}
  \circ
  \widehat{T}^{-1}
  \;=\;
  \widehat{S} 
  \circ
  \widehat{T}
  \circ
  \widehat{S}
  \,.
\]
Unwinding the definitions gives
\begin{equation}
  \label{RepresentingTSCubeRelation}
  \def\arraystretch{1.7}
  \begin{array}{l}
    \widehat{T}^{-1} \widehat{S}^{-1} \widehat{T}^{-1}
    {\vert n \rangle}
    \;=\;
    c_{\zeta}
    \widehat{T}^{-1} \widehat{S}^{-1}
    e^{ 
      - \tfrac{\pi \mathrm{i}}{\lattice} n^2 
    }
    {\vert n \rangle}
    \\
    \;=\;
    \widehat{T}^{-1}
    \tfrac
    {%
      \rule[-2pt]{0pt}{0pt}%
      c_{\zeta}
    }
    {
      \sqrt{\lattice}
    }
    \sum_{\widehat{n}}
    e^{ 
      \tfrac{\pi \mathrm{i}}{\lattice} 
      (
        - n^2 - 2 \widehat{n} n 
      )   
    }
    {\vert \widehat{n} \rangle}
    \\
    \;=\;
    \tfrac
    {%
      \rule[-3pt]{0pt}{0pt}%
      c_{\zeta}^2
    }
    {
      \sqrt{\lattice}
    }
    \sum_{\widehat{n}}
    e^{ 
      \tfrac{\pi \mathrm{i}}{\lattice} 
      (
        - n^2 - 2 \widehat{n} n - \widehat{n}^2
      )   
    }
    {\vert \widehat{n} \rangle}
    \\
    \;=\;
    \tfrac
    {%
      \rule[-3pt]{0pt}{0pt}%
      c_{\zeta}^2
    }
    {
      \sqrt{\lattice}
    }
    \sum_{\widehat{n}}
    e^{
      - 
      \tfrac{\pi \mathrm{i}}{\lattice} 
      (\widehat{n} + n)^2
    }
    {\vert \widehat{n} \rangle}
  \end{array}
  \hspace{.2cm}
  \mbox{and}
  \hspace{.6cm}
  \def\arraystretch{1.7}
  \begin{array}{l}
    \widehat{S} 
    \widehat{T} 
    \widehat{S}
    {\big\vert n \big\rangle}
  \;=\;
  \widehat{S} 
  \widehat{T} 
  \frac{1}{\sqrt{\lattice}}
  \sum_{
      \widehat n
  }
  \,
  e^{
    \tfrac
      { 2 \pi \mathrm{i} }
      { \lattice }
    \widehat{n} \, n
  }
  {\big\vert \widehat{n} \big\rangle}
  \\
  \;=\;
  \widehat{S}
  \,
  \tfrac
    {%
      \rule[-3pt]{0pt}{0pt}
      c_{\zeta}^{-1} 
    }
    {\sqrt{\lattice}}
  \sum_{
      \widehat n
  }
  \,
  e^{
    \tfrac{\pi \mathrm{i}}{\lattice}
    (
      2 \widehat{n} \, n
      +
      \widehat{n}^2
    )
  }
  {\big\vert \widehat{n} \big\rangle}
  \\
  \;=\;
  \frac
    {%
     \rule[-3pt]{0pt}{0pt}
     c_{\zeta}^{-1} 
   }
   {\lattice}
  \sum_{
    \widehat n
    ,\,
    \widehat{\widehat{n}}
  }
  \,
  e^{
    \tfrac{\pi \mathrm{i}}{\lattice}
    (
      2 \widehat{n} \, n
      +
      \widehat{n}^2
      +
      2 \widehat{\widehat{n}} \widehat{n}
    )
  }
  {\big\vert \widehat{\widehat{n}} \big\rangle}
  \\
  \;=\;
  \frac
    {%
      \rule[-3pt]{0pt}{0pt}
      c_{\zeta}^{-1} 
    }
    {\sqrt{\lattice}}
  \sum_{ \widehat{\widehat{n}} }
  \underbrace{
  \tfrac{1}{\sqrt{\lattice}}
  \textstyle{\sum_{ \widehat n }}
  \,
  e^{
    \tfrac{\pi \mathrm{i}}{\lattice}
      (\widehat{n} + (n + \widehat{\widehat{n}}))^2
  }
  }_{ 
      c_{\zeta}^3
    }
  e^{
    -
    \tfrac{\pi \mathrm{i}}{\lattice}
      (n + \widehat{\widehat{n}})^2
  }
  {\big\vert \widehat{\widehat{n}} \big\rangle}
  \,.
\end{array}
\end{equation}
The term over the brace is a constant in $n$ and $\widehat{\widehat{n}}$, by the assumption that $\lattice$ is even,%
\footnote{
  \label{NeedForEvenk}
  Since the summands in 
  $
   \sum_{n = 0}^{\lattice-1}
   e^{
     \tfrac{\pi \mathrm{i}}{\lattice}
     n^2
   }
 $
 are $\lattice$-periodic for even $\lattice$, 
  $
    e^{\tfrac{\pi \mathrm{i}}{\lattice} 
    (n + \lattice)^2}
    =
    e^{
      \tfrac{\pi \mathrm{i}}{\lattice} 
      n^2
    }
    e^{\pi \mathrm{i}(2n + \lattice)}
   \; \underset{
      \mathclap{\scalebox{.6}{$\lattice$ even}}
    }{=}\;
    e^{\tfrac{\pi \mathrm{i}}{\lattice} n^2}
  $,
  the sum is invariant under replacing $n \mapsto n + a$ for $a \in \mathbb{N}$.
}
whence the relation is satisfied if the normalization factor $c_{\zeta}$ is chosen as claimed in \cref{TheNormalizationFactor}.

\noindent {\bf (iii)} Finally, we need to see that the semidirect product structure is respected, hence that 
$$
  \ActedWOperator{a}{b}{M}
  \,
  \widehat{M} 
  {\vert n \rangle}
  \;=\;   
  \widehat{M}
  \,
  \WOperator{a}{b}
    {\vert n \rangle}
  \hspace{.6cm}
  \forall
  \left\{\!\!\!
  \def\arraystretch{1}
  \begin{array}{ccl}
    M &\in&  \mathrm{SL}_2(\mathbb{Z})
    \\
    (a,b) &\in& \mathbb{Z}^{2}
    \\
    {\vert n \rangle} &\in&
    \HilbertSpace{H}_{T^2}
    \mathrlap{\,.}
  \end{array}
  \right.
$$
It is sufficient to check this on the generators; the central generator is fixed because its operator is scalar. In the computations below, all summation and state indices are taken modulo $\lattice$. Explicitly:
$$
  \def\arraystretch{1.9}
  \begin{array}{ccl}
    \ActedWOperator{1}{0}{S}
    \,
    \widehat{S}
      {\big\vert
        [n]
      \big\rangle}
    &\equiv&
    \WOperator{0}{1}^{-1}
    \Big(
  \frac{1}{\sqrt{\vert \lattice \vert}}
  \sum_{ \widehat n }
  \,
  e^{
    \tfrac
      { 2 \pi \mathrm{i} }
      { \lattice }
    \widehat{n} \, n
  }
  {\vert\widehat{n}\rangle}
  \Big)
  \\
  &=&
  \frac{1}{\sqrt{\lattice}}
  \sum_{ \widehat n }
  \,
  e^{
    \tfrac
      { 2 \pi \mathrm{i} }
      { \lattice }
    (\widehat{n} + 1) \, n
  }
  {\big\vert \widehat{n} \big\rangle}
  \\
  &=&
  e^{
    \tfrac
      { 2 \pi \mathrm{i} }
      { \lattice }
    n
  }
  \,
    \widehat{S}
    {\big\vert n \big\rangle}
  \\
  &=&
  \widehat{S}
  \,
    \WOperator{1}{0}
    {\vert n \rangle}
  \mathrlap{\,,}
  \end{array}
  \hspace{1.5cm}
  \def\arraystretch{1.9}
  \begin{array}{ccl}
    \ActedWOperator{0}{1}{S}
    \,
    \widehat{S}
    {\vert n \rangle}
    &\equiv&
    \WOperator{1}{0}
    \Big(
  \frac{1}{\sqrt{\lattice}}
  \sum_{ \widehat n }
  \,
  e^{
    \tfrac
      { 2 \pi \mathrm{i} }
      { \lattice }
    \widehat{n} \, n
  }
  {\vert \widehat{n} \rangle}
  \Big)
  \\
  &=&
  \frac{1}{\sqrt{\lattice}}
  \sum_{ \widehat n }
  \,
  e^{
    \tfrac
      { 2 \pi \mathrm{i} }
      { \lattice }
    \widehat{n}
  }
  e^{
    \tfrac
     { 2 \pi \mathrm{i} }
     { \lattice }
    \widehat{n} \, n
  }
  {\vert \widehat{n} \rangle}
  \\
  &=&
  \frac{1}{\sqrt{\lattice}}
  \sum_{ \widehat n }
  \,
  e^{
    \tfrac
      { 2 \pi \mathrm{i} }
      { \lattice }
    \widehat{n} \, (n+1)
  }
  {\vert \widehat{n} \rangle}
  \\
  &=&
  \widehat{S}
  \,
  \WOperator{0}{1}
    {\vert n \rangle}
  \mathrlap{\,,}
  \end{array}
$$
and
$$
  \def\arraystretch{1.9}
  \begin{array}{ccl}
    \ActedWOperator{1}{0}{T}
    \,
    \widehat{T}
      {\vert  n \rangle}
    &\equiv&
    \WOperator{1}{0} 
    \,
    \tfrac{1}{c_{\zeta}}
    e^{
      \tfrac
        { \pi \mathrm{i} }
        { \lattice }
      n^2
    }
    {\vert n \rangle}
    \\
    &=&
    \tfrac{1}{c_{\zeta}}
    e^{
      \tfrac
        { 2 \pi \mathrm{i} }
        { \lattice }
      n
    }
    e^{
      \tfrac
        { \mathrm{i} \pi }
        { \lattice }
      n^2
    }
    {\vert n \rangle}
    \\
    &=&
    \widehat{T}
    \,
      \WOperator{1}{0}
      {\vert n \rangle}
    \mathrlap{\,,}
  \end{array}
  \hspace{1.5cm}
  \def\arraystretch{1.9}
  \begin{array}{ccl}
    \ActedWOperator{0}{1}{T}
    \,
    \widehat{T}
      {\vert n \rangle}
    &\equiv&
    \tfrac{1}{c_{\zeta}}
    \WOperator{0}{1} 
    \, 
    \WOperator{1}{0}
    \,
    e^{
      \tfrac
        { \pi \mathrm{i} }
        { \lattice }
    }
    \,
    e^{
      \tfrac
        { \pi \mathrm{i} }
        { \lattice }
      n^2
    }
    {\vert n \rangle}
    \\
    &=&
    \tfrac{1}{c_{\zeta}}
    e^{
      \tfrac
        { \pi \mathrm{i} }
        { \lattice }
      (
        n^2
        +
        2n
        +
        1
      )
    }
    {\vert n + 1 \rangle}
    \\
    &=&
    \tfrac{1}{c_{\zeta}}
    e^{
      \tfrac
        { \pi \mathrm{i} }
        { \lattice }
      (n+1)^2
    }
    {\vert n + 1 \rangle}
    \\
    &=&
    \widehat{T}
    \,
    \WOperator{0}{1}
    {\vert n \rangle}
    \mathrlap{\,,}
  \end{array}
$$
where in the first step of the last case we used \eqref{ModularRepresentationOf11}.
\end{proof}

\begin{remark}[\bf Comparison to modular data of abelian Chern-Simons theory on the torus]
\label{ComparisonToModularDataOfAbelianCS}
 The content of Prop.~\ref{2CohomotopicalQuantumStatesOnTorus} captures the {\it modular data} (cf. \cite{Gannon05}) of abelian Chern-Simons theory:

\begin{itemize}[leftmargin=.8cm]  
\item[{\bf (i)}]
The algebra 
\eqref{HeisenbergGroupCommutator}
of the $\WOperator{a}{b}$
is just that expected \cite[(5.28)]{Tong16} of quantum observables for anyonic topological order on the torus as predicted \cite[(17)]{BosNair89}\cite[(32)]{Polychronakos90}\cite[Prop. 2.2]{GelcaUribe10} by abelian Chern-Simons theory at {\it level} $\level = \lattice/2 \in \mathbb{Z}$ \eqref{ChernSimonsLevel}, and equivalently by $\mathrm{U}(1)$-WZW conformal field theory \cite[(4.3-4)]{Verlinde88}.

\item[{\bf (ii)}] Similarly, the operators $\widehat{S}$ and $\widehat{T}$ according to \eqref{IrrepOfSemidirectProduct} 
 implement a linear normalization of the known modular action on quantum states of abelian Chern-Simons theory \cite[(5.3)]{Wen90}\cite[pp. 65--67, arXiv version]{Manoliu98} (following \cite{Gocho90}\cite[(5,7)]{Funar95}) and equivalently of conformal characters of the $\mathrm{U}(1)$ 2D CFT \cite[Ex. 1]{Gannon05}.%
 \footnote{
   The canonical choice $c_{\zeta} = e^{2\pi \mathrm{i} / 24}$ in \cref{TheNormalizationFactor} supplies the scalar factor $c_{\zeta}^{-1}$ in $\widehat{T}$ for our linear normalization (cf. \cite[(3.1b)]{Gannon05}\cite[(26)]{Schweigert14}). Projective treatments such as \cite[pp. 65--67, arXiv version]{Manoliu98} encode the phase anomaly through a projective multiplier or a central extension. The absence of this scalar from the projective matrices therefore does not imply that the anomaly was missed.
 }

\item[{\bf (iii)}]
Prop.~\ref{2CohomotopicalQuantumStatesOnTorus} also records the compatibility of these two actions: jointly, the operators give a representation of the semidirect product of the modular group with the integer Heisenberg group.

\end{itemize} 

\end{remark}

\smallskip

The basic construction of Prop.~\ref{2CohomotopicalQuantumStatesOnTorus}
captures the braiding angle \emph{unit fractions}
$\braidingAngle=1/\lattice=1/(2\level)$ with even denominator.
Further irreps of the covariantized flux monodromy group yield general
rational braiding angles, including non-unit numerators and odd
denominators. These are statements about braiding phases;
assigning a filling fraction requires additional physical data
(cf. Rem. \ref{RelationBetweenBraidingPhasesAndFillingFraction}).
The comparison with $K$-matrix predictions is developed below in
Props. \ref{ComparingToKMatrixPredictionsInJainstates} and
\ref{AnalysisOfHalperinStates}. The more general representations are as follows:

\begin{lemma}[\bf More general representations]
\label{GeneralRepsOverTheTorus}
For
\begin{subequations}
\label{GeneralZeta}
\begin{align}
&
\zeta
:=
e^{\pi\mathrm{i}\frac{\p}{\lattice}}
\,,
\textup{ where }
\\
\label{GeneralKp}
&
(\lattice,\p)
\in
\mathbb{N}_{>0}\times\mathbb{Z}
\,,\;
\gcd(\p,\lattice) = 1
\mathrlap{\,,}
\end{align}
\end{subequations}
\begin{enumerate}[label=\textup{(\roman*)}]
\item
\label
{GeneralRepAtEvenKP}
For  $\lattice\p\in2\mathbb{Z}$
the same formulas \eqref{IrrepOfSemidirectProduct}
constitute a unitary representation $\rho^{(0)}_\zeta$ of
$\mathrm{SL}_2(\mathbb{Z})\ltimes\widehat{\mathbb{Z}^2}$.
\item
Generally, 
writing
\begin{subequations}
\begin{align}
  \varepsilon 
  &
  := \lattice\p \mod 2 \in \{0,1\}
  \mathrlap{\,,}
  \\
  \label{GeneralTorusTwistParameters}
  \eta
  &
  :=
  (-1)^\varepsilon\zeta  
  =
  e^{
    \pi\mathrm{i}
    \frac{\p+\varepsilon\lattice}{\lattice}
   }
   \mathrlap{\,,}
   \\
\label
{ModularInvariantHeisenbergSignCharacter}
&
\begin{aligned}
\mathllap{
  \chi
  \;:\;
}
\mathrm{SL}_2(\mathbb{Z})
\ltimes
\widehat{\mathbb{Z}^2}
&\longrightarrow
\{+1,-1\},
\\
   \Big(
    M
    ,\,
    \scalebox{1.4}{$($}
    \adjustbox{
      scale=.6,
      raise=1.5pt
    }{$
    \left[
    \adjustbox{raise=-1pt}{$
    \def\arraystretch{.9}
    \def\arraycolsep{0pt}
    \begin{array}{c}
      a 
      \\ 
      b
    \end{array}
    $}
    \right]
    $}
    ,
    c
    \scalebox{1.4}{$)$}
  \Big)
&\longmapsto
  (-1)^{ab + a + b + c}
  \mathrlap{\,,}
\end{aligned}
\end{align}
\end{subequations}
where $M\in\mathrm{SL}_2(\mathbb{Z})$ is arbitrary, the representation
$\rho_\eta^{(0)}$ of part~\ref{GeneralRepAtEvenKP} exists, and its tensor product
with the character \eqref{ModularInvariantHeisenbergSignCharacter},
\begin{equation}
\label{GeneralTorusRepresentationBySignTwist}
\rho_\zeta
  :=
\chi^\varepsilon
  \otimes
\rho^{(0)}_\eta
\mathrlap{\,,}
\end{equation}
is a unitary representation of
$\mathrm{SL}_2(\mathbb{Z})
\ltimes
\widehat{\mathbb{Z}^2}$
with central phase $\zeta$ and satisfying
\[
\WOperator{\lattice}{0}
=
\WOperator{0}{\lattice}
=
(-1)^\varepsilon\,\mathrm{id},
\]
corresponding to the scalar sector
$(s,t)=(\varepsilon,0)$
in 
\textup{Lem.~\ref{SomeIrreps}}.
\end{enumerate}
\end{lemma}
\begin{proof}
\noindent
{\bf (i)}
The proof of
Prop.~\ref{2CohomotopicalQuantumStatesOnTorus}
extends with all factors of
$e^{\pi\mathrm{i}/\lattice}$
replaced by $\zeta$ \eqref{GeneralZeta}.
Fourier orthogonality still applies because $\zeta^2$ has order
$\lattice$. Both the Gauss-sum manipulation in
\eqref{RepresentingTSCubeRelation} and the boundary case of
$T$-covariance use the $\lattice$-periodicity of the quadratic phase,
which follows from $\lattice\p\in2\mathbb{Z}$:
\begin{equation}
\label{ZetaToNPlusLatticeSquare}
\zeta^{(n+\lattice)^2}
=
\zeta^{n^2}
e^{2\pi\mathrm{i}\p n}
e^{\pi\mathrm{i}\p\lattice}
=
\zeta^{n^2}
\mathrlap{\,.}
\end{equation}
It just remains to see that the normalization factor
$c_\zeta$ \cref{TheNormalizationFactor}
has unit norm, so that $\widehat T$ is unitary.
For this, observe that
\begin{equation}
\begin{aligned}
\left|
\sum_{m=0}^{\lattice-1}
\zeta^{m^2}
\right|^2
&=
\sum_{m=0}^{\lattice-1}
\sum_{n=0}^{\lattice-1}
\zeta^{(n+m)^2}
\zeta^{-n^2}
&\phantom{--}&
\text{by \cref{ZetaToNPlusLatticeSquare}}
\\
&=
\sum_{m=0}^{\lattice-1}
\zeta^{m^2}
\sum_{n=0}^{\lattice-1}
e^{
2\pi\mathrm{i}
\frac{\p}{\lattice}mn
}
\\
&=
\sum_{m=0}^{\lattice-1}
\zeta^{m^2}
\lattice\delta_{0,m}
&&
\text{by $\gcd(\p,\lattice)=1$}
\\
&=
\lattice
\mathrlap{\,.}
\end{aligned}
\end{equation}
This shows that
$\left|c_\zeta\right|=1$
in \cref{TheNormalizationFactor}.

\medskip
\noindent
{\bf (ii)}
First we verify the character
\eqref{ModularInvariantHeisenbergSignCharacter}.
The function
\[
(a,b,c)\longmapsto(-1)^{ab+a+b+c}
\]
is multiplicative by the group law of
$\widehat{\mathbb{Z}^2}$.
Indeed, the exponent evaluated on the product of
$(a,b,c)$ and $(a',b',c')$ is, modulo $2$,
\[
\begin{aligned}
&
(a+a')(b+b')
+(a+a')+(b+b')
+c+c'+ab'-a'b
\\
&\qquad\equiv
(ab+a+b+c)
+
(a'b'+a'+b'+c').
\end{aligned}
\]
Moreover, this character is invariant under both
modular generators
\[
S:(a,b,c)\longmapsto(b,-a,c),
\qquad
T:(a,b,c)\longmapsto(a+b,b,c),
\]
since
\[
ab+a+b
\;\equiv\;
-ab+b-a
\;\equiv\;
(a+b)b+(a+b)+b
\qquad
\mathrm{mod}\,2.
\]
Since $S$ and $T$ generate
$\mathrm{SL}_2(\mathbb{Z})$,
the character therefore extends to the semidirect
product by being trivial on the modular group,
as asserted in
\eqref{ModularInvariantHeisenbergSignCharacter}.

Next, if $\varepsilon=0$, then
$\lattice\p$ is even and $\eta=\zeta$.
If $\varepsilon=1$, then both $\lattice$ and $\p$
are odd, so that $\p+\lattice$ is even.
In either case,
\[
\gcd(\p+\varepsilon\lattice,\lattice)=1,
\qquad
\lattice(\p+\varepsilon\lattice)\in2\mathbb{Z}.
\]
Hence part \textup{(i)} supplies the representation
$\rho^{(0)}_\eta$.

Tensoring with the unitary character
$\chi^\varepsilon$ preserves unitarity and
all representation relations.
Since $\chi$ takes the value $-1$ on the
distinguished central generator, the central phase
of \eqref{GeneralTorusRepresentationBySignTwist} is
\[
(-1)^\varepsilon\eta=\zeta.
\]
Likewise, $\chi$ takes the value $-1$ on each of the
two standard Heisenberg generators.
Thus, if $U_\eta,V_\eta$ denote their operators
in $\rho^{(0)}_\eta$, their operators in $\rho_\zeta$
are
\[
U_\zeta=(-1)^\varepsilon U_\eta,
\qquad
V_\zeta=(-1)^\varepsilon V_\eta.
\]
The formulas \eqref{IrrepOfSemidirectProduct} give
\[
U_\eta^\lattice
=
V_\eta^\lattice
=
\mathrm{id},
\]
and therefore
\[
U_\zeta^\lattice
=
V_\zeta^\lattice
=
(-1)^{\varepsilon\lattice}\,\mathrm{id}
=
(-1)^\varepsilon\,\mathrm{id}.
\]
The last equality holds because
$\varepsilon=1$ implies that $\lattice$ is odd.
This is precisely the asserted scalar sector.
For $\varepsilon=1$, this sign twist changes the Heisenberg central
phase from $\eta$ to $\zeta$. By contrast, the mapping class group
characters used below are trivial on the Heisenberg subgroup.
\end{proof}

\begin{proposition}
[\bf General 2-Cohomotopical quantum states over the
$\mathrm{pp}$-torus]
\label{General2CohomotopicalStatesOverTorus}
For
\[
(\lattice,\p)
\in
\mathbb{N}_{>0}\times\mathbb{Z},
\qquad
\gcd(\p,\lattice)=1,
\qquad
\zeta:=e^{\pi\mathrm{i}\frac{\p}{\lattice}},
\]
let $\rho_\zeta$ be the representation constructed in
\textup{Lem.~\ref{GeneralRepsOverTheTorus}\,\textup{(ii)}}.
Then:
\begin{itemize}
\item[\bf (i)]
The representation $\rho_\zeta$ is unitary,
has dimension $\lattice$,
and is irreducible already when restricted to
$\widehat{\mathbb{Z}^2}$.
Its scalar sector in
\textup{Lem.~\ref{SomeIrreps}} is
\[
(s,t)=(\varepsilon,0),
\qquad
\varepsilon:=\lattice\p\,\mathrm{mod}\,2,
\]
so that
\[
\WOperator{\lattice}{0}
=
\WOperator{0}{\lattice}
=
(-1)^{\lattice\p}\,\mathrm{id}.
\]

\item[\bf (ii)]
Every finite-dimensional unitary representation of
$\mathrm{SL}_2(\mathbb{Z})
\ltimes\widehat{\mathbb{Z}^2}$
which is irreducible already upon restriction to
$\widehat{\mathbb{Z}^2}$
is isomorphic to one of the form
\[
(\lambda\circ\mathrm{pr})\otimes\rho_\zeta,
\]
where
\[
\mathrm{pr}
:
\mathrm{SL}_2(\mathbb{Z})
\ltimes\widehat{\mathbb{Z}^2}
\longrightarrow
\mathrm{SL}_2(\mathbb{Z})
\]
is the quotient projection and
\[
\lambda:
\mathrm{SL}_2(\mathbb{Z})
\longrightarrow \mathrm{U}(1)
\]
is a unitary character.
Conversely, every such tensor product has these properties.
For each fixed central phase $\zeta$
and fixed choice of $\rho_\zeta$,
distinct characters $\lambda$ give non-isomorphic
representations.
\end{itemize}
\end{proposition}

\begin{proof}
\noindent
{\bf (i)}
Existence and unitarity follow from
Lem.~\ref{GeneralRepsOverTheTorus}.
To see irreducibility, recall that its construction is
\[
\rho_\zeta
=
\chi^\varepsilon
\otimes
\rho^{(0)}_\eta,
\qquad
\eta=(-1)^\varepsilon\zeta.
\]
Since
\[
\eta^2=\zeta^2
=
e^{2\pi\mathrm{i}\frac{\p}{\lattice}},
\qquad
\gcd(\p,\lattice)=1,
\]
we have
\[
\mathrm{ord}(\eta^2)=\lattice.
\]
The restriction of $\rho^{(0)}_\eta$
to $\widehat{\mathbb{Z}^2}$
is therefore the irreducible representation of
Lem.~\ref{SomeIrreps}
with central phase $\eta$,
dimension $D=\lattice$,
and parameters $(s,t)=(0,0)$.
Tensoring with the one-dimensional character
$\chi^\varepsilon$ preserves its invariant subspaces
and its dimension.
Hence $\rho_\zeta$ is likewise irreducible upon restriction
to $\widehat{\mathbb{Z}^2}$ and has dimension $\lattice$.
The asserted scalar sector is established in
Lem.~\ref{GeneralRepsOverTheTorus}\,\textup{(ii)}.

\medskip
\noindent
{\bf (ii)}
Let $\widetilde\rho$ be a representation satisfying
the stated hypotheses.
By Lem.~\ref{FinDimHeisenbergIrrepsAndRootsOfUnity},
its central generator acts by a root of unity $\zeta$.
Writing
\[
\lattice:=\mathrm{ord}(\zeta^2),
\]
we may express this phase as
\[
\zeta=e^{\pi\mathrm{i}\frac{\p}{\lattice}}
\qquad
\text{with}
\qquad
\gcd(\p,\lattice)=1.
\]
By Prop.
\ref{ClassificationOfCovariantizableIrrepsOfIntegerHeisenberg},
its Heisenberg restriction is isomorphic to the
representation of Lem.~\ref{SomeIrreps}
whose parameters satisfy
\[
t=0,
\qquad
(-1)^s
=
\zeta^{\lattice^2}
=
(-1)^{\lattice\p}.
\]
These are precisely the central phase and scalar sector
of the restriction of $\rho_\zeta$ from part \textup{(i)}.
Hence the Heisenberg restrictions of
$\widetilde\rho$ and $\rho_\zeta$ are isomorphic.

After conjugating by a unitary intertwiner,
we may assume that these restrictions agree.
For any
\[
g\in
\mathrm{SL}_2(\mathbb{Z})
\ltimes\widehat{\mathbb{Z}^2},
\]
the operator
\[
\widetilde\rho(g)\rho_\zeta(g)^{-1}
\]
then commutes with their common Heisenberg representation.
By irreducibility and Schur's lemma, it is a scalar
$d(g)\,\mathrm{id}$ with $d(g)\in \mathrm{U}(1)$.
The representation laws imply that $d$ is multiplicative,
and agreement on the Heisenberg subgroup implies that
$d$ is trivial there.
Thus $d=\lambda\circ\mathrm{pr}$
for a unitary character
$\lambda:\mathrm{SL}_2(\mathbb{Z})\to \mathrm{U}(1)$,
which proves the asserted form.

Conversely, tensoring $\rho_\zeta$
with $\lambda\circ\mathrm{pr}$
preserves unitarity and leaves its Heisenberg
restriction unchanged, hence irreducible.

Finally, an intertwiner between two such tensor products
at fixed $\zeta$ and fixed $\rho_\zeta$
commutes with their common irreducible Heisenberg
restriction, so is scalar.
If it is invertible, the intertwining condition on
the modular subgroup forces the two characters
$\lambda$ to coincide.
\end{proof}

Thus every root of unity occurs as a central phase already for the
$\mathrm{pp}$-spin structure. The even-product subfamily of
Lem.~\ref{GeneralRepsOverTheTorus}\,\textup{(i)} remains useful for
comparison with familiar bosonic examples:

\begin{remark}[\bf Relation to bosonic FQH systems]
  Besides the standard FQH effect for two-dimensional electron gases,
  rapidly rotating Bose-Einstein condensates with repulsive contact
  interactions provide models of bosonic FQH states \cite{WilkinGunn00}.
  A basic example is the bosonic Laughlin state at $\nu=1/2$
  \cite[(79)]{Cooper08}.
  Exact-diagonalization studies give evidence for states in the
  Jain sequences $\nu=r/(r\pm1)$, including
  $\nu=2/3,\,3/4,\,4/3,\,5/4$
  \cite{RegnaultJolicoeur03}\cite{RegnaultJolicoeur04}
  (see also \cite[p.~31]{Cooper08}).
  Each of these reduced fractions has an even numerator or denominator,
  as in the condition $\lattice\p\in2\mathbb{Z}$ of
  Lem.~\ref{GeneralRepsOverTheTorus}\,\textup{(i)}.

  This arithmetic comparison concerns that subfamily.
  Prop.~\ref{General2CohomotopicalStatesOverTorus} also admits odd
  $\lattice\p$, with scalar sector $(s,t)=(1,0)$.
  Thus the full $\mathrm{pp}$ classification does not by itself
  distinguish bosonic from fermionic constituents, nor identify a
  filling fraction with a braiding angle.
  The occurrence of bosonic Laughlin correlations in a rotating system
  of two fermions \cite{LuntEtAl24}, and the effective bosonic behavior of
  flux-coupled quasiparticles (cf. \cite[\S 7.3.3]{Wen07}), further
  illustrate why these physical identifications require care.
\end{remark}

Our surfaces carry a spin structure (Rem.~\ref{SpinStructure}).
Using all of $\mathrm{SL}_2(\mathbb{Z})$ corresponds to the
$\mathrm{pp}$-spin structure in \eqref{SpinStructuresOnTheTorus}.
We now consider the different covariance condition for the
$\mathrm{aa}$-spin structure.

\paragraph
{Covariantization for the $\mathrm{aa}$-spin structure}

For the $\mathrm{aa}$-spin structure, the mapping class group is the
subgroup $\mathrm{MCG}(\Sigma^2_1)^{\mathrm{aa}}\subset\mathrm{SL}_2(\mathbb{Z})$
generated by $S$ and $T^2$ \eqref{aaSpinMappingClassGroupOfTorus}.
Both scalar sectors $(s,t)=(0,0)$ and $(1,0)$ then admit extensions
for every central phase, as follows:

\begin{proposition}
[\bf 2-Cohomotopical quantum states over the aa-spin torus]
\label{QuantumStatesOverTheAASpinTorus}
For
\[
  (\lattice,\p)
  \in
  \mathbb{N}_{>0}\times\mathbb{Z},
  \qquad
  \gcd(\p,\lattice)=1,
  \qquad
  \zeta := e^{\pi\mathrm{i}\p/\lattice},
\]
define
\[
  \widetilde{\p}
  :=
  \begin{cases}
    \p
    & \vert\; \p \in 2\mathbb{Z}+1,
    \\
    \p+\lattice
    & \vert\; \p \in 2\mathbb{Z},
  \end{cases}
  \qquad
  \xi := e^{\pi\mathrm{i}\widetilde{\p}/\lattice}.
\]
Here and below, fractional powers of $\xi$ mean
\[
  \xi^x
  :=
  e^{\pi\mathrm{i}\widetilde{\p}x/\lattice}.
\]
Then, for either $s\in\{0,1\}$, the following operators on
\[
  \HilbertSpace{H}_{T^2}
  :=
  \mathbb{C}^{\lattice}
  \simeq
  \mathrm{Span}
  \big(
    \vert 0\rangle,\ldots,\vert \lattice-1\rangle
  \big)
\]
with the displayed basis orthonormal, define a unitary representation
$\rho^{\mathrm{aa}}_{\zeta,s}$ of
$
  \mathrm{MCG}(\Sigma^2_1)^{\mathrm{aa}}
  \ltimes
  \widehat{\mathbb{Z}^2}
$ by:
\begin{equation}
\label{IrrepOfSemidirectProductForAASpin}
  \begin{aligned}
    \WOperator{1}{0}\vert n\rangle
    &=
    \xi^{2n+s}\vert n\rangle,
    \\[3pt]
    \WOperator{0}{1}\vert n\rangle
    &=
    \begin{cases}
      \vert n+1\rangle
      & \vert\; n<\lattice-1,
      \\
      (-1)^s\vert 0\rangle
      & \vert\; n=\lattice-1,
    \end{cases}
    \\[3pt]
    \widehat{\zeta}\vert n\rangle
    &=
    \zeta\vert n\rangle,
    \\[3pt]
    \widehat{S}\vert n\rangle
    &=
    \frac{1}{\sqrt{\lattice}}
    \sum_{m=0}^{\lattice-1}
    \xi^{2(n+s/2)(m+s/2)}
    \vert m\rangle,
    \\[3pt]
    \widehat{T^2}\vert n\rangle
    &=
    e^{-\pi\mathrm{i}/6}
    \xi^{2(n+s/2)^2}
    \vert n\rangle.
  \end{aligned}
\end{equation}
Here $n\in\{0,\ldots,\lattice-1\}$ and $\widehat{T^2}$ denotes
the operator assigned to the generator $T^2$.
These representations are irreducible already upon restriction
to $\widehat{\mathbb{Z}^2}$, and satisfy
\[
  \big(\WOperator{1}{0}\big)^{\lattice}
  =
  \big(\WOperator{0}{1}\big)^{\lattice}
  =
  (-1)^s\,\mathrm{id}.
\]
Hence, for each $\zeta$, their restrictions realize the two
isomorphism classes with $(s,t)=(0,0)$ and $(s,t)=(1,0)$
in Lem.~\ref{SomeIrreps}.
\end{proposition}

\begin{proof}
For orientation, in the case $\p=1$ and $s=0$, the displayed
modular matrices recover the rank-one lattice formulas
of \cite[\S 3, Example 1, p.~221]{Gannon05},
with lattice $\Lambda=\sqrt{\lattice}\,\mathbb{Z}$.
There it is observed that, for integral lattices which
are not even, the lattice construction still gives
a representation of the index-three subgroup
$\langle S,T^2\rangle\subset\mathrm{SL}_2(\mathbb{Z})$.
For this rank-one lattice, this includes precisely
the additional case of odd $\lattice$.
We verify below the full semidirect product
representation for arbitrary coprime
$(\lattice,\p)$ and both choices of $s$.

First observe that, if $\p$ is even, coprimality forces
$\lattice$ to be odd. Hence $\widetilde{\p}$ is always odd,
with
\[
  \gcd(\widetilde{\p},\lattice)=1,
  \qquad
  \xi^2=\zeta^2,
  \qquad
  \xi^{\lattice}=-1,
  \qquad
  \mathrm{ord}(\xi^2)=\lattice.
\]
For brevity, write
\[
  U:=\WOperator{1}{0},
  \qquad
  V:=\WOperator{0}{1},
  \qquad
  Q:=\widehat{T^2}.
\]

The operators $U$ and $V$ are unitary and satisfy
\[
  UV
  =
  \xi^2VU
  =
  \zeta^2VU,
  \qquad
  U^{\lattice}
  =
  V^{\lattice}
  =
  (-1)^s\,\mathrm{id}.
\]
Together with
$\widehat{\zeta}=\zeta\,\mathrm{id}$,
these give a unitary representation of
$\widehat{\mathbb{Z}^2}$.
Moreover, the eigenvalues $\xi^{2n+s}$ of $U$ are pairwise
distinct, while $V$ permutes their one-dimensional
eigenspaces transitively. Therefore this representation
is irreducible.

Next we verify the mapping class group relations.
The operator $\widehat{S}$ is unitary by finite Fourier
orthogonality for the primitive $\lattice$-th root
$\xi^2$, while $Q$ is manifestly unitary.
The presentation \eqref{aaSpinMappingClassGroupOfTorus} of the
$\mathrm{aa}$-spin mapping class group is
\[
  \mathrm{MCG}(\Sigma^2_1)^{\mathrm{aa}}
  =
  \big\langle
    S,T^2
    \,\big\vert\,
    S^4=1,\;
    S^2T^2=T^2S^2
  \big\rangle.
\]
Writing $[r]_{\lattice}$ for the representative of
$r$ modulo $\lattice$ in
$\{0,\ldots,\lattice-1\}$, Fourier orthogonality gives
\[
  \widehat{S}^{\,2}\vert n\rangle
  =
  \begin{cases}
    \vert [-n]_{\lattice}\rangle
    & \vert\; s=0,
    \\
    -\vert \lattice-1-n\rangle
    & \vert\; s=1.
  \end{cases}
\]
Indeed, the coefficient of $\vert j\rangle$ in
$\widehat{S}^{\,2}\vert n\rangle$ is
\[
  \frac{1}{\lattice}
  \sum_{m=0}^{\lattice-1}
  \xi^{2(m+s/2)(j+n+s)}.
\]
For $s=0$ this gives the first case above.
For $s=1$ it vanishes unless $j+n+1=\lattice$,
in which case its value is $\xi^{\lattice}=-1$.
In particular,
$\widehat{S}^{\,4}=\mathrm{id}$.

Furthermore, the diagonal coefficient of $Q$ is
$\lattice$-periodic in $n$, since
\[
  \frac{
    \xi^{2(n+\lattice+s/2)^2}
  }{
    \xi^{2(n+s/2)^2}
  }
  =
  \xi^{4\lattice n+2\lattice s+2\lattice^2}
  =
  1.
\]
It is also invariant under $n\mapsto -n-s$.
The displayed formula for $\widehat{S}^{\,2}$ therefore
implies
\[
  \widehat{S}^{\,2}Q
  =
  Q\widehat{S}^{\,2}.
\]
Thus the mapping class group relations hold.

Finally, we check covariance with the Heisenberg operators.
Substitution in the formulas gives
\[
  \widehat{S}U
  =
  V^{-1}\widehat{S},
  \qquad
  \widehat{S}V
  =
  U\widehat{S}.
\]
Here the boundary terms agree because
\[
  \xi^{2\lattice(n+s/2)}
  =
  (-1)^s.
\]
Likewise, $Q$ commutes with $U$, and the ratio of its
successive diagonal coefficients is
\[
  \frac{
    \xi^{2(n+1+s/2)^2}
  }{
    \xi^{2(n+s/2)^2}
  }
  =
  \xi^{4n+2s+2}.
\]
Using the periodicity just established, this also applies
at the boundary and yields
\[
  QUQ^{-1}
  =
  U,
  \qquad
  QVQ^{-1}
  =
  \xi^{-2}U^2V
  =
  \zeta^{-2}U^2V.
\]
These are exactly the covariance relations for the
actions of $S$ and $T^2$ on
$\widehat{\mathbb{Z}^2}$.
The central operator $\widehat{\zeta}$ is scalar and
hence fixed under conjugation.
Consequently the operators define the claimed unitary
representation of the semidirect product.

Since
$\mathrm{ord}(\zeta^2)=\lattice$,
the scalar identities
$U^{\lattice}=V^{\lattice}=(-1)^s\,\mathrm{id}$
identify its Heisenberg restriction with the
$(s,t)=(s,0)$ case of Lem.~\ref{SomeIrreps},
by Prop.
\ref{ClassificationOfCovariantizableIrrepsOfIntegerHeisenberg}.
The two choices of $s$ give inequivalent restrictions,
as their scalar values of $U^{\lattice}$ differ.
\end{proof}

Combining the classification of Heisenberg restrictions with these
extension constructions gives the result over the torus:
\begin{theorem}
[\bf Classification of 2-cohomotopical flux quantum states
on the torus]
\label{ClassificationOf2CohomotopicalFluxQuantumStatesOverTorus}
\leavevmode\newline
The finite-dimensional unitary spaces of
2-cohomotopical flux quantum states on the torus
whose Heisenberg restriction is irreducible
are classified as follows.

Let $\zeta$ range over all roots of unity, and write
\[
  \zeta
  =
  e^{\pi\mathrm{i}\p/\lattice},
  \qquad
  (\lattice,\p)
  \in
  \mathbb{N}_{>0}\times\mathbb{Z},
  \qquad
  \gcd(\p,\lattice)=1.
\]
Then necessarily
$\lattice=\mathrm{ord}(\zeta^2)$.
In the following, $\mathrm{pr}$ denotes the projection
from the respective semidirect product onto its
mapping class group.

\begin{enumerate}[label=\textup{(\roman*)}]

  \item
  Over the torus with aa-spin structure, these
  representations are precisely
  \[
    (\lambda\circ\mathrm{pr})
    \otimes
    \rho^{\mathrm{aa}}_{\zeta,s},
    \qquad
    s\in\{0,1\},
  \]
  where $\rho^{\mathrm{aa}}_{\zeta,s}$ is the
  representation constructed in Prop.~\ref{QuantumStatesOverTheAASpinTorus} and
  \[
    \lambda:
    \mathrm{MCG}(\Sigma^2_1)^{\mathrm{aa}}
    \longrightarrow
    \mathrm{U}(1)
  \]
  is any unitary group character.
  Their restrictions to $\widehat{\mathbb{Z}^2}$
  have parameters $(s,t)=(s,0)$ in
  Lem.~\ref{SomeIrreps}, with
  \[
    \big(\WOperator{1}{0}\big)^{\lattice}
    =
    \big(\WOperator{0}{1}\big)^{\lattice}
    =
    (-1)^s\,\mathrm{id}.
  \]

  \item
  Over the torus with pp-spin structure, these
  representations are precisely
  \[
    (\lambda\circ\mathrm{pr})
    \otimes
    \rho^{\mathrm{pp}}_{\zeta},
  \]
  where $\rho^{\mathrm{pp}}_{\zeta}:=\rho_\zeta$ is the
  representation constructed in
  Lem.~\ref{GeneralRepsOverTheTorus}\,\textup{(ii)}
  and classified in Prop.~\ref{General2CohomotopicalStatesOverTorus}, and
  \[
    \lambda:
    \mathrm{SL}_2(\mathbb{Z})
    \longrightarrow
    \mathrm{U}(1)
  \]
  is any unitary group character.
  Its restriction to $\widehat{\mathbb{Z}^2}$
  has parameters
  \[
    t=0,
    \qquad
    s\equiv\lattice\p\pmod{2},
  \]
  so that
  \[
    \big(\WOperator{1}{0}\big)^{\lattice}
    =
    \big(\WOperator{0}{1}\big)^{\lattice}
    =
    \zeta^{\lattice^2}\,\mathrm{id}
    =
    (-1)^{\lattice\p}\,\mathrm{id}.
  \]
  In particular, no parity restriction on
  $\lattice\p$ is imposed on the braiding phase.

  \item
  All these representations have dimension
  \textup{(``ground state degeneracy'')}
  \[
    \dim\HilbertSpace{H}_{T^2}
    =
    \lattice
    =
    \mathrm{ord}(\zeta^2)
    =
    \frac{
      \mathrm{ord}(\zeta)
    }{
      \gcd\big(2,\mathrm{ord}(\zeta)\big)
    }.
  \]

\end{enumerate}

For each $\zeta$, fix the reference representations
appearing above.
Then the parameters $(s,\lambda)$ in the aa case,
and $\lambda$ in the pp case, are unique up to
isomorphism of the resulting representations.
Representations with different $\zeta$ are inequivalent.
\end{theorem}

\begin{proof}
Let a finite-dimensional unitary representation
of the respective semidirect product be given,
whose restriction to $\widehat{\mathbb{Z}^2}$
is irreducible.
By
Lem.~\ref{FinDimHeisenbergIrrepsAndRootsOfUnity},
the central generator acts as
$\widehat{\zeta}=\zeta\,\mathrm{id}$
for some root of unity $\zeta$.
Writing
$\lattice=\mathrm{ord}(\zeta^2)$,
we may express
$\zeta=e^{\pi\mathrm{i}\p/\lattice}$
with $\gcd(\p,\lattice)=1$.

By
Prop.~\ref{ClassificationOfCovariantizableIrrepsOfIntegerHeisenberg},
the Heisenberg restriction must be one of
the $(s,t)=(s,0)$ cases of
Lem.~\ref{SomeIrreps},
where $s$ is arbitrary for the aa-spin structure
and is constrained by
\[
  (-1)^s
  =
  \zeta^{\lattice^2}
  =
  (-1)^{\lattice\p}
\]
for the pp-spin structure.
Their dimensions are $\lattice$.

The existence of extensions for every one of these
Heisenberg representations is established by
Prop.~\ref{QuantumStatesOverTheAASpinTorus} in the aa case and by
Prop.~\ref{General2CohomotopicalStatesOverTorus} in the pp case.
Tensoring any such extension with a unitary character
of the mapping class group preserves its Heisenberg
restriction and hence its irreducibility.

Conversely, after unitarily identifying the Heisenberg
restriction with that of the corresponding reference
representation, Lem.~\ref{ClassOfExtendibleRepIsGInvariant}\,\textup{(ii)}
implies that any other extension differs from it
by tensoring with a character of the quotient,
which is the respective mapping class group.
Since both extensions are unitary, this character
takes values in $\mathrm{U}(1)$.
This proves that the displayed families exhaust
the representations under consideration.

Finally, different values of $\zeta$ are distinguished
by the action of the central generator, and the two
aa cases are distinguished by the scalar
$(-1)^s$ in
$\big(\WOperator{1}{0}\big)^{\lattice}$.
For a fixed Heisenberg restriction, any intertwiner
between two character twists of the chosen reference
extension must be scalar, by Schur's lemma.
It can therefore intertwine the full representations
only if the two characters agree.
This proves the asserted uniqueness.
\end{proof}

\begin{remark}[\bf Linear modular extensions on the torus]
\label{FineTopologicalOrder}
\begin{itemize}[itemsep=2pt]
  \item[\bf (i)]
  The abelianization of $\mathrm{SL}_2(\mathbb{Z})$ is
  $\mathbb{Z}/12\mathbb{Z}$ (cf. \cite[Cor.~2.4]{Conrad-SL2Z}).
  Hence there are precisely twelve unitary characters, all taking
  values in twelfth roots of unity.
  By Thm.~\ref{ClassificationOf2CohomotopicalFluxQuantumStatesOverTorus},
  for every root of unity $\zeta$ there are precisely twelve
  isomorphism classes in the $\mathrm{pp}$ case under the theorem's
  irreducibility assumption.
  All have dimension $\lattice=\mathrm{ord}(\zeta^2)$ and the same
  Heisenberg restriction; they differ by their scalar phases under
  modular transformations.

  \item[\bf (ii)]
  For the $\mathrm{aa}$-spin structure, abelianizing the presentation
  \eqref{aaSpinMappingClassGroupOfTorus} gives
  \[
    \big(\mathrm{MCG}(\Sigma^2_1)^{\mathrm{aa}}\big)_{\mathrm{ab}}
    \simeq
    \mathbb{Z}/4\mathbb{Z}\oplus\mathbb{Z}.
  \]
  Indeed, the relation $S^4=1$ remains, while the commutation relation
  between $S^2$ and $T^2$ imposes no further condition in the abelianization.
  A unitary character is therefore specified independently by
  \[
    \lambda(S)\in\mu_4,
    \qquad
    \lambda(T^2)\in\mathrm{U}(1),
    \qquad
    \mu_4:=\{1,\mathrm{i},-1,-\mathrm{i}\}.
  \]
  Thus, for each $\zeta$ and each $s\in\{0,1\}$, the extensions
  form a family parameterized by $\mu_4\times\mathrm{U}(1)$ once a
  reference extension is fixed.
  Including both Heisenberg sectors gives two disjoint copies of this
  family, all of dimension $\lattice$.
  In particular, there are uncountably many isomorphism classes at
  each central phase, since no finite-order condition is imposed on
  the value of $\lambda(T^2)$.

  \item[\bf (iii)]
  The classification assumes irreducibility upon restriction to
  $\widehat{\mathbb{Z}^2}$.
  Dropping this assumption allows further irreducible representations.
  For example, let $\rho$ be one of the extensions above, let $\Gamma$
  be the corresponding mapping class group, and let $\tau$ be a
  finite-dimensional irreducible unitary representation of $\Gamma$.
  Then
  \[
    \rho\otimes(\tau\circ\mathrm{pr})
  \]
  is irreducible: any commuting operator has the form
  $\mathrm{id}\otimes A$ by irreducibility of the Heisenberg restriction
  of $\rho$, and $A$ must be scalar by irreducibility of $\tau$.
  If $\dim\tau>1$, its Heisenberg restriction is a direct sum of
  $\dim\tau$ copies of that of $\rho$.
  Thus the modular group may act nontrivially on a multiplicity space
  without changing the central braiding phase.

  Over punctured surfaces, considered next, defect motions determine
  mapping classes through the point-pushing homomorphism
  \eqref{BirmanSequence}. Pulling back a mapping-class representation
  gives a braid representation, which is trivial on the kernel of that
  homomorphism. Its image may describe non-abelian transformations of
  defect states. For the sphere with at least three punctures, the kernel is
  specified in \eqref{SphericalBraidGroupSequence}; it is trivial for
  one or two punctures.

  \item[\bf (iv)]
  The distinctions in \textup{(i)} and \textup{(ii)} concern isomorphism
  classes of linear unitary representations. For a fixed reference
  extension, scalar character twists have the same projective action:
  \[
    [\lambda(M)\rho(M)]=[\rho(M)]
    \qquad\text{in }\mathrm{PU}(\lattice).
  \]
  Thus the scalar phases are additional data of a linear lift. If this
  lift is part of the specified monodromy data, the character twists
  remain distinct as proved above. Their identification with distinct
  physical topological orders requires specifying a realization by
  ground-state transport and its stability under the allowed perturbations;
  see \S\ref{GroundStateMonodromyAndTopologicalOrder}. A representation of the mapping class group alone does
  not establish the additional fusion and compatibility conditions of
  modular data (cf. \cite[\S 2]{Gannon05}).
\end{itemize}
\end{remark}

\subsubsection
{Comparison to K-matrix formalism}
\label
{PredictionsAndComparisons}

We compare the representation-theoretic analysis of the 2-cohomotopic surpflux flux monodromy over the torus from above (\cref{FluxMonodromyClassification}) with existing descriptions of topological order in FQH systems, notably with K-matrix Chern-Simons theory (\cref{KMatrixCSTheory}).

\paragraph
{Ground-state monodromy and topological order}
\label{GroundStateMonodromyAndTopologicalOrder}
Our construction supplies candidate local systems directly from the flux
moduli space. Over a fixed spin-structured torus, the monodromy group is
the semidirect product in Prop. \ref{DiffeomorphismActionOverTorusIsModularAction}.
Its irreducible unitary representations define the algebraic superselection
sectors of the adopted observable algebra. For those whose Heisenberg
restriction is irreducible, Thm.
\ref{ClassificationOf2CohomotopicalFluxQuantumStatesOverTorus} and
Rem. \ref{FineTopologicalOrder} classify the resulting linear monodromy data.

For a fixed connected parameter space $\mathcal M_0$, these data consist
of the isomorphism class of a finite-dimensional unitary local system
\eqref{StateSpacesAsLocalSystemsOnModuliSpace}, equivalently of a representation
\begin{equation}
\label{GroundStateMonodromyRepresentation}
  \rho \colon \pi_1(\mathcal M_0,m)
  \longrightarrow \mathrm{U}\big(\HilbertSpace{H}_m\big)
  \mathrlap{\,.}
\end{equation}
This specifies how quantum states transform along parameter cycles,
including the scalar phases of that transformation.

When realized by robust adiabatic transport of gapped ground states,
such local systems encode a core signature of topological order.
This transport signature is part of the characterization of a physical
phase. Relating it to a complete characterization involves the following
further data and questions:
\begin{enumerate}[
  leftmargin=*,
  itemsep=3pt
]
\item \textbf{Ground-state realization.}
Specify a smooth gapped family of Hamiltonians over $\mathcal M_0$
whose ground-state spaces and adiabatic transport realize the candidate
local system. This supplies the physical meaning of its fibers and
monodromy operators.

\item \textbf{Topological transport.}
The ground-state spaces carry the Berry connection. Its holonomy gives
\eqref{GroundStateMonodromyRepresentation} when the connection is flat,
or when a flat topological part has been specified. Projectively flat
transport gives projective monodromy, with a linear lift, when it exists,
as additional data (cf. \cite[\S V]{Wen12}). Flatness is a separate
condition from the spectral gap that permits adiabatic transport.

\item \textbf{Stability and scalar phases.}
Identify the transport data that are robust under the allowed local
perturbations and specify how auxiliary gapped degrees of freedom are
treated. Even an auxiliary system with a one-dimensional ground space
can contribute geometric phases. Thus the character twists of
Rem. \ref{FineTopologicalOrder}\,\textup{(iv)} are to be compared using
the specified stability criterion and scalar-phase data. Scalar phases
belonging to the robust topological transport remain part of this
comparison, including those invisible in the projective representation.

\item \textbf{Completeness.}
The choice of $\mathcal M_0$ determines which transport processes are
probed by the monodromy. A full classification of the intended physical
topological orders requires establishing which candidate local systems
are realized and whether their isomorphism classes distinguish these
orders, or which additional invariants are needed.
\end{enumerate}

\paragraph
{Relation to Experimental Topological Order}
Toroidal FQH systems offer a natural setting in which to investigate
ground-state realizations of the classified local systems and compare
their monodromy operators with adiabatic quantum-state transport.

Fractional Chern insulators are of particular interest in this connection.
The ``quantum anomalous'' version of the fractional quantum Hall effect
(FQAH effect) has been experimentally established in these crystalline
materials (cf. \cite{CLM23,YMMSX24,MSM24}). Their intrinsic
\emph{Berry curvature} plays the role of external magnetic flux in the
momentum-space description (cf. \cite[Fig. 3]{SS25-FQAH}). The periodicity
of quasi-momenta makes the effective momentum space a
\emph{Brillouin torus} (\cite[\S 2.1]{Thiang25}\cite[p. 52]{FreedMoore13}),
providing a link to the torus geometry considered here
(cf. \cite{PPSBPG24}).

There is a further link in the two-band setting (cf. \cite[\S II]{CLM23}):
the complex Grassmannian $\mathrm{Gr}_{2}^1\simeq\mathbb{C}P^1\simeq S^2$
classifies the occupied line of a gapped two-band Bloch Hamiltonian
(cf. \cite[\S 3.1.2]{SS20-Orb}). Such Bloch data are therefore described
by maps from the Brillouin torus to $S^2$
(\cite[(8.3-4)]{Sergeev23}\cite[(4)]{SS25-FQAH}\cite[\S I.A]{SS25-Crys}).
This places two-band Bloch data in the same mapping-space framework as
the cohomotopical description of FQH flux
(cf. \cite[\S 5.1]{KSS26}).
Here the connected component matters: with a compatible orientation of
$S^2$, the degree $d$ of the Bloch map is the occupied band's Chern number.
The torus classification in \S\ref{2CohomotopicalFluxThroughTorus}
uses the degree-zero component, while a nonzero-Chern band lies in
$\mathrm{Map}_d(T^2,S^2)$ with $d\neq 0$. In this component,
\cite[Prop. 1.5]{Kallel01}\cite[Thm. 3.5]{KSS26} give
\[
  \pi_1\mathrm{Map}_d(T^2,S^2)
  \;\simeq\;
  \widehat{\mathbb{Z}^2}\big/\big\langle z^{2|d|}\big\rangle
  \,,\qquad d\neq 0,
\]
where $z$ is the central generator of the integer Heisenberg group
\eqref{IntegerHeisenbergGroup}. The constant-map section used in
Prop. \ref{ExtensionOfMCGByFluxMonodromy} is unavailable in these
components, so their covariant monodromy requires a separate analysis
before applying the representation classification. This motivates the
momentum-space investigations in \cite{SS25-FQAH,SS25-Crys,SS26-Drinfeld}.

\paragraph
{Comparison to effective K-Matrix CS-Theory}
The $K$-matrix Chern-Simons formalism also constructs ground-state
representations of modular transformations, including scalar signature
factors and spin-structure dependence
(\cite[\S 5.6]{BelovMoore05}; see \cref{KMatrixGroundStateMonodromy}).
Wen derives the corresponding transformations for abelian FQH wavefunctions,
retaining scalar information about the chiral central charge
\cite[\S\S V--VI]{Wen12}.
Thus the filling fraction \eqref{KMatrixFillingFraction} and degeneracy
\eqref{KMatrixGroundStateDegeneracy} are only part of its predictions.

A comparison with Thm. \ref{ClassificationOf2CohomotopicalFluxQuantumStatesOverTorus}
must match the Heisenberg and modular actions, the spin structure, and the
choice of linear lift. The basic even-order example agrees with the
one-component Chern-Simons operators as explained after
\eqref{KMatrixModularOperators}; a systematic comparison of all extensions
remains to be established. The following calculations compare the
specified braiding and dimension formulas. They do not by themselves
establish an equivalence or a distinction of physical topological orders.

\smallskip

To this end, recall (Thm. \ref{ClassificationOf2CohomotopicalFluxQuantumStatesOverTorus}\,iii) that a finite-dimensional irreducible Heisenberg representation with braiding phase $\zeta=e^{\pi\mathrm{i}\p/\lattice}$, where $\gcd(\p,\lattice)=1$, has dimension $\lattice$. If such a representation is realized as the full torus ground-state space, this gives the conditional prediction:
\begin{equation}
\label
{PredictionOfBraidingPhaseToGroundStateDegeneracy}
\begin{array}{ll}
  \text{braiding angle $\theta$:}
  &
  \p/\lattice
  \\
  \text{ground state degeneracy:}
  &
  \lattice
\end{array}
\;\;\;\;
\text{(irreducible Heisenberg restriction; $\gcd(\p,\lattice)=1$)}
\end{equation}
(Our approach does not explicitly refer to the filling fraction $\nu$, but traditional theory relates this to the braiding angle, cf.  Rem. \ref{RelationBetweenBraidingPhasesAndFillingFraction} below.)

\begin{remark}
\label
{ObservablePartOfBraidingAngle}
For a fixed quasi-particle species, the exchange phase
$\zeta=e^{\pi\mathrm{i}\theta}$ determines $\theta$ modulo $2$
\eqref{TheBraidingPhase}. The interferometry discussed in
Ex. \ref{ExperimentallyObservedBraidingPhases} probes the full-braid phase
$\zeta^2=e^{2\pi\mathrm{i}\theta}$, hence determines only the residue class
\begin{equation}
\label
{BraidingAngleModuloOne}
  \theta
  \mod
  1
  \,\in\,
  \mathbb{Q}/\mathbb{Z}
  \mathrlap{\,.}
\end{equation}
We call \eqref{BraidingAngleModuloOne} the \emph{full-braid parameter}
and use it in the following denominator comparisons. This does not
identify $\zeta$ with $-\zeta$: an odd integer shift of the total
exchange parameter changes the exchange phase. In the convention of
\cite[(10)]{GoldhaberJain1995}, a compensating fermionic exchange factor
keeps the total statistics unchanged. The classification in
Thm. \ref{ClassificationOf2CohomotopicalFluxQuantumStatesOverTorus}
retains $\zeta$ itself and distinguishes these signs.
\end{remark}

First we will show:
\begin{proposition}[\bf Comparing to $K$-matrix predictions in Jain states]
\label
{ComparingToKMatrixPredictionsInJainstates}
  The $K$-matrix formalism verifies the dimension--denominator relation
  \eqref{PredictionOfBraidingPhaseToGroundStateDegeneracy} for all
  fundamental quasi-particles
  \eqref{WeigthVectorsForMultifundamentalquasi-particles} in all Jain
  states \eqref{JainSequenceFillingFractions}, using the full-braid
  parameter \textup{(Rem. \ref{ObservablePartOfBraidingAngle})}:
  the torus ground state degeneracy equals the denominator of
  $\theta$ modulo $1$.
\end{proposition}
We prove this by the following Lemmas (esp.
\ref{KMatrixGroundStateDegeenracyInJainState} and
\ref{KMatrixBraidingAnglesInJainStates}), recalling the relevant
$K$-matrix predictions for Jain states and deriving the formulas
needed for the comparison.

\medskip

Beyond the Jain states \eqref{JainSequenceFillingFractions}, we consider nonsingular symmetric two-component $K$-matrices with coprime entries, including the Halperin regime specified in Ex. \ref{KMatrixDescriptionOfHalperinState}. Proposition \ref{AnalysisOfHalperinStates} verifies the same numerical relation for their fundamental quasi-particles. Example \ref{WouldBeHalperinStateWithoutInteraction} explains why the dimension of a full multi-component state space need not equal the denominator of the braiding angle of one component.

\medskip

These comparisons establish agreement of the specified dimension and denominator formulas. They do not identify irreducibility for a multi-component observable algebra with irreducibility upon restriction to one integer Heisenberg group, nor do they extend Thm. \ref{ClassificationOf2CohomotopicalFluxQuantumStatesOverTorus} to all covariant representations.

\medskip

We proceed to establish these results.

\begin{remark}[\bf Jain sequence of filling fractions]
\label
{JainSequence}
We recall that the \emph{Jain sequence} 
(\cite{Jain89,Jain92,Jain07})
of FQH filling fractions consists of those of the form
\begin{equation}
\label
{JainSequenceFillingFractions}
  \nu
  =
  \frac
    { n }
    { 2 p n \pm 1 }
  \,,
  \;\;
  n,p \in \mathbb{N}_{\geq 1}
  \mathrlap{\,.}
\end{equation}
It is filling fractions of this Jain form which currently dominate the experimentally observed situations and which are generally expected to be thermodynamically robust, whereas non-Jain FQH states are generally expected to be fragile. 
\end{remark}

For comparison with the Jain predictions, we recall two Abelian cases
studied by interferometry:

\begin{example}
[\bf Experimental full-braid phases]
\label
{ExperimentallyObservedBraidingPhases}
The interferometry results below are consistent with the indicated
full-braid phases $e^{2\pi\mathrm{i}\theta}$. We use the orientation
convention of \eqref{KMatrixBraidingPhase}; $\theta$ is specified modulo $1$.
\begin{center}
\def\arraystretch{1.4}
\begin{tabular}{ccc}
  \textbf{Filling fraction}
  &
  \textbf{Full-braid parameter}
  &
  \textbf{Interferometry}
  \\
  \hline\hline
  $\nu=1/3$
  &
  $\theta \equiv 1/3$
  &
  \cite{NakamuraEtAl20}
  \\
  \hline
  $\nu=2/5$
  &
  $\theta \equiv -2/5 \equiv 3/5$
  &
  \cite{NakamuraEtAl23}
\end{tabular}
\end{center}
In particular, \cite{NakamuraEtAl23} extracts the full-braid angle
$\theta_a=(-0.43\pm0.05)\,2\pi$, consistent with $-4\pi/5$.
This determines $\zeta^2$, not the choice of exchange phase
$\zeta=\pm e^{\pi\mathrm{i}\theta}$.
Current-correlation experiments provide further evidence for fractional
statistics \cite{BartolomeiEtAl20,Ruelle23,GlidicEtAl23}; this should be
distinguished from extracting a numerical braid phase. In particular,
\cite{GlidicEtAl23} explicitly qualifies the quantitative connection
between the measured correlations and a specific exchange phase.
\end{example}

\begin{remark}[\bf Braiding phases in the Jain sequence]
\label
{RelationBetweenBraidingPhasesAndFillingFraction}
For the positive Jain branch $\nu=n/(2pn+1)$, the quasi-particle
statistics parameter in \cite[(3)]{GoldhaberJain1995} is
$\big(2p(n-1)+1\big)/(2pn+1)$.
The same paper uses the parameter $-2p/(2pn+1)$ after separating
off fermionic exchange statistics \cite[(10)]{GoldhaberJain1995}.
These differ by an integer and hence give the same full-braid phase
$e^{2\pi\mathrm{i}\theta}$, as relevant to
\eqref{BraidingAngleModuloOne} (cf. \cite[\S 9.8]{Jain07}).
\footnote{For earlier work on fractional statistics and filling fractions,
see \cite{Su86}. We do not use a universal identification of the braiding
angle with $\pm\nu$: a comparison must specify the quasi-particle species
and the orientation convention for braiding.}

In the notation for both Jain branches, and with the braiding convention
of \eqref{KMatrixBraidingPhase}, the fundamental quasi-particle has
\begin{equation}
\label
{BraidingAngleAccordingToJain}
  \nu
  =\frac{n}{2pn+s}
  \quad\Longrightarrow\quad
  \theta_{\mathrm{Jain}}
  \equiv -\frac{2ps}{2pn+s}
  \pmod{1}
  \,,\qquad
  p,n\in\mathbb{N}_{\geq1},\quad s\in\{+1,-1\}.
\end{equation}
We recover this formula directly from the Jain $K$-matrices below,
in Lem. \ref{KMatrixBraidingAnglesInJainStates} with $r=1$.
For example, it gives $1/3$ modulo $1$ at $\nu=1/3$ and
$3/5\equiv-2/5$ modulo $1$ at $\nu=2/5$, agreeing with
Ex. \ref{ExperimentallyObservedBraidingPhases} in the stated orientation.
Reversing a braid conjugates its phase, replacing $\theta$ by $-\theta$;
this differs from an integer shift of the statistics parameter.

The quasi-particle species matters independently of orientation.
At $\nu=2/5$, the matrix \eqref{KMatrixForJainSequence} is
$K=\left[\begin{smallmatrix}3&2\\2&3\end{smallmatrix}\right]$ with
$Q=(1,1)^t$. The fundamental weight $W=(1,0)^t$ has charge magnitude
$e\lvert Q^tK^{-1}W\rvert=e/5$ and statistics parameter
$W^tK^{-1}W=3/5$. A full electromagnetic flux insertion instead has
weight $W=Q$, charge magnitude $2e/5$, and parameter
$Q^tK^{-1}Q=2/5$. It is a composite of fundamental quasi-particles;
its phase is not the fundamental quasi-particle phase measured in
\cite{NakamuraEtAl23}.

The denominator comparison in Prop.
\ref{ComparingToKMatrixPredictionsInJainstates} is insensitive to both
integer shifts and orientation reversal. It uses
\eqref{BraidingAngleAccordingToJain} for a fundamental quasi-particle;
the weight $W=Q$ instead describes the full multifundamental combination
\eqref{WeigthVectorsForMultifundamentalquasi-particles}.
\end{remark}

\begin{example}
[\bf Symmetric $K$-Matrix formulation for Jain states]
\label
{KMatrixFormulationForJainStates}
FQH liquids at Jain filling fractions $\nu = \frac{n}{2pn\pm1}$ \eqref{JainSequenceFillingFractions} are meant to be described in the $K$-matrix formalism (\cref{KMatrixCSTheory}, here in the ``symmetric basis'') by $n \times n$ $K$-matrices and an $n$-component charge vector $Q$ of the following form (slightly generalizing \cite[(5.9)]{WenZee92}\cite[(4.53)]{Wen95} to include the negative sign branch):
\begin{equation}
\label
{KMatrixForJainSequence}
    K
    := 
    \left(
    \begin{matrix}
      2p \pm 1 & 2p & 2p & 2p & \cdots
      \\
      2p & 2p \pm 1 & 2p & 2p & \cdots
      \\
      2p & 2p & 2p \pm 1 & 2p & \cdots
      \\
      2p & 2p & 2p & 2p \pm 1 & \cdots
      \\
      \vdots & \vdots & \vdots & \vdots & 
      \ddots
    \end{matrix}
    \right)
    \,,
    \;\;\;\;\;\;
    Q =
    \left(
    \begin{matrix}
      1
      \\
      1
      \\
      1
      \\
      1
      \\
      \vdots
    \end{matrix}
    \right)
    \mathrlap{\,.}
\end{equation}
Hence
\begin{equation}
  K 
    \;=\; 
  2 p \, Q Q^t
    \pm 
  I 
  \mathrlap{\,.}
\end{equation}

Moreover, the weight vector $W$ for a ``fundamental quasi-particle'' is meant \cite[p. 32]{Wen95} to have a single entry $1$ and all other entries $0$, while the ``dominating'' quasi-particle has a weight vector with all entries 1. Interpolating between these two cases, and since the ordering of the fundamental quasi-particle species is irrelevant, we will consider weight vectors:
\begin{equation}
\label
{WeigthVectorsForMultifundamentalquasi-particles}
  W_{1\cdots r}
  :=
  \left(
  \begin{matrix}
    1
    \\
    \vdots
    \\
    1
    \mathrlap{
      \;\;\;
      \text{\small $r$th entry}
    }
    \\
    0
    \\
    \vdots
    \\
    0
  \end{matrix}
  \right)
  \hspace{2.4cm}
  \text{for $1 \leq r \leq n$.}
\end{equation}
In this notation, the fundamental quasi-particle is $W_1 := W_{1 \cdots 1}$.

\end{example}
\begin{lemma}[{cf. \cite[\S 2]{GhoshEtAl2004}}]
  The inverse of the $K$-matrix \eqref{KMatrixForJainSequence} exists and is given by:
  \begin{equation}
    \label
    {InverseKMatrixForJainSequence}
    K^{-1}
    \;=\;
    \pm I -
    \frac
      { 2 p  }
      { 1 \pm 2 p n }
    QQ^t
    \mathrlap{\,.}
  \end{equation}
\end{lemma}
\begin{proof}
  We make the ansatz:
  \[
    K^{-1}
    =
    a I + b QQ^t
  \]
  and solve $K K^{-1} = I$ for $a, b \in \mathbb{R}$. The condition is that:
  \[
   \big(
     \pm I + 2p QQ^t
   \big)
   \big(
     a I + b QQ^t
    \big)
    =
    I
    \mathrlap{\,.}
  \]
  Observing that $(QQ^t)^2 = n QQ^t$, this is equivalent to the matrix equation
  \[
      \pm a I  
        + 
      (\pm b  + 2 p a + 2 p n b )QQ^t
      =
      I
      \mathrlap{\,.}
  \]
  Taking $a=\pm1$ and $b=-2p/(1\pm2pn)$ makes the first coefficient equal to $1$ and the second equal to $0$, proving the displayed inverse directly. This verification also applies when $n=1$, where $I$ and $QQ^t$ are not linearly independent.
\end{proof}

A special case of the following result is claimed in \cite[below (5.9)]{WenZee92}, but it is worth making this explicit:
\begin{lemma}
For the matrices in \eqref{KMatrixForJainSequence}, the relation \eqref{KMatrixFillingFraction} for the filling fraction holds, in that:
\begin{equation}
\label{QtKInvQ}
  Q^t  K^{-1}  Q
  \;=\;
  \frac
    { n }
    { 2p n \pm 1 }
  \mathrlap{\,.}
\end{equation}
\end{lemma}
\begin{proof}
  From \eqref{InverseKMatrixForJainSequence}, and observing that 
  $Q^t Q = n$, we obtain:
  \begin{equation}
  \label
  {ComputingKMatrixFillingFractionForJainState}
    Q^t  K^{-1} Q
    \;=\;
    \pm n
    -
    \frac
      { 2 p  }
      { 1 \pm 2 p n }
    n^2
    \;=\;
    \frac
      { \pm n(1 \pm 2 p n) - 2pn^2 }
      { 1 \pm 2 p n }
    \;=\;
    \frac
      { n }
      { 2 p n \pm 1 }
    \mathrlap{\,,}
  \end{equation}
  as claimed.
\end{proof}

For completeness, we derive the familiar Jain torus degeneracy
\eqref{KMatrixGroundStateDegeneracy} directly from the symmetric
$K$-matrix \eqref{KMatrixForJainSequence}; the formula is stated
explicitly in \cite[\S II, following (2.4)]{LopezFradkin99}:

\begin{lemma}
\label
{KMatrixGroundStateDegeenracyInJainState}
  The $K$-matrix ground state degeneracy 
  \eqref{KMatrixGroundStateDegeneracy}
  for the Jain state $K$-matrix
  \eqref{KMatrixForJainSequence} equals the denominator of the Jain filling fraction, in that:
  \begin{equation}
    \big\vert\mathrm{det}(K)\big\vert
    =
    2pn \pm 1
    \mathrlap{\,.}
  \end{equation}
\end{lemma}
\begin{proof}
  Using that determinants remain unchanged when adding one column to another, consider adding to the first column all other columns. This makes all entries in the first column equal to $2 p n  \pm 1$. Therefore:
  \[
    \mathrm{det}(K)
    =
    (2 p n \pm 1)
    \,
    \mathrm{det}(K')
    \,,
  \]
  where $K'$ is like $K$ except that all entries in its first column equal 1. 

  Using next that determinants are also unchanged when subtracting one row from another, consider subtracting the first row of $K'$ from each of the other rows. This results in a matrix $K''$ as shown in parentheses in the following equation:
  \[
    \mathrm{det}(K)
    =
    (2 p n  \pm 1)\,
    \mathrm{det}\big(K''\big)
    =
    (2 p n \pm 1)
    \,
    \mathrm{det}
    \left(
    \begin{matrix}
      1 & 2p & 2p & 2p & \cdots
      \\
      0 & \pm 1 & 0 & 0 & \cdots
      \\
      0 & 0 & \pm 1 & 0 & \cdots
      \\
      0 & 0 & 0 & \pm 1 & \cdots
      \\
      \vdots & \vdots & \vdots & \vdots & 
      \ddots
    \end{matrix}
    \right)
    \mathrlap{\,.}
  \]
  With $K''$ being upper triangular, its determinant equals the product of its diagonal entries and hence has unit absolute value. This proves the claim.
\end{proof}

Next, we evaluate the braiding parameters for the weight vectors
\eqref{WeigthVectorsForMultifundamentalquasi-particles}:
\begin{lemma}
\label
{KMatrixBraidingAnglesInJainStates}
  The $K$-matrix braiding angles \eqref{KMatrixBraidingPhase} corresponding to the weight vectors $W_{1\cdots r}$ \eqref{WeigthVectorsForMultifundamentalquasi-particles}
  are, modulo $\mathbb{Z}$ 
  \eqref{BraidingAngleModuloOne}:
  \begin{equation}
  \label
  {JainStateBraidingPhasesInKMatrixFormalism}
    \theta_{_{W_{1 \cdots r}}}
    \,\mathrm{mod}\,
    1
    \;=\;
    \mp
    \frac
      { 2pr^2 }
      { 2pn \pm 1 }
    \,\mathrm{mod}\,
    1
    \mathrlap{\,.}
  \end{equation}
\end{lemma}
\begin{proof}
Using that $W_{1 \cdots r}^t W_{1 \cdots r} = r = W_{1 \cdots r}^t Q$ and $W_{1 \cdots n} = Q$, we find from \eqref{InverseKMatrixForJainSequence} 
and in generalization of \eqref{ComputingKMatrixFillingFractionForJainState}:
\[
  \begin{aligned}
    W_{1 \cdots r}^t K^{-1} W_{1 \cdots r}
    &=\;
    \pm r
    -
    \frac
      { 2p }
      { 1 \pm 2pn }
    r^2
    \;=\;
    \frac
      { 
        \pm r(1 \pm 2pn) 
        -
        2pr^2
      }
      { 1 \pm 2pn }
    \;=\;
   \frac
     { \pm r + 2pr(n-r) }
     { 1 \pm 2pn }
  \\
  & =\;
   \frac
     { r \pm  2pr(n-r) }
     { 2pn \pm 1 }
   \;=\;
   \frac
     { \pm r (2pn \pm 1) \mp 2pr^2  }
     { 2pn \pm 1 }
   \;=\;
   \pm r
   \mp
   \frac
     { 2pr^2 }
     { 2pn \pm 1 }
    \mathrlap{\,.}
    \qedhere
  \end{aligned}
\]
\end{proof}
This essentially completes the proof of the claim of Prop. \ref{ComparingToKMatrixPredictionsInJainstates}; we just need to isolate which of the fractions in \cref{JainStateBraidingPhasesInKMatrixFormalism} are actually reduced.
\begin{proof}[Proof of \cref{ComparingToKMatrixPredictionsInJainstates}]
  By \cref{KMatrixGroundStateDegeenracyInJainState,KMatrixBraidingAnglesInJainStates} the claim holds for all multifundamental quasi-particle species $W_{1\cdots r}$ \cref{WeigthVectorsForMultifundamentalquasi-particles} for which the fraction in \cref{JainStateBraidingPhasesInKMatrixFormalism} is reduced, hence for which
  \begin{equation}
    \begin{aligned}
    &
    \gcd\bracket({
      2p r^2
      \,,
      2pn \pm 1
    })
    = 1
    \\
    \Leftrightarrow \;\;
    &
    \gcd\bracket({
      r^2
      \,,
      2pn \pm 1
    })
    = 1
    \\
    \Leftrightarrow \;\;
    &
    \gcd\bracket({
      r
      \,,
      2pn \pm 1
    })
    = 1
    \mathrlap{\,.}
    \end{aligned}
  \end{equation}
  This clearly subsumes the fundamental quasi-particle $W_{1\cdots 1}$, as claimed. (It also works for the full multicomponent quasi-particle $W_{1 \cdots n}$.) 
\end{proof}

\begin{remark}
For each specified root $\zeta$, Thm.
\ref{ClassificationOf2CohomotopicalFluxQuantumStatesOverTorus} supplies
extensions of the indicated Heisenberg representations for both the
$\mathrm{pp}$ and $\mathrm{aa}$ spin structures. The even numerator in
\eqref{JainStateBraidingPhasesInKMatrixFormalism} is not an existence
criterion for these extensions: that equation specifies an angle modulo
$1$, and a choice of representative fixes the sign of
$\zeta=e^{\pi\mathrm{i}\theta}$, while $\zeta^2$ is independent of it.
Comparison with a spin Chern-Simons theory requires matching the spin
structure, quadratic refinement, and scalar phases of a linear lift.
Such spin-dependent modular data are available in
\cite[\S 5.6.1]{BelovMoore05}; see also
\cref{KMatrixGroundStateMonodromy}.
\end{remark}

 Of course, it is easy to find $K$-matrix scenarios that violate the relation \eqref{PredictionOfBraidingPhaseToGroundStateDegeneracy}:
\begin{example}
\label
{WouldBeHalperinStateWithoutInteraction}
$K$-matrix theory predicts that two FQH liquids of filling fraction $1/p$, for $p\in\mathbb{N}_{>0}$, may inhabit the same physical space \emph{without} interacting, as such described by
  \begin{equation}
  \label
  {Noninteracting2ComponentKMatrix}
    K :=
    \left(
    \begin{matrix}
      p & 0
      \\
      0 & p
    \end{matrix}
    \right)
    \,,
    \;\;\;\;
    Q 
      :=
    \left(
    \begin{matrix}
      1
      \\
      1
    \end{matrix}
    \right)
    \mathrlap{\,.}
  \end{equation}
  The predicted ground state degeneracy \eqref{KMatrixGroundStateDegeneracy} would be $\vert\mathrm{det}(K)\vert = p^2$, while the braiding angles \eqref{KMatrixBraidingPhase} of the fundamental quasi-particles would have denominator just $p$.
\end{example}
The dimension $p^2$ in \eqref{Noninteracting2ComponentKMatrix}
is the dimension of the tensor product of the two $p$-dimensional
component Hilbert spaces. Such a tensor product is irreducible for the
product of the two component observable algebras when each factor is
irreducible. Thus its dimension does not contradict
\eqref{PredictionOfBraidingPhaseToGroundStateDegeneracy}, which assumes
an irreducible representation of a single integer Heisenberg group.
For example, restricting to the first component alone gives $p$ copies
of its $p$-dimensional representation; this restriction does not describe
the full two-component observable algebra.

For the following symmetric two-component matrices, the same numerical
relation does hold when the entries are coprime and the matrix is
nonsingular. This is an arithmetic statement and does not assert that
nonzero coupling alone forces an indecomposable topological order.
\begin{example}
[\bf Symmetric two-component $K$-matrices and the Halperin regime]
\label
{KMatrixDescriptionOfHalperinState}
Consider the following family of symmetric two-component $K$-matrices
and charge vectors, which includes the coprime Halperin cases
(cf. \cite[p. 2299]{WenZee92}\cite[\S IV.A]{GRG08}):
\begin{equation}
\label
{SymmetricKMatrixForHalperinState}
  K
  =
  \left(
  \begin{matrix}
    p & l
    \\
    l & p
  \end{matrix}
  \right)
  \,,
  \;\;\;\;
  Q 
  :=
  \left(
  \begin{matrix}
    1
    \\
    1
  \end{matrix}
  \right)
  \;\;
  \text{ for }
  \left\{
  \begin{aligned}
    &
    p \in \mathbb{N}_{\geq 0}
    ,\,
    l \in \mathbb{Z} \setminus \{0\}
    \\
    &
    \gcd(p,l) = 1,\qquad p^2 \neq l^2
    \mathrlap{\,.}
  \end{aligned}
  \right.
\end{equation}
The conventional homogeneous holomorphic Halperin $(p,p,l)$ regime has
$p>l\geq 0$, with $p$ odd for fermionic constituents and even for bosonic
constituents. In the nonzero-coupling family displayed here this means
$p>l>0$. In particular, the eigenvalues $p\pm l$ are positive;
the plasma analysis in \cite[\S IV.A, (14)]{GRG08} explains the
phase-separation instability for $l>p$. The following arithmetic result
holds throughout the larger nonsingular family
\eqref{SymmetricKMatrixForHalperinState}, without requiring a
homogeneous Halperin-wavefunction realization.
\end{example}

\begin{proposition}
\label
{AnalysisOfHalperinStates}
  For every symmetric two-component $K$-matrix and charge vector in
  \eqref{SymmetricKMatrixForHalperinState}, the torus ground state
  degeneracy predicted by \eqref{KMatrixGroundStateDegeneracy}
  coincides with the denominator of the full-braid parameter
  \eqref{BraidingAngleModuloOne} computed from
  \eqref{KMatrixBraidingPhase} for either fundamental quasi-particle.
\end{proposition}
\begin{proof}
This is a straightforward computation:
The ground state degeneracy \eqref{KMatrixGroundStateDegeneracy} is:
\begin{equation}
\label
{GroundStateDegeneracyOfHalperinState}
  \big\vert 
    \mathrm{det}(K)
  \big\vert
  =
  \vert p^2 - l^2 \vert
  \mathrlap{\,.}
\end{equation}
Moreover, observing that the inverse $K$-matrix is
\[
  K^{-1}
  =
  \frac
    { 1 }
    { p^2 - l^2 }
  \left(
  \begin{matrix}
    \, p & -l
    \\
    -l & \, p
  \end{matrix}
  \right),
\]
we find the braiding angle \eqref{KMatrixBraidingPhase} of the fundamental quasi-particles $W_i$ \eqref{WeigthVectorsForMultifundamentalquasi-particles} to be
\[
  W_i^t K^{-1} W_i
  \;=\;
  \frac
    { p }
    { p^2 - l^2 }
  \mathrlap{\,.}
\]
This fraction has nonzero denominator and is in lowest terms, since $\gcd(p,p^2-l^2)=\gcd(p,l^2)=1$. Hence the absolute value of its denominator equals $\vert p^2-l^2\vert$, as claimed.
\end{proof}

\subsection{On punctured surfaces}
\label{FluxThroughPuncturedSurface}

Here we derive the observables on 2-cohomotopically quantized topological flux over {\it $n$-punctured surfaces}, which in practice will mean: Surfaces of conducting material where magnetic flux is {\it expelled} from (the vicinity of) $n$ defects (cf. Rem. \ref{FluxDensityQuantizedInCohomotopy}). 

\smallskip 
Covariantization of these observables reveals an action of the surface's
$n$-braid group (cf. Prop. \ref{ExtensionOfMCGByFluxMonodromy}).
We determine here the action of the Artin half-twists supported in a disk
containing the punctures. The puncture contribution is the lattice
$\mathbb{Z}^n/\mathbb{Z}(1,\cdots,1)$: integer labels at the punctures,
considered modulo a common shift. This is the dual of the sum-zero lattice,
and distinguishing these two integral forms of the standard representation
is essential for interpreting the resulting groups in terms of framed
braids \eqref{FramedBraidGroup}, as considered in \S\ref{OnPuncturedDisks}.

\smallskip
The group calculations in this subsection use orientation-preserving
diffeomorphisms. For a fixed spin structure, their covariant versions
are restricted to its stabilizer, as in Rem. \ref{SpinStructure}.
In \S\ref{OnPuncturedDisks}--\ref{QBitQuantumGatesOperableOn2PuncturedOpenDisk}
we choose the spin structure extending across the punctures from the
unique spin structure on the sphere or disk. Every oriented mapping
class in those calculations preserves this choice.

\begin{lemma}[\bf Homotopy type of compactified $n$-punctured surface]
 \label{HomotopyTypeOfCompactifiedPuncturedSurface}
  For $n \in \mathbb{N}_{\geq 1}$, the one-point compactification of the $n$-puncturing of a compact surface (possibly with boundary) $\Sigma^2_{g,b}$ \eqref{TheSurface}  is homotopy equivalent to the wedge sum 
  \eqref{WedgeSum}
  of that surface with $(n-1)$ circles:
  \begin{equation}
    \label{ShapeOfCptOfPuncturedSurface}
    \big(
      \Sigma^2_{g,b,n}
    \big)_{\cpt}
    \;\;
    \shapeEquivalence
    \;\;
    \Sigma^2_{g,b}
      \,\vee\,
    \textstyle{\bigvee_{n-1}} 
    S^1
    \,.
  \end{equation}
\end{lemma}
\begin{proof}
  For $n = 1$ the statement is immediate. 
  
  For $n = 2$ consider the topological space $X$ obtained by attaching to $\Sigma^2_{g,b}$ an interval with endpoints glued to two distinct points $s_1, s_2 \in \Sigma^2_{g,b}$ (the would-be positions of the punctures), hence consider this pushout of topological spaces:
  $$
    \begin{tikzcd}[row sep=small, column sep=huge]
      S^0 
      \ar[
        r,
        "{ 
          (s_1, s_2)
        }"
      ]
      \ar[d, hook]
      \ar[
        dr,
        phantom,
        "{
          \scalebox{.7}{
            \color{gray}
            (po)
          }
        }"{pos=.8}
      ]
      &
      \Sigma^2_{g,b}
      \ar[d, hook]
      \\
      D^1
      \ar[
        r, 
        "{ \iota_{\mathrm{ext}} }"{swap}
      ]
      &
      X \,.
    \end{tikzcd}
  $$
  Moreover, consider another arc {\it inside} $\Sigma^2_{g,b}$ connecting these two points
  $$
    \begin{tikzcd}
    D^1 
    \ar[
      r,
      hook,
      "{
        \iota_{\mathrm{int}} 
      }"
    ]
    &
    \Sigma^2_{g,b}
    \ar[
      r,
      hook
    ]
    &
    X
    \,.
    \end{tikzcd}
  $$
  Both of these arcs are evidently contractible sub-complexes of $X$, and so the quotient projections obtained by identifying either arc with a single point are weak homotopy equivalences (cf. \cite[p. 11]{Hatcher02}):
  
  \begin{equation}
    \label{PinchingEquivalences}
    \begin{tikzcd}[
      row sep=4pt, 
      column sep=50pt
    ]
      &
      X
      \ar[
        dl,
        "{ \shapeEquivalence }"{sloped}
      ]
      \ar[
        dr,
        "{ \shapeEquivalence }"{sloped}
      ]
      \\
        X / \iota_{\mathrm{ext}}(D^1)
      \ar[
        rr,
        dashed,
        "{ \shapeEquivalence }"{swap}
      ]
      &&
      X / \iota_{\mathrm{int}}(D^1)
      \,.
    \end{tikzcd}
  \end{equation}

  \noindent
  \begin{minipage}{9.5cm}
  Now, as indicated in \eqref{PinchingEquivalences}, the ``external'' quotient on the left is evidently homeomorphic to the desired one-point compactification, while the ``internal'' quotient on the right is evidently homeomorphic to the claimed wedge sum. 
  This proves the claim for $n = 2$.  

  \smallskip 
  The graphic on the right illustrates the situation for the case $g,b = 0$.

  \smallskip

  The general statement, including the case $n > 2$, follows analogously by attaching further arcs in this fashion, cf. \hyperlink{FigureArcs}{Fig. A}.
  \qedhere

  \end{minipage}
  \;\;
  \adjustbox{
    raise=-2cm,
    scale=1.3
  }{
\begin{tikzpicture}
  \node at (0,0) {
  \includegraphics[width=5cm]
  {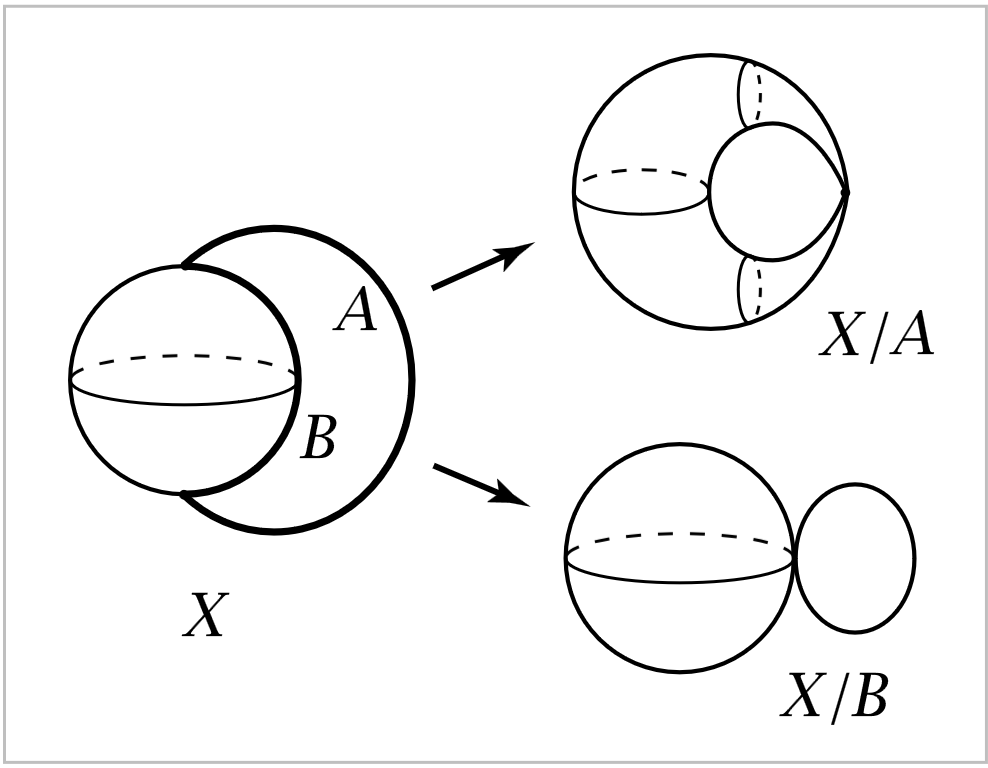}
  };
  \node[
    scale=.5,
    rotate=-40
  ] at (-1.87,.86)
  {
    \color{purple}
    puncture
  };
  \begin{scope}[
    yscale=-1
  ]
  \node[
    scale=.5,
    rotate=+40
  ] at (-1.85,.85)
  { \color{purple}
    puncture
  };
  \end{scope}
  \node[
    scale=.8
  ] 
    at (-2.15,-.53)
  {$\Sigma^2_0$};
  \node[
    scale=.8
  ] 
    at (.8,.14)
  {$(\Sigma^2_{0,0,2})_{\cpt}\,
  \rotatebox[origin=c]{10}{$=$}\,$};
  \node[scale=.8] at (.7,-1.7)
  {
    $\Sigma^2_0 \vee S^1
   \;
  \adjustbox{rotate=10,raise=2pt}{$=$}\,$
  };
  \node[
    rotate=-20,
    scale=.8
  ]
    at (-.2,-.6)
    {$\sim$};
  \node[
    rotate=+20,
    scale=.8
  ]
    at (-.18,+.66)
    {$\sim$};
  \node[
    scale=.7
  ]
  at (1.98,.97) 
  {$\infty$};
  \begin{scope}[
    shift={(-1.47,.57)}
  ]
  \draw[purple]
    (180-5:.2 and .1) arc
    (180-5:360+5:.11 and .05);
  \end{scope}
  \begin{scope}[
    yscale=-1,
    shift={(-1.47,.52)},
  ]
  \draw[purple]
    (180-5:.2 and .1) arc
    (180-5:360+5:.11 and .05);
  \end{scope}

  \draw[
    draw=white,
    fill=white
  ] 
    (-.7,.35) circle
    (.16);
  \node[scale=.7] at (-.7,.35) 
    {$\iota_{\mathrm{ext}}$};

  \draw[
    draw=white,
    fill=white  
  ]
    (-.86,.-.27) circle
    (.16);
  \node[scale=.7] at 
    (-.84,.-.27)
    {$\iota_{\mathrm{int}}$};

  \draw[
    draw=white,
    fill=white
  ]
    (2.15,.22) circle
    (.16);

  \node[scale=.7] at (2.2,.2)
    {$\iota_{\mathrm{ext}}$};

  \draw[
    draw=white,
    fill=white
  ]
    (1.95,-1.58) circle
    (.16);
  \node[scale=.7] 
    at (1.99,-1.64) 
     {$\iota_{\mathrm{int}}$};

\end{tikzpicture}
  }
  \hspace{-10pt}
  \adjustbox{
    rotate=-90,
    scale=.7
  }{
    \color{gray}
    \clap{
    graphics adapted from 
    \cite[p. 11]{Hatcher02}
    }
  }
\end{proof}

\vspace{-1.5cm} 
\begin{center} 
\begin{minipage}{10.2cm}
\hypertarget{FigureArcs}{}
\small
{\bf Figure A.}
There are several ways to attach arcs for $n > 2$ punctures in the above proof of Lem. \ref{HomotopyTypeOfCompactifiedPuncturedSurface},
all equivalent in the resulting homotopy type. But for the analysis of braiding that follows in Thm. \ref{BraidGroupActionOnFluxMonodromyOverPuncturedSurface} it is useful (cf. \hyperlink{FigureS}{\it Fig. S}\,) to single out one puncture $v_n$ and take the $n-1$ arcs to connect this one puncture to each of the $n-1$ remaining ones. 

The case $g,b = 0$ and $n =3$ is illustrated on the right.
 \end{minipage}
  \hspace{-.6cm}
  \adjustbox{
    raise=-3cm, 
    scale=0.75
  }{
\begin{tikzpicture}

\draw[
  line width=1.4
] (0,0) circle (1.5);

\draw[
  line width=1.2,
  dashed,
  gray
] 
  (1.5,0) arc 
  (0:180:1.5 and .55);
\draw[
  line width=1.2,
  gray
] 
  (-1.5,0) arc 
  (180:360:1.5 and .55);

\draw[
  draw opacity=0,
  fill=white
] (0,1.5) 
  ellipse (.18 and .1);
\draw
  (-.16,1.5) arc
   (185:355:.18 and .1);

\begin{scope}[
  rotate=-148
]
\draw[
  draw opacity=0,
  fill=white
] (0,1.5) 
  ellipse (.18 and .1);
\draw
  (-.16,1.5) arc
   (185:355:.18 and .1);
\end{scope}

\begin{scope}[
  rotate=+149
]
\draw[
  draw opacity=0,
  fill=white
] (0,1.5) 
  ellipse (.18 and .1);
\draw
  (-.16,1.5) arc
   (185:355:.18 and .1);
\end{scope}

\draw[
  line width=.4,
  gray
] (0,0) circle (1.5);

\draw[
  line width=1.8
] 
  (90:1.45) .. controls
  (70:4.5) and 
  (-40:4.5) ..
  (-60:1.46);

\begin{scope}[
  xscale=-1
]
\draw[
  line width=1.8
] 
  (90:1.45) .. controls
  (70:4.5) and 
  (-40:4.5) ..
  (-60:1.46);
\end{scope}

\node[
  scale=1
] at (0,1.25) 
  {$v_3$};

\node[
  scale=1
] at (240:1.25) 
  {$v_1$};

\node[
  scale=1
] at (300:1.22) 
  {$v_2$};

\end{tikzpicture}  
  }
\end{center} 
\vspace{-.6cm}

With this result in hand, it is straightforward to compute the solitonic flux monodromy \eqref{SolitonicTopologicalObservablesMoreGenerally} through a punctured surface:

\begin{proposition}[\bf Flux monodromy through punctured surface]
  \label{FluxMonodromyThroughPuncturedSurface}
  For $g,b \in \mathbb{N}$ and $n \in \mathbb{N}_{> 0}$,
  we have an isomorphism
  \begin{equation}
    \label{FluxMonodromyOnNPuncturedSurface}
      \pi_1\Big(
        \mathrm{Map}^\ast_0\big(
          (\Sigma^2_{g,b,n})_{\cpt}
          ,\,
          S^2
        \big)
      \Big)
      \;\simeq\;
      \pi_1\Big(
        \mathrm{Map}^\ast_0\big(
          \Sigma^2_{g,b}
          ,\,
          S^2
        \big)
      \Big)
      \,\times\,
      \mathbb{Z}^{n-1}
      \,.
  \end{equation}
\end{proposition}
\begin{proof}
  We may compute as follows:
  $$
    \def\arraystretch{2}
    \begin{array}{rcll}
      \pi_1\Big(
        \mathrm{Map}^\ast_0\big(
          (\Sigma^2_{g,b,n})_{\cpt}
          ,\,
          S^2
        \big)
      \Big)
      &\simeq&
      \pi_1\Big(
        \mathrm{Map}^\ast_0\big(
          \Sigma^2_{g,b}
          \vee 
          \bigvee_{n-1}S^1
          ,\;
          S^2
        \big)
      \Big)
      &
      \proofstep{
        by
        \eqref{ShapeOfCptOfPuncturedSurface}
      }
      \\
      &=&
      \pi_1\Big(
        \mathrm{Map}^\ast_0\big(
          \Sigma^2_{g,b}
          ,\,
          S^2
        \big)
        \times
        \prod_{n-1}
        \mathrm{Map}^\ast\big(
          S^1
          ,\,
          S^2
        \big)
      \Big)
      &
      \proofstep{
        by \eqref{MapsOutOfWedgeSum}
      }
      \\
      &=&
      \pi_1\Big(
        \mathrm{Map}^\ast_0\big(
          \Sigma^2_{g,b}
          ,\,
          S^2
        \big)
      \Big)
        \times
        \prod_{n-1}
     \pi_1\Big(
        \mathrm{Map}^\ast\big(
          S^1
          ,\,
          S^2
        \big)
      \Big)
      \\
      &=&
      \pi_1\Big(
        \mathrm{Map}^\ast_0\big(
          \Sigma^2_{g,b}
          ,\,
          S^2
        \big)
      \Big)
        \times
        \prod_{n-1}
        \pi_2(S^2)
      &
      \proofstep{
        by 
        \eqref{SpheresFromSmashProducts},
      }
    \end{array}
  $$
  whence the claim follows by $\pi_2(S^2) \simeq \mathbb{Z}$.
\end{proof}

\begin{remark}[\bf Internal punctures and the rim of the sample]
\label{InternalPuncturesAndTheRim}
\begin{itemize}[leftmargin=.8cm] 
\item[{\bf (i)}] In words, 
Prop. \ref{FluxMonodromyThroughPuncturedSurface} may be understood as saying that the (monodromy of) solitonic flux on an $n$-punctured surface, for $n \geq 1$, has (not $n$ but) $n-1$ generators associated with $n-1$ of the punctures, while associated with the remaining puncture is the solitonic flux of the surface itself, regarded as containing its point-at-infinity.

\item[{\bf (ii)}] Below in Thm. \ref{BraidGroupActionOnFluxMonodromyOverPuncturedSurface} we give further analysis of this situation, showing that the Artin half-twists act on the puncture factor by the dual integral standard representation of the symmetric group $\mathrm{Sym}_n$ (Def. \ref{StandardRepresentationOfSymmetricGroup} below), but first note that the discrepancy between $n$ and $n-1$ here has a clear physical meaning at least for $g = 0$, recalling that the $n$-punctured sphere is homeomorphic to the $(n-1)$-punctured plane, hence the $(n-1)$-punctured open disk $\Sigma^2_{0,0,n} \simeq \mathbb{R}^2 \setminus \{x_1, \cdots, x_{n-1}\} $ \eqref{ExamplesOfSurfaces}: 

\begin{minipage}{10.8cm}
We may think of the $n$th puncture as modeling the outer rim of the slab of material that is represented by $\Sigma^2$, and of the remaining $n-1$ punctures as the actual material defects as one will recognize them in the laboratory.
\end{minipage}
\;
    \adjustbox{
      raise=-.5cm
    }{
    \begin{tikzpicture}[scale=1.3]
      \draw[
        draw opacity=0,
        fill=gray!40
      ]
        (0,0)
        ellipse
        (2 and .5);
      \foreach \n in {1,...,3} {
      \draw[
        draw opacity=0,
        draw=gray!80,
        fill=white
      ]
        ({80+\n*(360/3)}:.85 and .85*.26)
        ellipse
        (2*.1 and .5*.1);
      \node[
        scale=.8
      ] at
        ({80+\n*(360/3)}:.85 and .85*.26)
        {$\mathbf{v_{\n}}$};
      }
      \begin{scope}[scale=1.1]
      \node[
        scale=1.1
      ] at
        (30:2 and .5)
        {\scalebox{.8}{$
          \mathbf{v_4}
        $}};      
      \end{scope}
      \begin{scope}[scale=1.18]
      \node[
        scale=1.1
      ] at
        (43:2 and .5)
        {
            \adjustbox{
              scale=.7,
              rotate=-15
            }
            {
              \llap{
                \color{darkblue}
                \bf
                rim
              }
            }        
        };
      \end{scope}
      \node[
        scale=.7,
      ] at (-.1,-.16) {
        \bf
        \color{darkblue}
        \def\arraystretch{.85}
        \begin{tabular}{c}
          internal
          \\
          punctures
        \end{tabular}
      };
    \end{tikzpicture}
    }

\item[{\bf (iii)}] With this physical interpretation, the rim puncture is
distinguished: motions of the internal defects preserve it. The full
mapping class group of the punctured sphere also allows permutations
that move this puncture. The examples in \S\ref{OnTheOpenAnnulus} and
\S\ref{QBitQuantumGatesOperableOn2PuncturedOpenDisk} first treat that
full group and then describe the restriction preserving the rim.
A disk with its actual boundary fixed pointwise is the separate case
of Prop. \ref{CovariantFluxMonodromyOnPuncturedDisks}\textup{(ii)}.
\end{itemize} 
\end{remark}

\vspace{1mm} 
\begin{definition}[\bf Standard representation and its integral lattices]
  \label{StandardRepresentationOfSymmetricGroup}
  For $n \geq 2$ the ``\emph{standard}'' $\mathbb{C}$-linear representation of the symmetric group $\mathrm{Sym}_n$ is the $n-1$-dimensional complex irrep classified by the partition $(n-1, 1)$, 
  hence by the Young diagram 
  \adjustbox{
    raise=-.2cm
  }{
  \begin{tikzpicture}[scale=0.8]
    \draw (0,0) rectangle (.3,.3);
    \draw (0+.3,0) rectangle (.3+.3,.3);    
    \draw (0+2*.3,0) rectangle (.3+2*.3,.3);    
    \draw[densely dotted] (0+3*.3,0) rectangle (.3+3*.3,.3);    
    \draw(0+4*.3,0) rectangle (.3+4*.3,.3);    
    \draw (0,-.3) rectangle (.3,.3-.3);
  \end{tikzpicture}
  }:
  this is the
  quotient of the defining $n$-dimensional permutation representation 
  by the trivial 1-dimensional representation (cf. \cite[p. 9 \& Ex. 4.6]{FultonHarris91}\cite[Def. 2.5]{Kao10}).

  More concretely, with respect to the canonical linear basis
  \begin{equation}
    \label{DefiningRepresentation}
    \mathbb{C}^n
    \;\simeq\;
    \mathbb{C}\langle
      v_1, v_2, \cdots, v_{n}
    \big\rangle
    \,,
  \end{equation}
  \begin{itemize}
  \item the {\it defining permutation representation} of $\mathrm{Sym}_n$ is given, 
  for $\sigma \in \mathrm{Sym}_n$, by
  $\sigma (v_i) \,:=\, v_{\sigma(i)}$,
  hence in terms of the Artin generators $(b_i)_{i=1}^{n-1}$ \eqref{ArtinGeneratorsForSymmetricGroup} by
  \vspace{-1mm} 
  $$
    b_i(v_j)
    \;=\;
    \left\{\!\!
    \def\arraystretch{1}
    \begin{array}{lcl}
      v_{i+1} &\vert& j = i
      \\
      v_i &\vert& j = i + 1
      \\
      v_j &\vert& \mbox{otherwise}\,,
    \end{array}
    \right.
  $$

  \item the trivial 1-dimensional irrep inside this is
  \vspace{-1mm} 
  \begin{equation}
    \label{TrivialIrrepInsideDefining}
    \mathbf{1}
    \,\simeq\,
    \mathbb{C}\big\langle
      \underbrace{
        v_1 + v_2 + \cdots + v_n
      }_{
        =:\, t
      }
    \big\rangle
    \;\; \xhookrightarrow{\quad}\;
    \mathbb{C}^n
    \,,
  \end{equation}

  \item and the {\it standard representation} is: 
  \begin{equation}
    \label{TheStandardRep}
    \mathbf{n\!\!-\!\!1}
    \,\simeq\,
    \mathbb{C}\big\langle
      \underbrace{
        v_i - v_n
      }_{ =: \, e_i }
    \big\rangle_{i =1}^{n-1}
    \;\;\xhookrightarrow{\phantom{--}}\;
    \mathbb{C}^n
    \,,
  \end{equation}
  with the adjacent transpositions acting on these basis vectors as follows
  (in the first line, $1 \leq i < n-1$):
  \begin{align}
    \label{LowerArtinGeneratorsInStandardRep}
    \def\arraystretch{1}
      b_i(e_j)
      &\;=
      \left\{\!
      \def\arraystretch{1}
      \begin{array}{lcl}
        e_{i + 1} &\vert& j = i
        \\
        e_i &\vert& j = i + 1
        \\
        e_j &\vert& \mbox{otherwise}
      \end{array}
      \right.
      \\[+5pt]
      \label{TopArtinGeneratorsInStandardRep}
      b_{n-1}(e_j)
      &\;=
      \left\{\!
      \def\arraystretch{1}
      \begin{array}{rcl}
        e_j - e_{n-1} &\vert& j < n-1
        \\
        - e_{n-1} &\vert& j = n-1 \,.
      \end{array}
      \right.
  \end{align}
The $\mathbb{Z}$-span of these basis vectors is the sum-zero lattice
$$
  A_{n-1}
  \;:=\;
  \Big\{
    (a_1,\cdots,a_n) \in \mathbb{Z}^n
    \;\Big\vert\;
    \textstyle\sum_{i=1}^n a_i = 0
  \Big\}
  \;=\;
  \mathbb{Z}\langle e_1,\cdots,e_{n-1}\rangle
  \,.
$$
We denote its action by $\mathrm{st}$ and call it the
{\it integral standard representation}:
\begin{equation}
  \label{IntegralStandardRepresentation}
  \begin{aligned}
    \mathrm{st}:
    \mathrm{Sym}_n
    &\;\longrightarrow\;
    \mathrm{Aut}_{\mathbb{Z}}(A_{n-1})
    \mathrlap{\,,}
    \\
    \mathrm{st}(\sigma)
    \Big(\sum_{i=1}^n a_i v_i\Big)
    &\;:=\;
    \sum_{i=1}^n a_i v_{\sigma(i)}
    \mathrlap{\,.}
  \end{aligned}
\end{equation}
\end{itemize}
Hence over $\mathbb{C}$ we have a reduction of the defining $\mathrm{Sym}_n$-representation explicitly like this:
$$
  \begin{tikzcd}[row sep=0pt, column sep=small]
    \hspace{-25pt} 
    t 
    &\longmapsto& 
    v_1 + \cdots + v_n
    \\
    e_i 
    \hspace{-15pt} 
    &\longmapsto& 
    v_i - v_n
    \\
    \mathbf{1}
    \,\oplus\,
    \mathbf{n}\!-\!\mathbf{1}
    \ar[
      from=rr,
      shift left=4pt
    ]
    \ar[
      rr, phantom, "{\sim}"
    ]
    \ar[
      rr,
      shift left=5pt
    ]
    &&
    \mathbb{C}^n_{{}_{\mathrm{def}}}
    &
    \in 
    \, 
    \mathrm{Rep}_{\mathbb{C}}(\mathrm{Sym}_n)
    \\
    \frac{
      t - (e_1 + \cdots + e_{n-1})
      \mathclap{\phantom{\vert_{\vert}}}
    }{n}
    &\longmapsfrom&
    v_n
    \\
    e_i +
    \frac{
      t - (e_1 + \cdots + e_{n-1})
      \mathclap{\phantom{\vert_{\vert}}}
    }{n}
    &\longmapsfrom&
    v_{i < n}
    \,,
  \end{tikzcd}
$$
which also shows that over the integers we only have a monomorphism
\begin{equation}
  \label{DefiningRepDecompositionOverZ}
  \begin{tikzcd}
    \mathbb{Z}
    \;\oplus\;
    A_{n-1}
    \ar[
      rr,
      hook
    ]
    &&
    \mathbb{Z}^{n}_{{}_{\mathrm{def}}}
    &
    \in
    \,
    \mathrm{Rep}_{\mathbb{Z}}(\mathrm{Sym}_n)
  \end{tikzcd}
\end{equation}
with image the subgroup of $n$-tuples whose sum is divisible by $n$.

The {\it dual integral standard representation}, denoted $\mathrm{st}^{\vee}$,
is
\begin{equation}
  \label{DualIntegralStandardRepresentation}
  \begin{aligned}
    \mathrm{st}^{\vee}:
    \mathrm{Sym}_n
    &\;\longrightarrow\;
    \mathrm{Aut}_{\mathbb{Z}}(A_{n-1}^{\vee})
    \mathrlap{\,,}
    \\
    \big(\mathrm{st}^{\vee}(\sigma)\varphi\big)(x)
    &\;:=\;
    \varphi\big(\mathrm{st}(\sigma^{-1})x\big)
    \mathrlap{\,.}
  \end{aligned}
\end{equation}
Here $\mathrm{st}$ is defined in \eqref{IntegralStandardRepresentation},
and the dual lattice is
$$
  A_{n-1}^{\vee}
  \;:=\;
  \mathrm{Hom}_{\mathbb{Z}}(A_{n-1},\mathbb{Z})
  \;\simeq\;
  \mathbb{Z}^n/\mathbb{Z}(1,\cdots,1)
  \,.
$$
Here the isomorphism is induced by the map from $\mathbb{Z}^n$ sending
$(a_1,\cdots,a_n)$ to the functional
$x \mapsto \sum_{i=1}^n a_i x_i$ on $A_{n-1}$.
This map is surjective since the values on $e_i=v_i-v_n$ can be prescribed
arbitrarily, and its kernel consists exactly of the constant tuples.
The permutation action on the quotient agrees with the dual action
\eqref{DualIntegralStandardRepresentation}.
In particular, the evaluation coordinates on this dual lattice are
$$
  \begin{tikzcd}[row sep=0pt, column sep=small]
    A_{n-1}^{\vee}
    \ar[rr, phantom, "{\simeq}"]
    &&
    \mathbb{Z}^{n-1}
    \\
    {[(a_1,\cdots,a_n)]}
    &\longmapsto&
    (q_1,\cdots,q_{n-1})
    \,=\,
    (a_1-a_n,\cdots,a_{n-1}-a_n)
    \,.
  \end{tikzcd}
$$
Under $\mathrm{st}^{\vee}$ \eqref{DualIntegralStandardRepresentation},
the adjacent transposition $b_i$ exchanges $q_i$ and $q_{i+1}$ for $i<n-1$,
while the last one acts by
$$
  \begin{aligned}
    \mathrm{st}^{\vee}(b_{n-1})
    (q_1,\cdots,q_{n-1})
    \;&=\;
    (q_1-q_{n-1},\cdots,q_{n-2}-q_{n-1},-q_{n-1})
    \,.
  \end{aligned}
$$
These are formulas for the values of a functional on the basis $e_i$;
\eqref{LowerArtinGeneratorsInStandardRep} and
\eqref{TopArtinGeneratorsInStandardRep} instead give the images of the
basis vectors themselves.

The two lattices have isomorphic complexifications, but for $n\geq 3$
they are not isomorphic as integral $\mathrm{Sym}_n$-representations.
For example, their coinvariants (the quotients by all elements
$\sigma x-x$) are
$$
  (A_{n-1})_{\mathrm{Sym}_n}
  \;=\; 0
  \,,
  \qquad
  (A_{n-1}^{\vee})_{\mathrm{Sym}_n}
  \;\simeq\; \mathbb{Z}/n\mathbb{Z}
  \qquad (n\geq 3)
  \,.
$$
Indeed, all differences $v_i-v_j$ with $i\neq j$ have the same class in
the first quotient, and
$(v_1-v_2)+(v_2-v_3)=v_1-v_3$ forces this class to vanish.
In the second quotient all $v_i$ have the same class, with the sole
remaining relation $n[v_1]=0$.
For $n=2$, both lattices carry the sign representation; for $n=1$, we
set $A_0=A_0^{\vee}=0$.
\end{definition}
Now, the main result of this section:

\begin{theorem}[\bf Braid group action on flux monodromy over punctured surface]
  \label{BraidGroupActionOnFluxMonodromyOverPuncturedSurface}
  For $n \geq 1$, choose a disk in $\Sigma^2_{g,b}$ containing all
  punctures and use arcs in this disk for the decomposition
  \eqref{FluxMonodromyOnNPuncturedSurface}.
  The action \eqref{ActionOfMCGOnModuliMonodromy} of the Artin half-twists
  $b_i \in \mathrm{Br}_n(\Sigma^2_{g,b}) \xrightarrow{\;}
  \mathrm{MCG}(\Sigma^2_{g,b,n})$ \eqref{ExtensionByBraidGroup},
  represented by diffeomorphisms supported in this disk, is via the
  \emph{dual integral standard representation} $\mathrm{st}^{\vee}$
  \eqref{DualIntegralStandardRepresentation} on the
  puncture factor
  $$
    \mathbb{Z}^{n-1}
    \;\simeq\;
    A_{n-1}^{\vee}
    \;\simeq\;
    \mathbb{Z}^n/\mathbb{Z}(1,\cdots,1)
  $$
  and the identity on the unpunctured-surface factor.
  In particular, this action of the disk braid group factors through
  $\mathrm{Sym}_n$.
\end{theorem}
\begin{proof}
See \hyperlink{FigureS}{\it Fig. S} for an illustration of the following analysis.

  We may, without restriction, assume the punctures to jointly sit within an open disk inside the surface, $\{v_1, \cdots v_n\} \subset D^2 \subset \Sigma^2_{g,b}$. Then the one-point compactification $(\Sigma^2_{g,b,n})_{\cpt}$ may be obtained by enlarging each puncture to a little missing open disk, erecting a little cone (horn) over the boundary of this disk, and making these cones bend over to make (just) their tips touch -- this joint tip is the point $\infty$.
  Let then 
  \begin{equation}
    \label{DashedLoops}
    \ell_i \;\in\;
    \pi_0
    \,
    \mathrm{Map}^\ast\big(
      S^1
      ,\,
      (\Sigma^2_{g,b,n})_{\cpt}
    \big)
    \;\simeq\;
    \pi_1   
    \big(
      (\Sigma^2_{g,b,n})_{\cpt}
    \big)
     \,,
     \;\;\;\;\;
     i \in \{1, \cdots, n-1\}
  \end{equation}
  denote the homotopy class of a loop that starts at $\infty$, runs down through the $i$th cone and back through the $n$th cone

  \vspace{-3mm}
  $$
    \ell_i \;:=\;
  \adjustbox{
    raise=-.3cm
  }{
  \begin{tikzpicture}
    \draw[
      -Latex, 
      rounded corners]
      (0,0) -- 
        node[scale=.7, xshift=-6pt] {$i$}
      (-.4,-.8) --
      (+.4, -.8) -- 
        node[scale=.7, xshift=6pt] {$n$}
      (0,0);
    \node[scale=.7] at (0,.1) 
      { $\infty$ };
  \end{tikzpicture}
  }
  .
  $$
  \vspace{2mm} 

\noindent  This is also illustrated by the dashed arrows in the top panel of \hyperlink{FigureS}{Fig. S}, which, for graphical convenience, shows not the cones themselves, but the arcs whose contraction to $\{\infty\}$ produces the cones, according to the proof of Lem. \ref{HomotopyTypeOfCompactifiedPuncturedSurface}.

  Now consider a based map $f:(\Sigma^2_{g,b,n})_{\cpt} \to \Omega S^2$,
  representing a class
  $$  
    \def\arraystretch{1.6}
      [f] \;\in\;
      \pi_1
      \,
      \mathrm{Map}^\ast\big(
        (\Sigma^2_{g,b,n})_{\cpt}
        ,\,
        S^2
      \big)
      \underset{
        \mathclap{
          \scalebox{.7}{ \eqref{FundamentalGroupOfPointedMappingSpace}
          }
        }
      }
      {
        \;\simeq\;
      }
      \pi_0
      \, 
      \mathrm{Map}^\ast\big(
        (\Sigma^2_{g,b,n})_{\cpt}
        ,\,
        \Omega S^2
      \big)
  $$
  with its homotopy class decomposed, according to Prop. \ref{FluxMonodromyThroughPuncturedSurface}, as
  \begin{equation}
    \label{DecompositionOfMap}
    \begin{tikzcd}[sep=0pt]
    \pi_0 \, 
    \mathrm{Map}^\ast\big(
      (\Sigma^2_{g,b,n})_{\cpt}
      ,\,
      \Omega S^2
    \big)
    \ar[
      rr,
      phantom,
      "{ \simeq }"
    ]
    &&
    \pi_0 
    \, \mathrm{Map}^\ast\big(
      \Sigma^2_{g,b}
      ,\,
      \Omega S^2
    \big)
    \times
    \prod_{n-1}
    \pi_0 
    \,
    \mathrm{Map}^\ast\big(
      S^1
      ,\,
      \Omega S^2
    \big)
    \\
    {[f]} 
      &\longmapsto&
    \big(
      [\widetilde f]
      ,
      (q_1, \cdots, q_{n-1})
    \big)
    \mathrlap{\,,}
    \end{tikzcd}
  \end{equation}
  where the 
  integer classes $q_i := [f_\ast \ell_i] \in \mathbb{Z}$
  come from the restriction of $f$ to these loops  $\ell_i$ \eqref{DashedLoops}:
  \begin{equation}
    \label{PushforwardOfLoops}
    \begin{tikzcd}[
      sep=0pt
    ]
      \pi_0
      \,
      \mathrm{Map}^\ast\big(
        S^1
        ,\,
        (\Sigma^2_{g,b,n})_{\cpt}
      \big)
      \ar[
        rr,
        "{ f_\ast }"
      ]
      &&
      \pi_0
      \,
      \mathrm{Map}^\ast\big(
        S^1
        ,\,
        \Omega S^2
      \big)
        \,\simeq\,
      \mathbb{Z}
      \\
      \ell_i  
        &\longmapsto&
      q_i
      \mathrlap{\,.}
    \end{tikzcd}
  \end{equation}

  We need to determine the effect on these components of precomposition with a  diffeomorphism of $ \Sigma^2_{g,b,n}$ representing the mapping class of the $i$th Artin generator $b_i$ \eqref{ArtinPresentationOfPlainBraids}. This diffeo may be chosen such that its unique continuous extension to the one-point compactification, 
  $$
    b_i \,\in\,
    \pi_0
    \,
    \mathrm{Map}^\ast\big(
      (\Sigma^2_{g,b,n})_{\cpt}
      ,\,
      (\Sigma^2_{g,b,n})_{\cpt}
    \big)
    \,,
  $$
  restricts for each $j \in \{1, \cdots, n\}$ to a homeomorphism from the $j$th cone onto the $\sigma(j)$th cone, where $\sigma$ is the permutation underlying the Artin generator.
  Then direct inspection (illustrated in \hyperlink{FigureS}{Fig. S}\,) shows

  \begin{itemize} 
  \item  for $i < n-1$ that
  \begin{equation}
    \label{PushingLoopsAlongLowerArtinGenerator}
    \begin{tikzcd}[
      ampersand replacement=\&,
      sep = 0pt
    ]
      \pi_0
      \,
      \mathrm{Map}^\ast\big(
        S^1 
        ,\,
        (\Sigma^2_{g,b,n})_{\cpt}
      \big)
      \ar[
        rr, 
        "{ {b_i}_\ast }" 
      ]
      \&\&
      \pi_0
      \,
      \mathrm{Map}^\ast\big(
        S^1 
        ,\,
        (\Sigma^2_{g,b,n})_{\cpt}
      \big)
      \ar[
        rr, 
        "{ f_\ast }"
      ]
      \&\&
      \pi_0
      \,
      \mathrm{Map}^\ast\big(
        S^1 
        ,\,
        \Omega S^2
      \big)
      \\
      \ell_j
        \&
        \underset{
          \mathclap{
            \scalebox{.7}{\def\arraystretch{.9}\begin{tabular}{c}
             left of 
             \\
             \hyperlink{FigureS}{Fig. S}
             \end{tabular}
            }
          }
        }{
          \longmapsto
        }
        \&
      \left\{\!
      \def\arraycolsep{2pt}
      \begin{array}{ccl}
        \ell_{i + 1} &\vert& j = i
        \\
        \ell_i &\vert& j = i + 1
        \\
        \ell_j &\vert& \mathrm{otherwise}
      \end{array}
    \!  \right\}
      \&
      \;\;
        \underset{
          \mathclap{
            \scalebox{.7}{
              \eqref{PushforwardOfLoops}
            }
          }
        }{
          \longmapsto
        }
      \&
\quad       \left\{\!
      \def\arraycolsep{2pt}
      \begin{array}{ccl}
        q_{i + 1} &\vert& j = i
        \\
        q_i &\vert& j = i + 1
        \\
        q_j &\vert& \mathrm{otherwise}
      \end{array}
     \! \right\}
      \,,
    \end{tikzcd}
  \end{equation}

\item  for $i = n-1$ that
  \begin{equation}
    \label{PushingLoopsAlongTopArtinGenerator}
    \hspace{-6mm} 
    \begin{tikzcd}[
      ampersand replacement=\&,
      sep = 0pt
    ]
      \pi_0
      \,
      \mathrm{Map}^\ast\big(
        S^1 
        ,\,
        (\Sigma^2_{g,b,n})_{\cpt}
      \big)
      \ar[
        rr, 
        "{ {b_i}_\ast }" 
      ]
      \&\&
      \pi_0
      \,
      \mathrm{Map}^\ast\big(
        S^1 
        ,\,
        (\Sigma^2_{g,b,n})_{\cpt}
      \big)
      \ar[
        rr, 
        "{ f_\ast }"
      ]
      \&\&
      \pi_0
      \,
      \mathrm{Map}^\ast\big(
        S^1 
        ,\,
        \Omega S^2
      \big)
      \\
      \ell_j
        \&
        \underset{
          \mathclap{
            \scalebox{.7}{\def\arraystretch{.9}\begin{tabular}{c}
             right of 
             \\
             \hyperlink{FigureS}{Fig. S}
             \end{tabular}
            }
          }
        }{
          \longmapsto
        }
        \&
      \left\{\!
      \def\arraycolsep{2pt}
      \def\arraystretch{1.4}
      \begin{array}{lcl}
        \ell_{n-1}^{-1} \circ \ell_{j}  &\vert& j < n-1
        \\
        \ell_{n-1}^{-1} &\vert& j = n-1
      \end{array}
     \! \right\}
      \&
      \quad 
        \underset{
          \mathclap{
            \scalebox{.7}{
              \eqref{PushforwardOfLoops}
            }
          }
        }{
          \longmapsto
        }
      \&
\quad       \left\{\!
      \def\arraycolsep{2pt}
      \def\arraystretch{1.4}
      \begin{array}{rcl}
        q_j - q_{n-1} &\vert& j < n-1
        \\
        -q_{n-1} &\vert& j = n-1
      \end{array}
     \! \right\}
      \,.
    \end{tikzcd}
  \end{equation}
\end{itemize}
  These formulas describe the action on the integer values $q_j$,
  hence on functionals on the puncture-loop lattice. More precisely,
  identify the abelianized puncture loops $[\ell_j]$ with the basis
  $e_j=v_j-v_n$ of $A_{n-1}$. A tuple $(q_1,\cdots,q_{n-1})$
  specifies the functional $\varphi\in A_{n-1}^{\vee}$ with
  $\varphi(e_j)=q_j$.
  The transformations
  \eqref{PushingLoopsAlongLowerArtinGenerator} and
  \eqref{PushingLoopsAlongTopArtinGenerator} are exactly the dual action
  $\mathrm{st}^{\vee}$ \eqref{DualIntegralStandardRepresentation}:
  for an adjacent transposition, its inverse is itself, so by
  \eqref{IntegralStandardRepresentation} we have
  $\varphi\circ\mathrm{st}(b_i)=\mathrm{st}^{\vee}(b_i)\varphi$.
  This identifies the puncture contribution with
  $A_{n-1}^{\vee}\simeq\mathbb{Z}^n/\mathbb{Z}(1,\cdots,1)$.

  To see also that the first component is unchanged, let
  $P:=\{v_1,\cdots,v_n\}$ and write
  $$
    \pi:\Sigma^2_{g,b}
    \;\longrightarrow\;
    \Sigma^2_{g,b}/P
    \;\simeq\;
    (\Sigma^2_{g,b,n})_{\cpt}
  $$
  for the quotient map. Under the decomposition
  \eqref{DecompositionOfMap}, the unpunctured-surface component is
  $[\widetilde f]=[f\circ\pi]$, with $v_n$ as basepoint on
  $\Sigma^2_{g,b}$.
  If $h_i$ represents $b_i$ on $\Sigma^2_{g,b}$ and $\bar h_i$ is
  its induced map on the quotient, then
  $\bar h_i\circ\pi=\pi\circ h_i$.
  After forgetting the punctures, $h_i$ is isotopic to the identity:
  its mapping class is the trivial image under the composite
  $$
    \begin{tikzcd}[
      column sep = 3pt,
    ]
      \mathrm{Br}_n\big(
        \Sigma^2_{g,b}
      \big)
      \ar[rr, shorten=-2pt]
      &&
      \mathrm{MCG}\big(
        \Sigma^2_{g,b,n}
      \big)
      \ar[rr, shorten=-2pt]
      &&
      \mathrm{MCG}\big(
        \Sigma^2_{g,b}
      \big)
      \,.
    \end{tikzcd}
  $$
  This is also the exactness assertion in the generalized Birman sequence
  \eqref{BirmanSequence}.
  Consequently, $f\circ\bar h_i\circ\pi$ and $f\circ\pi$ are freely
  homotopic. Since $\Omega S^2$ is a connected $H$-space, the action of
  its fundamental group on based homotopy classes is trivial, so their
  based homotopy classes agree as well
  \cite[Prop. 4A.2 \& Ex. 4A.3]{Hatcher02}.

  Altogether, using $\mathrm{st}^{\vee}$
  \eqref{DualIntegralStandardRepresentation}, the action on the components
  \eqref{DecompositionOfMap} is
  $$
    {b_i}_\ast
    \;:\;
    \Big(
      [\widetilde{f}]
      ,\,
      (q_1, \cdots, q_{n-1})
    \Big)
    \;\;\longmapsto\;\;
    \Big(
      [\widetilde{f}]
      ,\,
      \mathrm{st}^{\vee}(b_i)
      (q_1, \cdots, q_{n-1})
    \Big)
    \,.
  $$
  The calculation concerns the Artin half-twists in the chosen disk;
  additional surface-braid generators, such as motions around handles,
  are not included in this assertion.
\end{proof}

\begin{center}
\hypertarget{FigureS}{}

\begin{tikzpicture}[scale=0.85]

\draw[ 
  line width=2,
  -Latex, 
  black!70
]
  (1,0) -- (3.5,-5.2);
\draw[ 
  line width=2,
  -Latex,
  black!70
]
  (-1,0) -- (-3.5,-5.2);

\node[scale=.9] at (-4, -4.8)
 {\eqref{PushingLoopsAlongLowerArtinGenerator}};

\node[scale=.9] at (+4, -4.8)
 {\eqref{PushingLoopsAlongTopArtinGenerator}};


\begin{scope}[shift={(-7.5,0)}]

\def\xOne{8}
\def\yOne{1.4}

\def\xTwo{8.5-3.8}
\def\yTwo{1.4-.8}

\def\xThree{8.5+1.2}
\def\yThree{1.4-.9}

 \draw[draw opacity=0,fill=white]
   (3.7,-.2) rectangle
   (11, 2.7);
 \clip
   (3.7,-.2) rectangle
   (11, 2.7);

\draw[
  draw opacity=0,
  fill=gray!30
]
  (0,0) --
  (10,0) --
  (10+6-.8,1.8) --
  (0+6+.8,1.8) -- cycle;

 \draw[
   draw opacity=0,
   fill=white
 ]
   (\xOne,\yOne) ellipse
   (.23 and .04);

 \draw[
   draw opacity=0,
   fill=white
 ]
   (\xTwo,\yTwo) ellipse
   (.23 and .04);

 \draw[
   draw opacity=0,
   fill=white
 ]
   (\xThree,\yThree) ellipse
   (.23 and .04);

\draw[
  line width=3
]
  (\xOne, \yOne) .. controls
  (\xOne, \yOne+1.7) and 
  (\xTwo,\yTwo+1.7) ..
  (\xTwo, \yTwo);

\draw[
  line width=3
]
  (\xOne, \yOne) .. controls
  (\xOne, \yOne+1.7) and 
  (\xThree,\yThree+1.7) ..
  (\xThree, \yThree);

\draw[
  densely dashed,
  black!40,
  line width=1.2,
  -Latex,
]
  (\xTwo, \yTwo) --
  node[
    black, 
    yshift=0pt
  ] 
    {\colorbox{gray!30}{$\,\mathclap{\ell_1}\,$}}
  (\xOne,\yOne);

\draw[
  densely dashed,
  black!40,
  line width=1.2,
  -Latex,
]
  (\xThree, \yThree) --
  node[
    black, 
    yshift=0pt
  ] 
    {\colorbox{gray!30}{$\,\mathclap{\ell_2}\,$}}
  (\xOne,\yOne);

\draw[
  densely dashed,
  black!40,
  line width=1.2,
  -Latex,
]
  (\xTwo, \yTwo) --
  node[
    black, 
    yshift=0pt,
    pos=.61
  ] 
    {\colorbox{gray!30}{$\ell_2^{\mathrlap{{}^{-1}}} \circ \ell_1$}}
  (\xThree,\yThree);

 \draw[line width=2, gray]
   (3.7,-.2) rectangle
   (11, 2.7);

\end{scope}


\begin{scope}[shift={(-3,-4)}]

\def\xOne{8}
\def\yOne{1.4}

\def\xTwo{8.5-3.8}
\def\yTwo{1.4-.8}

\def\xThree{8.5+1.2}
\def\yThree{1.4-.9}

 \draw[draw opacity=0,fill=white]
   (3.7,-.2) rectangle
   (11, 2.7);
\clip
   (3.7,-.2) rectangle
   (11, 2.7);

\draw[
  draw opacity=0,
  fill=gray!30
]
  (0,0) --
  (10,0) --
  (10+6-.8,1.8) --
  (0+6+.8,1.8) -- cycle;

 \draw[
   draw opacity=0,
   fill=white
 ]
   (\xOne,\yOne) ellipse
   (.23 and .04);

 \draw[
   draw opacity=0,
   fill=white
 ]
   (\xTwo,\yTwo) ellipse
   (.23 and .04);

 \draw[
   draw opacity=0,
   fill=white
 ]
   (\xThree,\yThree) ellipse
   (.23 and .04);

\draw[
  dotted,
  line width=1,
  -Latex
]
  (\xOne-.2,\yOne-.07) .. controls
  (\xOne-.4, \yOne-.07) and
  (\xThree-3, \yThree-.3) ..
  (\xThree-.2, \yThree);

\draw[
  dotted,
  line width=1,
  -Latex
]
  (\xThree+.2,\yThree+.1) .. controls
  (\xThree+.4, \yThree+.1) and
  (\xOne+2.7, \yOne+.3) ..
  (\xOne+.2, \yOne);

 \draw[line width=2, gray]
   (3.7,-.2) rectangle
   (11, 2.7);

\node[scale=.9] 
  at (\xOne, \yOne)
  {$\mathbf{v}_3$};

\node[scale=.9] 
  at (\xTwo, \yTwo)
  {$\mathbf{v}_1$};

\node[scale=.9] 
  at (\xThree, \yThree)
  {$\mathbf{v}_2$};

\node at 
 ({\xOne+.5*(\xThree-\xOne)},
  {\yOne+.5*(\yThree-\yOne)}
 ) {$b_2$};

\end{scope}


\begin{scope}[shift={(-11.5,-4)}]

\def\xOne{8}
\def\yOne{1.4}

\def\xTwo{8.5-3.8}
\def\yTwo{1.4-.8}

\def\xThree{8.5+1.2}
\def\yThree{1.4-.9}

 \draw[draw opacity=0,fill=white]
   (3.7,-.2) rectangle
   (11, 2.7);
\clip
   (3.7,-.2) rectangle
   (11, 2.7);

\draw[
  draw opacity=0,
  fill=gray!30
]
  (0,0) --
  (10,0) --
  (10+6-.8,1.8) --
  (0+6+.8,1.8) -- cycle;

 \draw[
   draw opacity=0,
   fill=white
 ]
   (\xOne,\yOne) ellipse
   (.23 and .04);

 \draw[
   draw opacity=0,
   fill=white
 ]
   (\xTwo,\yTwo) ellipse
   (.23 and .04);

 \draw[
   draw opacity=0,
   fill=white
 ]
   (\xThree,\yThree) ellipse
   (.23 and .04);

\draw[
  dotted,
  line width=1,
  -Latex
]
  (\xTwo+.2,\yTwo-.07) .. controls
  (\xTwo+.2, \yTwo-.4) and
  (\xThree-2, \yThree-.5) ..
  (\xThree-.2, \yThree-.05);

\draw[
  dotted,
  line width=1,
  -Latex
]
  (\xThree-.1,\yThree+.2) .. controls
  (\xThree-.2, \yThree+.6) and
  (\xTwo+1, \yTwo+.6) ..
  (\xTwo+.2, \yTwo+.05);

\node at (7.3,.6) 
  { $b_1$ };

 \draw[line width=2, gray]
   (3.7,-.2) rectangle
   (11, 2.7);

\node[scale=.9] 
  at (\xOne, \yOne)
  {$\mathbf{v}_3$};

\node[scale=.9] 
  at (\xTwo, \yTwo)
  {$\mathbf{v}_1$};

\node[scale=.9] 
  at (\xThree, \yThree)
  {$\mathbf{v}_2$};

\end{scope}


\begin{scope}[shift={(-11.5,-8)}]

\def\xOne{8}
\def\yOne{1.4}

\def\xTwo{8.5-3.8}
\def\yTwo{1.4-.8}

\def\xThree{8.5+1.2}
\def\yThree{1.4-.9}

 \draw[draw opacity=0,fill=white]
   (3.7,-.2) rectangle
   (11, 2.7);
 \clip
   (3.7,-.2) rectangle
   (11, 2.7);

\draw[
  draw opacity=0,
  fill=gray!30
]
  (0,0) --
  (10,0) --
  (10+6-.8,1.8) --
  (0+6+.8,1.8) -- cycle;

 \draw[
   draw opacity=0,
   fill=white
 ]
   (\xOne,\yOne) ellipse
   (.23 and .04);

 \draw[
   draw opacity=0,
   fill=white
 ]
   (\xTwo,\yTwo) ellipse
   (.23 and .04);

 \draw[
   draw opacity=0,
   fill=white
 ]
   (\xThree,\yThree) ellipse
   (.23 and .04);

\draw[
  line width=3
]
  (\xOne, \yOne) .. controls
  (\xOne, \yOne+1.7) and 
  (\xTwo,\yTwo+1.7) ..
  (\xTwo, \yTwo);

\draw[
  line width=3
]
  (\xOne, \yOne) .. controls
  (\xOne, \yOne+1.7) and 
  (\xThree,\yThree+1.7) ..
  (\xThree, \yThree);

\draw[
  densely dashed,
  black!40,
  line width=1.2,
  -Latex,
]
  (\xTwo, \yTwo) --
  node[
    black, 
    yshift=0pt
  ] 
    {\colorbox{gray!30}{$\,\mathclap{\ell_2}\,$}}
  (\xOne,\yOne);

\draw[
  densely dashed,
  black!40,
  line width=1.2,
  -Latex,
]
  (\xThree, \yThree) --
  node[
    black, 
    yshift=0pt
  ] 
    {\colorbox{gray!30}{$\,\mathclap{\ell_1}\,$}}
  (\xOne,\yOne);

 \draw[line width=2, gray]
   (3.7,-.2) rectangle
   (11, 2.7);

\end{scope}


\begin{scope}[shift={(-3,-8)}]

\def\xThree{8}
\def\yThree{1.4}

\def\xTwo{8.5-3.8}
\def\yTwo{1.4-.8}

\def\xOne{8.5+1.2}
\def\yOne{1.4-.9}

 \draw[draw opacity=0,fill=white]
   (3.7,-.2) rectangle
   (11, 2.7);
 \clip
   (3.7,-.2) rectangle
   (11, 2.7);

\draw[
  draw opacity=0,
  fill=gray!30
]
  (0,0) --
  (10,0) --
  (10+6-.8,1.8) --
  (0+6+.8,1.8) -- cycle;

 \draw[
   draw opacity=0,
   fill=white
 ]
   (\xOne,\yOne) ellipse
   (.23 and .04);

 \draw[
   draw opacity=0,
   fill=white
 ]
   (\xTwo,\yTwo) ellipse
   (.23 and .04);

 \draw[
   draw opacity=0,
   fill=white
 ]
   (\xThree,\yThree) ellipse
   (.23 and .04);

\draw[
  line width=3
]
  (\xOne, \yOne) .. controls
  (\xOne, \yOne+1.7) and 
  (\xThree,\yThree+1.7) ..
  (\xThree, \yThree);

\begin{scope}
\clip 
  (\xThree-.13,\yThree+.15)
  rectangle
  (\xThree+.2,\yThree+.3);

\draw[
  line width=6.5,
  gray!30
]
  (\xOne, \yOne) .. controls
  (\xOne, \yOne+1.7) and 
  (\xTwo,\yTwo+1.7) ..
  (\xTwo, \yTwo);
\end{scope}

\begin{scope}
\clip 
  (\xThree-.13,\yThree+.3)
  rectangle
  (\xThree+.2,\yThree+.6);

\draw[
  line width=6.5,
  white
]
  (\xOne, \yOne) .. controls
  (\xOne, \yOne+1.7) and 
  (\xTwo,\yTwo+1.7) ..
  (\xTwo, \yTwo);

\end{scope}

\draw[
  line width=3
]
  (\xOne, \yOne) .. controls
  (\xOne, \yOne+1.7) and 
  (\xTwo,\yTwo+1.7) ..
  (\xTwo, \yTwo);

\draw[
  densely dashed,
  black!40,
  line width=1.2,
  -Latex,
]
  (\xThree, \yThree) --
  node[
    black, 
    shift={(3pt,3pt)}
  ] 
    {$\,\mathclap{\ell_2}\,$}
  (\xOne,\yOne);

\draw[
  densely dashed,
  black!40,
  line width=1.2,
  -Latex,
]
  (\xOne-.14, \yOne-.1) --
  node[
    black, 
    shift={(-5pt,-5pt)},
    pos=.66
  ] 
    {$\,\mathclap{\ell_2^{\mathrlap{{}^{-1}}}}\,$}
  (\xThree-.14,\yThree-.1);

\draw[
  densely dashed,
  black!40,
  line width=1.2,
  -Latex,
]
  (\xTwo, \yTwo) --
  node[
    black, 
    yshift=0pt
  ] 
    {\colorbox{gray!30}{$\ell_2^{\mathrlap{{}^{-1}}}\circ \ell_1$}}
  (\xThree,\yThree);

 \draw[line width=2, gray]
   (3.7,-.2) rectangle
   (11, 2.7);

\end{scope}

\end{tikzpicture}

\vspace{.4cm}

\begin{minipage}{13.5cm}
  \footnotesize
  {\bf Figure S -- Effect of braiding on  compactification of punctured surface} --
  proof of Thm. \ref{BraidGroupActionOnFluxMonodromyOverPuncturedSurface}.
  Using that the homotopy type of the 1-point compactification $(\Sigma^{2}_{g,b,n})_{\cpt}$ of a punctured surface is obtained by attaching, for all $j < n$,  an arc from the $j$th to the $n$th puncture (cf. the proof of Lem. \ref{HomotopyTypeOfCompactifiedPuncturedSurface}),  the above graphic shows the effect of the Artin generator \eqref{ArtinGeneratorsForSymmetricGroup}
  mapping classes $b_i \in \mathrm{Br}_n(\Sigma^2_{g,b}) \xrightarrow{\;} \mathrm{MCG}\big(\Sigma^2_{g,b,n}\big)$ on these arcs --- and on the indicated generators $\ell_j \in \pi_1\big((\Sigma^2_{g,b,n})_{\cpt},\infty\big)$ obtained after contracting the arcs to $\{\infty\}$ --- making manifest that the induced action on the integer evaluations $q_j=[f_\ast\ell_j]$ is the dual integral standard representation $\mathrm{st}^{\vee}$ \eqref{DualIntegralStandardRepresentation} of the symmetric group $\mathrm{Sym}_n$.
\end{minipage}

\end{center}

\smallskip 
\subsection{On punctured disks}
\label{OnPuncturedDisks}

We now apply Thm. \ref{BraidGroupActionOnFluxMonodromyOverPuncturedSurface}
to punctured disks and spheres. To describe the resulting groups, we
distinguish the sum-zero framing lattice $A_{n-1}$ from the quotient lattice
$$
  Q_{n-1}
  \;:=\;
  A_{n-1}^{\vee}
  \;\simeq\;
  \mathbb{Z}^n/\mathbb{Z}(1,\cdots,1)
  \,.
$$
Their permutation actions are $\mathrm{st}$
\eqref{IntegralStandardRepresentation} and $\mathrm{st}^{\vee}$
\eqref{DualIntegralStandardRepresentation}, respectively.
It is the latter lattice that occurs in the flux monodromy.

\begin{definition}[{\bf Framed braids and algebraic framing groups}
  {\cite{MelvinTufillaro91}\cite{KoSmolinsky92}, cf.
  \cite[\S 3.2]{Kokkinakis25}}]
  For a surface $\Sigma^2$, define the \emph{algebraic framing group}
  by the permutation action of
  $\mathrm{Br}_n(\Sigma^2)\twoheadrightarrow\mathrm{Sym}_n$
  \eqref{SurfaceBraidGroup}:
  \begin{equation}
    \label{FramedBraidGroup}
    \mathrm{FBr}^{\mathrm{alg}}_n(\Sigma^2)
    \;:=\;
    \mathbb{Z} \wr \mathrm{Br}_n(\Sigma^2)
    \;\simeq\;
    \mathbb{Z}^n
      \rtimes_{\mathrm{def}}
    \mathrm{Br}_n(\Sigma^2)
    \mathrlap{\,.}
  \end{equation}
  On the disk, a choice of trivialization of the tangent bundle identifies
  this with the usual \emph{framed} or \emph{ribbon braid group}, denoted
  $$
    \mathrm{FBr}_n
    \;:=\;
    \mathrm{FBr}^{\mathrm{alg}}_n(D^2)
    \;\simeq\;
    \mathbb{Z}^n\rtimes_{\mathrm{def}}\mathrm{Br}_n
    \,.
  $$
  Each integer records a full twist of one ribbon, as in
  \hyperlink{FigureFL}{\it Fig. FL}; the closure is a framed link.

  For a general closed surface, \eqref{FramedBraidGroup} is an algebraic
  construction. Geometric framed surface braids are instead defined using
  configurations of nonzero tangent vectors, and the resulting framing
  extension need not split \cite[Def. 1 \& Thm. 5]{BellingeriGervais12}.
  Already for a single strand on the sphere, the geometric group is
  $\pi_1(US^2)\simeq\pi_1(\mathrm{SO}(3))\simeq\mathbb{Z}_2$,
  where $US^2$ denotes the unit tangent bundle, whereas
  $\mathrm{FBr}^{\mathrm{alg}}_1(S^2)\simeq\mathbb{Z}$.
\end{definition}

\smallskip

The quotient framing lattice and the sum-zero lattice enter two different
exact sequences:

\begin{lemma}[\bf Framing quotients and sum-zero subgroups]
  \label{FramingQuotientsAndSumZeroSubgroups}
  For $n\geq 1$, let $H$ be a group equipped with a homomorphism
  $H\to\mathrm{Sym}_n$, and put $\mathbf{1}:=(1,\cdots,1)\in\mathbb{Z}^n$.
  Then we have a central exact sequence
  \begin{equation}
    \label{FramingQuotientExactSequence}
    \begin{tikzcd}[row sep=2pt, column sep=small]
      1
      \ar[r]
      &
      \mathbb{Z}
      \ar[r, hook, "{\delta}"]
      &
      \mathbb{Z}^n\rtimes_{\mathrm{def}}H
      \ar[r, ->>, "{p}"]
      &
      Q_{n-1}\rtimes_{\mathrm{st}^{\vee}}H
      \ar[r]
      &
      1
      \\
      &
      k
      \ar[r, mapsto]
      &
      (k\mathbf{1},1_H)
      \\
      &&
      (a,h)
      \ar[r, mapsto]
      &
      ([a],h)
      \mathrlap{\,.}
    \end{tikzcd}
  \end{equation}
  Thus its quotient identifies framing tuples that differ by a common
  integer shift. There is also an exact sequence of $\mathrm{Sym}_n$-modules
  (with trivial action on the last term)
  \begin{equation}
    \label{SumZeroLatticeInsideFramingQuotient}
    \begin{tikzcd}[row sep=2pt, column sep=small]
      0
      \ar[r]
      &
      A_{n-1}
      \ar[r, hook, "{\jmath}"]
      &
      Q_{n-1}
      \ar[r, ->>, "{\overline\Sigma}"]
      &
      \mathbb{Z}/n\mathbb{Z}
      \ar[r]
      &
      0
      \\
      &
      a
      \ar[r, mapsto]
      &
      {[a]}
      \\
      &&
      {[b]}
      \ar[r, mapsto]
      &
      \displaystyle\sum_{i=1}^n b_i\;\mathrm{mod}\;n
      \mathrlap{\,.}
    \end{tikzcd}
  \end{equation}
  Here $\jmath$ intertwines $\mathrm{st}$
  \eqref{IntegralStandardRepresentation} with $\mathrm{st}^{\vee}$
  \eqref{DualIntegralStandardRepresentation}.

  For comparison, the sum-zero lattice itself gives the following subgroup
  inclusions, using $\mathrm{st}$ \eqref{IntegralStandardRepresentation}:
  \begin{equation}
    \label{FramedBraidsOfVanishingTotalFraming}
    \hspace{-1.2cm} 
    \begin{tikzcd}[row sep=0pt, column sep=5pt]
      \mathclap{
      \scalebox{.7}{
        \color{darkblue}
        \bf
        \def\arraystretch{.9}
        \begin{tabular}{c}
          with vanishing
          \\
          total framing
        \end{tabular}
      }
      }
      &&
      \mathclap{
      \scalebox{.7}{
        \color{darkblue}
        \bf
        \def\arraystretch{.9}
        \begin{tabular}{c}
          with total framing
          \\
          divisible by 
          \scalebox{1.2}{$n$}
        \end{tabular}
      }
      }
      &&
      \mathclap{
      \scalebox{.7}{
        \color{darkblue}
        \bf
        \def\arraystretch{.9}
        \begin{tabular}{c}
          with arbitrary
          \\
          total framing
        \end{tabular}
      }
      }
      \\
      \mathllap{
        \smash{
        \adjustbox{
          scale=.66,
          rotate=40,
          raise=-10pt
        }{
          \color{darkblue}
          \bf
          \def\arraystretch{.9}
          \begin{tabular}{c}
          algebraically
          \\
          framed braids
          \end{tabular}
        }
        }
      }
      \mathbb{Z}^{n-1}
      \underset{
        \mathclap{
          \adjustbox{
            scale=.7,
            raise=4pt
          }{
            \color{gray}\rm
            standard rep
          }
        }
      }{
        \rtimes_{{}_{\mathrm{st}}} 
      }
      \mathrm{Br}_n(\Sigma^2)
      \ar[
        rr,
        hook
      ]
      &&
      \mathbb{Z}
      \times
      \big(
      \mathbb{Z}^{n-1}
      \underset{
        \mathclap{
          \adjustbox{
            scale=.7,
            raise=4pt
          }{
            \color{gray}\rm
            standard rep
          }
        }
      }{
        \rtimes_{{}_{\mathrm{st}}} 
      }
      \mathrm{Br}_n(\Sigma^2)
      \big)
      \ar[
        rr,
        hook
      ]
      &&
      \mathbb{Z}^{n}
      \underset{
        \mathclap{
          \adjustbox{
            scale=.7,
            raise=4pt
          }{
            \color{gray}\rm
            defining rep
          }
        }
      }{
        \rtimes_{{}_{\mathrm{def}}}
      }
      \mathrm{Br}_n(\Sigma^2)
      \mathrlap{
      \;\defneq\,
      \mathrm{FBr}^{\mathrm{alg}}_n(\Sigma^2)
      }
      \\[-4pt]
      e_i
      &\longmapsto&
      e_i
      &\longmapsto&
      v_i
        - 
      v_n 
      \\
      && 
      t 
        &\longmapsto&
      v_1 + \cdots + v_n
      \\
      b_i &\longmapsto& b_i
      &\longmapsto& b_i
    \end{tikzcd}
  \end{equation}
  These inclusions concern $A_{n-1}$; the quotient in
  \eqref{FramingQuotientExactSequence} instead uses $Q_{n-1}$ and
  $\mathrm{st}^{\vee}$ \eqref{DualIntegralStandardRepresentation}.
\end{lemma}
\begin{proof}
  Since $H$ permutes coordinates, it fixes $\mathbf{1}$, so
  $\delta(\mathbb{Z})$ is central. The map $p$ is a surjective homomorphism
  with precisely this kernel, proving \eqref{FramingQuotientExactSequence}.

  In \eqref{SumZeroLatticeInsideFramingQuotient}, the map $\jmath$ is
  injective since
  $A_{n-1}\cap\mathbb{Z}\mathbf{1}=0$.
  The map $\overline\Sigma$ is well-defined because replacing $a$ by
  $a+k\mathbf{1}$ changes its coordinate sum by $nk$.
  It is surjective, and if $\sum_i a_i=nk$, then
  $a-k\mathbf{1}\in A_{n-1}$ represents the same class in $Q_{n-1}$.
  This proves exactness. Both maps commute with coordinate permutations.

  Finally, the inclusions \eqref{FramedBraidsOfVanishingTotalFraming}
  follow from \eqref{DefiningRepDecompositionOverZ}: the first lattice is
  $A_{n-1}$, and the image of $\mathbb{Z}\oplus A_{n-1}$ in
  $\mathbb{Z}^n$ consists of the tuples whose sum is divisible by $n$.
\end{proof}

\begin{proposition}[{\bf 2-Cohomotopical covariant flux monodromy on punctured disks}]
  \label{CovariantFluxMonodromyOnPuncturedDisks}
  Equip the punctured sphere or disk with the spin structure extending
  across the punctures from the sphere or disk. For $n\geq 1$, put
  $H_n:=\mathrm{MCG}(\Sigma^2_{0,0,n})$.
  The 2-Cohomotopical covariant flux monodromy is as follows:
  \begin{itemize}
  \item[{\bf (i)}] On the $n$-punctured sphere, it is the product of the
  unpunctured-sphere contribution with the semidirect product for
  $\mathrm{st}^{\vee}$ \eqref{DualIntegralStandardRepresentation}:
  \begin{equation}
    \label{MonodromyOverNPuncturedSphere}
    \begin{aligned}
      G_n^{\mathrm{sph}}
      \;&:=\;
      \pi_1\Big(
        \mathrm{Map}^\ast_0\big(
          (\Sigma^2_{0,0,n})_{\cpt}
          ,\,
          S^2
        \big)
        \sslash
        \mathrm{Diff}^{+,\partial}\big(
          \Sigma^2_{0,0,n}
        \big)
      \Big)
      \\
      &\;\simeq\;
      \mathbb{Z}\times
      \big(
        Q_{n-1}\rtimes_{\mathrm{st}^{\vee}}H_n
      \big)
      \mathrlap{\,.}
    \end{aligned}
  \end{equation}
  The first $\mathbb{Z}$ is the factor $\pi_3(S^2)$.

  \item[{\bf (ii)}] On the $n$-punctured closed disk, it is the quotient of
  the usual framed braid group by the simultaneous unit framing shift:
  \begin{equation}
    \label{FluxMonodromyOfPuncturedDisk}
    \begin{aligned}
      G_n^{\mathrm{disk}}
      \;&:=\;
      \pi_1\Big(
        \mathrm{Map}^\ast_0\big(
          (\Sigma^2_{0,1,n})_{\cpt}
          ,\,
          S^2
        \big)
        \sslash
        \mathrm{Diff}^{+,\partial}\big(
          \Sigma^2_{0,1,n}
        \big)
      \Big)
      \\
      &\;\simeq\;
      Q_{n-1}\rtimes_{\mathrm{st}^{\vee}}\mathrm{Br}_n
      \\
      &\;\simeq\;
      \mathrm{FBr}_n\big/\big\langle(\mathbf{1},1)\big\rangle
      \mathrlap{\,.}
    \end{aligned}
  \end{equation}
  The action here is again \eqref{DualIntegralStandardRepresentation};
  the element $(\mathbf{1},1)$ changes every framing by one and has
  trivial underlying braid.
  \end{itemize}
  For the spherical mapping class group, the small cases and the
  spherical braid description are
  $$
    H_n
    \;\simeq\;
    \left\{
    \def\arraystretch{1.4}
    \begin{array}{lcl}
      1 &\vert& n=1
      \\
      \mathbb{Z}_2 &\vert& n=2
      \\
      \mathrm{Br}_n(S^2)/\langle\fullRotation\rangle
      &\vert& n\geq 3
      \,.
    \end{array}
    \right.
  $$
  In the last line, $\fullRotation$ is the central full rotation of
  order two from \eqref{SphericalBraidGroupSequence}; it induces the
  trivial permutation and hence acts trivially on $Q_{n-1}$.
  In particular,
  $$
    G_1^{\mathrm{sph}}\simeq\mathbb{Z}
    \,,\qquad
    G_2^{\mathrm{sph}}\simeq
    \mathbb{Z}\times(\mathbb{Z}\rtimes_{-1}\mathbb{Z}_2)
    \,,\qquad
    G_1^{\mathrm{disk}}=1
    \,.
  $$
\end{proposition}
\begin{proof}
  The constant map is fixed by every diffeomorphism of the punctured
  surface. It therefore gives a section of the homotopy-quotient fibration
  over $B\mathrm{Diff}^{+,\partial}(\Sigma^2_{0,b,n})$.
  Its homotopy exact sequence consequently yields a split extension
  of the flux-monodromy group by
  $\pi_0\mathrm{Diff}^{+,\partial}(\Sigma^2_{0,b,n})
  =\mathrm{MCG}(\Sigma^2_{0,b,n})$.
  This argument also applies for $n=1,2$; it does not require the
  identity component of the diffeomorphism group to be contractible.

  For $b\in\{0,1\}$, we thus have
  \begin{equation}
    \label{SimplifyingFluxMonodromyOverPuncturedDisks}
    \def\arraystretch{2}
    \setlength\arraycolsep{4pt}
    \begin{array}{ll}
      \mathrlap{
        \pi_1\Big(
          \mathrm{Map}^\ast_0\big(
            (\Sigma^2_{0,b,n})_{\cpt},S^2
          \big)
          \sslash
          \mathrm{Diff}^{+,\partial}\big(
            \Sigma^2_{0,b,n}
          \big)
        \Big)
      }
      \\
      \;\simeq\;
      \pi_1\Big(
        \mathrm{Map}^\ast_0\big(
          (\Sigma^2_{0,b,n})_{\cpt},S^2
        \big)
      \Big)
      \rtimes\mathrm{MCG}(\Sigma^2_{0,b,n})
      &
      \proofstep{by the fixed constant map}
      \\
      \;\simeq\;
      \Big(
        \pi_1\Big(
          \mathrm{Map}^\ast_0\big(
            \Sigma^2_{0,b},S^2
          \big)
        \Big)
        \times Q_{n-1}
      \Big)
      \rtimes\mathrm{MCG}(\Sigma^2_{0,b,n})
      &
      \proofstep{by \eqref{FluxMonodromyOnNPuncturedSurface}}
      \\
      \;\simeq\;
      \pi_1\Big(
        \mathrm{Map}^\ast_0\big(
          \Sigma^2_{0,b},S^2
        \big)
      \Big)
      \times
      \Big(
        Q_{n-1}\rtimes_{\mathrm{st}^{\vee}}
        \mathrm{MCG}(\Sigma^2_{0,b,n})
      \Big)
      &
      \proofstep{
        by Thm. \ref{BraidGroupActionOnFluxMonodromyOverPuncturedSurface}
      }
      \mathrlap{\,.}
    \end{array}
  \end{equation}
  Here $\mathrm{st}^{\vee}$ is \eqref{DualIntegralStandardRepresentation}.
  For these genus-zero surfaces, the mapping class group is generated by
  the Artin half-twists from a disk containing the punctures, so the cited
  theorem determines its entire action.

  For $b=1$, the surface $\Sigma^2_{0,1}=D^2$ is contractible, so the first
  factor in \eqref{SimplifyingFluxMonodromyOverPuncturedDisks} is trivial.
  Its mapping class group with $n$ punctures is $\mathrm{Br}_n$, and
  \eqref{FramingQuotientExactSequence} gives
  \eqref{FluxMonodromyOfPuncturedDisk}.

  For $b=0$, the first factor is
  $\pi_1\mathrm{Map}^\ast_0(S^2,S^2)\simeq\pi_3(S^2)\simeq\mathbb{Z}$,
  giving \eqref{MonodromyOverNPuncturedSphere}.
  The descriptions of $H_n$ follow from
  \eqref{MappingClassGroupsOfPuncturedSpheres} and
  \eqref{SphericalBraidGroupSequence}, with $H_1=1$ and $H_2=\mathrm{Sym}_2$
  treated separately. For $n=2$, the dual integral standard representation
  is the sign action on $Q_1\simeq\mathbb{Z}$; for $n=1$, $Q_0=0$.
\end{proof}

\subsection{On the open annulus}
\label{OnTheOpenAnnulus}

On the open annulus $\Sigma^2_{0,0,2}$, hence on the 2-punctured sphere
\eqref{ExamplesOfSurfaces}, choose the spin structure extending across
both punctures. Prop. \ref{CovariantFluxMonodromyOnPuncturedDisks}
gives the covariantized flux monodromy group
\begin{equation}
  \label{CovFluxMonodromyOnOpenAnnulus}
  \mathbb{Z}
  \times
  \big(
    Q_1
    \rtimes_{{}_{\mathrm{st}^{\vee}}}
    \mathrm{Sym}_2
  \big)
  \;\simeq\;
  \mathbb{Z}
  \times
  \big(
    \mathbb{Z}
    \rtimes_{-1}
    \mathbb{Z}_2
  \big)
  \,.
\end{equation}
Here $Q_1=\mathbb{Z}^2/\mathbb{Z}(1,1)\simeq\mathbb{Z}$,
and the dual integral standard representation $\mathrm{st}^{\vee}$
\eqref{DualIntegralStandardRepresentation} is the sign action.
The nontrivial mapping class interchanges the two punctures, hence the two
ends of the open annulus. The corresponding semidirect product is the
{\it infinite dihedral group}:
\begin{equation}
  \label{ZRtimesZTwo}
  \mathbb{Z} 
    \rtimes 
  \mathbb{Z}_2
  \;\;
  \simeq
  \;\;
  \Big\{
  \big(
    n, 
    \modulo{s}{2}
  \big)
  \,\in\,
  \mathbb{Z} \times \mathbb{Z}_2
  \,,
  \;\;\;
  \big(
    n, \modulo{s}{2}
  \big)  
  \cdot
  \big(
    n', \modulo{s'}{2}
  \big)  
  \;=\;
  \big(
    n + (-1)^s n'
    ,\,
    \modulo{(s + s')}{2}
  \big)
  \Big\}
  \,.
\end{equation}

\begin{lemma}[\bf Finite-dimensional unitary irreps over the open annulus]
  \label{FinDimIrrepsOfCovariantFluxMonodromyOverOpenAnnulus}
  The finite-dimensional irreducible unitary representations of
  $\mathbb{Z}\rtimes_{-1}\mathbb{Z}_2$ \eqref{ZRtimesZTwo} are,
  up to unitary isomorphism, the four 1-dimensional representations
  \begin{equation}
    \label{1DIrrepsOfZoRtimesZ2}
    \begin{tikzcd}[row sep=-2pt, 
     column sep=0pt
    ]
      &&
      \mathbb{Z} 
        \rtimes 
      \mathbb{Z}_2
      \ar[rr]
      &&
      \mathrm{U}(\mathbb{C})
      \\
      \rchi_{++}
      &:&
      \big(
       n
       ,\,
       \modulo{s}{2}
     \big)
     &\longmapsto&
     1
     \\
      \rchi_{+-}
      &:&
      \big(
       n
       ,\,
       \modulo{s}{2}
     \big)
     &\longmapsto&
     e^{\pi \mathrm{i} s}
     \\
      \rchi_{-+}
      &:&
      \big(
       n
       ,\,
       \modulo{s}{2}
     \big)
     &\longmapsto&
     e^{\pi \mathrm{i} n}
     \\
      \rchi_{--}
      &:&
      \big(
       n
       ,\,
       \modulo{s}{2}
     \big)
     &\longmapsto&
     e^{\pi \mathrm{i} (n+s)}
    \end{tikzcd}
  \end{equation}
  and, for $\lambda\in\mathrm{U}(1)\setminus\{+1,-1\}$,
  the 2-dimensional representations
  \begin{equation}
    \label{2DIrrepsOfZoRtimesZ2}
    \begin{tikzcd}[sep=0pt]
      \mathllap{
        \rho_\lambda
        \,:\,
        \;
      }
      \mathbb{Z} \rtimes 
      \mathbb{Z}_2
      \ar[
        rr,
      ]
      &&
      \mathrm{U}(\mathbb{C}^2)
      \\[-2pt]
      \big(
       1
       ,\,
       \modulo{0}{2}
     \big) 
     &\longmapsto&
     D_\lambda
     &:=&
     \scalebox{0.8}{$
     \begin{bmatrix}
       \lambda^{-1}
       &
       0
       \\
       0 & 
       \lambda
     \end{bmatrix}
     $}
     \\
      \big(
       0
       ,\,
       \modulo{1}{2}
     \big)      
     &\longmapsto&
     X
     &:=&
    \scalebox{0.8}{$   \begin{bmatrix}
         0 & 1 
         \\
         1 & 0
       \end{bmatrix}
       $}
     \mathrlap{\,.}
    \end{tikzcd}
  \end{equation}
  Here $\rho_\lambda\simeq\rho_{\lambda'}$ precisely when
  $\lambda'\in\{\lambda,\lambda^{-1}\}$.

  Every finite-dimensional irreducible unitary representation of the full
  group \eqref{CovFluxMonodromyOnOpenAnnulus} is uniquely obtained by
  choosing one of these irreps $\rho$, up to the stated equivalence,
  and a phase $\eta\in\mathrm{U}(1)$, and setting
  \begin{equation}
    \label{CentralCharacterOfAnnulusIrrep}
    \begin{tikzcd}[row sep=-2pt, column sep=0pt]
      \mathbb{Z}\times
      \big(\mathbb{Z}\rtimes_{-1}\mathbb{Z}_2\big)
      \ar[rr]
      &&
      \mathrm{U}(\HilbertSpace{H}_\rho)
      \\
      \big(m,(n,\modulo{s}{2})\big)
      &\longmapsto&
      \eta^m\,\rho\big(n,\modulo{s}{2}\big)
      \mathrlap{\,.}
    \end{tikzcd}
  \end{equation}
\end{lemma}
\begin{proof}
  Write $r=(1,\modulo{0}{2})$ and $s=(0,\modulo{1}{2})$ for the
  generators of \eqref{ZRtimesZTwo}, so that $s^2=1$ and $srs=r^{-1}$.
  In a finite-dimensional unitary irrep, choose a unit eigenvector $v$
  of $\rho(r)$ with eigenvalue $\lambda^{-1}\in\mathrm{U}(1)$.
  Then $\rho(s)v$ is an eigenvector with eigenvalue $\lambda$, and the span
  of $v$ and $\rho(s)v$ is invariant under both generators.
  By irreducibility this is the whole representation space.

  If $\lambda\neq\lambda^{-1}$, these two unit eigenvectors are orthogonal,
  and on this basis the generators act as in \eqref{2DIrrepsOfZoRtimesZ2}.
  Conversely, the two distinct eigenspaces of $D_\lambda$ are interchanged
  by $X$, so these representations are irreducible. Their spectra on $r$
  show that the only possible equivalence of parameters is
  $\lambda\leftrightarrow\lambda^{-1}$, which is implemented by $X$.

  If $\lambda=\lambda^{-1}$, then $\rho(r)=\pm\mathrm{id}$ on the above
  span. An eigenvector of the unitary involution $\rho(s)$ spans an
  invariant line, which by irreducibility is the whole representation.
  The two independent sign choices give \eqref{1DIrrepsOfZoRtimesZ2}.

  Finally, in an irrep of \eqref{CovFluxMonodromyOnOpenAnnulus}, the
  generator of the central $\mathbb{Z}$-factor acts by a scalar
  $\eta\in\mathrm{U}(1)$, by Schur's lemma. The restriction to the
  dihedral factor is therefore irreducible, giving
  \eqref{CentralCharacterOfAnnulusIrrep} and its uniqueness.
\end{proof}

In particular, for any $o\in\mathbb{N}_{>0}$, the irreps that factor through
the finite dihedral quotient
\begin{equation}
  \label{FineQuotientsOfCovFluxMonodromyOverOpenAnnulus}
  \begin{tikzcd}[sep=-2pt]
  \mathbb{Z} \rtimes \mathbb{Z}_2 
  \ar[
    rr,
    ->>
  ]
  &&
  \mathbb{Z}_o \rtimes \mathbb{Z}_2
  \\
  \big(
    n, \modulo{s}{2}
  \big)
  &\longmapsto&
  \big(
    \modulo{n}{o}
    ,\, 
    \modulo{s}{2}
  \big)
  \end{tikzcd}
\end{equation}
are exactly those for which $\rho(r)^o=\mathrm{id}$.
Thus \eqref{1DIrrepsOfZoRtimesZ2} gives $\rchi_{++}$ and $\rchi_{+-}$
for every $o$, and also $\rchi_{-+}$ and $\rchi_{--}$ when $o$ is even.
The 2-dimensional irreps are \eqref{2DIrrepsOfZoRtimesZ2} with
\[
  \lambda=e^{2\pi\mathrm{i}\tfrac{k}{o}}
  \,,\qquad
  k\in\big\{1,\cdots,\lfloor(o-1)/2\rfloor\big\}
  \,,\qquad
  D_\lambda=R_z(4\pi\tfrac{k}{o})
  \,.
\]
The parameter $\eta$ in \eqref{CentralCharacterOfAnnulusIrrep} remains
independent of this finite-quotient condition on the dihedral factor.
In particular, the two puncture classes in $Q_1$ sum to zero, so their
operators are inverses; their product is the identity, independently of
the character of the Hopf $\mathbb{Z}$-factor.

If one end is distinguished as the outer rim and must be preserved,
the end-interchanging mapping class is excluded. The mapping class
group of the open annulus preserving each end is trivial, so the
covariant flux monodromy reduces to
$\mathbb{Z}\times Q_1\simeq\mathbb{Z}^2$ and its finite-dimensional
unitary irreps are all one-dimensional. Thus the two-dimensional
sectors above use the symmetry interchanging the two ends.

\smallskip

\subsection{On the 2-punctured disk}
\label{QBitQuantumGatesOperableOn2PuncturedOpenDisk}

We now specialize \S\ref{OnPuncturedDisks} to the 2-punctured open disk,
hence to the 3-punctured sphere
(cf. Rem. \ref{InternalPuncturesAndTheRim}).
We first allow all permutations of the three punctures. For the spin
structure extending over the sphere, the full mapping class group is
$\mathrm{Sym}_3$, acting on the puncture lattice by $\mathrm{st}^{\vee}$
\eqref{DualIntegralStandardRepresentation}. This yields irreducible
unitary representations of dimensions $1$, $2$, $3$, and $6$, as follows.
The restriction preserving the rim puncture is described at the end
of this subsection.

\begin{example}[\bf Flux monodromy over the 3-punctured sphere]
  Prop. \ref{CovariantFluxMonodromyOnPuncturedDisks}, with
  \eqref{MappingClassGroupsOfPuncturedSpheres}, gives the group
  \begin{equation}
    \label{CovariantFluxMonodromyOn2PuncturedOpenDisk}
    \pi_1\Big(
      \mathrm{Map}^\ast_0\big(
        (\Sigma^2_{0,0,2+1})_{\cpt}
        ,\,
        S^2
      \big)
      \sslash
      \mathrm{Diff}^{+,\partial}\big(
        \Sigma^2_{0,0,2+1}
      \big)
    \Big)
    \;\simeq\;
    \mathbb{Z}
    \times
    \big(
      Q_2
      \rtimes_{{}_{\mathrm{st}^{\vee}}}
      \mathrm{Sym}_3
    \big)
    \;=:\;
    G_3
    \,.
  \end{equation}
  Here $Q_2=\mathbb{Z}^3/\mathbb{Z}(1,1,1)$, with $\mathrm{Sym}_3$
  permuting the classes $[v_1]$, $[v_2]$, $[v_3]$, subject to
  $[v_1]+[v_2]+[v_3]=0$.
  By \eqref{FramingQuotientExactSequence}, the puncture factor
  $Q_2\rtimes\mathrm{Sym}_3$ is the quotient of
  $\mathbb{Z}^3\rtimes\mathrm{Sym}_3$ by the common framing shift.
  The first $\mathbb{Z}$-factor in
  \eqref{CovariantFluxMonodromyOn2PuncturedOpenDisk} is independently
  the Hopf factor $\pi_3(S^2)$.
\end{example}

\begin{proposition}[\bf Finite-dimensional unitary irreps over the 3-punctured sphere]
  \label{FinDimIrrepsOverThreePuncturedSphere}
  The unitary characters of $Q_2$ are parametrized by
  \[
      \widehat{Q_2}
      \;\simeq\;
      \Big\{
        \mathbf{z}=(z_1,z_2,z_3)\in\mathrm{U}(1)^3
        \;\Big\vert\;
        z_1z_2z_3=1
      \Big\}
      \,.
  \]
  Explicitly, the corresponding characters are
  \begin{equation}
    \label{CharactersOfThreePunctureLattice}
    \begin{tikzcd}[row sep=-2pt, column sep=0pt]
        \chi_{\mathbf{z}}\;:\;
        Q_2
        \ar[rr]
        &&
        \mathrm{U}(\mathbb{C})
        \\
        \big[(a_1,a_2,a_3)\big]
        &\longmapsto&
        z_1^{a_1}z_2^{a_2}z_3^{a_3}
        \mathrlap{\,.}
    \end{tikzcd}
  \end{equation}
  Choose one representative $\mathbf{z}$ of each $\mathrm{Sym}_3$-orbit
  under permutation of coordinates, and write $H_{\mathbf{z}}$ for its
  stabilizer. The finite-dimensional irreducible unitary representations
  of $G_3$ \eqref{CovariantFluxMonodromyOn2PuncturedOpenDisk} are,
  up to unitary isomorphism,
  \begin{equation}
    \label{MackeyIrrepsOverThreePuncturedSphere}
    \rho_{\eta,\mathbf{z},\sigma}(m,g)
    \;:=\;
    \eta^m\,
    \mathrm{Ind}_{Q_2\rtimes H_{\mathbf{z}}}^{Q_2\rtimes\mathrm{Sym}_3}
    \big(\chi_{\mathbf{z}}\otimes\sigma\big)(g)
    \,.
  \end{equation}
  Here $\eta\in\mathrm{U}(1)$, and $\sigma$ runs over the irreducible
  unitary representations of $H_{\mathbf{z}}$; the representation being
  induced is
  $(q,h)\mapsto\chi_{\mathbf{z}}(q)\sigma(h)$.
  For the chosen orbit representatives, these parameters are unique up to
  unitary equivalence of $\sigma$.
  Their dimensions are given by
  \begin{equation}
    \label{DimensionsOfIrrepsOverThreePuncturedSphere}
    \def\arraystretch{1.4}
    \begin{array}{lccc}
      \mathbf{z}
      & H_{\mathbf{z}}
      & [\mathrm{Sym}_3:H_{\mathbf{z}}]
      & \dim(\rho_{\eta,\mathbf{z},\sigma})
      \\
      \hline
      (z,z,z),\;z^3=1
      & \mathrm{Sym}_3 & 1 & 1,\,1,\,2
      \\
      (z,z,z^{-2}),\;z^3\neq1
      & \mathrm{Sym}_2 & 3 & 3,\,3
      \\
      \text{three distinct coordinates}
      & \{1\} & 6 & 6
    \end{array}
    \mathrlap{\,.}
  \end{equation}
  In the last column, each entry lists the dimensions obtained from the
  inequivalent irreps of the indicated stabilizer, for fixed $\eta$ and
  fixed orbit of $\mathbf{z}$.
\end{proposition}
\begin{proof}
  A character of $\mathbb{Z}^3$ descends to $Q_2$ precisely when it is
  trivial on $(1,1,1)$, giving
  \eqref{CharactersOfThreePunctureLattice}.
  For $h\in H_{\mathbf{z}}$, the identity
  $\chi_{\mathbf{z}}(h q)=\chi_{\mathbf{z}}(q)$ makes
  $(q,h)\mapsto\chi_{\mathbf{z}}(q)\sigma(h)$ a representation of
  $Q_2\rtimes H_{\mathbf{z}}$.

  Now the Mackey algorithm (Prop. \ref{MackeyAlgorithm}) gives
  \eqref{MackeyIrrepsOverThreePuncturedSphere}.
  Explicitly, in a finite-dimensional unitary representation the commuting
  operators from $Q_2$ admit a simultaneous orthogonal decomposition into
  character spaces. The symmetric group permutes these spaces, and
  irreducibility forces their characters to form a single orbit.
  The stabilizer acts irreducibly on any one character space, and
  induction reconstructs the representation from this action.
  Conversely, inducing an irreducible stabilizer representation in this
  way is irreducible: every invariant subspace decomposes into character
  spaces, its intersection with one such space is stabilizer-invariant,
  and the symmetric group acts transitively on them.
  The independent central $\mathbb{Z}$ acts by $m\mapsto\eta^m$,
  by Schur's lemma.

  The stabilizers are $\mathrm{Sym}_3$, $\mathrm{Sym}_2$, or the trivial
  group according as the three coordinates are all equal, exactly two
  are equal, or all are distinct. Multiplying their irreducible
  representation dimensions by the corresponding orbit sizes gives
  \eqref{DimensionsOfIrrepsOverThreePuncturedSphere}.
\end{proof}

For example, if $\omega=e^{2\pi\mathrm{i}/3}$, then
$\mathbf{z}=(1,\omega,\omega^2)$ has trivial stabilizer and yields a
6-dimensional irrep. Its $Q_2\rtimes\mathrm{Sym}_3$-representation
already factors through the finite quotient
$(Q_2/3Q_2)\rtimes\mathrm{Sym}_3$.
Thus even among representations of finite quotients, the 2-dimensional
standard representation of $\mathrm{Sym}_3$ does not exhaust the
irreducibles of dimension greater than one.

\medskip

The 2-dimensional irreps in
\eqref{DimensionsOfIrrepsOverThreePuncturedSphere} have a particularly
simple description. Write $t$ for the generator of the Hopf
$\mathbb{Z}$-factor, and use $[v_i]$ and $s\in\mathrm{Sym}_3$ also for
their canonical inclusions in $G_3$. For
$\eta\in\mathrm{U}(1)$ and $\xi\in\mathrm{U}(1)$ with $\xi^3=1$,
these irreps are
\begin{equation}
  \label{RestrictionOfCanonicalRepOfFramed3Braid}
  \begin{tikzcd}[row sep=-2pt, column sep=0pt]
    \rho_{\eta,\xi}\;:\;
    \mathbb{Z}
    \times
    \big(
      Q_2
      \rtimes_{{}_{\mathrm{st}^{\vee}}}
      \mathrm{Sym}_3
    \big)
    \ar[rr]
    &&
    \mathrm{U}(\mathbb{C}^2)
    \\
    t
    &\longmapsto&
    \eta\,\mathrm{id}
    \\
    {[v_i]}
    &\longmapsto&
    \xi\,\mathrm{id}
    &\qquad (i=1,2,3)
    \\
    s
    &\longmapsto&
    U(s)
    \mathrlap{\,,}
  \end{tikzcd}
\end{equation}
where $U$ is the unitarization of the standard representation $\mathbf{2}$
of $\mathrm{Sym}_3$. The condition $\xi^3=1$ is exactly the relation
$[v_1]+[v_2]+[v_3]=0$; the Hopf phase $\eta$ is independent of $\xi$.
For $\xi=1$, \eqref{RestrictionOfCanonicalRepOfFramed3Braid} factors
through the projection $G_3\to\mathbb{Z}\times\mathrm{Sym}_3$.

\begin{proposition}[\bf Unitarization of the standard representation of $\mathrm{Sym}_3$]
\label{UnitarizationOfStandardRepOfSym3}
The standard irrep $\mathbf{2}$ of $\mathrm{Sym}_3$
{\rm (Def. \ref{StandardRepresentationOfSymmetricGroup})}
admits an orthonormal basis in which the generating permutations,
written in one-line notation, act by the following matrices:
\begin{itemize}
\item[(i)]
The Pauli $Z$-matrix:
$$
  U(213)
  \;=\;
  Z \;:=\;
  \begin{bmatrix}
    1 & 0 
    \\
    0 & -1
  \end{bmatrix}
$$
\item[(ii)]
A rotation matrix {\rm (cf. \cite[(4.4-6)]{NielsenChuang00})}
\begin{equation}
  \label{ARotationGate}
  U(231)
  \;=\;
  R_y(8 \pi /3)
  \;=\;
  -
  R_y(2\pi/3)
  \;:=\;
  \left[
  \begin{matrix}
     \mathrm{cos}(4\pi/3)
     &
     -\mathrm{sin}(4\pi/3)
     \\
     \mathrm{sin}(4\pi/3)
     &

     \mathrm{cos}(4\pi/3)
  \end{matrix}
  \right]
  \;=\;  
  -
  \left[
    \begin{matrix}
      1/2 & - \sqrt{3}/2
      \\
      \sqrt{3}/2 & 1/2
    \end{matrix}
  \right]
  .
\end{equation}
\end{itemize}
\end{proposition}
\begin{proof}
  The defining representation on $\mathbb{C}^3 \simeq \mathrm{Span}_{\mathbb{C}}(v_1, v_2, v_3)$ \eqref{DefiningRepresentation} is evidently unitary with respect to the canonical inner product $\langle v_i \vert v_j \rangle = \delta_{i j}$, but the basis $\big(e_1 := v_1 - v_3,\, e_2 := v_2 - v_3 \big)$, from 
  \eqref{TheStandardRep}, for the standard sub-representation $\mathbf{2}$,
  is not orthonormal with respect to this inner product. One choice of orthonormal basis for this subspace is given by
  \[
    \left[
    \begin{matrix}
      1 
      \\
      0
    \end{matrix}
    \right]
    \;:=\;
    \tfrac{1}{\sqrt{6}}\big(
      v_1 + v_2 - 2 \, v_3
    \big)
    \,,
    \;\;\;\;\;\;
    \left[
    \begin{matrix}
      0
      \\
      1
    \end{matrix}
    \right]
    \;:=\;
    \tfrac{1}{\sqrt{2}}\big(
      v_1 - v_2
    \big)
    \,.
  \]
  A straightforward computation shows that on this basis the permutations $(213)$ and $(231)$ act as claimed.
\end{proof}

\begin{remark}[\bf Finite projective image]
  Writing $C:=U(231)$ in \eqref{ARotationGate}, the displayed matrices
  satisfy
  \[
    Z^2=C^3=\mathrm{id}
    \,,\qquad
    ZCZ=C^{-1}
    \,.
  \]
  They generate the six distinct matrices $C^jZ^\epsilon$, for
  $j\in\{0,1,2\}$ and $\epsilon\in\{0,1\}$, giving a faithful copy of
  $\mathrm{Sym}_3$. None of its nonidentity elements is scalar.
  Hence the projective image of every irrep
  \eqref{RestrictionOfCanonicalRepOfFramed3Braid} is also
  $\mathrm{Sym}_3$: both the Hopf factor and the puncture lattice act by
  scalar phases. These phases do not supply a relative phase between
  the two basis states. In particular, these representations by
  themselves do not provide a universal set of single-qubit gates.
\end{remark}

\smallskip
\noindent{\bf Preserving the outer rim.}
If $v_3$ represents the outer rim, the permitted mapping classes lie
in its stabilizer $\mathrm{Sym}_2=\langle(213)\rangle$ in
$\mathrm{Sym}_3$, where we retain the one-line permutation notation.
The corresponding subgroup of
\eqref{CovariantFluxMonodromyOn2PuncturedOpenDisk} is
\[
  G_{3,\mathrm{rim}}
  \;:=\;
  \mathbb{Z}\times
  \big(Q_2\rtimes_{{}_{\mathrm{st}^{\vee}}}\mathrm{Sym}_2\big)
  \;\subset\;
  G_3
  \,.
\]
The representations \eqref{RestrictionOfCanonicalRepOfFramed3Braid}
restrict reducibly to this subgroup: the Hopf factor and $Q_2$ still
act by scalars, while
$\mathbf{2}|_{\mathrm{Sym}_2}\simeq\mathbf{1}\oplus\mathrm{sgn}$.
In the basis of Prop. \ref{UnitarizationOfStandardRepOfSym3}, the two
coordinate lines are invariant and the internal-defect exchange acts
by $Z$. The three-cycle represented by $U(231)$ moves $v_3$ and hence
does not preserve the rim. Thus the projective image of these
restricted representations is $\{[\mathrm{id}],[Z]\}$.

Other two-dimensional irreps of $G_{3,\mathrm{rim}}$ can occur: a
character of $Q_2$ not fixed by the transposition has an orbit of size
two and induces an irreducible representation. More generally, the
same character-space argument as above shows that the irreducible
dimensions for this subgroup are $1$ and $2$.
For an actual closed disk with two internal punctures and boundary
fixed pointwise, Prop. \ref{CovariantFluxMonodromyOnPuncturedDisks}
\textup{(ii)} instead gives $Q_1\rtimes\mathrm{Br}_2$.
This is a different surface and boundary condition; its braid-group
factor is retained, rather than replaced by its permutation image.


\section
{Conclusion \& Outlook}

Seeking further theoretical understanding of topological quantum order, specifically in experimentally promising FQH systems, we analyzed a novel non-Lagrangian description of FQH anyons, arising by taking seriously the fact that these are bound states of surplus \emph{flux quanta} seen against the backdrop of an exact rational filling fraction. This suggests that effective exotic surplus \emph{flux quantization} may go to the heart of the nature of FQH anyons. We began by recalling motivation for such surplus flux quantization in the nonabelian cohomology theory known as 2-Cohomotopy, a deformation of standard Dirac flux quantization, and then provided a thorough algebro-topological and representation-theoretic analysis of the topological quantum order predicted thereby. 

At the heart of this analysis is the proof that the monodromy of 2-cohomotopical flux:
\begin{enumerate}
\item  on the plane witnesses abelian anyon braiding, in fact reproduces abelian Chern-Simons Wilson loop quantum observables, directly equipped with their traditional renormalization choices (which the CS-Lagrangian alone does not provide, cf. \cite{SS26-Renormalization}),

\item on the torus witnesses in addition the integer Heisenberg group algebra of torus anyon observables, including their modular equivariance depending on the choice of spin structure, and enforcing in superselection sectors the quantization of the braiding phase to a root of unity (cf. \cref{2CohomotopicalFluxThroughTorus}).
\end{enumerate}

This shows that in familiar situations and for abelian anyons, surplus flux quantization in 2-Cohomotopy matches in fine detail the expected properties of abelian anyons seen in experiment and via effective abelian Chern-Simons theory.

But  2-Cohomotopical flux quantization also makes predictions in further situations that have not been described by other means:
\begin{enumerate}
  \item
  The same general covariance that on the torus implies the expected modular data predicts on punctured surfaces the monodromy action of nonabelian braid groups describing the braiding of defect anyons attached to the punctures. Physically the punctures here describe defects where (surplus) flux is expelled, suggesting that doping FQH liquids by superconducting islands may be a route towards stable nonabelian anyons in FQH systems (cf. \hyperlink{FigureD}{Fig. D}). 
  
  This prediction is further discussed in \cite{SS26-Islands}.

  \item 
  \label
  {FQAHAnyonPrediction}
  The same 2-sphere coefficient that we have analyzed here as classifying exotic effective surplus FQH flux quantization appears naturally as the (Bloch sphere) coefficient of Bloch Hamiltonians in 2-band Fractional Quantum \emph{Anomalous} Hall systems (FQAH). Indeed, under position/momentum space duality the discussion of FQH systems here maps to an analogous discussion of FQAH systems, where it first of all predicts a novel route towards realizing (abelian) anyons in fractional Chern insulators, in the guise of Hopf solitons in the Bloch Hamiltonian, localized in momentum space (cf. \cref{2CohomotopyOfSurfaces}).

  This prediction is further discussed in \cite{SS25-FQAH}.

  \item
  With the mapping space to the 2-sphere thus understood as the space of (spectrally flattened) Bloch Hamiltonians in 2-band fractional Chern insulators, one may consider imposing further natural physical conditions. One such condition is a bound on the averaged \emph{quantum metric}, reflecting that deformations of realistic systems will not explore the full abstract space of Bloch Hamiltonians. We find that with such a constraint imposed, the situation in \cref{FQAHAnyonPrediction} decomposes effectively into a cyclic layer of $\vert C \vert$ FQH bands ($C$ the system's Chern class), permuted under lattice translation, just as previously argued by \cite[\S I]{Qi2011}\cite[\S 2]{BarkeshliQi12} but now exhibiting nonabelian anyonic topological order.

  This prediction will be further discussed in \cite{GSS26-Bounded}.
\end{enumerate}

These results support the suggestion that the non-Lagrangian 2-cohomotopical quantization of surplus FQH flux may be a valuable theoretical alternative for the description and understanding of anyonic topological order in fractional quantum (anomalous) Hall systems, especially in the search for new experimental pathways to stable nonabelian anyons desired for serious application in quantum computing technologies.




\appendix

\section{Background}
\label{Background}

\stoptoc

Here we briefly recall and cite some background material
referred to in the main text:

\vspace{-2mm} 
\begin{itemize}[
  topsep=5pt,
  itemsep=.2pt,
  leftmargin=1.8cm
]
\item[\S\ref{FractionalQuantumHallSystems}] -- Effective CS for FQH

\item[\S\ref{SomeCategoryTheory}] -- Some category theory

\item[\S\ref{SomeAlgebraicTopology}] -- Some algebraic topology

\item[\S\ref{SurfacesAnd2Cohomotopy}] -- Surfaces \& 2-Cohomotopy

\item[\S\ref{SomeRepresentationTheory}] -- Some representation theory

\item[\S\ref{QuadraticGaussSums}] -- Quadratic Gauss sums
\end{itemize}

\subsection{Effective CS for FQH}
\label{FractionalQuantumHallSystems}

We recall effective Chern-Simons descriptions of FQH systems, distinguishing their local response formulas, global quantum theories, and phenomenological interpretation. In particular, their ground-state monodromy representations are part of the comparison with our construction.

\medskip

\subsubsection
{Traditional Effective Chern-Simons for FQH systems at unit-fraction fillings}
\label
{ProblemWithFluxQuantizationForCSLagrangian}
The traditional ansatz for an effective field theory description of fractional quantum Hall systems
at unit-fraction filling $\nu = 1/\lattice$ postulates that the effective field is a 1-form potential $a$ for the electric current density 2-form $J$ (``statistical gauge field''), itself minimally coupled to the {\it quasi-hole current} $j$, 
and with effective dynamics encoded by the level = $\level = \lattice/2$ Chern-Simons (CS) Lagrangian  \cite{Zhang92}\cite{Wen95}:
\begin{equation}
  \label
  {EffectiveCSLagrangian}
  \begin{tikzcd}[
    sep=0pt
  ]
  {}
  \\[5pt]
  \scalebox{.7}{
    \bf
    \def\arraystretch{.9}
    \begin{tabular}{c}
      \color{darkblue}
      Electron current
      \\
      \color{gray}
      density 2-form
    \end{tabular}
  }
  &[-10pt]
  J 
  &=&
  \overset{
    \mathclap{
      \adjustbox{
        raise=-2pt,
        scale=.7,
        rotate=+12
      }
      {
        \rlap{
        \color{gray}
        \bf
        current 3-vector
        }
      }
    }
  }{
    \vec J
  }
  \;
  \scalebox{1.2}{$\lrcorner$}
  \;
  \,
  \overset{
    \mathclap{
      \adjustbox{
        raise=-2pt,
        scale=.7,
        rotate=+12
      }
      {
        \rlap{
        \color{gray}
        \bf
        volume form
        }
      }
    }
  }{
    \mathrm{dvol}
  }
  \;\;
  &=:&
  \mathrm{d}\;a
  \hspace{-1pt}
  \adjustbox{
    scale=.7,
    raise=1pt
  }{\;\;
    \color{darkblue}
    \bf
    \def\tabcolsep{0pt}
    \begin{tabular}{c}
      Effective gauge field
    \end{tabular}
  }
  \\
  \scalebox{.7}{
    \bf
    \def\arraystretch{.9}
    \begin{tabular}{c}
      \color{darkblue}
      Quasi-particle current
      \\
      \color{gray}
      density 2-form
    \end{tabular}
  }
  &
  j 
    &=&
  \vec j 
  \;
  \scalebox{1.2}{$\lrcorner$}
  \;
  \,
  \mathrm{dvol}
  \\
  \scalebox{.7}{
    \bf
    \def\arraystretch{.9}
    \begin{tabular}{c}
      \color{darkblue}
      Background flux
      \\
      \color{gray}
      density 2-form
    \end{tabular}
  }
  &
  F &=&
  \mathrm{d}\, A
  \mathrlap{
    \;
    \adjustbox{
      scale=.7,
      raise=1pt
    }{\;\;
      \color{darkblue}
      \bf
      External gauge field
    }
  }
  \;\;\;\;\;\;\;
  \\
  \scalebox{.7}{
  \bf
  \def\arraystretch{.9}
  \begin{tabular}{c}
    \color{darkblue}
    Effective Lagrangian
    \\
    \color{gray}
    density 3-form
  \end{tabular}
 }
 &
  L
  &:=&
  \mathrlap{
   \tfrac{\lattice}{2}
   \,
   \grayunderbrace
     { a \, \mathrm{d}a }
     { 
       \mathclap{ \mathrm{CS}(a) }
     }
   \;-\;
   \grayunderbrace
     { A \, \mathrm{d}a }
     {
       \mathclap{ A \, J }
     }
   \;-\;
   a \, j
   \;\;\;\;
   \scalebox{.9}{
     \cite[(2.11)]{Wen95}, 
     cf. \cite[(7.3.10)]{Wen07}
   }
  }
  \;\;\;\;\;\;\;\;\;\;\;\;
  \end{tikzcd}
\end{equation}

\vspace{-.1cm}
\noindent 
This is motivated by observing that 
the Euler-Lagrange equations of this $L$, 
\begin{equation}
  \label{EoMForEffectiveCS}
  \frac{
    \delta L
  }{
    \delta a
  }
  \,=\,
  0
  \;\;\;\;\;
  \Leftrightarrow
  \;\;\;\;\;
  \adjustbox{
    margin=3pt,
    bgcolor=lightolive
  }{$
  J \;=\;
  \tfrac{1}{\lattice}\big(
    \,
    F \,+\, j
    \,
  \big)
  $}
  \mathrlap{\,,}
\end{equation}
in the relevant case of longitudinal electron current
and static quasi-particles
\begin{equation}
  \begin{tikzcd}[
    sep=-2pt
  ]
  J 
  &\defneq&
  J_0 \, \mathrm{d}x \, \mathrm{d}y
  &-&
  J_x \, \mathrm{d}t \, \mathrm{d}y
  \\
  j 
  &\defneq&
  j_0 \, \mathrm{d}x \, \mathrm{d}y
  \\
  F 
    &\defneq&
  B\, \mathrm{d}x \, \mathrm{d}y
  &-&
  E_y \, \mathrm{d}t \, \mathrm{d}y
  \end{tikzcd}
\end{equation}
express just the hallmark properties of the FQHE at filling fraction $\nu \,=\, 1/\lattice$:
\begin{equation}
  \label{HallConductivityImplied}
  \mbox{\eqref{EoMForEffectiveCS}}
  \;
  \Leftrightarrow
  \;
  \left\{\!\!\!
  \def\arraystretch{1.6}
  \begin{array}{ccccl}
    J_x 
      &=&
    \tfrac{1}{\lattice}
    E_y
    &
    \Leftrightarrow
    &
    \mbox{
      Hall conductivity 
      law at $1/\lattice$ filling 
    }
    \\[-3pt]
    J_0 
    &=&
    \tfrac{1}{\lattice}
    \, 
    B
    &
    \Leftrightarrow
    &
    \mbox{
      each electron absorbs
      $\lattice$ flux quanta,
      but
    }
    \\[-5pt]
    &&
    \mathllap{+}
    \tfrac{1}{\lattice}
    \,
    j_0
    &&
    \mbox{
      $1/\lattice$th surplus electron
      for each quasi-particle\,.
    }
  \end{array}
  \right.
\end{equation}

Proceeding from this Lagrangian density, it is furthermore traditional (cf. \cite[p. 35]{Witten16}\cite[p. 159]{Tong16}) to consider ``integrating out'' the effective gauge field $a$ \eqref{EffectiveCSLagrangian}, by solving its equations of motion \eqref{EoMForEffectiveCS} in the case $j = 0$, which is
\begin{equation}
  \label
  {SimplifiedEoM}
  \adjustbox{
    margin=3pt,
    bgcolor=lightolive
  }{$
  J = \tfrac{1}{\lattice} F
  $}
  \mathrlap{\,,}
\end{equation}
locally via
\begin{equation}
\label
{SolvingForTheEffectiveGaugeField}
  a := \tfrac{1}{\lattice}A
  \mathrlap{\,,}
\end{equation}
and plugging this expression for $a$ back into the Lagrangian density \eqref{EffectiveCSLagrangian} to obtain: 
\begin{equation}
\label
{InducedCSLagrangian}
  \mathllap{
  \scalebox{.7}{
  \bf
  \def\arraystretch{.9}
  \begin{tabular}{c}
    \color{darkblue}
    Induced CS-Lagrangian
    \\
    \color{gray}
    for anyonic flux quanta
  \end{tabular}
 }
 }
 \;
  L'
  :=
  L_{
    \left\vert
      \begin{subarray}{l}
        j = 0,
        \\
        a = A/\lattice
      \end{subarray}
    \right.}
  =
  -
  \tfrac{1}{2\lattice}
  A \,\mathrm{d}A
  \mathrlap{\,.}
\end{equation}

Expression \eqref{InducedCSLagrangian} gives the local Hall response.
It does not replace the full theory globally: for a general background
connection $A$, division by $\lattice$ in
\eqref{SolvingForTheEffectiveGaugeField} need not define a connection on
a line bundle. As emphasized in \cite[p.~35]{Witten16}, the dynamical field
$a$ must then be retained to describe the topological sectors and
fractionally charged excitations.

Setting the external probe $A=0$ and omitting source insertions $j$ instead
isolates the intrinsic Chern-Simons sector
(\cite[p.~50]{Witten16}\cite[p.~160]{Tong16}):
\begingroup
\setlength{\abovedisplayshortskip}{3.5pt}
\begin{equation}
\label
{AdHocTruncation}
  \mathllap{
  \scalebox{.7}{
  \bf
  \def\arraystretch{.9}
  \begin{tabular}{c}
    \color{darkblue}
    Intrinsic CS sector
    \\
    \color{gray}
    for anyonic flux quanta
  \end{tabular}
 }
 }
 \;
  L''
  :=
  L_{
    \left\vert
      \begin{subarray}{l}
        j = 0
        \\
        A = 0
      \end{subarray}
    \right.
  }
  =
  \tfrac{\lattice}{2}
  a\, \mathrm{d}a
  \mathrlap{\,.}
\end{equation}
\endgroup
The action associated with \cref{AdHocTruncation} is globally defined with the appropriate level quantization and, for odd $\lattice$, spin structure. This specialization permits computing intrinsic ground-state and braiding data; the coupling to $A$ is retained when computing electromagnetic response.

\subsubsection
{A background choice for surplus flux}

To relate the effective fields to surplus flux, we now make a restricted background choice within ordinary integral cohomology. This is a choice of variables and background, rather than a prerequisite for the existence of the quantum Chern-Simons theory (cf. also \cite{Witten2026}).

\medskip

For smooth field configurations solving \cref{EoMForEffectiveCS}, the corresponding necessary conditions on flux classes are:
\begin{enumerate}
\item $[J]$ must be integral --- the number of electrons, 

\item $[F + j]$ must be integrally divisible by $K$ --- this characterizes the number of quasi-particles (whose presence reflects the failure of each electron to be bound to $K$ flux quanta).
\end{enumerate}

For $j=0$, suppose that the background connection $A$ is supplied with a $\lattice$th root $a_{\mathrm{bnd}}$. At the level of curvatures this gives
\begin{equation}
  \label
  {FluxAsKMultipleOfBoundElectrons}
  F \defneq K J_{\mathrm{bnd}}
\end{equation}
(read: $K$ flux quanta per number of \emph{bound} electrons, the defining property of exact filling fraction $\nu = 1/K$), for $J_{\mathrm{bnd}}$ a flux density
with local potential 
\begin{equation}
  \begin{aligned}
  J_{\mathrm{bnd}} 
  & = \mathrm{d} a_{\mathrm{bnd}}
  \\
  A & = K a_{\mathrm{bnd}}
  \mathrlap{\,,}
  \end{aligned} 
\end{equation}
with the integral flux classes satisfying:
\begin{equation}
\label
{ImprovedIntegralFluxQuantization}
  \begin{aligned}
  [J_{\mathrm{bnd}}]
  & \in
  H^2(-;\mathbb{Z})
  \\
  [F]
  &
  = K [J_{\mathrm{bnd}}]
  \mathrlap{\,.}
  \end{aligned}
\end{equation}

This restricts the background to one with a chosen root; it is not a condition imposed on every electromagnetic probe of the theory. On this background, \cref{InducedCSLagrangian} becomes:
\begin{equation}
  L'
  \defneq
  -\tfrac{1}{2K}
  A \, \mathrm{d}A
  =
  - \tfrac{K}{2}
  a_{\mathrm{bnd}}
  \,
  \mathrm{d} a_{\mathrm{bnd}}
  \mathrlap{\,,}
\end{equation}
which has the usual integral or spin Chern-Simons level quantization.

Before integrating out $a$, the chosen root permits completing the square in the parent Lagrangian \cref{EffectiveCSLagrangian}:
\begin{equation}
\label
{DecoupledFormOfL}
  \begin{aligned}
  L 
    &
    =
  \tfrac{K}{2}
  a \, \mathrm{d}a
  -
  K a_{\mathrm{bnd}} \, \mathrm{d}a
  -
  a j
  \\
  & \simeq
  \tfrac{K}{2}
  (a - a_{\mathrm{bnd}})
  \,
  \mathrm{d}
  (a - a_{\mathrm{bnd}})
  -
  \tfrac{K}{2}
  a_{\mathrm{bnd}}
  \,
  \mathrm{d} a_{\mathrm{bnd}}
  -
  a j
  \\
  & =
  \grayunderbrace{
  \Big(
  \tfrac{K}{2}
  a_{\mathrm{sur}}
  \,
  \mathrm{d}
  a_{\mathrm{sur}}
  -
  a_{\mathrm{sur}} j
  \Big)
  }{
    L_{\mathrm{dyn}}
  }
  -
  \grayunderbrace{
  \Big(
  \tfrac{K}{2}
  a_{\mathrm{bnd}}
  \,
  \mathrm{d} a_{\mathrm{bnd}}
  +
  a_{\mathrm{bnd}} j
  \Big)
  }{
    L_{\mathrm{bckg}}
  }
  \mathrlap{\,,}
  \end{aligned}
\end{equation}
where ``$\simeq$'' means equality up to boundary terms, and where $a_{\mathrm{sur}} := a - a_{\mathrm{bnd}}$ is the potential for the number of \emph{surplus} electrons, hence where we have split
\begin{equation}
\label
{ElectronCurrentSplit}
  \begin{aligned}
    J & = J_{\mathrm{bnd}} +
     J_{\mathrm{sur}}
    \\
    a & = a_{\mathrm{bnd}}
    + a_{\mathrm{sur}}
    \mathrlap{\,,}
  \end{aligned}
\end{equation}
where, for smooth fields, $a$, $a_{\mathrm{bnd}}$, and $a_{\mathrm{sur}}$
are connections on line bundles. Against this fixed background, the
equation of motion \cref{EoMForEffectiveCS} reads:
\begin{equation}
\label
{EoMInTermsOfSurplusCurrent}
  J_{\mathrm{sur}}
  =
  \tfrac{1}{K}
  j
  \mathrlap{\,.}
\end{equation}
For smooth integral flux classes this imposes a global compatibility
condition on the total source class $[j]$. It does not prohibit individual
anyon insertions: a Wilson line is a defect, and the associated gauge field
is considered on its complement with prescribed holonomy. Conditions on
total charge on a closed surface must be distinguished from the existence
of a local anyon species (cf. \cite[\S 3.1]{Witten16}).
The dynamical and background terms are:
\begin{subequations}
\label
{DecomposingTheLagrangian}
\begin{align}
  \label{DynamicalLagrangian}
  L_{\mathrm{dyn}}
  & =
  \tfrac{K}{2}
  a_{\mathrm{sur}} 
  \, 
  \mathrm{d}a_{\mathrm{sur}}
  - 
  a_{\mathrm{sur}} j
  \\
  \label{BackgroundLagrangian}
  L_{\mathrm{bckg}}
  & =
  \tfrac{K}{2}
  a_{\mathrm{bnd}} 
  \, 
  \mathrm{d}a_{\mathrm{bnd}}
  + 
  a_{\mathrm{bnd}} j
  \\
  L & \simeq 
  L_{\mathrm{dyn}}
  -
  L_{\mathrm{bckg}}
  \mathrlap{\,.}
\end{align}
\end{subequations}
When $j$ here is taken to be the distributional current concentrated on a knot/link, the path integral for the dynamical Lagrangian \cref{DynamicalLagrangian} computes (after the usual renormalization choices, cf. \cite[\S 2.2]{SS26-Renormalization}) the Wilson loop observable of this link, which expresses the corresponding anyon braiding phase.

We learn here that: 
\begin{standout}
With the chosen background root, the dynamical field may be written as the surplus current potential $a_{\mathrm{sur}}=a-a_{\mathrm{bnd}}$ in \cref{DecomposingTheLagrangian}. This is the background-dependent rewriting of the original effective field $a$ that motivates our comparison with surplus flux quantization.
\end{standout}

This finally relates to our conceptual starting point \cref{2SphereInsideBU1}, where Lagrangian Chern-Simons theory for surplus flux quanta is replaced by surplus flux quantization in 2-Cohomotopy.

\subsubsection
{Hierarchical K-Matrix Chern-Simons for FQH at other filling fractions}
\label
{KMatrixCSTheory}
The argument for the effective Lagrangian \eqref{EffectiveCSLagrangian}, which applies only to unit-fraction fillings $\nu = \sfrac{1}{\lattice}$, has been generalized (\cite{BlokWen90}\cite{WenZee92}, review in \cite[\S 2.1]{Wen95}\cite[\S 7.3.3]{Wen07}\cite[(5.31)]{Tong16}) by appeal to the {\it hierarchical} picture of FQH quasi-particles \cite{Haldane83}\cite{Halperin84}, which envisions that the quasi-particles appearing at some unit-fraction filling $1/K_1$ experience a secondary $K_2$-fractional quantum Hall effect, now with respect not to the actual magnetic gauge field but the effective gauge field $a^{(1)} := a$ \eqref{EffectiveCSLagrangian}. 

Repeating the above assumptions \eqref{EffectiveCSLagrangian} on effective FQH fields, this picture motivates the introduction of a {\it secondary} effective gauge field $a^{(2)}$, for the quasi-particle current $j$,
with its own Chern-Simons dynamics, and hence the generalization of the Lagrangian density \eqref{EffectiveCSLagrangian} to (\cite[p. 11]{Wen95}\cite[p. 300]{Wen07}):
\def\arraystretch{2}
\begin{align}    \label{SecondLayerEffectiveLagrangianDensity}
  L
  &
  :=
  \tfrac{K_1}{2}
  \grayunderbrace{
  a^{(1)} \, \mathrm{d} a^{(1)}
  }{
    \mathrm{CS}(a^{(1)})
  }
  \,-\, 
  \grayunderbrace{
  A \, \mathrm{d}a^{(1)}
  }{
    A\, J
  }
  \,-\,
  \grayunderbrace{
    a^{(1)} \, 
    \mathrm{d} a^{(2)}
  }{
    a j
  }
  \,+\,
  \tfrac{K_2}{2}
  \grayunderbrace{
    a^{(2)} \, \mathrm{d} a^{(2)}
  }{
    \mathrm{CS}(a^{(2)})
  }
  \\
  \label{KMatrixEffectiveLagrangian}
  & \simeq
  K_{i j}
  \, 
  \tfrac{1}{2}
  a^{(i)}\, \mathrm{d}a^{(j)}
  \,-\,
  Q_i
  \,
  A \, \mathrm{d}a^{(i)}
  \,,
\end{align}
where $\simeq$ denotes equality modulo a total derivative: the density
in \eqref{SecondLayerEffectiveLagrangianDensity} equals the expression in
\eqref{KMatrixEffectiveLagrangian} plus
$\tfrac12\,\mathrm{d}\big(a^{(1)}\wedge a^{(2)}\big)$.
On a spacetime with boundary, this contributes the corresponding
boundary term to the action. The {\it K-matrix}
$K\defneq(K_{ij})$ and {\it charge vector} $Q\defneq(Q_i)$ are
\begin{equation}
  \label{SeondLayerKMatrixAndChargeVector}
  K
  \,:=\,
  \left[
  \def\arraystretch{1.2}
  \begin{matrix}
    K_1 & -1
    \\
    -1 & K_2
  \end{matrix}
  \right]
  \,,
  \;\;\;\;\;\;\;\;
  Q \,:=\, 
  \left[
  \def\arraystretch{1.2}
  \begin{matrix}
    1 
    \\
    0
  \end{matrix}
  \right]
  \,.
\end{equation}
This is argued to be an effective description for the filling fraction
$$
  \nu 
  \;\defneq\;
  \frac{1}{
    K_1 - \tfrac{1}{K_2}
  }
  \;=\;
  Q^t K^{-1} Q
  \,.
$$
Proceeding in this manner it is envisioned that
$n$th-order FQH systems are effectively described by a nonsingular symmetric integer $n \times n$ matrix $K$ and an integer charge $n$-vector $Q$ via a Lagrangian density \eqref{SecondLayerEffectiveLagrangianDensity},  
\begin{itemize}

\item whose \textbf{filling fraction} is (cf. \cite[(2.31)]{Wen95}):
\begin{equation}
\label
{KMatrixFillingFraction}
  \nu 
    \;=\; 
  Q^t  K^{-1}  Q
  \mathrlap{\,,}
\end{equation}

\item whose \textbf{ground state degeneracy} on the closed surface of genus $g$ is argued to be (cf. \cite[(1.2)]{WenZee92}):
\begin{equation}
\label
{KMatrixGroundStateDegeneracy}
  \mathrm{dim}\big(
    \HilbertSpace{H}_{\Sigma^2_g}
  \big)
  \;=\;
  \big\vert 
    \mathrm{det}(K)
  \big\vert^g
\end{equation}

\item which host anyons that are labeled by weight vectors $W \in \mathbb{Z}^n$ and exhibit \textbf{braiding phases} (cf. \cite[(2.30)]{Wen95}):
\begin{equation}
\label
{KMatrixBraidingPhase}
  \zeta_{_W}
  \;=\;
  e^{ 
    \pi \mathrm{i} 
    \big(
      W^t  K^{-1} W
    \big)  
  }
  \mathrlap{\,.}
\end{equation}

Here the \emph{fundamental} anyon species are represented by the integer basis vectors $W_i$ (hence those with components $(W_i)_j = \delta_{i j}$), whose braiding phase among themselves is hence
\begin{equation}
\label
{KMatrixFundamentalBraidingPhase}
  \zeta_{_{W_i}}
  \;=\;
  e^{ 
    \pi \mathrm{i} 
    (K^{-1})_{ii}
  }
  \mathrlap{\,.}
\end{equation}
General weight vectors $W = \sum_i a_i W_i$ are thought to describe some kind of bound states of the fundamental anyon species $W_i$, with $a_i$-fold admixture of the species $W_i$.
\end{itemize}

\paragraph{Ground-state monodromy}
\label{KMatrixGroundStateMonodromy}
For a nonsingular symmetric integral $K$, the torus state space has a basis
indexed by $D_K:=\mathbb{Z}^n/K\mathbb{Z}^n$, with a spin structure fixed
when required.
If $K$ has even diagonal entries, its discriminant pairing and quadratic
form are
\begin{equation}
\label{KMatrixDiscriminantData}
  b_K([u],[v])=u^tK^{-1}v\;\mathrm{mod}\,1,
  \qquad
  q_K([u])=\tfrac12 u^tK^{-1}u\;\mathrm{mod}\,1.
\end{equation}
Writing $\sigma(K)$ for the signature, the generators
\eqref{ModularGenerators} act, in one linear normalization, by
\begin{equation}
\label{KMatrixModularOperators}
\begin{aligned}
  \widehat S\vert[u]\rangle
  & =\frac{1}{\sqrt{|D_K|}}
      \sum_{[v]\in D_K}e^{2\pi\mathrm{i}b_K([u],[v])}\vert[v]\rangle,
  \\
  \widehat T\vert[u]\rangle
  & =e^{-2\pi\mathrm{i}\sigma(K)/24}
      e^{2\pi\mathrm{i}q_K([u])}\vert[u]\rangle.
\end{aligned}
\end{equation}
Fourier orthogonality and the Gauss--Milgram identity give
$\widehat S^4=1$ and $(\widehat T\widehat S)^3=1$ in this generator
convention; moreover, $\widehat S^2\widehat T=\widehat T\widehat S^2$.
For odd lattices, spin structures and their quadratic refinements enter
the corresponding formulas (\cite[\S\S 5.6--6]{BelovMoore05}).
Thus Chern-Simons theory provides linear transformation data, including
scalar phases, in addition to dimensions and braiding angles.

For $K=(\lattice)$ with $\lattice>0$ even, \eqref{KMatrixModularOperators}
reduces to the modular operators in \eqref{IrrepOfSemidirectProduct} with
$\zeta=e^{\pi\mathrm{i}/\lattice}$ and $c_\zeta=e^{\pi\mathrm{i}/12}$.
Matching more general cases requires fixing the spin structure and
linear normalization. Ground-state transport for abelian FQH wavefunctions,
including scalar phases, is treated in \cite[\S V]{Wen12}.

\begin{remark}
[\bf Phenomenology of the K-matrix formulation]
\label
{PhenomenologyOfKMatrixFormulation}
\begin{itemize}
    \item[\bf (i)] 
The microscopic hierarchy picture and the composite-fermion picture have distinct motivations and make different microscopic claims (\cite[\S 12.1]{Jain07}\cite{Jain14}). Such questions concern the realization of a proposed effective theory; they do not invalidate the mathematical construction of its ground-state representations.

     \item[\bf (ii)] The rank of $K$ counts effective gauge fields and need not equal the number of microscopic electron species or physical layers. The formalism describes abelian orders in single-component systems as well as in systems with several layers or spin components (cf. \cite[\S II.E--G]{Wen12}\cite{BarkeshliWen17}\cite{Zeng21}\cite{HuLiuZhu23}\cite{Zeng24}).

    \item[\bf (iii)] A proposed $K$ and charge vector $Q$ specify an effective abelian theory. Establishing which material realizes that theory, or deriving it from a microscopic Hamiltonian, is an additional physical task. Our classification of flux-monodromy representations likewise requires a ground-state realization for comparison with experiment (\S\ref{GroundStateMonodromyAndTopologicalOrder}).
  \end{itemize}
\end{remark}


\subsection{Some category theory}
\label{SomeCategoryTheory}
\def\arraystretch{1}

Category theory is the {\it algebra} in {\it algebraic topology} (cf. \cite[\S 3]{Kroemer07}\cite[\S 2]{Marquis06} and \S\ref{SomeAlgebraicTopology}). 
We need only basics, but we do need the actual theory (namely adjunctions, hence universal constructions) and not just monoidal (braided-, fusion-, ...) categories regarded as algebraic structures themselves, as is now common in (topological) quantum theory (cf. \cite{HeunenVicary19}\cite{KongZhang22}).  Broader motivation of category theory for mathematical physicists is in \cite{Geroch85}; detailed lecture notes tailored towards our needs here are provided in \cite{Sc18}; for a further introduction we recommend \cite{AdamekHerrlichStrecker90}\cite{Awodey06}; standard monographs are \cite{Borceux94}\cite{MacLane97}.

\medskip

A {\it category} $\mathcal{C}$ is a class of {\it objects} $X$, $Y$, ... with prescribed sets $\mathrm{Hom}(X,\,Y)$ of (homo-){\it morphisms} (``homs'') between them, regarded abstractly as maps $f \,\colon\,X \to Y$ and ultimately defined by their composition law, 
\begin{equation}
  \label{CompositionOfMorphisms}
  (\mbox{-})\circ (\mbox{-})
  :
  \!\!
  \begin{tikzcd}
    \mathrm{Hom}(Y,Z)
    \times
    \mathrm{Hom}(X,Y)
    \to
    \mathrm{Hom}(X,Z)
  \end{tikzcd}
  \,,
  \;\;\;
  \begin{tikzcd}[
    column sep=15pt
  ]
    X
    \ar[
      r, 
      shorten=-2pt,
      "{ f }"
    ]
    \ar[
      rr,
      rounded corners,
      to path={
           ([yshift=+00pt]\tikztostart.south)      
        -- ([yshift=-10pt]\tikztostart.south)      
        --  node[yshift=6pt] {
            \scalebox{.8}{$
              g \circ f
            $}
        }
           ([yshift=-10pt]\tikztotarget.south)      
        -- ([yshift=+00pt]\tikztotarget.south)      
      }
    ]
    &
    Y
    \ar[
      r, 
      shorten=-2pt,
      "{ g }"  
    ]
    &
    Z
    \smash{\mathrlap{\,,}}
  \end{tikzcd}
\end{equation}
which is required to be associative and unital in the evident sense. 

The archetypical example is the category $\mathrm{Set}$ of sets, but also each individual set $S$ may be equivalently regarded as a category whose only homs happen to be identities:
\begin{equation}
  \label{CategoryOfSets}
  \mathrm{Set}
  \,=\,
  \left\{
  \hspace{-7pt}
  \scalebox{.9}{
    \def\tabcolsep{2pt}
    \begin{tabular}{rl}
      {\color{gray}objects:} & sets
      \\
      {\color{gray}homs:} & functions
    \end{tabular}
  }
  \hspace{-5pt}
  \right\}
  \,,
  \hspace{1.4cm}
  \mathrm{Set}\;  \ni \;S
  \;=\;
  \left\{
  \hspace{-7pt}
  \scalebox{.9}{
    \def\tabcolsep{2pt}
    \begin{tabular}{rl}
      {\color{gray}objects:} & elements
      \\
      {\color{gray}homs:} & trivial
    \end{tabular}
  }
  \hspace{-5pt}
  \right\}
  \,.
\end{equation}
We use this occasion to highlight for lay readers the distinction between the notation ``$\longrightarrow$'' for maps (and generally for homs) between sets/objects, in contrast to the notation ``\,$\longmapsto$\,'' for the corresponding {\it assignments} between (generalized) elements:
$$
  \begin{tikzcd}[sep=0pt]
    S
    \ar[
      rr,
      "{ f }"
    ]
    &&
    S'
    \\
    \mathllap{
      S \ni
      \;
    }
    s &\longmapsto& f(s)
    \mathrlap{
      \; \in S'
      \,.
    }
  \end{tikzcd}
$$
More generally, there are the ``concrete'' categories (cf. \cite[\S I.5]{AdamekHerrlichStrecker90}) of sets with algebraic structure (vector spaces, groups, algebras, representations, modules, ...),
with their structure-preserving functions between them (the concrete {\it homomorphisms}), e.g.:
\begin{equation}
  \label{SomeConcreteCategories}
  \mathrm{Vec}
  \,=\,
  \left\{
  \hspace{-7pt}
  \scalebox{.9}{
    \def\tabcolsep{2pt}
    \begin{tabular}{rl}
      {\color{gray}objects:} & vector spaces
      \\
      {\color{gray}homs:} & linear maps
    \end{tabular}
  }
  \hspace{-5pt}
  \right\}
  \,,\;\;
  \mathrm{Grp}
  \,=\,
  \left\{
  \hspace{-7pt}
  \scalebox{.9}{
    \def\tabcolsep{2pt}
    \begin{tabular}{rl}
      {\color{gray}objects:} & groups
      \\
      {\color{gray}homs:} & homomorphisms
    \end{tabular}
  }
  \hspace{-5pt}
  \right\}
  \,,\;\;
  G \mathrm{Rep}
  \,=\,
  \left\{
  \hspace{-7pt}
  \scalebox{.9}{
    \def\tabcolsep{2pt}
    \begin{tabular}{rl}
      {\color{gray}objects:} & representations
      \\
      {\color{gray}homs:} & intertwiners
    \end{tabular}
  }
  \hspace{-5pt}
  \right\}
  .
\end{equation}
But general categories may have homs without underlying functions of sets. For instance, given a group $G$, its {\it delooping groupoid} is the category $\mathbf{B}G$ with a single object $\ast$, with homs $\ast \xrightarrow{g} \ast$ labeled by group elements, and composition being the group operation \eqref{DeloopingGroupoid}; while for $X$ a topological space, there is the category called its {\it fundamental groupoid} $\Pi_1(X)$ (cf. \cite[(1.7)]{Whitehead78}) whose homs are the homotopy classes (fixing endpoints) of continuous paths in $X$ with composition the concatenation of paths:
\begin{equation}
  \label{DeloopingGroupoid}
  \mathbf{B}G
  \,=\,
  \left\{
  \hspace{-7pt}
  \scalebox{.9}{
    \def\tabcolsep{2pt}
    \begin{tabular}{rl}
      {\color{gray}objects:} & single one
      \\
      {\color{gray}homs:} & group elements
    \end{tabular}
  }
  \hspace{-5pt}
  \right\}
  \,,
  \hspace{1cm}
  \Pi_1(X)
  \,=\,
  \left\{
  \hspace{-7pt}
  \scalebox{.9}{
    \def\tabcolsep{2pt}
    \begin{tabular}{rl}
      {\color{gray}objects:} & points of $X$
      \\
      {\color{gray}homs:} & 
      htmpy classes of paths in $X$
    \end{tabular}
  }
  \hspace{-5pt}
  \right\}
  \,.
\end{equation}

Generally, a (small) category all of whose homs are invertible is called a {\it groupoid} (cf. \cite{IbortRodriguez21}). Here invertible homs are like gauge transformations in that they identify the objects they relate, while retaining the information of {\it how} the identification happens, whence it is fruitful to think of groupoids as sets of {\it elements with gauge transformations} between them (cf. \cite[pp. 6]{Sc25}):
\begin{equation}
  \label{GroupoidAsSetsOfelementsWithGauge}
  \mbox{generic groupoid}
  \;=\;
  \left\{
  \hspace{-7pt}
  \scalebox{.9}{
    \def\tabcolsep{2pt}
    \begin{tabular}{rl}
      {\color{gray}objects:} & elements
      \\
      {\color{gray}homs:} & gauge transformations
    \end{tabular}
  }
  \hspace{-5pt}
  \right\}
  \,.
\end{equation}

\medskip

Next, a {\it functor} $F : \mathcal{C} \xrightarrow{\;} \mathcal{D}$ between a pair of categories is a function between objects and between sets of morphisms which respects this compositional structure:
$$
  \begin{tikzcd}[sep=0pt]
    \mathcal{C}
    \ar[rr, "{ F }"]
    &&
    \mathcal{D}
    \\
    X
    \ar[d, "{f}"]
    \ar[
      dd,
      rounded corners,
      to path={
           ([xshift=-00pt]\tikztostart.west)
        -- ([xshift=-10pt]\tikztostart.west)
        -- node[]
           {\scalebox{.7}{\colorbox{white}{$g \circ f$}}}
           ([xshift=-10pt]\tikztotarget.west)
        -- ([xshift=-00pt]\tikztotarget.west)
      }
    ]
    &\longmapsto&
    F(X)
    \ar[d, "{F(f)}"]
    \ar[
      dd,
      rounded corners,
      to path={
           ([xshift=+00pt]\tikztostart.east)
        -- ([xshift=+14pt]\tikztostart.east)
        -- node[]
           {\scalebox{.7}{\colorbox{white}{$F(g \circ f)$}}}
           ([xshift=+14pt]\tikztotarget.east)
        -- ([xshift=+00pt]\tikztotarget.east)
      }
    ]
    \\[+18pt]
    Y
    \ar[d, "{g}"]
    &\longmapsto&
    F(Y)
    \ar[d, "{F(g)}"]
    \\[+18pt]
    Z
    &\longmapsto&
    F(Z)
  \end{tikzcd}
$$

There is an evident composition of functors, which is evidently associative and unital, whence categories with functors between them form a (``very large'') category:
\begin{equation}
  \label{CategoryOfCategories}
  \mathrm{CAT}
  \,=\,
  \left\{\!\!\!
  \scalebox{.9}{
    \def\tabcolsep{2pt}
    \begin{tabular}{rl}
      {\color{gray}objects:} & categories
      \\
      {\color{gray}homs:} & functors
    \end{tabular}
  }
  \!\right\}
  \mathrlap{\,.}
\end{equation}

\medskip

A {\it natural transformation} $\alpha : F \Rightarrow G$ between a pair of parallel functors $F,G : \mathcal{C} \rightrightarrows \mathcal{D}$ is for each object $X$ of $\mathcal{C}$ a morphism $F(X) \xrightarrow{\alpha(X)} G(X)$ of $\mathcal{D}$ such that the following squares ``commute'', meaning that the two possible diagonal composites of morphisms coincide:
\vspace{-2mm} 
\begin{equation}
  \label{NaturalTransformation}
  \begin{tikzcd}[sep=0pt]
    \mathcal{C}
    \ar[
      rrr,
      bend left=25,
      "{ F }"{description, name=s},
    ]
    \ar[
      rrr,
      bend right=25,
      "{ G }"{description, name=t},
    ]
    \ar[
      from=s,
      to=t,
      Rightarrow,
      shorten=3pt,
      "{\; \alpha }"
    ]
    &&
    &
    \mathcal{D}
    \\[15pt]
    X
    \ar[
      d,
      "{ f }"
    ]
    &\mapsto&
    F(X)
    \ar[d, "{ F(f) }"]
    \ar[
      rr,
      "{ \alpha(X)  }"
    ]
    &\phantom{--}&
    G(X)
    \ar[d, "{ G(f) }"]
    \\[24pt]
    Y
    &\mapsto&
    F(Y)
    \ar[
      rr,
      "{ \alpha(Y)  }"
    ]
    &&
    G(Y)
    \,.
  \end{tikzcd}
\end{equation}
There is an  evident ``vertical'' composition of such natural transformations $\begin{tikzcd}[column sep=12pt] F \ar[r, Rightarrow, "{ \alpha }"]
& G \ar[r, Rightarrow, "{ \beta }"] & 
H\end{tikzcd}$ which is evidently associative and unital, and hence functors $\mathcal{C} \xrightarrow{\;} \mathcal{D}$ with natural transformations between them form a category, called the {\it functor category}:
\begin{equation}
  \label{FunctorCategory}
  \mathcal{D}^{\mathcal{C}}
  \,=\,
  \left\{\!\!\!
  \scalebox{.9}{
    \def\tabcolsep{2pt}
    \begin{tabular}{rl}
      {\color{gray}objects:} & functors $\mathcal{C} \xrightarrow{\;} \mathcal{D}$
      \\
      {\color{gray}homs:} & natural transformations
    \end{tabular}
  }
 \!\! \right\}
  \mathrlap{\,.}
\end{equation}
A good example of the interplay of these notions is the observation
that the functor category from a delooping groupoid \eqref{DeloopingGroupoid} to the category of vector spaces \eqref{SomeConcreteCategories} is the category of group representations (cf. \S\ref{SomeRepresentationTheory}):
\begin{equation}
  \label{RepAsFunc}
  G \mathrm{Rep}
  \;\;
  \simeq
  \;\;
  \mathrm{Func}\big(
    \mathbf{B}G
    ,\,
    \mathrm{Vec}
  \big)
  \,.
\end{equation}

\medskip

\noindent
{\bf Adjunctions.}
These notions (categories, functors and natural transformations) famously constitute the substrate of category theory --- but what makes it a theory, with non-trivial theorems, is the further notion of {\it adjunctions}.

First, note that there is also a ``horizontal'' composition of natural transformations by functors:
$$
  \begin{tikzcd}[
    row sep=10pt,
    column sep=40pt
  ]
    \mathcal{C}'
    \ar[r, "L"]
    &
    \mathcal{C}
    \ar[
      rr,
      bend left=20,
      "{ F }"{description, name=s},
    ]
    \ar[
      rr,
      bend right=20,
      "{ G }"{description, name=t},
    ]
    \ar[
      from=s,
      to=t,
      Rightarrow,
      shorten=3pt,
      "{\; \alpha }"
    ]
    &&
    \mathcal{D}
    \ar[r, "R"]
    &
    \mathcal{D}'
    &
    \\
    X'
    &
    {}
    \ar[r, |->]
    &
    {}
    &
    {}&
    \mathclap{
      R \!\circ\! F \!\circ\! L(X')
      \xrightarrow{
        R \circ \alpha \circ L(X')
      }
      R \!\circ\! G \!\circ\! L(X')
      \mathrlap{\,.}
    }
  \end{tikzcd}
$$

Now, an {\it adjunction} $L \dashv R$ between a pair of back-and-forth functors, $L \colon \mathcal{C} \rightleftarrows \mathcal{D} \colon R$, is a pair of  natural transformations 
\begin{equation}
  \label{UnitAndCounit}
  \underset
  {
    \mathclap{
      \adjustbox{
        scale=.7,
        raise=-3pt
      }{``the unit''}
    }
  }
  {
    \unit{}{}
  }
  : 
  \mathrm{id}_{\mathcal{C}} \Rightarrow R\circ L 
  \hspace{1cm}
  \mbox{and}
  \hspace{1cm}
  \underset{
    \mathclap{
      \adjustbox{
        scale=.7,
        raise=-3pt
      }{``the co-unit''}
    }
  }
  {
    \counit{}{}
  }
  : 
  L \circ R \Rightarrow \mathrm{id}_{\mathcal{D}}
\end{equation}
such that:
\begin{equation}
  \begin{tikzcd}[
    column sep=35pt,
    row sep=10pt
  ]
    {}
    &
    {}
    \ar[
      d,
      Rightarrow,
      "{ \unit{}{} }"
    ]
    &
    {}
    &
    {}
    \\
    \mathcal{C}
    \ar[
      r, 
      "{ L }"{description}
    ]
    \ar[
      rr,
      bend left=50,
      "{ \mathrm{id}_{{}_{\mathcal{C}}} }"{description}
    ]
    &
    \mathcal{D}
    \ar[
      r, 
      "{ R }"{description}
    ]
    \ar[
      rr,
      bend right=50,
      "{ \mathrm{id}_{{}_{\mathcal{D}}} }"{description}
    ]
    &
    \mathcal{C}
    \ar[
      d,
      Rightarrow,
      "{ \counit{}{} }"
    ]
    \ar[
      r, 
      "{ L }"{description}
    ]
    &
    \mathcal{D}
    \\
    {} & {} & {} & {}
  \end{tikzcd}
  \;\;\;\;\;
    =
  \;\;\;\;\;
  \begin{tikzcd}[row sep=10pt]
    {} & 
    {} 
    \ar[
      dd,
      Rightarrow,
      shorten=10pt,
      "{ \mathrm{id}_{{}_{L}} }"
    ]
    & {}
    \\
    \mathcal{C}
    \ar[
      rr,
      bend left=35pt,
      "{
        L
      }"{description}
    ]
    \ar[
      rr,
      bend right=35pt,
      "{
        L
      }"{description}
    ]
    &&
    \mathcal{D}
    \\
    {} & {} & {}
  \end{tikzcd}
\end{equation}
and
\begin{equation}
  \begin{tikzcd}[
    column sep=35pt,
    row sep=10pt
  ]
    {}
    &
    {}
    &
    {}
    \ar[
      d,
      Rightarrow,
      "{ \unit{}{} }"
    ]
    &
    {}
    \\
    \mathcal{D}
    \ar[
      r, 
      "{ R }"{description}
    ]
    \ar[
      rr,
      bend right=50,
      "{ \mathrm{id}_{{}_{\mathcal{D}}} }"{description}
    ]
    &
    \mathcal{C}
    \ar[
      r, 
      "{ L }"{description}
    ]
    \ar[
      rr,
      bend left=50,
      "{ \mathrm{id}_{{}_{\mathcal{C}}} }"{description}
    ]
    \ar[
      d,
      Rightarrow,
      "{ \counit{}{} }"
    ]
    &
    \mathcal{D}
    \ar[
      r, 
      "{ R }"{description}
    ]
    &
    \mathcal{C}
    \\
    {} & {} & {} & {}
  \end{tikzcd}
  \;\;\;\;\;
    =
  \;\;\;\;\;
  \begin{tikzcd}[row sep=10pt]
    {} & 
    {} 
    \ar[
      dd,
      Rightarrow,
      shorten=10pt,
      "{ \mathrm{id}_{{}_{R}} }"
    ]
    & {}
    \\
    \mathcal{D}
    \ar[
      rr,
      bend left=35pt,
      "{
        R
      }"{description}
    ]
    \ar[
      rr,
      bend right=35pt,
      "{
        R
      }"{description}
    ]
    &&
    \mathcal{C}
    \smash{\mathrlap{\,.}}
    \\
    {} & {} & {}
  \end{tikzcd}
\end{equation}
Remarkably, if an adjunction exists then it is essentially unique and equivalent to there being a natural bijection $\widetilde{(-)}$ (``forming adjuncts'') between hom-sets, of this form:
$$
  \begin{tikzcd}[
    row sep=30pt
  ]
    X
    \ar[
      from=d,
      "{ g }"{swap}
    ]
    &[-20pt]
    Y
    \ar[
      d,
      "{ f }"
    ]
    &\mapsto&
    \mathrm{Hom}\big(
      L(X)
      ,\,
      Y
    \big)
    \ar[
      rr,
      <->,
      "{ \widetilde{(-)}_{X,Y} }"
    ]
    \ar[
      d,
      "{
         f \circ (-) \circ L(g)
      }"{description}
    ]
    &&
    \mathrm{Hom}\big(
      X
      ,\,
      R(Y)
    \big)
    \ar[
      d,
      "{
         R(f) \circ (-) \circ g
      }"{description}
    ]
    \\
    X'
    &
    Y'
    &\mapsto&
    \mathrm{Hom}\big(
      L(X')
      ,\,
      Y'
    \big)
    \ar[
      rr,
      <->,
      "{ \widetilde{(-)}_{X',Y'} }"
    ]
    &&
    \mathrm{Hom}\big(
      X'
      ,\,
      R(Y')
    \big)
  \end{tikzcd}
$$
(whence the terminological allusion to adjoint operators), and under this identification the (co)unit is the adjunct of the identity:
$$
  \unit{}{X}
  \;=\;
  \widetilde{
    \mathrm{id}_{L(X)}
  }
  \,,
  \hspace{1cm}
  \counit{}{Y}
  \;=\;
  \widetilde{
    \mathrm{id}_{R(Y)}
  }
  \,.
$$

\smallskip

Finally, an (adjoint-){\it equivalence} between categories, $\mathcal{C} \,\simeq\, \mathcal{D}$ is an adjunction where both the unit and the counit are invertible.

\begin{example}[{\bf Homotopy fibers and cosets}, {\cite[Def. 1.14, Ex. 1.12]{FSS23-Char}}]
  For $F : \hotype{X} \xrightarrow{} \hotype{Y}$ a functor between (small) groupoids and $y$ an object of $\hotype{Y}$, 
  the {\it homotopy fiber} of $F$ at $y$ 
  is  the groupoid 
  $\mathrm{hofib}_y(F)$
  whose objects are pairs $(x,c_x)$ with $x\in\hotype{X}$ and
  $c_x:F(x)\to y$, and whose morphisms from $(x,c_x)$ to $(x',c_{x'})$
  are morphisms $f:x\to x'$ satisfying $c_x=c_{x'}\circ F(f)$.
  Thus its morphisms are the following commuting triangles:
  \begin{equation}
    \label{HomotopyFiberViaResolution}
    \mathrm{hofib}_y(F)
    \;=\;
    \left\{\!\!
    \begin{tikzcd}[row sep=0pt]
       x \ar[rr, "{f}"{description}] & & x'
       \\
       F(x) 
         \ar[dr, "{  c_x }"{description}]
         \ar[rr, "{F(f)}"{description}] & & 
       F(x')
         \ar[dl, "{ c_{x\prime} }"{description}]
       \\[10pt]
       & y
    \end{tikzcd}
  \!  \right\}
    \,.
  \end{equation}
  
  For $\iota : H \xhookrightarrow{\;} G$ a subgroup inclusion, the homotopy fiber of $\mathbf{B} \iota : \mathbf{B}H \xrightarrow{\;} \mathbf{B}G$ \eqref{DeloopingGroupoid} is equivalent to the quotient set $G/H$:
  \begin{equation}
  \mathrm{hofib}_\ast(\mathbf{B}\iota)
  \;\defneq\;
  \left\{\!\!
  \begin{tikzcd}[row sep=10pt]
    \ast 
    \ar[rr, "{ h }"]
    \ar[dr, "{ g }"{description}]
      &&
    \ast
    \ar[dl, "{ g' }"{description}]
    \\
    &
    \ast
  \end{tikzcd}
  \,\bigg\vert\,
    \def\arraystretch{.9}
    \begin{array}{c}
      g,g' \in G
      \\
      h \in H
    \end{array}
\!\!  \right\}
    \;\;
    \simeq
    \;\;
    \big\{
      g \cdot H
      \,\big\vert\,
      g \in G
    \big\}
    \;=\;
    G/H
    \,.
  \end{equation}
\end{example}

The key example of adjunctions
appearing in the main text is the following classical phenomenon of representation theory (cf. \S\ref{SomeRepresentationTheory}), there also known as {\it Frobenius reciprocity}:

\begin{proposition}[{\bf Induced representations}, {cf. \cite{Hristova19}\cite{nlab-InducedRepresentation}}]
\label{InducedRepresentations}
Here $G\mathrm{Rep}_{\mathbb{C}}$ denotes all algebraic complex representations, allowing vector spaces of arbitrary dimension. For finite-index subgroup inclusions, induction and coinduction preserve the finite-dimensional subcategories.
\begin{itemize}
    \item[\bf (i)] 
For $H \xhookrightarrow{\iota} G$ a subgroup inclusion, the restriction functor $\iota^\ast \,:\, G \mathrm{Rep}_{\mathbb{C}} \xrightarrow{\;} H \mathrm{Rep}_{\mathbb{C}}$ has 
\begin{itemize}
\item[--] a left adjoint given by
\begin{equation}
  \label{LeftInducedRepresentation}
  \HilbertSpace{V}
  \,\in\,
  H\mathrm{Rep}_{\mathbb{C}}
  \;\;\;\;\;\;\;\;\;
  \vdash
  \;\;\;\;\;\;\;\;\;
  \iota_! \HilbertSpace{V}
  \;\;
  :=
  \;\;
  \mathbb{C}[G]
  \otimes_{{}_{\mathbb{C}[H]}}
  \HilbertSpace{V}
\end{equation}
where on the right we have the tensor product of (left-with-right) $\mathbb{C}[H]$-modules equipped with the $G$-action given by
$$
  \left.
  \begin{array}{l}
    {[a,v]} \,\in\,
    \mathbb{C}[G]
    \otimes_{{}_{\mathbb{C}[H]}}
    \HilbertSpace{V}
    \\
    g \,\in\, G
  \end{array}
 \!\! \right\}
  \hspace{.7cm}
  \vdash
  \hspace{.7cm}
  g \cdot [a,\, v]
  \;:=\;
  [g\cdot a,\, v]
  \,.
$$
\item[--] a right adjoint given by
\begin{equation}
  \label{RightInducedRepresentation}
  \HilbertSpace{V}
  \,\in\,
  H\mathrm{Rep}_{\mathbb{C}}
  \;\;\;\;\;\;\;\;\;
  \vdash
  \;\;\;\;\;\;\;\;\;
  \iota_\ast \HilbertSpace{V}
  \;\;
  :=
  \;\;
  \mathrm{hom}_{{}_{\mathbb{C}[H]}}
  \big(
    \mathbb{C}[G]
    ,\,
    \HilbertSpace{V}
  \big)
  \,,
\end{equation}
where on the right we have the vector space of left $\mathbb{C}[H]$-module homomorphisms equipped with the action
$$
  \left.
  \begin{array}{l}
    f \,\in\,
    \mathrm{hom}_{{}_{\mathbb{C}[H]}}
    \big(
      \mathbb{C}[G]
      ,\,
      \HilbertSpace{V}
    \big)
    \\
    g \,\in\, G
    \\
    a \,\in\,
    \mathbb{C}[G]
  \end{array}
 \!\! \right\}
  \hspace{.7cm}
  \vdash
  \hspace{.7cm}
  (g \cdot f)(a)
  \;:=\;
  f(a \cdot g)
  \,.
$$
\end{itemize}
making an adjoint triple of functors
\begin{equation}
  \label{InducedRepsAdjointTriple}
  \begin{tikzcd}
    H\mathrm{Rep}_{\mathbb{C}}
    \ar[
      rr,
      shift left=14pt,
      "{
        \iota_!
      }"{description}
    ]
    \ar[
      rr,
      shift left=8pt,
      phantom,
      "{ \bot }"{scale=.7}
    ]
    \ar[
      from=rr,
      "{
        \iota^\ast
      }"{description}
    ]
    \ar[
      rr,
      shift right=7pt,
      phantom,
      "{ \bot }"{scale=.7}
    ]
    \ar[
      rr,
      shift right=14pt,
      "{
        \iota_\ast
      }"{description}
    ]
    &&
    G\mathrm{Rep}_{\mathbb{C}}
    \,.
  \end{tikzcd}
\end{equation}
\item[\bf (ii)] Moreover, 
if $H$ has finite index in $G$, then the left and right adjoints are naturally isomorphic:
$$
  \mathllap{
  [G:H] \,:=\, \vert G/H\vert
  \,<\, \infty
  \;\;\;\;\;\;\;\;\;
  \vdash
  \;\;\;\;\;\;\;\;\;
  }
  \iota_! \,\simeq\, \iota_\ast
  \,.
$$
\end{itemize}
\end{proposition}

\smallskip 
\begin{remark}[\bf CoUnit of induced representations]
 \label{CoUnitOfInducedRepresentations}
  In more detail, inspection shows (cf. \cite{nlab-InducedRepresentation}) that 
  the unit of the left induced representation \eqref{LeftInducedRepresentation} is given by ``inserting'' the neutral element
  $$
    \begin{tikzcd}[sep=0pt]
      \HilbertSpace{V}
      \ar[
        rr,
        "{
          \unit{}{\HilbertSpace{V}}
        }"
      ]
      &&
      \mathbb{C}[G]
      \otimes_{{}_{\mathbb{C}[H]}}
      \HilbertSpace{V}
      \\[-2pt]
      v &\longmapsto& 
      {[\mathrm{e},v]}
      \,.
    \end{tikzcd}
  $$
  and the counit of the right induced representation 
  \eqref{RightInducedRepresentation} is given by evaluating at the neutral element:
  $$
    \begin{tikzcd}[sep=0pt]
      \mathrm{hom}_{{}_{\mathbb{C}[H]}}
      \big(
        \mathbb{C}[G]
        ,\,
        \HilbertSpace{V}
      \big)
      \ar[
        rr,
        "{
          \counit{}{\HilbertSpace{V}}
        }"
      ]
      &&
      \HilbertSpace{V}
      \\[-2pt]
      f &\longmapsto& f(\mathrm{e})
      \mathrlap{\,.}
    \end{tikzcd}
  $$
\end{remark}

But under the equivalence \eqref{RepAsFunc}, this adjoint triple \eqref{InducedRepsAdjointTriple} is of the form 
$\begin{tikzcd}
  \mathrm{Func}\big(
    \mathbf{B}H
    ,\,
    \mathrm{Vec}
  \big)
  \ar[r, shift left=8pt]
  \ar[r, shift left=4pt, phantom, "{\bot}"{scale=.6}]
  \ar[from=r]
  \ar[r, shift right=4pt, phantom, "{\bot}"{scale=.6}]
  \ar[r, shift right=8pt]
  &
  \mathrm{Func}\big(
    \mathbf{B}G
    ,\,
    \mathrm{Vec}
  \big)
\end{tikzcd}$
and this is an example of an extremely general kind of adjunctions known as (left and right) {\it Kan extensions}  or {\it base change} between categories of (higher) local systems, which we just state in the form needed in the main text:

\begin{proposition}[{\bf Base change for local systems}, {cf. \cite[Ex. A.18]{SS23-EoS}}]
\label{BaseChangeForLocalSystems}
For $F : \hotype{X} \xrightarrow{\;} \hotype{Y}$ a functor between {\rm (small)} groupoids, its induced precomposition functor on functor categories \eqref{FunctorCategory}
$$
  F^\ast : 
    \big(
      \hotype{Y} 
        \xrightarrow{\;} 
      \mathrm{Vec}_{\mathbb{C}}
    \big) 
    \;\longmapsto\; 
    \big(
      \hotype{X}
        \xrightarrow{F}
      \hotype{Y} 
        \xrightarrow{\;} 
      \mathrm{Vec}_{\mathbb{C}}
    \big) 
$$
has both a left and a right adjoint:
\begin{equation}
  \label{BaseChangeTriple}
  \begin{tikzcd}
    \mathrm{Vec}
      _{\mathbb{C}}
      ^{\hotype{X}}
    \;\;
    \ar[
      rr,
      shift left=16pt,
      "{
        F_!
      }"{description}
    ]
    \ar[
      rr,
      shift left=9pt,
      phantom,
      "{ \bot }"{scale=.7}
    ]
    \ar[
      from=rr,
      "{
        F^\ast
      }"{description}
    ]
    \ar[
      rr,
      shift right=8pt,
      phantom,
      "{ \bot }"{scale=.7}
    ]
    \ar[
      rr,
      shift right=16pt,
      "{
        F_\ast
      }"{description}
    ]
    &&
    \;\;
    \mathrm{Vec}
      _{\mathbb{C}}
      ^{\hotype{Y}}
    \,,
  \end{tikzcd}
\end{equation}
which agree when the homotopy fibers \eqref{HomotopyFiberViaResolution} of $F$ are finite sets.
\end{proposition}

\medskip

\subsection{Some algebraic topology}
\label{SomeAlgebraicTopology}

Algebraic topology is the study of topological spaces {\it up to homotopy} by algebraic means, namely with tools of category theory 
(cf. \cite[\S 3]{Kroemer07}\cite[\S 2]{Marquis06} and \S\ref{SomeCategoryTheory}).
General background on algebraic topology and homotopy theory may be found in \cite{STHu59}\cite{Giblin77}\cite{Whitehead78}\cite{James84}\cite{AguilarGitlerPrieto02} \cite{Strom11}\cite[\S 1]{FSS23-Char}.

\medskip

\noindent
{\bf Topological spaces.}
We write 

\begin{itemize}[
  itemsep=2pt
]
\item $\mathrm{Top}$ for the category of {\it compactly generated} topological spaces (cf. \cite[Ntn. 1.0.16]{SS25-EBund})

with mapping spaces (cf. \cite[\S 1]{AguilarGitlerPrieto02}) denoted $\mathrm{Map}(-,-)$ and their underlying (hom-)sets denoted $\mathrm{Hom}(-,-)$, 

\item $\mathrm{Top}^\ast$ for pointed such spaces with pointed maps between them

with mapping spaces denoted $\mathrm{Map}^\ast(-,-)$ and their underlying (hom-)sets denoted $\mathrm{Hom}^\ast(-,-)$.
\end{itemize}

The mapping spaces are characterized by natural homeomorphisms (cf. \cite[(3.98)]{James84}\cite[Thm. 3.47(a)]{Strom11})
\begin{equation}
  \label{MappingSpaceAdjunction}
  \def\arraystretch{1.5}
  \begin{array}{ccc}
  \mathrm{Map}\big(
    X \times Y
    ,\,
    Z
  \big)
  &\simeq&
  \mathrm{Map}\big(
    X
    ,\,
    \mathrm{Map}(Y,Z)
  \big)
  \\
  \mathrm{Map}^\ast\big(
    X \wedge Y
    ,\,
    Z
  \big)
  &\simeq&
  \mathrm{Map}^\ast\big(
    X
    ,\,
    \mathrm{Map}^\ast(Y,Z)
  \big)
  \mathrlap{\,,}
  \end{array}
\end{equation}
where for a space $X$ pointed by $\infty_X \in X$ and a space $Y$ pointed by $\infty_Y \in Y$ 

\begin{itemize}
\item[(i)]
their {\it smash product} is
\begin{equation}
  \label{SmashProduct}
  X \wedge Y
  \;:=\;
  \frac
    { X \times Y }
    {
      \{\infty_{{}_X}\} \!\times\! Y
      \;\cup\;
      X \!\times\! \{\infty_{{}_Y}\}
    }
  \,,
\end{equation}
which is symmetric via natural homeomorphisms
\begin{equation}
  \label{SymmetryOfSmashProduct}
  X \wedge Y 
  \;\simeq\;
  Y \wedge X
  \,;
\end{equation}
for instance:
\begin{equation}
  \label{SpheresFromSmashProducts}
  S^1 \wedge S^n \,\simeq\, S^{n+1}
  \,,
  \;\;\;\mbox{so that}\;\;\;
  \pi_{n}
  \mathrm{Map}^\ast(S^{m}, X)
  \;\simeq\;
  \pi_0
  \mathrm{Map}^\ast(S^{n+m}, X)
  \;\simeq\;
  \pi_{n+m}
  (X)
\end{equation}
(here the smash product with the circle is called reduced {\it suspension} and usually denoted $\Sigma := S^1 \wedge(-)$, but we stick with writing ``$S^1 \wedge$'' in order not to clash with our use of ``$\Sigma^2$'' for the generic surface);

\item[(ii)]
their pointed mapping space is the fiber over the base point of $Y$ of the map $\mathrm{ev}$ that evaluates unpointed maps at the base point of $X$:
\begin{equation}
  \label{EvaluationSequence}
  \begin{tikzcd}[
    column sep=40pt
  ]
    \mathrm{Map}^\ast(X,Y)    
    \ar[r, "{ \mathrm{fib}_{{}_{(\infty_Y)}} }"]
    &
    \mathrm{Map}(X,Y)
    \ar[r, "{ \mathrm{ev}_{{}_{(\infty_X)}} }"]
    &
    Y
    \mathrlap{\,.}
  \end{tikzcd}
\end{equation}
\end{itemize}

The coproduct of $ X, Y \in \mathrm{Top}^\ast$ is the {\it wedge sum}
\begin{equation}
  \label{WedgeSum}
  X \vee Y
  \;:=\;
  \frac{
    X \coprod Y
  }{
    \{\infty_{{}_X},\, \infty_{{}_Y}\}
  }
  \,,
\end{equation}
which in particular means that we have a natural bijection
\begin{equation}
  \label{MapsOutOfWedgeSum}
  \mathrm{Hom}^\ast\big(
    X \vee Y
    ,\,
    Z
  \big)
  \;\simeq\;
  \mathrm{Hom}^\ast\big(
    X 
    ,\,
    Z
  \big)  
  \times
  \mathrm{Hom}^\ast\big(
    Y
    ,\,
    Z
  \big)  
  \,.
\end{equation}
Incidentally, the smash product \eqref{SmashProduct}
naturally distributes over finite wedge sums
$$
  X \wedge \big(
    Y \vee Z
  \big)
  \;\simeq\;
  \big(
    X \wedge Y
  \big)
  \;\vee\;
  \big(
    X \wedge Z
  \big)
  \,.
$$

\medskip

\noindent
{\bf One-point compactification.}
Here for  $X \in \mathrm{Top}^\ast$ we generically denote its basepoint by $\infty_{{}_X} \in X$, also speaking of the ``point at infinity'', and for $X$ a locally compact Hausdorff space we write $X_{\cpt} \in \mathrm{Top}^\ast$ for its {\it one-point compactification} (cf. \cite[p. 199]{Bredon93}), thinking of it as {\it adjoining a point at infinity}.

For example, stereographic projection gives
\begin{equation}
  \label{CompactificationOfRn}
  \mathbb{R}^n_{\cpt}
  \;\simeq\;
  S^n
\end{equation}
and 
if a space is already compact, then the adjoined point-at-infinity is disjoint and pointed maps out of the space are identified with plain maps:
\begin{equation}
  \label{UnpointedMapsFromPointed}
  \mbox{$X$ compact Hausdorff}
  \hspace{.8cm}
  \yields
  \hspace{.8cm}
  X_{\cpt}
  \simeq
  X \sqcup \{\infty\}
  \;\;\;\;\;\;
  \mbox{and}
  \;\;\;\;\;\;
  \mathrm{Maps}^\ast\big(
    X_{\cpt},\,
    Y
  \big)
  \;\simeq\;
  \mathrm{Maps}(
    X,\,
    Y
  )\,.
\end{equation}

We have natural homeomorphisms (cf. \cite[Prop. 3.7]{James84} \cite[Prop. 1.6]{Cutler20})
\begin{equation}
  (X \sqcup Y )_{\cpt}
  \;\simeq\;
  X_{\cpt} \vee Y_{\cpt}
  \,,
  \;\;\;\;\;\;\;\;
  (X \times Y)_{\cpt}
  \;\simeq\;
  X_{\cpt} \wedge Y_{\cpt}
  \,,
\end{equation}

For $Y \in \mathrm{Top}^\ast$ we denote the connected component of the map constant on $\infty_{{}_Y}$ by
\begin{equation}
  \label{ConnectedComponentOfMappingSpace}
  \mathrm{Map}_0(-,Y)
  \,\subset\,
  \mathrm{Map}(-,Y)
  \hspace{1cm}
  \mbox{and}
  \hspace{1cm}
  \mathrm{Map}^\ast_0(-,Y)
  \,\subset\,
  \mathrm{Map}^\ast(-,Y)\,.
\end{equation}

\begin{proposition}[{\bf One-point compactification functorial on proper maps} {\cite[p. 70]{James84}\cite[Prop. 1.6]{Cutler20}}]
\label{OnePointCompactificationFunctorialOnProperMaps}
The operation of one-point compactification extends to a functor on the category of locally compact Hausdorff spaces with \emph{proper maps} between them
\begin{equation}
  \label{FunctorialCpt}
  (-)_{\cpt}
  \;:\;
  \begin{tikzcd}
    \mathrm{LCHaus}_{\mathrm{PrpMaps}}
    \ar[r]
    &
    \mathrm{CptHaus}^\ast
    \,.
  \end{tikzcd}
\end{equation}
Since homeomorphisms are proper, this implies in particular functoriality on homeomorphisms.
\end{proposition}

\medskip

\noindent
{\bf Loop spaces.}
The {\it based loop space} of $X \in \mathrm{Top}^\ast$ is 
\begin{equation}
  \label{BasedLoopSpace}
  \Omega X 
  \,:=\,
  \mathrm{Map}^\ast\big(
    S^1
    ,\,
    X
  \big)
\end{equation}
whose connected components form the {\it fundamental group} at the basepoint (cf. \cite[\S 2.5]{AguilarGitlerPrieto02}):
\begin{equation}
  \label{ConnectedComponentsOfLoopSpace}
  \pi_1 \, X
  \;:=\;
  \pi_0 \, \Omega X
\end{equation}

For example, the fundamental group of pointed mapping spaces $X \to Y$ (based at the map constant on the basepoint of $Y$) has these alternative expressions:
\begin{equation}
  \label{FundamentalGroupOfPointedMappingSpace}
  \def\arraystretch{1.6}
  \begin{array}{rcll}
    \pi_1 \, \mathrm{Map}^\ast(X,Y)
    &\defneq&
    \pi_0\, 
    \Omega \, \mathrm{Map}^\ast(X,Y)
    &
    \proofstep{
    by
    \eqref{ConnectedComponentsOfLoopSpace}
    }
    \\
    &\defneq&
    \pi_0
    \,
    \mathrm{Map}^\ast\big(
      S^1
      ,\,
      \mathrm{Map}^\ast(X,Y)
    \big)
    &
    \proofstep{
      by
      \eqref{BasedLoopSpace}
    }
    \\
    &\simeq&
    \pi_0 \, 
    \mathrm{Map}^\ast\big(
      S^1 \wedge X
      ,\,
      Y
    \big)
    &
    \proofstep{
      by
      \eqref{MappingSpaceAdjunction}
    }
    \\
    &\simeq&
    \pi_0 \, 
    \mathrm{Map}^\ast\big(
      X \wedge S^1
      ,\,
      Y
    \big)
    &
    \proofstep{
      by \eqref{SymmetryOfSmashProduct}
    }
    \\
    &\simeq&
    \pi_0
    \,
    \mathrm{Map}^\ast\big(
      X
      ,\,
      \mathrm{Map}^\ast(S^1, Y)
    \big)
    &
    \proofstep{
      by
      \eqref{MappingSpaceAdjunction}
    }
    \\
    &\defneq&
    \pi_0
    \,
    \mathrm{Map}^\ast\big(
      X
      ,\,
      \Omega Y
    \big)
    &
    \proofstep{
      by \eqref{BasedLoopSpace}.
    }
  \end{array}
\end{equation}

\medskip

\noindent
{\bf Homotopy.}
We write $\mathrm{Grpd}_\infty$ for the $\infty$-category of homotopy types and
$$
  \shape
  \;:\;
  \mathrm{Top}
  \xrightarrow{\quad}
  \mathrm{Grpd}_\infty
$$
for the underlying functor. This means that a {\it weak homotopy equivalence} 
between topological spaces is equivalently an equivalence under $\shape$:
\begin{equation}
  \label{WeakHomotopyEquivalenceNotation}
  X,Y \in \mathrm{Top}  
  \hspace{1.2cm}\vdash
  \hspace{1.2cm}
  X \;\shapeEquivalence\; Y
  \;\;\;\;\;\;\;\;
  \Leftrightarrow
  \;\;\;\;\;\;\;\;
  \begin{tikzcd}
    \shape X
    \ar[
      rr,
      "{
        \shape f
      }",
      "{
        \sim
      }"{swap}
    ]
    &&
    \shape \, Y.
  \end{tikzcd}
\end{equation}

\medskip

Given $f : Y \to Z$ a map of pointed topological spaces, with homotopy fiber $X$
$$
  \begin{tikzcd}[
    column sep=35pt
  ]
    X
    \ar[
      r, 
      "{ \mathrm{hofib}(f) }"
    ]
    &
    Y
    \ar[r, "{ f }"]
    &
    Z
  \end{tikzcd}
$$
the resulting long homotopy fiber sequences
\begin{equation}
  \label{LongHomotopyFiberSequence}
  \begin{tikzcd}[
    row sep=8pt
  ]
    \Omega X
    \ar[r]
    &
    \Omega Y
    \ar[r]
    \ar[
      d,
      phantom,
      ""{coordinate, name=mid1}
    ]
    &
    \Omega Z
    \ar[
      dll,
      rounded corners,
      to path={
        -- ([xshift=8pt]\tikztostart.east)
        |- (mid1)
        -| ([xshift=-10pt]\tikztotarget.west)
        -- (\tikztotarget)
      }
    ]
    \\
    X
    \ar[r]
    &
    Y
    \ar[r]
    &
    Z
  \end{tikzcd}
\end{equation}
pass under $\pi_0(-)$ to long exact sequences of homotopy groups
\begin{equation}
  \label{LongExactSequenceOfHomotopyGroups}
  \begin{tikzcd}[
    row sep=8pt
  ]
    \pi_{n+1}(X)
    \ar[r]
    &
    \pi_{n+1}(Y)
    \ar[r]
    \ar[
      d,
      phantom,
      ""{coordinate, name=mid1}
    ]
    &
    \pi_{n+1}(Z)
    \ar[
      dll,
      rounded corners,
      to path={
        -- ([xshift=8pt]\tikztostart.east)
        |- (mid1)
        -| ([xshift=-10pt]\tikztotarget.west)
        -- (\tikztotarget)
      }
    ]
    \\
    \pi_{n}(X)
    \ar[r]
    &
    \pi_{n}(Y)
    \ar[r]
    &
    \pi_{n}(Z)
    \,.
  \end{tikzcd}
\end{equation}

\medskip

\noindent
{\bf Homotopy quotients and Borel construction.}

\begin{definition}
  For $G \acts X$ a Hausdorff topological group acting 
  continuously on a topological space $X$,
  we write
  \begin{equation}
    \label{HomotopyQuotientViaBorel}
    \begin{tikzcd}
      X 
      \ar[
        r,
        "{ q }"
      ]
      &
      X \sslash G
      \;:=\;
      X \times_G E G
    \end{tikzcd}
  \end{equation}
  for its {\it Borel construction} (cf. \cite[Ex. 2.3.5]{SS25-EBund}), and call its
  homotopy type the {\it homotopy quotient} of the action.
\end{definition}

In the special case when $X = \ast$ we get the traditional {\it classifying space} (namely of principal $G$-bundles, cf. \cite[Thm. 3.5.1]{RudophSchmidt17}\cite[Thm. 4.1.13]{SS25-EBund})
\begin{equation}
  \label{ClassifyingSpaceAsHomotopyQuotient}  
  \ast \sslash G 
  \;\simeq\;
  BG
\end{equation}
whose loop space recovers $G$ up to weak homotopy equivalence
$$
  \Omega B G
  \;\shapeEquivalence\;
  G
  \,,
$$
hence whose homotopy groups are those of $G$ shifted up in degree:
\begin{equation}
  \label{HomotopyGroupsOfBG}
  \pi_{n+1}(B G)
  \;\simeq\;
  \pi_n(G)
  \mathrlap{\,.}
\end{equation}
This makes a long homotopy fiber sequence \eqref{LongHomotopyFiberSequence}
\begin{equation}
  \label{BorelFiberSequence}
  \begin{tikzcd}[
    column sep=25pt
  ]
    G
    \ar[
      rr,
      "{
        g \,\mapsto\, g(x_0)
      }"
    ]
    &&
    X
    \ar[r, "{ q }"]
    &
    X \sslash G
    \ar[r]
    &
    B G
    \,.
  \end{tikzcd}
\end{equation}
Hence if $G$ preserves the connected components of $X$ (such as if $X$ only has one connected component), then the long exact sequence of homotopy groups 
\eqref{LongExactSequenceOfHomotopyGroups}
implies that
\begin{equation}
  \label{TrivialComponentActionMeansHoQuotientPreservesPi0}
  \mbox{
    $\pi_0(G) \acts \, \pi_0(X)$
    is trivial
  }
  \;\;\;\;\;\;\;
  \Rightarrow
  \;\;\;\;\;\;\;
  \begin{tikzcd}[
    column sep=20pt
  ]
    \pi_0\big(
      X
    \big)
    \ar[
      rr,
      "{ \pi_0(q) }",
      "{\sim}"{swap}
    ]
    &&
    \pi_0\big(
      X \sslash G
    \big).
  \end{tikzcd}
\end{equation}

\medskip

\subsection{Surfaces \& 2-Cohomotopy}
\label{SurfacesAnd2Cohomotopy}

We give a streamlined review of the analysis of 2-Cohomotopy moduli of closed surfaces, due to \cite{Hansen74}, that is used in \S\ref{OnClosedSurfaces}, and record some related facts needed there for identifying the modular action on 2-cohomotopical flux monodromy.

\medskip

\noindent
{\bf Fundamental polygons of closed oriented surfaces.}
The homeomorphism class of oriented closed surfaces of genus $g$ is represented (cf. \cite[Thm. 2.8]{Giblin77}) by the quotient space of the regular $4g$-gon (called a {\it fundamental polygon} of the surface) obtained by identifying all boundary vertices with a single point and, going clockwise for $k \in \{0, \cdots, g-1\}$, the $4k + 1$st boundary edge with the reverse of the $4k+3$rd, and the $4k + 2$nd with the reverse of the $4k+4$th. For small $g$ this is illustrated in \eqref{SurjectionsOfClosedSurfaces}, cf. \cite[p. 5]{Hatcher02}.
A more homotopy-theoretic formulation of this statement is as follows.

\smallskip

The fundamental group $\pi_1$ of a wedge sum \eqref{WedgeSum} of circles is the free group on the set of summands, whose $i$th generator is represented by the loop that goes identically through the $i$th circle summand.
For the classification of surfaces of genus $g$ \eqref{TheSurface} we are concerned with wedge sums of $2g$ circles to be denoted $\bigvee_{i=1}^g \big( S^1_a  \vee S^1_b \big)$, whose generators we accordingly denote $(a_i, b_i)_{i=1}^g$.

With this, the classical presentation by fundamental polygons becomes:
\begin{proposition}[{\bf Homotopy type of closed oriented surfaces}, {cf. \cite[p. 151]{Hansen74}}]
\label{FundamentalPolygonOfClosedOrientedSurfaces}
The homeomorphism type of the closed oriented surface $\Sigma^2_g$ \eqref{TheSurface} of genus $g \in \mathbb{N}$ is that of the cell attachment {\rm (cf. \cite[\S 3.1]{AguilarGitlerPrieto02})} shown on the left here:
\begin{equation}
  \label{ClosedSurfaceAsPushout}
  \begin{tikzcd}[
    column sep=30pt
  ]
    S^1 
    \ar[
      rr,
      "{ 
        \prod_{i=1}^g [a_i, b_i]
      }"
    ]
    \ar[d]
    \ar[
      drr,
      phantom,
      "{  
        \scalebox{.7}{
          \color{gray}
          \rm (po)
        }
      }"{pos=.8}
    ]
    &&
    \bigvee_{i =1}^g
    \big(
      S^1_a \vee S^1_b
    \big)
    \ar[
      d,
      "{ i_g }"
    ]
    \ar[
      rr
    ]
    \ar[
      drr,
      phantom,
      "{  
        \scalebox{.7}{
          \color{gray}
          \rm (po)
        }
      }"{pos=.8}
    ]
    &&
    \ast
    \ar[
      d
    ]
    \\
    D^2
    \ar[
      rr
    ]
    &&
    \Sigma^2_g
    \ar[
      rr,
      "{ q^g_0 }"
    ]
    &&
    S^2
    \mathrlap{\,,}
  \end{tikzcd}
\end{equation}
whence its homotopy type sits in a long homotopy cofiber sequence of this form:
\begin{equation}
  \label{LongHomotopyFiberSequenceOfSurface}
  \begin{tikzcd}[
    column sep=26pt
  ]
    S^1
    \ar[
      rr,
      "{
        \prod_{i=1}^g [a_i, b_i]
      }"
    ]
    &&
    \bigvee_{i=1}^g
    \big( 
      S^1_a \vee S^1_b
    \big)
    \ar[
      r,
      "{ 
        i_g 
      }"
    ]
    &
    \Sigma^2_g
    \ar[
      r,
      "{ q^g_0 }"
    ]
    &
    S^2
    \ar[
      rr,
      "{
        S^1 
        \wedge
        \big(
          \prod_{i=1}^g [a_i, b_i]
        \big)
      }"
    ]
    &\phantom{--}&
    \bigvee_{i = 1}^g
    \big(
      S^2_a \vee S^2_b
    \big)\,.
  \end{tikzcd}
\end{equation}
\end{proposition}
In \eqref{ClosedSurfaceAsPushout}  the attaching map
$
  \prod_i [a_i, b_i] = a_1 b_1 a_1^{-1} b_1^{-1} \, a_2 b_2 a_2^{-1} b_2^{-1} \cdots
$,
is a representative for the element of $\pi_1\big(\bigvee_i (S^1_a \vee S^1_b)\big)$ of that same name (the consecutive sequence of edges and reverse edges in the boundary of the fundamental polygon), $D^2$ is the fundamental polygon itself and the pushout enforces the identification of pairs of its boundary edges. 
Finally, the connecting map $q^g_0$ sends all of these previous boundary edges to the base point.
\begin{remark}[\bf Compatible surjections of closed surfaces to 2-sphere]
By the pasting law for pushouts, \eqref{ClosedSurfaceAsPushout} also shows that we have canonical projection maps $q^{g+1}_g : \Sigma^2_{g+1} \xrightarrow{\;} \Sigma^2_g$ 
\eqref{SurjectionsOfClosedSurfaces}
compatible with their maps $q_0$ to $S^2$ \eqref{ClosedSurfaceAsPushout}, given by sending just the $g+1$st pair of edges to the point:
\begin{equation}
  \label{ConsecutiveProjectionMapsBetweenClosedSurfaces}
  \begin{tikzcd}[column sep=large]
    S^1 
    \ar[
      rr,
      "{ 
        \prod_i [a_i, b_i]
      }"
    ]
    \ar[d]
    \ar[
      drr,
      phantom,
      "{  
        \scalebox{.7}{
          \color{gray}
          \rm (po)
        }
      }"{pos=.8}
    ]
    &&
    \bigvee_{i =1}^{g+1}
    \big(
      S^1_a \vee S^1_b
    \big)
    \ar[
      d,
      "{ i_{g+1} }"
    ]
    \ar[
      r,
      ->>
    ]
    \ar[
      dr,
      phantom,
      "{  
        \scalebox{.7}{
          \color{gray}
          \rm (po)
        }
      }"{pos=.8}
    ]
    &
    \bigvee_{i =1}^{g}
    \big(
      S^1_a \vee S^1_b
    \big)    
    \ar[rr]
    \ar[
      d,
      "{ i_g }"
    ]
    \ar[
      drr,
      phantom,
      "{  
        \scalebox{.7}{
          \color{gray}
          \rm (po)
        }
      }"{pos=.8}
    ]
    &&
    \ast
    \ar[
      d
    ]
    \\
    D^2
    \ar[
      rr
    ]
    &&
    \Sigma^2_{g+1}
    \ar[
      r,
      "{ 
        q^{g+1}_g 
      }"{description}
    ]
    \ar[
      rrr,
      rounded corners,
      to path={
       -- ([yshift=-02pt]\tikztostart.south)
       -- ([yshift=-7pt]\tikztostart.south)
       -- node[yshift=5pt, scale=.7]
         {$q^g_0$}
          ([yshift=-11pt]\tikztotarget.south)
       -- ([yshift=-00pt]\tikztotarget.south)
     }
    ]
    &
    \Sigma^2_{g}
    \ar[
      rr, 
      "{ q^{g}_0 }"{description}
    ]
    &&
    S^2
    \mathrlap{\,.}
  \end{tikzcd}
\end{equation}
\end{remark}
We also obtain from this the following re-derivation of the integral cohomology of closed surfaces, which is needed in the main text for identifying the modular group action on 2-cohomotopical flux monodromy but also 
serves as the blueprint for its 2-cohomotopical variant shown further below in Lem. \ref{RestatingHansenProp2}:
\begin{proposition}[\bf Integral cohomology of closed surfaces]
  \label{OrdinaryCohomologyOfClosedSurfaces}
  For closed oriented surfaces, 
  their ordinary integral cohomology in $\mathrm{deg} = 1$ is:
  \begin{equation}
    \label{IntegralCohomologyOfClosedOrientedSurface}
      \widetilde{H}^1(
        \Sigma^2_g
        ;\, 
        \mathbb{Z}
      )
      \;\simeq\;
      \mathbb{Z}^{2g}.
  \end{equation}
\end{proposition}
\begin{proof}
  The long exact sequence of homotopy groups \eqref{LongExactSequenceOfHomotopyGroups}
  which is induced by the homotopy fiber sequence obtained by mapping \eqref{LongHomotopyFiberSequenceOfSurface} into the classifying space $B \mathbb{Z}$ \eqref{OrdinaryCohomologyViaClassifyingSpaces} is, in the relevant part, of this form:
  $$
    \begin{tikzcd}[
      column sep=7pt,
      row sep=6pt
    ]
      \pi_0\, \mathrm{Map}^\ast\big(
        S^2
        ,\,
        B \mathbb{Z}
      \big)
      \ar[
        rr,
        "{
          (q^g_0)^\ast
        }"
      ]
      \ar[d, equals]
      &\phantom{-}&
      \pi_0\, \mathrm{Map}^\ast\big(
        \Sigma^2_g
        ,\,
        B \mathbb{Z}
      \big)
      \ar[d, equals]
      \ar[
        rr
      ]
      &&
      \pi_0\, \mathrm{Map}^\ast\big(
        \bigvee_i (S^1_a \vee S^1_b)
        ,\,
        B \mathbb{Z}
      \big)      
      \ar[d, equals]
      \ar[
        rr,
        "{
          \big(
            \prod_i [a_i, b_i]
          \big)^\ast
        }"
      ]
      &\phantom{---}&
      \pi_0\,
      \mathrm{Map}^\ast\big(
        S^1
        ,\,
        B \mathbb{Z}
      \big)\,.
      \ar[d, equals]
      \\
      \underbrace{
      \widetilde{H}^1\big(
        S^2
        ;\,
        \mathbb{Z}
      \big)
      }_{ 0 }
      &&
      \widetilde{H}^1\big(
        \Sigma^2_g
        ;\,
        \mathbb{Z}
      \big)
      &&
      \prod_{i=1}^g
      \big(
        \underbrace{
        \widetilde{H}^1(S^1;\, \mathbb{Z})
        }_{ \mathbb{Z} }
      \big)^2
      &&
      \underbrace{
      \widetilde{H}^1(S^1;\, \mathbb{Z})
      }_{ \mathbb{Z} }
    \end{tikzcd}
  $$
  Hence, to conclude, we need to show that the map on the right --- forming the pullback of $B\mathbb{Z}$-valued maps along $\prod_i [a_i, b_i]$ --- is the zero-map under $\pi_0$. This is the case because the pullback classes are still group commutators and as such they vanish since the $\pi_0$ group in question is $\simeq \mathbb{Z}$ and hence abelian:
  $$
    \begin{tikzcd}[
      row sep=0pt
    ]
      \underbrace{
      \pi_0
      \, 
      \mathrm{Map}^\ast\big(
        \textstyle{\bigvee_i} 
        (S^1_a \vee S^1_b)
        ,\,
        B \mathbb{Z}
      \big)
      }_{
        \mathbb{Z}^{2g}
      }
      \ar[
        rr,
        "{
          \big(
            \prod_i [a_i, b_i]
          \big)^\ast
        }"
      ]
      &&
      \underbrace{
      \pi_0 \, 
      \mathrm{Map}^\ast\big(
        S^1
        ,\,
        B \mathbb{Z}
      \big)
      }_{
        \mathbb{Z}
      }\,.
      \\[-2pt]
      \big(
        (n_i, m_i)
      \big)_{i=1}^g
      &\longmapsto&
      \prod_i 
      \grayunderbrace{
        [ n_i, m_i ]
      }{ 0 }
    \end{tikzcd}
  $$

  \vspace{-6mm} 
\end{proof}
In an analogous manner, we find:
\begin{lemma}[{\cite[Prop. 2]{Hansen74}}]
  \label{RestatingHansenProp2}
  The long exact sequence of homotopy groups \eqref{LongExactSequenceOfHomotopyGroups}
  which is induced by the homotopy fiber sequence obtained by mapping \eqref{LongHomotopyFiberSequenceOfSurface} into $S^2$ truncates to a short exact sequence:
  \begin{equation}
    \label{HansenProp2}
    \begin{tikzcd}[
      column sep=20pt
    ]
      1
      \ar[r]
      &
      \grayunderbrace{
      \pi_1 \, \mathrm{Map}^\ast\big(
        S^2
        ,\,
        S^2
      \big)
      }{ \mathbb{Z} }
      \ar[
        rr,
        "{ 
          (q^g_0)^\ast
        }"
      ]
      &&
      \pi_1 \, \mathrm{Map}^\ast\big(
        \Sigma^2_g
        ,\,
        S^2
      \big)
      \ar[
        rr,
        "{ (i_g)^\ast }"
      ]
      &&
      \prod_{i = 1}^g
      \big(
      \grayunderbrace{
      \pi_1 \, \mathrm{Map}^\ast(
        S^1
        ,\,
        S^2
      )
      }{ \mathbb{Z} }
      \big)^2
      \ar[r]
      &
      1
      \mathrlap{\,.}
    \end{tikzcd}
  \end{equation}
\end{lemma}
\begin{proof}
  It is sufficient to see that pullback of  $S^2$-valued maps along $\prod_i [a_i,b_i]$ is the zero map under (suspension and) taking $\pi_1$: This is because the pullback classes are still group commutators and as such they vanish since the $\pi_1$ groups in question are $\simeq \mathbb{Z}$ and hence abelian:
  $$
    \begin{tikzcd}[
      row sep=0pt
    ]
      \underbrace{
      \pi_1
      \, 
      \mathrm{Map}^\ast\big(
        \textstyle{\bigvee_i} (S^1_a \vee S^1_b)
        ,\,
        S^2
      \big)
      }_{
        \mathbb{Z}^{2g}
      }
      \ar[
        rr,
        "{
          \big(
            \prod_i [a_i, b_i]
          \big)^\ast
        }"
      ]
      &&
      \underbrace{
      \pi_1 \, 
      \mathrm{Map}^\ast\big(
        S^1
        ,\,
        S^2
      \big)
      }_{
        \mathbb{Z}
      }
      \\[-2pt]
      \big(
        (n_i, m_i)
      \big)_{i=1}^g
      &\longmapsto&
      \prod_i 
      \grayunderbrace{
        [ n_i, m_i ]
      }{ 0 }
    \end{tikzcd}
  $$
  \vspace{-1mm} 

  \noindent 
  (and analogously, under suspension, with all copies of $S^1$ in this formula replaced by $S^2$).
  Here in evaluating these groups we have used \eqref{MapsOutOfWedgeSum} and stages of \eqref{FundamentalGroupOfPointedMappingSpace}, thereby identifying all copies of $\mathbb{Z}$ with $\mathbb{Z} \simeq \pi_2(S^2)$ (and with $\mathbb{Z} \simeq \pi_3(S^2)$ for the case with suspension).
\end{proof}
Finally, to see that the statement \eqref{HansenProp2} for the {\it pointed} mapping space implies the variant statement \eqref{ExtensionSequenceOfFLuxMonodromyOnClosedSurface}
for the unpointed mapping space which is of our actual interest:

\begin{lemma}[{\cite[Thm. 1]{Hansen74}}]
\label{HansenThm1}
For $g \in \mathbb{N}$, we have a short exact sequence of this form:
$$
  \begin{tikzcd}[
    column sep=10pt,
    row sep=3pt
  ]
    1
    \ar[r]
    &
    \pi_1\,
      \mathrm{Map}_0\big(
        S^2
        ,\,
        S^2
      \big)
    \ar[
      rr,
      "{
        (q^g_0)^\ast
      }"
    ]
    &\phantom{---}&
    \pi_1 \,
      \mathrm{Map}_0\big(
        \Sigma^2_g
        ,\,
        S^2
      \big)
    \ar[
      rr
    ]
    &&
    \pi_1 \,
      \mathrm{Map}^\ast_0\big(
        \bigvee_g
        (S^1_a \vee S^1_b)
        ,\,
        S^2
      \big)
    \ar[r]
    &
    1\;.
  \end{tikzcd}
$$
\end{lemma}

\vspace{-5mm} 
\begin{proof}
  Consider the long exact sequences of homotopy groups \eqref{LongExactSequenceOfHomotopyGroups} induced by the evaluation sequences \eqref{EvaluationSequence} on $\Sigma^2_g$ and on $S^2$, respectively, with the map between them induced by pullback along $q^g_0$ \eqref{ClosedSurfaceAsPushout} extended to the short exact sequence from \eqref{HansenProp2}:
  $$
    \begin{tikzcd}[row sep=10pt, 
      column sep=27pt
    ]
      &
      \pi_2(S^2)
      \ar[
        dd,
        "{ \delta^0 }"
      ]
      \ar[rr, equals]
      &&
      \pi_2(S^2)
      \ar[
        dd,
        "{
          \delta^g
        }"
      ]
      \\
      \\
      1
      \ar[r]
      &
      \pi_1 \, \mathrm{Map}^\ast\big(
        S^2
        ,\,
        S^2
      \big)
      \ar[
        rr,
        "{ 
            (q^g_0)^\ast
        }"
      ]
      \ar[
        dd,
        "{ j^0_\ast }"
      ]
      &&
      \pi_1 \, \mathrm{Map}^\ast\big(
        \Sigma^2_g
        ,\,
        S^2
      \big)
      \ar[
        rr,
        "{
            (i_g)^\ast
        }"
      ]
      \ar[
        dd,
        "{ 
          j^g_\ast
         }"
      ]
      &&
      \pi_1 \, 
      \mathrm{Map}^\ast\big(
        \bigvee_i 
        (S^1_a \vee S^1_b)
        ,\, 
        S^2
      \big)
      \ar[r]
      &
      1
      \\
      \\
      1
      \ar[
        r, 
        dashed
      ]
      &
      \pi_1 \, \mathrm{Map}\big(
        S^2
        ,\,
        S^2
      \big)
      \ar[
        rr,
        dashed,
        "{ 
          (q^g_0)^\ast
        }"
      ]
      \ar[d]
      &&
      \pi_1 \, \mathrm{Map}\big(
        \Sigma^2_g
        ,\,
        S^2
      \big)
      \ar[d]
      \ar[
        uurr, 
        dashed,
        "{
          (i_g)^\ast
          \circ\,
          \overline{j^g_\ast}
        }"{swap, sloped}
      ]
      \\
      &
      \grayunderbrace{
        \pi_1(S^2)
      }{ 1 }
      \ar[rr, equal]
      &&
      \grayunderbrace{
        \pi_1(S^2)
      }{ 1 }
    \end{tikzcd}
  $$
  Here $j^0_\ast$ and $j^g_\ast$ denote the maps induced by inclusion of pointed maps into unpointed maps. All solid sequences are exact, by construction, horizontally as well as vertically. Using this, a routine diagram chase shows
  \footnote{
    To spell it out:
    Since $j^g_\ast$ is seen to be surjective, 
    we may define
    $(i_g)^\ast \circ\, \overline{j^g_\ast}$ on a given element $\phi$
    by choosing any preimage through $j^g_\ast$. To see that this is well defined: If $\widehat{\phi}$, $\widehat{\phi}'$ are a pair of  preimages, their difference is in the image of $\delta^g \,=\, (q^g_0)^\ast \circ \delta^0$, hence in the image of $(q^g_0)^\ast$ and hence vanishes under $(i_g)^\ast$.
    With the map thus existing, surjectivity is immediate from $(i_g)^\ast$ being surjective.

    To see that the dashed $(q^g_0)^\ast$ is injective: Consider $\phi, \phi'$ a pair of elements in the domain with the same image. Since $j^0_\ast$ is surjective we may find $j^0_\ast$-preimages $\widehat{\phi}$, $\widehat{\phi'}$. By commutativity of the middle square we then have $j^g_\ast \circ (q^g_0)^ \ast (\widehat{\phi}) \,=\, j^g_\ast \circ (q^g_0)^ \ast (\widehat{\phi'})$, and so the difference between $(q^g_0)^ \ast (\widehat{\phi})$ and $(q^g_0)^ \ast (\widehat{\phi'})$ is in the image of $\delta^g =  (q^g_0)^\ast \circ \delta^0$. But since $(q^g_0)^\ast$ is injective, this means that already the difference between $\widehat{\phi}$ and $\widehat{\phi}'$ is in the image of $\delta^0$, hence vanishes under $j^0_\ast$, hence $\phi = \phi'$, which was to be seen.

    Finally, to see that the dashed sequence is exact in the middle: By the previous construction, the kernel of $(i_g)^\ast \circ \, \overline{j^g_\ast}$ consists exactly of those $\phi$ whose $j^g_\ast$-preimage is in the kernel of $(i_g)^\ast$, hence in the image of $(q^g_0)^\ast$, hence of those $\phi$ in the image of $(q^g_0)^\ast \circ j^0_\ast$, hence in the image of $(q^g_0)^\ast$ -- which was to be shown.
  }
  that the dashed sequence exists and is exact.
\end{proof}
As a corollary, we note:
\begin{lemma}[\bf Flux monodromy mapping to ordinary cohomology]
\label[lemma]{FLuxMonodromyMappingToOrdinaryCohomology}
The cohomology operation from 2-Cohomotopy in degree -1 to integral cohomology in degree 1, induced by the looping of the unit class $1^2 : S^2 \xrightarrow{\;} B^2 \mathbb{Z}$ \eqref{BnZUnitClass},
produces a morphism of short exact sequences:
$$
  \begin{tikzcd}[
    column sep=22pt
  ]
    1 
    \ar[r]
    &
    \pi_0 
    \,
    \mathrm{Map}(
      S^2
      ,\,
      \Omega S^2
    )
    \ar[
      d,
      "{
        (\Omega 1^2)_\ast
      }"
    ]
    \ar[
      rr,
      "{
        (q^g_0)^\ast
      }"
    ]
    &&
    \grayoverbrace{
    \pi_0 
    \,
    \mathrm{Map}(
      \Sigma^2_g
      ,\,
      \Omega S^2
    )
    }{
      \pi_1
      \,
      \mathrm{Map}(
        \Sigma^2_g
        ,\,
        S^2
      )
    }
    \ar[rr]
    \ar[
      d,
      "{
        (\Omega 1^2)_\ast
      }"
    ]
    &&
    \pi_0 
    \,
    \mathrm{Map}^\ast\big(
      \bigvee_i (S^1_a \vee S^1_b)
      ,\,
      \Omega S^2
    \big)    
    \ar[r]
    \ar[
      d,
      "{
        (\Omega 1^2)_\ast
      }",
      "{ \sim }"{sloped, swap}
    ]
    &
    1
    \\
    1 
    \ar[r]
    &
    \grayunderbrace{
    \pi_0 
    \,
    \mathrm{Map}(
      S^2
      ,\,
      B \mathbb{Z}
    )
    }{
      H^1(S^2;\, \mathbb{Z})
      \;\simeq\;
      0
    }
    \ar[
      rr,
      "{
        (q^g_0)^\ast
      }"
    ]
    &&
    \grayunderbrace{
    \pi_0 
    \,
    \mathrm{Map}\big(
      \Sigma^2_g
      ,\,
      B \mathbb{Z}
    \big)
    }{
      H^1(\Sigma^2_g,\, \mathbb{Z})
      \;\simeq\;
      \mathbb{Z}^{2g}
    }
    \ar[
      rr,
      "{ \sim }"
    ]
    &&
    \grayunderbrace{
    \pi_0 
    \,
    \mathrm{Map}^\ast\big(
      \textstyle{\bigvee_i}
      (S^1_a \vee S^1_b)
      ,\,
      B \mathbb{Z}
    \big)    
    }{
      \mathbb{Z}^{2g}
    }
    \ar[r]
    &
    1\,.
  \end{tikzcd}
$$
\end{lemma}
\begin{proof}
  The top exact sequence is from Lem. \ref{HansenThm1}, and inspection of the proof there shows immediately that it applies verbatim also with the coefficient $\Omega S^2$ replaced by $B \mathbb{Z}$, throughout (for the exactness of the horizontal sequence in the proof, this is in the proof of Prop. \ref{OrdinaryCohomologyOfClosedSurfaces}). This gives the bottom exact sequence and the compatibility of the two under the cohomology operation, as claimed.
\end{proof}

\subsection{Some Representation Theory}
\label{SomeRepresentationTheory}

We need only basic representation theory (cf. \cite{Serre77}\cite{FultonHarris91}\cite{EtingofEtAl11}) and the Mackey classification of irreps of finite semidirect product groups.

\medskip

\noindent
{\bf Linear representations.}
Consider $G$ a group. The cardinality of its underlying set, hence the {\it order} of $G$, is denoted $\vert G \vert$. A finite-dimensional $\mathbb{C}$-linear {\it representation} of $G$ is a group homomorphism $\rho : G \xrightarrow{\;} \mathrm{GL}(V)$ for $V$ a finite-dimensional vector space, also to be denoted $G \acts_{\rho}\, V$.

\begin{itemize}[itemsep=2pt] 
\item Given a pair $(\rho, \rho')$ of representations, an isomorphism $\eta :   \rho \xrightarrow{\sim} \rho'$ is a linear isomorphism of the underlying vector spaces, $\eta : V \xrightarrow{\sim} V'$ such that $\eta \circ \rho \,=\, \rho' \circ \eta$ (an ``intertwiner''). In particular, with any choice of linear basis $V \xrightarrow{\sim} \mathbb{C}^{\mathrm{dim}(V)}$ a representation $\rho$ is isomorphic to a matrix representation $G \xrightarrow{\;} \mathrm{GL}_n(\mathbb{C})$ for $n = \mathrm{dim}(V)$.

\item Given $\rho, \rho'$ a pair of (matrix) representations, their direct sum $\rho \oplus \rho' : G \xrightarrow{\;} \mathrm{GL}_{n + n'}(\mathbb{C})$ is represented by the corresponding block-diagonal matrices.
A representation $\rho$ is called {\it irreducible} (``irrep'') if $V\neq0$ and its only $G$-invariant linear subspaces are $0$ and $V$. A representation is {\it semisimple} if it is a direct sum of irreducibles. Finite-dimensional complex representations of finite groups are semisimple by Maschke's theorem; finite-dimensional unitary representations of any group are semisimple because invariant subspaces have invariant orthogonal complements. In these cases irreducibility is equivalent to having no nontrivial direct-sum decomposition. For general representations the latter condition only defines {\it indecomposability}.

\item For any group $G$,
we denote the set of isomorphism classes $[-]$ of irreducible $\mathbb{C}$-linear representations by
\begin{equation}
  \label{SetOfIrrepClasses}
  \mathrm{Irr}(G)
  \;\simeq\;
  \big\{
    [\rho_i]
  \big\}_{i \in I}
  \,.
\end{equation}
\item {\it Schur's Lemma} says, in particular, that every intertwining operator $\eta$ from an irreducible representation to itself (hence every linear operator that commutes with all the representation operators of an irreducible representation) is a scalar multiple of the identity (cf. \cite[Cor. 1.17]{EtingofEtAl11}):
\begin{equation}
  \label{SchurLemma}
  \eta \in \mathrm{End}_{\mathbb{C}}(V),\quad n=\dim V
  \,,\;\;\;\;\;\;
  \underset{g \in G}{\forall}
  \;\;\;
  \eta \circ \rho(g)
  \,=\,
  \rho(g) \circ \eta
  \hspace{1cm}
  \Rightarrow
  \hspace{1cm}
  \underset{ z \in \mathbb{C} }{\exists}
  \;\;\;\;
  \,
  \eta \,=\, z \cdot \mathrm{id}_n
  \,.
\end{equation}

\item For finite $G$, the {\it sum-of-squares formula} (cf. \cite[Thm. 3.1(ii)]{EtingofEtAl11}) says that the sum of the squared dimensions of the distinct irreps equals the order of the group:
\begin{equation}
  \label{SumOfSquaresFormula}
  \sum_{i \in I}
  \,
  \big(\mathrm{dim}(\rho_i)\big)^2
  \;=\;
  \vert G \vert
  \,.
\end{equation}
\end{itemize}

\begin{example}[\bf Irreps of $\mathbb{Z}_d$]
\label{IrrepsOfZd}
For $d \in \mathbb{N}_{> 0}$, the $\mathbb{C}$-linear irreps of $\mathbb{Z}_d$, to be denoted $\mathbf{1}_{[n]} \in \mathrm{Rep}_{\mathbb{C}}(\mathbb{Z}_d)$, are all 1-dimensional and labeled by $[n] \in \mathbb{Z}_d$, where $[1] \in \mathbb{Z}_d$ acts on $\mathbf{1}_{[n]}$ by multiplication with $e^{2\pi \mathrm{i} \tfrac{n}{d}}$.
\end{example}
\begin{example}[{\bf Irreps of $\mathrm{Sym}_3$}, {cf. \cite[\S 1.3]{FultonHarris91}}]
\label{IrrepsOfSym3}
Irreps of $\mathrm{Sym}_3$ include the trivial 1-dimensional representation $\mathbf{1}$, the sign representation $\mathbf{1}_{\mathrm{sgn}}$, and the standard representation $\mathbf{2}$
\eqref{TheStandardRep}. Since $1^2 + 1^2 + 2^2 = 6 = \vert \mathrm{Sym}_3\vert$, the sum-of-squares formula \eqref{SumOfSquaresFormula} implies that  there are no further distinct irreps.
\end{example}

\begin{proposition}[{\bf Unitarization of linear representations}, {cf. \cite[Prop. 4.6]{Knapp96}}]
\label{UnitarizationOfReps}
For $G$ a finite group and $\HilbertSpace{H}$ a finite-dimensional complex Hilbert space, every $\mathbb{C}$-linear representation $R : G \xrightarrow{\;} \mathrm{GL}(\HilbertSpace{H})$ on the underlying complex vector space of $\HilbertSpace{H}$ is isomorphic to a unitary representation $U : G \xrightarrow{\;} \mathrm{U}(\HilbertSpace{H}) \xhookrightarrow{\;} \mathrm{GL}(\HilbertSpace{H})$.
\end{proposition}

\smallskip

\noindent
{\bf Representation characters.}
The {\it character} $\rchi^{\rho}$ of a finite-dimensional representation $\rho$ is the function on $G$ assigning the traces of the representation matrices
\begin{equation}
  \begin{tikzcd}[row sep=-3pt, column sep=0pt]
    G 
      \ar[
        rr, 
        "{ \rchi^{\rho} }"
      ] 
    && 
    \mathbb{C}
    \\
    g &\longmapsto&
    \mathrm{tr}\big(\rho(g)\big)
    \,.
  \end{tikzcd}
\end{equation}

\begin{itemize} 
\item Passage to characters is evidently linear in direct sums of representations, in that
$
  \rchi^{\rho \oplus \rho'}
  \;=\;
  \rchi^{\rho}
  +
  \rchi^{\rho'}
$. On finite-dimensional semisimple representations it is injective on isomorphism classes (cf. \cite[Thm.~2.17]{EtingofEtAl11}):
\begin{equation}
  \label{CharacterConstructionIsInjective}
  \begin{tikzcd}[row sep=-3pt, column sep=0pt]
  \mathrm{Rep}_{\mathrm{ss}}(G)_{/\sim}
  \ar[
    rr,
    hook
  ]
  &&
  \mathbb{C}^G
  \\
  {[\rho]}
    &\longmapsto&
  \rchi^{\rho}
  \mathrlap{\,.}
  \end{tikzcd}
\end{equation}
Here $\mathrm{Rep}_{\mathrm{ss}}(G)$ denotes the finite-dimensional semisimple
complex representations. Thus \eqref{CharacterConstructionIsInjective}
applies to all finite-dimensional unitary representations, and to all
complex representations of finite groups in the present finite-dimensional
setting. Without semisimplicity a character determines only the
semisimplification.

\item For finite $G$, $\vert G \vert < \infty$, the evident normalized inner product on functions $\rchi : G \xrightarrow{\;} \mathbb{C}$
\begin{equation}
  \label{SchurInnerProduct}
  \big\langle
    \rchi
    ,\,
    \rchi'
  \big\rangle
  \;:=\;
  \tfrac{1}{
    \vert G \vert
  }
  \,
  \sum_{g\in G}\rchi(g)\,\overline{\rchi'(g)}
\end{equation}
exhibits {\it Schur orthonormality}, which is the statement that for the characters of irreducible representations $\rho_i$, $\rho_j$ \eqref{SetOfIrrepClasses} we have
(cf. \cite[Thm. 2.12]{FultonHarris91}\cite[Thm. 3.8]{EtingofEtAl11})
\begin{equation}
  \label{SchurOrthonormality}
  \big\langle
    \rchi^{\rho_i}
    ,\,
    \rchi^{\rho_j}
  \big\rangle
  \;=\;
  \delta_{i j}
  \,.
\end{equation}
\end{itemize} 

\smallskip

\noindent
{\bf Tensor products of representations.}
Given a pair of linear representations $(\rho, \rho')$ over, respectively,
a pair of groups $(G, G')$,  their {\it external tensor product} $\rho \boxtimes \rho'$ (cf. \cite[Ex. 2.36]{FultonHarris91}) is the representation of the direct product group $G \times G'$ given by
\vspace{-2mm} 
\begin{equation}
  \label{ExternalTensorProductOfReps}
  \rho \boxtimes \rho'
  \;:\;
  \begin{tikzcd}
    G \times G'
    \ar[rr, "{ \rho \times \rho' }"]
    &&
    \mathrm{GL}_n(\mathbb{C})
    \times
    \mathrm{GL}_{n'}(\mathbb{C})
    \ar[r]
    &
    \mathrm{GL}_{n n'}(\mathbb{C})
    \,.
  \end{tikzcd}
\end{equation}
\begin{itemize} 
\item If $\rho$ and $\rho'$ are irreducible then so is $\rho \boxtimes \rho'$ and all irreps of $G \times G'$ arise this way, hence (cf. still \cite[Ex. 2.36]{FultonHarris91}):
\begin{equation}
  \label{IrrepsOfDirectProductGroup}
  \begin{tikzcd}[row sep=-3pt, column sep=0pt]
  \mathrm{Irr}(G)
  \times
  \mathrm{Irr}(G')
  \ar[
    rr,
    "{ \sim }"
  ]
  &&
  \mathrm{Irr}(G \times G')
  \\
  \big([\rho], [\rho']\big)
  &\longmapsto&
  \big[
    \rho \boxtimes \rho'
  \big]
  \,.
  \end{tikzcd}
\end{equation}

\item Given a pair $\rho$ and $\rho'$ of linear representations of the same group $G$, their plain tensor product is the 
restriction of their external tensor product \eqref{ExternalTensorProductOfReps} along the diagonal of $G$:
\begin{equation}
  \label{TensorProductOfReps}
  \rho \otimes \rho'
  \;:\;
  \begin{tikzcd}  
    G
    \ar[rr, "{ \mathrm{diag} }"]
    &&
    G \times G
    \ar[
      rr, 
      "{ \rho \,\boxtimes\, \rho' }"
    ]
    &&
    \mathrm{GL}_{n n'}(\mathbb{C})
    \,.
  \end{tikzcd}
\end{equation}
\end{itemize} 

\smallskip

\noindent
{\bf Group algebra.}
For $G$ a group, its {\it group convolution algebra} $\mathbb{C}[G]$ (or just {\it group algebra} for short,
cf. \cite[\S 3.4]{FultonHarris91}\cite[p. 51]{BalachandranJoMarmo10}\cite[\S 2.4]{CiagliaIbortMarmo19}\cite[(4)]{CiagliaDiCosmoIbortMarmo20}) is the $\mathbb{C}$-linear span of $G$ equipped with the algebra structure which on its canonical basis elements $\vec e_g$ is the group product ($\vec e_g \cdot \vec e_{g'} := \vec e_{g g'}$):
\begin{equation}
  \label{GroupAlgebra}
  \mathbb{C}[G]
  \,:=\,
  \Bigg\{
    \sum_{g \in G} c_g \cdot \vec e_g
    \,\in\,
    \mathrm{Span}_{\mathbb{C}}(G)
    ,\;
    \bigg(
      \sum_{g \in G} 
      c_g \cdot \vec e_g
    \bigg)
    \cdot
    \bigg(
      \sum_{g' \in G} 
      c'_{g'} \cdot \vec e_{g'}
    \bigg)
    \;\;
    :=
    \;\;
    \sum_{g \in G}
    \bigg(
      \sum_{h \in G}
        c_h \, c'_{h^{-1} g}
    \bigg)
    \cdot
    \vec e_g
  \Bigg\}
  .
\end{equation}
\begin{itemize} 
\item Hence group representations $G \xrightarrow{\;} \mathrm{GL}(V)$ are equivalently algebra homomorphisms $\mathbb{C}[G] \xrightarrow{\;} \mathrm{End}(V)$, hence are equivalently {\it modules} over the group algebra.

\item For example, the group algebra understood as a module over itself is the {\it regular representation} $G \acts\, \mathbb{C}[G]$.
\end{itemize} 
\smallskip

\noindent
{\bf Induced representations.}
Here $\mathrm{Rep}(G)$ includes representations of arbitrary dimension, as in Prop. \ref{InducedRepresentations}.
For $H \xhookrightarrow{i} G$ a subgroup inclusion, and $\rho : H \xrightarrow{\;} \mathrm{GL}(V)$ a representation of $H$, its left or right {\it induced $G$-representation along $i$}, to be denoted $i_! \rho$ or  $i_\ast \rho$, respectively, is (cf. \cite[\S 4.8]{EtingofEtAl11})
\begin{equation}
  \label{InducedRepresentation}
  \left.
  \def\arraystretch{1.9}
  \begin{array}{ccl}
  i_! \rho
  &:=&
  G\acts\; \Big(
  \mathbb{C}[G]
  \otimes_{\mathbb{C}[H]}
  V
  \Big)
  \\
  i_\ast \rho 
  &:=&
  G \acts\; \Big(
    \mathrm{Hom}_{\mathbb{C}[H]}\big(
      \mathbb{C}[G]
      ,\,
      V
    \big)
  \!\Big)
  \end{array}
\!\!\!  \right\}
  \,\in\,
  \mathrm{Rep}(G)
  \,,
\end{equation}
where the constructions on the right are the tensor product and hom-space of modules over the group algebra $\mathbb{C}[H]$ \eqref{GroupAlgebra}, understood as equipped with their residual $\mathbb{C}[G]$-module structure given by left multiplication (for $i_! \rho$) or by right multiplication (for $i_\ast \rho$) of the group algebra $\mathbb{C}[G]$ on itself.

\begin{itemize}
\item If $G$ is finite, $\vert G \vert < \infty$, left and right inductions are naturally isomorphic (and generally when $H \subset G$ is of finite index) and one writes
\begin{equation}
  \label{InductionFunctor}
  \mathrm{Ind}_H^G \, \rho
  \,:=\,
  \iota_! \rho
  \,\simeq\,
  \iota_\ast \rho
  \;:\;
  \begin{tikzcd}
    \mathrm{Rep}(H)
    \ar[r]
    &
    \mathrm{Rep}(G)\,.
  \end{tikzcd}
\end{equation}
\end{itemize} 

\smallskip
\noindent
{\bf Mackey theory.} Consider $G = A \rtimes H$ a finite semidirect product group where the normal subgroup $A$ is abelian. The following Prop. \ref{MackeyAlgorithm} classifies its irreducible representations from knowledge of the irreps of $A$ and of certain subgroups of $G$. The statement involves the following notions:  

\begin{itemize}[
  itemsep=1.5pt
]
\item[--] $\widetilde A := \mathrm{Hom}(A,\mathbb{C}^\times)$  the dual group of irreps of $A$,

\item[--] $\big\{\chi_i \big\}_{i \in \widetilde{A}/H}$ a set of representatives of orbits in this group under the action of $H$,

\item[--] $H_i := \mathrm{Stab}_H\big(\chi_i\big)$ the corresponding stabilizer subgroups,

\item[--] $\widehat{\chi}_i$ the rep of $A \rtimes H_i$ given by $\widehat{\chi}_i(a,h) := \chi_i(a)$ for the chosen orbit representative $\chi_i \in \widetilde A$,

\item[--] $\widehat{\rho}$ the pullback of any irrep $\rho$  along $A \rtimes H_i \twoheadrightarrow H_i$.
\end{itemize}

\begin{proposition}[{\bf Mackey algorithm}, {cf. \cite[Prop. 25 p. 62]{Serre77}}]
  \label{MackeyAlgorithm}
  The irreps of $A \rtimes H$ are, up to isomorphism, exactly the induced representations \eqref{InductionFunctor} of tensor products \eqref{TensorProductOfReps} of the form
  $$
    \mathrm{Ind}
      _{A \rtimes H_i}
      ^{A \rtimes H}
    \big(
      \widehat{\chi}_i
      \otimes
      \widehat{\rho}
    \big)
    \,,
  $$
  for $i \in \widetilde{A}/H$ and $\rho \in \mathrm{Irr}(H_i)_{/\sim}$; and two such are isomorphic iff they come from the same index $i$ and isomorphic irreps $\rho$.
\end{proposition}

\smallskip
\noindent
{\bf Wreath product groups}.
Given $G$ a group and $H \xrightarrow{s} \mathrm{Sym}_n$ a group equipped with a homomorphism to a symmetric group, their {\it wreath product} is the semidirect product of $G^n$ with $H$ acting by permutation of factors (cf. \cite[Def. 8.1]{BMMN06}\cite[\S 4.1]{JamesKerber84})
\begin{equation}
  \label{WreathProduct}
  G \wr_s H
  \;:=\;
  G^n \rtimes_{s} H
  \,.
\end{equation}

\begin{itemize} 
\item For $G$ a finite group with set $I$ of classes $[\rho^i]$ of irreducible representations \eqref{SetOfIrrepClasses}, say that an {\it $I$-partition} of some $n \in \mathbb{N}$ is a tuple 
\begin{equation}
  \label{IPartition}
  (n) \,:=\, \big(n_i \in \mathbb{N}\big)_{i \in I}
  \;\;\;\;
  \mbox{with}
  \;\;\;
  \sum_i n_i = n
  \,.
\end{equation}
\item Given such, write
$$
  \mathrm{Sym}_{(n)}
  \;:=\;
  \prod_{i \in I}
  \,
  \mathrm{Sym}_{n_i}
  \;
  \xhookrightarrow{\quad}
  \;
  \mathrm{Sym}_n
$$
for the subgroup of permutations of $n$ elements that permute the $n_i$ elements among each other, for all $i \in I$. There is an evident linear representation of the corresponding wreath product \eqref{WreathProduct} on 
\begin{equation}
  \label{IIndexedExternalTensorProduct}
  \underset
    {i \in I}
    {\scalebox{1.4}{$\boxtimes$}}
  \,
    \rho_i^{\boxtimes^{n_i}}
  \;\;
  \in
  \;\;
  \mathrm{Rep}\big(
    G \wr \mathrm{Sym}_{(n)}  
  \big)
  \,,
\end{equation}
where the subgroup $G^n$ is represented by the given external tensor product \eqref{ExternalTensorProductOfReps} 
of its irreps \eqref{SetOfIrrepClasses}
and the subgroup $\mathrm{Sym}_{(n)}$ acts by permutation of tensor factors.
At the same time, an irrep $\sigma$ of $\mathrm{Sym}_{(n)}$ is of the form \eqref{IrrepsOfDirectProductGroup}
\begin{equation}
  \label{IrrepOfPartitionedSym}
  \sigma 
  \;\simeq\; 
  \underset
    {i \in I}
    {\scalebox{1.4}{$\boxtimes$}}    
    \,
    \sigma_{j(i)}
  \,,
  \,\;\;\;\;\;
  [\sigma_{j(i)}]
  \,\in\,
  \mathrm{Irr}\big(
    \mathrm{Sym}_{n_i}
  \big)
  \,,
\end{equation}
and we may regard this as a representation of $G \wr \mathrm{Sym}_{(n)} \twoheadrightarrow \mathrm{Sym}_{(n)}$.
\end{itemize}

\medskip 
\begin{proposition}[{\bf Irreps of finite wreath product groups}, {\cite[Thm. 4.3.34]{JamesKerber84}}]
  \label{IrrepsOfWreathProduct}
  For finite $G$,
  the irreducible representations of $G \wr \mathrm{Sym}_n$ \eqref{WreathProduct} are, up to isomorphism, exactly the induced representations 
  \eqref{InductionFunctor}
  along the inclusion
  $$
    \begin{tikzcd}
      \Big(
        \prod_{i \in I}
        \big(
        G \wr \mathrm{Sym}_{n_i}
        \big)
      \Big)
      \;\simeq\;
      G \wr \mathrm{Sym}_{(n)} 
      \ar[r, hook]
      &
      G \wr \mathrm{Sym}_n
    \end{tikzcd}
  $$ 
  of the representations which are tensor products \eqref{TensorProductOfReps}
  of the representations \eqref{IIndexedExternalTensorProduct} 
  corresponding to $I$-partitions \eqref{IPartition} 
  with irreps $\sigma$
  \eqref{IrrepOfPartitionedSym}
  of $\mathrm{Sym}_{(n)}$:  \begin{equation}
    \label{InducedIrrepsOfWreathProduct}
    \mathrm{Ind}
      _{G \wr \mathrm{Sym}_{(n)}}
      ^{G \wr \mathrm{Sym}_{n}}
    \bigg(\!
    \Big(
    \underset
      {i \in I}
      {\scalebox{1.4}{$\boxtimes$}}
    \,
    \rho_i^{\boxtimes^{n_i}}
    \Big)
    \otimes
    \sigma
   \! \bigg)
    \;\;\;
    \in
    \;\;
    \mathrm{Rep}\big(
      G \wr \mathrm{Sym}_{n}
    \big)
    \,.
  \end{equation}
\end{proposition}

\smallskip 
\begin{example}[\bf Irreps of $\mathbb{Z}_d \wr \mathrm{Sym}_3$]
  \label{IrrepsOfWreathForSingleGIrrep}
  The $\mathbb{C}$-linear irreps of $\mathbb{Z}_d \wr \mathrm{Sym}_3$ according to Prop. \ref{IrrepsOfWreathProduct} may be organized into three classes, depending on the nature of the partition \eqref{IPartition} --- recall for the following the irreps $\mathbf{1}_{[n]}$ of $\mathbb{Z}_d$ from Ex. \ref{IrrepsOfZd} and the irreps $\mathbf{1}$, $\mathbf{1}_{\mathrm{sgn}}$, and $\mathbf{2}$ of $\mathrm{Sym}_3$ from Ex. \ref{IrrepsOfSym3}. 
  
  \begin{itemize}[
    leftmargin=1.5cm, itemsep=3pt
  ]
  \item[\underline{Case 1}:]
  {\bf Partition involves a single irrep} $\mathbf{1}_n$.
  In this case $\mathrm{Sym}_{(3)} \simeq \mathrm{Sym}_3$ and the induction functor is the identity, so that the corresponding irreps of $G \wr \mathrm{Sym}_3$ according to \eqref{InducedIrrepsOfWreathProduct} are of the form
  \begin{equation}
    \label{Class1IrrepsOfZdwrSym3}
    \mathbf{1}_{[n]}^{\boxtimes^3}
    \,\otimes\,
    \sigma
    \;\;\;
    \in
    \;
    \mathrm{Rep}_{\mathbb{C}}\big(
      \mathbb{Z}_d \wr \mathrm{Sym}_3
    \big)
  \end{equation}
  for $[n] \in \mathbb{Z}_{d}$
  and $\sigma \in \{\mathbf{1}, \mathbf{1}_{\mathrm{sgn}}, \mathbf{2}\}$.

  \item[\underline{Case 2}:]
  {\bf Partition involves two distinct irreps} $\mathbf{1}_{[n]}$, $\mathbf{1}_{[n']}$.
  In this case $\mathrm{Sym}_{(3)} \simeq 1 \times \mathrm{Sym}_2 \simeq \mathrm{Sym}_2$, so that the corresponding irreps of $G \wr \mathrm{Sym}_3$ according to \eqref{InducedIrrepsOfWreathProduct} are of the form
  \begin{equation}
    \label{Class2IrrepsOfZdwrSym3}
    \mathbb{C}\big[
      \mathbb{Z}_d
      \wr \mathrm{Sym}_3
    \big]
    \otimes_{
      \mathbb{C}\big[
        \mathbb{Z}_d
        \times
        (
          \mathbb{Z}_d 
            \wr 
          \mathrm{Sym}_2
        )
      \big]
    }
    \left(
      \Big(
        \mathbf{1}_{[n]}
        \boxtimes
        \mathbf{1}^{\boxtimes^2}_{[n']}
      \Big)
      \otimes
      \sigma
    \right)
  \end{equation}  
  for $[n] \neq [n'] \in \mathbb{Z}_d$ and $\sigma \in \{\mathbf{1}, \mathbf{1}_{\mathrm{sgn}}\}$

  \item[\underline{Case 3}:]
  {\bf Partition involves three distinct irreps} $\mathbf{1}_{[n_1]}$, $\mathbf{1}_{[n_2]}$, $\mathbf{1}_{[n_3]}$.
  In this case $\mathrm{Sym}_{(3)} \simeq 1 \times 1 \times 1 \simeq 1$, so that the corresponding irreps of $G \wr \mathrm{Sym}_3$ according to \eqref{InducedIrrepsOfWreathProduct} are of the form
  \begin{equation}
    \label{Class3IrrepsOfZdwrSym3}
    \mathbb{C}\big[
      \mathbb{Z}_d \wr \mathrm{Sym}_3
    \big]
    \otimes_{
      \mathbb{C}[\mathbb{Z}_d^3]
    }
    \big(
      \mathbf{1}_{[n_1]}
      \boxtimes
      \mathbf{1}_{[n_2]}
      \boxtimes
      \mathbf{1}_{[n_3]}
    \big)
  \end{equation}  
  for $[n_i] \in \mathbb{Z}_d$ pairwise distinct.
  \end{itemize}
\end{example}

\medskip

\subsection{Quadratic Gauss sums}
\label{QuadraticGaussSums}

Here we briefly compile some facts about Gauss sums used in \S\ref{2CohomotopicalFluxThroughTorus}, see \cite{BerndtEvans81} for more pointers to the literature (see also \cite{Doyle16} but beware of typos in (1.1) there).

\medskip

First, it may be worth recalling the simple cousin of the Gauss sums:
\begin{proposition}[\bf Discrete Fourier transform of Kronecker delta]
  For $\lattice \in \mathbb{N}_{> 0}$ and $\p \in \mathbb{Z}$ we have
  \begin{equation}
    \label{DiscreteFourierTransformOfKroneckerDelta}
    \sum_{n = 0}^{\lattice-1}
    \,
    e^{
      2 \pi \mathrm{i} 
      \tfrac
        { \p }
        { \lattice }
      n
    }
    \;\;
    =
    \;\;
    \begin{cases}
      \lattice 
        &
        \textrm{ if } 
        \p = 0 \pmod \lattice
      \\
      0 
        &
        \textrm{ if }
        \p \neq 0 \pmod \lattice
      \,.
    \end{cases}
  \end{equation}
\end{proposition}
\begin{proof}
  The statement for $\p = 0 \pmod \lattice$ is immediate. For $\p \neq 0 \pmod \lattice$
  observe that
  \[
    \Big(
      1 -
      e^{
        2 \pi \mathrm{i} 
        \tfrac
          { \p }
          { \lattice }          
      }
    \Big)
    \sum_{n = 0}^{\lattice-1}
    \,
    e^{
      2 \pi \mathrm{i} 
      \tfrac
        { \p }
        { \lattice }
      n
    }
    \;=\;
    1 
    -
    e^{
       2 \pi \mathrm{i}
      \p
    }
    \;\;
    =
    \;\;
    0
    \,.
  \]
\end{proof}

Now:
\begin{proposition}[{\bf Classical quadratic Gauss sum evaluation}, {cf. \cite[p. 87]{Lang70}\cite{RamMurtyPathak17}}]
For $\lattice \in \mathbb{N}_{> 0}$ we have
\begin{equation}
  \label{EvaluationOfClassicalQuadraticGaussSum}
  \sum_{n=0}^{\lattice-1}
  \,
  e^{
    2 \pi \mathrm{i}
    \tfrac
      { 1 }
      { \lattice }
    n^2
  }
  \;=\;
  \begin{cases}
    (1+ \mathrm{i})
    \sqrt{\lattice}
    &
    \textrm{ if }
    \lattice = 0 \pmod 4
    \\
    \sqrt{\lattice} 
    &
    \textrm{ if }
    \lattice = 1 \pmod 4
    \\
    0 
    &
    \textrm{ if }
    \lattice = 2 \pmod 4
    \\
    \mathrm{i}\, \sqrt{\lattice}
    &
    \textrm{ if }
    \lattice = 3 \pmod 4
    \,.
  \end{cases}
\end{equation}
\end{proposition}
\begin{proposition}[{\bf Quadratic Gauss sum with multiple exponents}, {cf. \cite[``QS4'' p. 86 ]{Lang70}\cite[Lem. 5.1]{RamMurtyPathak17}}]
For odd $\lattice \in 2\mathbb{N} + 1$ we have more generally,%
\footnote{The sign in \cref{QuadraticGaussSumWithMultipleExponents} is called the \emph{Jacobi symbol}; we leave it unevaluated here.}
for $\p \in \mathbb{Z}$:
\begin{enumerate}
\item 
if $\gcd(\p,\lattice) = 1$:%
\begin{equation}
  \label
  {QuadraticGaussSumWithMultipleExponents}
  \sum_{n=0}^{\lattice-1}
  e^{
    2 \pi \mathrm{i}
    \tfrac
      {\p}
      {\lattice}
    n^2
  }
  \;=\;
  \pm
  \begin{cases}
    \sqrt{\lattice} 
    & 
    \textrm{ if } 
    \lattice = 1  \pmod 4
    \\
    \mathrm{i}\sqrt{\lattice}
    & 
    \textrm{ if } 
    \lattice = 3 \pmod 4
    \mathrlap{\,,}
  \end{cases}
\end{equation}
\item for arbitrary $\p$ with $d := \gcd(\p,\lattice)$:
\begin{equation}
  \sum_{n=0}^{\lattice-1}
  e^{
    2 \pi \mathrm{i}
    \tfrac
      {\p}
      {\lattice}
    n^2
  }
  \;=\;
  d
  \sum
    _{n=0}
    ^{\lattice/d-1}
  e^{
    2 \pi \mathrm{i}
    \tfrac
      {\p/d}
      {\lattice/d}
    n^2
  }
  \mathrlap{\,.}
\end{equation}
\end{enumerate}
\end{proposition}

In \S\ref{2CohomotopicalFluxThroughTorus} we are crucially concerned with the variant of the classical quadratic Gauss sum that has \emph{half} the usual exponents. In its plain form it is elementary to reduce this to the ordinary quadratic Gauss sum:
\begin{proposition}[{\bf Quadratic Gauss sum with halved exponents}]
For $\lattice \in 2\mathbb{N}_{>0}$ we have
\begin{equation}
  \label{QuadraticGaussSumWithHalfExponent}
  \sum_{n = 0}^{\lattice-1}
  \,
  e^{
    \tfrac{\pi \mathrm{i}}{\lattice}
    n^2
  }
  \;=\;
  e^{\pi \mathrm{i}/4}
  \,
  \sqrt{\lattice}
  \,.
\end{equation}
\end{proposition}
\begin{proof}
Setting $r := \lattice/2 \,\in\, \mathbb{N}$, we may compute as follows:
$$
  \def\arraystretch{1.8}
  \begin{array}{ccll}
    \sum_{n=0}^{\lattice-1}
    \,
    e^{
        { \pi \mathrm{i} }
      \tfrac
        { 1 }
        { \lattice }
      n^2
    }
    &=&
    \sum_{n=0}^{2r-1}
    \,
    e^{
      \tfrac
        { \pi \mathrm{i} }
        { 2r }
      n^2
    }
    &
    \proofstep{
      by definition of $r$
    }
    \\
    &=&
    \tfrac{1}{2}
    \Big(
    \sum_{n=0}^{2r-1}
    +
    \sum_{n=2r}^{4r-1}
    \Big)
    \,
    e^{
      \tfrac
        { \pi \mathrm{i} }
        { 2r }
      n^2
    }
    &
    \proofstep{
      since summands are $2r$-periodic,
      cf. ftn. \ref{NeedForEvenk}
    }
    \\
    &=&
    \tfrac{1}{2}
    \sum_{n=0}^{4r-1}
    \,
    e^{
      \tfrac
        { 2\pi \mathrm{i} }
        { 4r }
      n^2
    }
    \\
    &=&
    \tfrac{1}{2}
    (1 + \mathrm{i}) \sqrt{4r}
    &
    \proofstep{
      by \eqref{EvaluationOfClassicalQuadraticGaussSum}
    }
    \\
    &=&
    e^{ \pi \mathrm{i}/4 }\, 
    \sqrt{2r}
    \\
    &=&
    e^{ \pi \mathrm{i}/4 }
    \, 
    \sqrt{\lattice}
    &
    \proofstep{
      by definition of $r$.
    }
  \end{array}
$$

\vspace{-4mm}
\end{proof}

More generally, there is the following reciprocity relation for the parameters of the quadratic Gauss sum with halved exponents, which relates it to the ordinary quadratic Gauss sum:
\begin{proposition}[{\bf Landsberg-Schaar identity} {\cite{Schaar1850}, cf. \cite{ArmitageRogers00}\cite{Ustinov22}\cite{Harcos}}]
  For $\lattice \in 2\mathbb{N}_{>0}$ and $\p \in \mathbb{N}_{> 0}$ we have
\begin{equation}
  \label{LandsbergSchaarIdentity}
  \sum_{n = 0}^{\lattice-1}
  \,
  e^{
    \pi \mathrm{i}
    \tfrac{\p}{ \lattice }
    n^2
  }
  \;\;
  =
  \;\;
  \tfrac
    {
      e^{\pi \mathrm{i} / 4}
    }
   {
     \mathclap{\phantom{\vert^{\vert^{\vert}}}}
     \sqrt{\p / \lattice}
   }   
  \sum_{n = 0}^{\p -1}
  \,
  e^{
    -
    \pi \mathrm{i}
    \tfrac{\lattice}{\p}
    n^2
  }
  \,.
\end{equation}
\end{proposition}

In summary, this implies the evaluation which we use in the main text:
\begin{proposition}[\bf Quadratic Gauss sum with multiple halved exponents]
  For $\lattice \in 2 \mathbb{N}_{>0}$ and 
  $\p \in 2\mathbb{N} + 1$
  we have:  
  \begin{enumerate}
  \item
  if $\gcd(\p,\lattice) = 1$:
  \begin{equation}
    \label{QuadraticGaussSumWithMultipleHavedExponents}
    \sum_{n=0}^{\lattice-1}
    \,
    e^{
      \pi \mathrm{i}
      \tfrac{\p}{ \lattice }
      n^2
    }
  \underset{
    \mathclap{
      \scalebox{.7}{
        \eqref{LandsbergSchaarIdentity}
      }
    }
  }{
    \;\;=\;\;
  }
  \tfrac
    {
      e^{\pi \mathrm{i} / 4}
    }
   {
     \mathclap{\phantom{\vert^{\vert^{\vert}}}}
     \sqrt{\p / \lattice}
   }   
  \sum_{n = 0}^{\p - 1}
  \,
  e^{
    -
    2\pi \mathrm{i}
    \tfrac{\lattice/2}{\p}
    n^2
  }
  \underset{
    \mathclap{
      \scalebox{.7}{
        \eqref{QuadraticGaussSumWithMultipleExponents}
      }
    }
  }{
    \;\;=\;\;
  }
  \pm
  \begin{cases}
    e^{+\pi \mathrm{i} / 4}
    \sqrt{\lattice}
    &
    \textrm{ if }
    \p = 1 \pmod 4
    \\
    e^{-\pi \mathrm{i} / 4}
    \sqrt{\lattice}
    &
    \textrm{ if }
    \p = 3 \pmod 4 
    \,.
  \end{cases}
\end{equation}
\smallskip
\item for arbitrary odd $\p$ with 
$d := \gcd(\p,\lattice)$:
\begin{equation}
  \sum_{n=0}^{\lattice-1}
  \,
  e^{
    \pi \mathrm{i}
    \tfrac{\p}{ \lattice }
    n^2
  }
  \;=\;
  d
    \sum_{n=0}^{\lattice/d - 1}
    \,
    e^{
      \pi \mathrm{i}
      \tfrac{\p/d}{ \lattice/d }
      n^2
    }
  \mathrlap{\,,}
\end{equation}
\end{enumerate}
hence
\begin{equation}
  \left\vert{
    \sum
      ^{\lattice-1}
      _{n=0}
    \,
    e^{
      \pi \mathrm{i}
      \tfrac
        {\p}
        {\lattice}
      n^2
    }
  }\right\vert
  \;=\;
  \sqrt{
    \lattice
    \gcd(\p,\lattice) 
  }
  \mathrlap{\,.}
\end{equation}
\end{proposition}

These relations are used in \S\ref{2CohomotopicalFluxThroughTorus}
to classify torus monodromy representations with irreducible Heisenberg
restriction.


\noindent
{\bf Declarations.}
\begin{description}
\item[] 
The authors declare that no data were generated or analyzed in this theoretical research.
\item[]
The authors declare that they have no conflict of interest.
\end{description}


\end{document}